\documentclass{dcthesis}

\committee[]{}{}{}{}
\school{National University of Singapore}{Singapore}
\degree{Doctor of Philosophy}
\field{Financial Mathematics}
\usepackage{multirow}
\usepackage{microtype}
\usepackage{graphicx}
\usepackage{subfigure}
\usepackage{booktabs} 
\usepackage{hyperref}
\usepackage{amsmath,amssymb,amsthm}
\usepackage{junstyle}
\usepackage{natbib}
\usepackage{bibentry}
\usepackage{algorithm}
\usepackage{algorithmic}
\usepackage{enumerate}
\usepackage{dsfont}
\usepackage{xcolor}
\usepackage{mathrsfs}
\usepackage{amsmath}
\numberwithin{equation}{chapter}

\newtheorem{theorem}{Theorem}[section]
\newtheorem{lemma}[theorem]{Lemma}
\newtheorem{proposition}[theorem]{Proposition}

\def\df{\displaystyle\frac}
\newcommand{\be}{\begin{equation}}
\newcommand{\ee}{\end{equation}}
\newcommand{\ba}{\begin{array}}
\newcommand{\ea}{\end{array}}
\newcommand{\bea}{\begin{eqnarray}}
\newcommand{\eea}{\end{eqnarray}}
\newcommand{\beas}{\begin{eqnarray*}}
\newcommand{\eeas}{\end{eqnarray*}}

\title{TWO STOCHASTIC CONTROL PROBLEMS
IN CAPITAL STRUCTURE AND PORTFOLIO CHOICE}
\author{Shan Huang}
\date{}
\field{Financial Mathematics}
\degree{Doctor of Philosophy}
\committee{Min Dai, Ph.D.}{Chao Zhou, Ph.D.}{Weiqing Ren, Ph.D.}{Yue-Kuen Kwok, Ph.D.}

\begin{document}
\frontmatter
\newgeometry{left=1.5in,top=1in,bottom=1in,right=1in}
\maketitle

\chapter*{Declaration}
I, Shan Huang, hereby declare that the thesis is my original work and it has been written by me in its entirety. I have duly acknowledged all the sources of information which have been used in the thesis.

\chapter*{Acknowledgments}
Thanks for my supervisor Professor Min Dai and my co-supervisor Associate Professor Jussi Keppo from NUS Business School. They have shown great light on the ideas and schemes for my research work and given me a lot of suggestions, instructions, and guidance.  Their inspirations  will become the cornerstones on my further research and will benefit me throughout my whole life. I learned that the appearance of difficulty is a good thing, but at least it shows the significance and meaning behind the research work.
 
I also would like to thank my parents and my dear husband, Han Jun. In the past four years, my family members encouraged me a lot. Although they can not help me in the academic studies, they stand by my side and build a warm harbour. Especially, I really appreciate the support and soothe from my husband. It is him who changes my outlook on life and makes me become calm, patient, and brave in front of failure and pressure. 
 
 My sincere thanks also go to my colleagues for helpful discussion, including Mr. Jiang Wei,  Dr. Lei Yaoting, Dr. Yang Chen, Dr. Xu Jing, and Dr. Chen Yingshan. Especially, I would thank the Post-Doc Seyoung Park due to his suggestions and contributions on our joint projects.

\tableofcontents

\listoffigures

\listoftables

\backmatter


\chapter{Introduction}
\label{chp: Intro}
\numberwithin{equation}{chapter}
This thesis mainly focuses on two stochastic control problems in capital structure and individual's life-cycle portfolio choice. 
In Chapter \ref{part 1} we derive a stochastic control model to optimize banks' dividend and recapitalization policies and calibrate that to a sample of U.S. banks in the situation where we model banks' true accounting asset values as partially observed variables due to the opaqueness in banks' assets.  
In Chapter \ref{part 2}, we present an optimal portfolio selection model with voluntary retirement option in an economic situation, 
where an investor faces borrowing and short sale constrains, as well as the cointegration between the stock and labor markets.  The detailed introduction of each chapter is presented below.

\section{Stochastic Control Problem in Capital Structure}
\subsection{Literature Review on Relevant Works }

 In the real market, it is  difficult for outsiders to judge and observe the risks of banks. Such opaqueness intuitively affects bank capital and  policy decision-making.  In this subsection, we will review studies  on bank opaqueness, bank capital and regulatory decisions in the following.

The risks of banks are difficult to judge and observe for outsiders.
Therefore, banks' assets are not fully observed for the bank shareholders and regulators.
This is due to several reasons. First, according to \cite{FH04} and  \cite{Mg01},  the banking business is complex, and banks' nonmarketable loans are difficult to assess.
Second, banks have an incentive to minimize the fluctuation of the debt they have produced since the debt is a money-like security (see e.g. \cite{D13}).
To produce the money-like debt,  banks choose to minimize information leakage because the debt needs to be information-insensitive to serve as an efficient transaction
medium. This means that policies designed to enhance bank transparency reduce the ability of banks to produce debt.
Third,  \cite{LW03} find that insiders protect their private control benefits by using earnings management to conceal firm performance from outsiders.\footnote{ \cite{HC95} finds that firms that reported small profits greatly outnumber those that reported small losses. Consistent with that study, \cite{BD97},  \cite{BKP02}, and  \cite{SC05} find that firms avoid losses and missing last years' earnings by intentionally manipulating earnings.  \cite{E13} shows that banks engage in earnings management through loan syndications.}
By  \cite{Mg02}, banks' high leverage compounds the uncertainty over their earnings, and as a result, they are inherently more opaque than other types of firms. However, according to \cite{JRL14},  \cite{HH04}, and \cite{FG08}, strong institutions and competition between banks dampen the incentives to hide actual performance through earnings management practices.

\cite{Mg02} shows that banks are relatively more opaque than industrial firms. The uncertainty of  banks come from certain assets, loans and the trading assets,  of which the risks  are hard to observe or easy to change. 
Conventional wisdom says that bank loans (assets) are informationally opaque and this has been justified on a variety of theoretical works (see, e.g.  \cite{CK80},  \cite{BJ88},  \cite{D89}, \cite{D91} and \cite{KC98}). Bank loans (assets) quite often lack transparency and liquidity,  and this makes their risks  difficult to quantify and manage (\cite{G96}). Since the latest financial crisis, many observers have linked the financial panic that occurred during the crisis to bank opacity and there exists non-symmetric information between the outsiders and insiders of banks. How to model bank assets and equity to make optimal regulatory policy under non-symmetric information has become an important field in recent years, especially after the latest financial crisis.

As explained above, banks have the incentive to hide part of uncertainty in assets and earnings. This phenomena, i.e., opaqueness has also been studied in other empirical works.  \cite{A99} find that the loan loss provisions (LLP) has a negative correlation with the earnings process.  \cite{PMS10} investigate the ``stress test'', the extraordinary examination conducted by federal bank supervisors in 2009, and debate over bank opaqueness.  Moreover, opaqueness may foster price contagion that exacerbates the speculative cycles of bubbles and crashes that create financial instability (\cite{JLY12}).
 \cite{BL13} focus on research examining the relation between bank financial reports and risk assessments of equity and debt outside, the relation between bank financial reporting discretion, regulatory capitals and earnings management, and banks' economic decisions under different regimes by empirical analyses.
 \cite{Fl13} examine bank the trading characteristics of equity during ``normal'' periods and two ``crisis'' periods. They find only limited (mixed) evidence that banks are unusually opaque during normal periods. In addition, they point out that the balance sheet composition of a bank affects its equity opacity  significantly. 

However, most of the existing works focus mainly on empirical analysis and choose  bank assets composition as the primary measure of opaqueness. Few structure models have been set up to quantify the level of bank opaqueness through static noise level in accounting report.  This  static noise in accounting report will affect bank capital and regulatory decisions.

According to Basel Accord III,  Tier 1 Capital must be at least 6.0\% of risk-weighted assets at all times and the total Capital (Tier 1 Capital plus Tier 2 Capital) must be at least 8.0\% of risk-weighted assets at all times.  From this perspective, the violation of minimum capital buffer requirement is costly for banks, and  it incurs other costs with portfolio adjustment and recapitalization as well. Hedging capital buffer against minimum capital violation is well founded subject to these conditions.

An early continuous time model of a capital-constrained firm is presented in \cite{MR96}.  \cite{HT99} have extended the basic model into an insurance company setting by assuming that risk reduction at a proportional cost is available, which is interpreted as a cheap reinsurance.  \cite{MW01} have extended the model to allow for a recapitalization option.  \cite{KML10} use a similar modeling framework to study the unintended consequences of banking regulations in terms of rising default probability.
 \cite{CCT06} have analyzed the optimal timing and amount of dividends under fixed dividend cost.
 \cite{SC13} extend that model to consider a proportional dividend tax. \cite{LoZ08} study the dividend and equity issuance policy under an assumption that the company's reserves follow a diffusion process.  \cite{BM10} and \cite{SH18} consider dividend optimization under implementation delays. \cite{DS16} solve for the optimal dividend and default strategy under Chapter 11 of the U.S. bankruptcy code.
Consistent with our model,  \cite{MT11} find that bank value is positively cross-sectionally related to bank capital.
 \cite{E04} uses a variant of the classical inventory or cash management models to study the cyclicality of bank capital.
 \cite{BCW11},   \cite{DMV04}, \cite{C14}, \cite{IMR14} solve for a firm's optimal dynamic cash balance policy in terms of a trade-off between the gains from investment and the opportunity cost of spending cash.

Our model builds on the basic continuous time model in  \cite{PK06} and the partially observed models in Chapter 4 in  \cite{B04} and \cite{BKS09}.
Our modelling innovation is to allow true accounting asset values to be partially observed, and in this sense, banks' assets are opaque.
We show that the opaqueness substantially changes bank shareholders and regulators' decisions.

\subsection{Conributions of the Thesis}
Since the financial crisis of 2007--2008, many observers have linked the financial panic that occurred during the crisis to bank opacity (see e.g. \cite{G90}), \cite{L08}, \cite{D09},  \cite{A09}, \cite{Du09}, and \cite{A16}).\footnote{See also e.g. DealBook, New York Times, March 11, 2010, ``Court-Appointed Lehman Examiner Unveils Report'' that explains how Lehman Brothers used an accounting rule called Repo 105 to temporarily shuffle about 50 billion USD off the firm's balance sheet for the two fiscal quarters before it collapsed.} According to this narrative, bank assets are not fully observed for bank shareholders and regulators.\footnote{\cite{BW12} and \cite{Iy13} show that financial reporting opacity can negatively affect outsiders' ability to effectively monitor banks.}
Therefore, we model the true accounting asset values as partially observed variables for the shareholders and the regulators.
This means that for them quarterly accounting reports are noisy signals on the true accounting values.
As discussed above, there is a risk that these reported accounting values do not correspond to the current business situation, because some banks' assets are difficult to assess and bankers have an incentive to smooth earnings. Bankers  are insiders and, thus, they might know the true accounting asset values better than the shareholders and the regulators.
We abstract away the details of why (unintentional or intentional) there is noise in the accounting reports using the partially observed model, 
and in the model calibration, we estimate the parameters corresponding to the accounting noise.
This way we avoid complex signaling games between the bankers, shareholders, and the regulators.
However, by the model calibration, bankers clearly have an incentive to raise accounting noise that smooths  asset values over time if they have equity-based compensation such as stock options, because the smoothing raises the market value of equity substantially.

Given the noisy reported accounting values, the shareholders and regulators obtain the conditional probability distribution of the true accounting values. Then the bank shareholders solve for the optimal dividend and recapitalization policy of the bank, and the bank regulators decide to close the bank if the expected equity conditional on the accounting signals falls too low.
We focus on the bank's dividends and recapitalization option since they are used more than, for instance, asset sales (see e.g. \cite{BH14} and \cite{BLS16}).
According to our data and the model, the threshold when regulators close the bank is low, because the regulators weight more the risk of liquidating a solvent bank than the risk of not liquidating an insolvent bank; when they close a bank, they need to be certain that the bank is insolvent.\footnote{\label{footnote4}This is consistent with  \cite{BW10} and  \cite{HL12} who suggest that regulators may actually exploit financial reporting choices in order to not intervene in troubled banks, often referred to as forbearance. Forbearance could prevent panic runs on healthy banks (see  \cite{MW13}) and, thus, be optimal from a macroprudential perspective during a financial crisis. By \cite{BD11}, regulators are more likely to prefer to forbear on weak banks, risky banks, and when the banking sector is weaker.  \cite{BT93},  \cite{MF00}, and \cite{BD05} argue that forbearance can be motivated by self-interested reputational concerns or political pressures, ultimately leading to a weaker banking sector (\cite{R99}) and costlier failed bank resolutions (\cite{WC15}).}

The capital structure decision of banks is in its very essence a risk management decision. A bank practitioner views bank equity capital not primarily as a form of financing, but as a buffer against asset risks which needs to be managed so that the bank can satisfy its regulatory minimum capital requirement even under relatively adverse future scenarios. It is implicit in this view that the violation of the minimum capital requirement is costly for the bank, and that the bank faces costs or constraints associated with portfolio adjustment and recapitalization. Subject to these conditions the role of equity capital as a hedging mechanism against minimum capital violation is well founded.
Consistent with this, in our model the bank shareholders maximize the bank value by adjusting its equity capital level through dividends and equity issuances under an illiquid bank portfolio, imperfections in capital raising transactions, and loss of franchise value associated with the violation of the minimum capital requirement.

We also solve the corresponding model where the true accounting values are perfectly observed, which is called the fully observed model.  The difference between the partially observed and the fully observed models is that in the partially observed model, part of the uncertainties in the asset values is static because of the uncertainty in the reported accounting values, while in the fully observed model, all the uncertainties in the asset values are dynamic, driven by the future shocks in the asset values.
Since the uncertainties are different and since banks' capital structure decision is a risk management decision, the optimal dividend and recapitalization policies of the two models are different.
More specifically, under the noisy accounting values, bank investors hedge the opaqueness by paying less dividends and issuing more equity.\footnote{Consistent with this,  \cite{KS14},  \cite{FR06}, \cite{H06}, \cite{KT07}, \cite{LR05}, and \cite{L05} find empirically that firms gradually adjust their capital structure in response to various shocks. }


We calibrate the two structural models to a sample of U.S. banks during 1993--2015.
We find several results.
First, consistent with \cite{Fl13}, the banks are more opaque during the financial crisis of 2007--2009 than outside that.
This is also consistent with the value smoothing of debt in  \cite{D13}).
Further, the asset smoothing due to the accounting noise is substantial; on average, the banks' asset smoothing hides about one-third of the true asset volatility.
The asset smoothing raises the bank value; on average, the noise in the reported accounting asset values raises the banks' market equity value by $7.8\%$ because this way, the banks can hide their solvency risk from banking regulators.
Second, we find that banks with a high level of loan loss provisions, nonperforming assets, and real estate loans, and with a low volatility of total assets returns have a higher level of accounting noise. This is consistent with  \cite{A99},  \cite{LR95}, and \cite{Fl13}, who find that loan loss provisions have a negative correlation with the earnings process and that the balance sheet composition of a bank affects its opacity.
Third, our partially observed model explains substantially better the banks' actions in out-of-sample than the fully observed benchmark model.
More specifically, cross-sectionally, the partially observed model explains $52\%$, $14\%$, and $85\%$ of the variations in the banks' dividends, recapitalization decisions, and market equity values in out-of-sample, while the corresponding numbers for the fully observed model are $35\%$, $4\%$, and $79\%$, respectively.
Further, the partially observed model's average equity-to-debt ratio is closer to the sample ratio of $13.51\%$.
These results indicate that bank owners and regulators take into account the noisiness in the accounting values when deciding their actions.

Mathematically, we show that the value function is of at most linear growth. In the presence of issuance delay, as it is very difficult to prove the priori regularity of the value function such as measurability and continuity, we give a weak dynamic programming principle instead of classical dynamic programming. Based on the weak dynamic programming principle, we prove that the value function is a unique viscosity solution to the associated Hamilton-Jacobi-Bellman (HJB) equation. Moreover, in  the fully observed model, we derive a semi-explicit solution of the value function and shareholders' optimal policies under some conditions.

\section{Stochastic Control Problem in Portfolio Choice}
\subsection{Literature Review on Relevant Works}

The study of optimal consumption and investment over the life cycle has provided individuals with the economic justification for their own portfolio and saving decisions. The optimal framework for consumption and asset allocation based on realistic calibration ultimately improves social welfare by forming the basis of policy design for pension, insurance,
and retirement. Along with this line, a multitude of life-cycle models with flexible 
labor supply have been presented.\footnote{\cite{CHMM11} argue that having flexibility over working hours or retirement time is determined to be a leading 
factor when studying optimal portfolio choice over the life cycle.} There are at least two types of 
flexibility in labor supply when addressing the interactions among savings, portfolio choice, and 
retirement over the life cycle. On the one hand, individuals can adjust working hours on their 
jobs (\cite{BMS92}). On the other hand, individuals can freely choose their 
voluntary retirement time (\cite{FP07}, hereafter FP; \cite{DL10}, hereafter DL). 

Based on the pioneering work by Merton (\cite{M69}, \cite{M71}), an increasing emphasis on uninsurable risks associated with labor income have emerged. Once 
uninsurable income risks are considered to match stylized facts in relation to portfolio 
choice, new aspects of advancing the life-cycle model are revealed by taking into account infinite-horizon incomplete markets with stochastic labor income (see e.g.  \cite{HL97}, \cite{DFSZ97}, 
 \cite{K98}, and \cite{V01}), precautionary saving motive against background risk (\cite{CDK03}), a market 
incompleteness induced by uninsurable income risks and borrowing constraints in a finite-horizon setting
(\cite{CGM05}), additive and endogenous habit formation preferences (\cite{P07}),
cointegration between labor income and stock dividends (\cite{BDG07}), insufficient insurance
against a large and negative wealth shock (\cite{GLZ10}), nonhomothetic utility over
basic and luxury goods for households (\cite{WY10}), stochastic interest rates and labor income
streams in which the expected income growth is affine in short-term interest rates (\cite{MS10}),
labor income dynamics at business-cycle frequencies (Lynch and Tan \cite{LT11}), recursive utility and
illiquidity induced by transaction costs in a general equilibrium framework (\cite{BUV14}),
U.S. Social Security rules and family status (\cite{HMM16}), and the delegation of portfolio
management (\cite{KMM16}). However, none of the preexisting literature matches empirical stylized 
facts observed in retirement-induced optimal portfolio choice.

\subsection{Contributions of the Thesis}
As far as preexisting literature on optimal retirement  (see e.g. \cite{FP07}, \cite{DL10}, and \cite{BJP16}) is concerned, one knows that it is optimal to invest more in risky assets when considering increased flexibility in retirement. The 
retirement flexibility makes labor income's beta with the market negative when income has low 
market risk exposure. That is, working longer (shorter) becomes more attractive when the market is down (up). Subsequently, investment in the stock market can be effectively utilized
as a hedging instrument against labor income risks. However, this conclusion might not be the case under  realistic considerations in labor income. It is very important to keep in mind that income 
shocks and stock returns are not highly correlated, consistent with the data (\cite{CGM05} and \cite{DW13}).\footnote{The existing empirical evidence demonstrates that the correlation 
between the shocks to labor income and stock returns is low, for example, 0.15 (see e.g. \cite{CCGM01} and \cite{GM05}).} Instead, the  conintegration phenomenon has been observed between the stock and labor markets (see e.g. \cite{BJ97},  \cite{MSV04},  \cite{SV06}, and \cite{BDG07}). As \cite{DL10} claim in the 
conclusion of their paper,
\begin{quote}\textit{
It would be nice to add more state variables to the model. For example, it has long
been known that wages are sticky and it is reasonable that they respond to shocks
in the stock market, but with a delay. ... Unfortunately, models with additional state
variables seem almost impossible to be solved analytically given current tools and
numerical solution is also very difficult.}
\end{quote}

This is precisely the study we would like to undertake here. We investigate the effects of retirement 
flexibility on the optimal portfolio choice in a range of more elaborate settings that move the optimal portfolio 
selection problem closer to the one solved by real-world investors. More specifically, we allow for cointegration between 
the stock and labor markets, which is one of the most important characteristics of labor income. In addition to the concept of cointegration, there is at least one major departure from 
 \cite{FP07} and  \cite{DL10}. We proceed our analysis under the economically plausible constraints with which an investor is prevented from 
borrowing against the net present value of his labor income and shorting  securities at no cost.\footnote{Borrowing 
constraints are consistent with the realistic ramifications present in actual capital markets: many investors 
are constrained from borrowing against human capital, partly because of some market frictions, like informational 
asymmetry, agency conflicts, and limited enforcement. Short sale constraints have also been imposed because there 
are a wide range of legal and institutional restrictions on short selling in the U.S. equity markets. An extensive 
treatment of short sale constraints and their effects on asset prices is fully treated by \cite{BCW06}.} 

We consider a representative investor's utility maximizing framework. The investor exhibits 
standard constant relative risk aversion (CRRA) utility preferences (recursive utility is also provided for robustness) and encounters a constant investment opportunity 
set. The investor receives stochastic labor income while working. In particular, there are two kinds of uninsurable 
risk characteristics in labor income (\cite{WWY16}): (1) diffusive and continuous shock, and (2) discrete 
and jump shock.\footnote{The labor income risks are uninsurable or undiversifiable because of a lack of explicit 
insurance markets for the income risks (\cite{CGM05}). Social securities and private insurance markets 
are not perfectly sufficient to hedge against large and negative wealth shocks (\cite{GLZ10}).} In 
line with  \cite{FP07} and \cite{DL10}, 
we endogenize the labor supply along the extensive margin that the investor either works full-time or retires permanently. That is, the investor is allowed to endogenously (optimally) determine his irreversible retirement time.\footnote{The incentive of voluntary retirement results from more preferences 
for not working (or more leisure preferences) than staying in the workforce, in the sense that the marginal  utility is increased once retired. 
}

Four main distinct implications that we obtain  are as follows.

First, with reasonable parameter values, there exists a target wealth-to-income ratio under
which an investor does not participate in the stock market at all, whereas above which the 
investor increases the proportion of financial wealth invested in the stock market as he 
accumulates wealth. It has been quite economically plausible, and even numerically
possible that the share of wealth invested in equity (or the portfolio share) rises in 
wealth (see e.g. \cite{FP07}, \cite{DL10}, \cite{P07}, \cite{WY10}, and  \cite{CS14}).\footnote{This 
result is obtained in \cite{FP07} with a constant wage in some special cases, in \cite{DL10} with a stochastic
wage that varies significantly along with the stock market, and in  \cite{P07} with 
additive and endogenous habit formation preferences. \cite{WY10} calibrate the 
life-cycle model to the Survey of Consumer Finances and demonstrate that 
the portfolio share rises in wealth by considering nonhomothetic utility over basic and luxury 
goods. Empirically, \cite{CS14} show that the portfolio share is increasing and 
concave in financial wealth. Compared to the papers above, we provide
new economic insights into the optimal portfolio choice over the life cycle and view the increasing portfolio share in wealth from a significantly different angle.} Because of cointegration 
between the stock and labor markets, returns to human capital and stock market returns are highly 
positively correlated (\cite{BDG07}), as a result, an investor with little wealth who 
is away from retirement does not participate in the stock market at all, even when the market risk 
premium is positive.\footnote{This result can be a resolution to the non-participation puzzle and 
regarded as a complement to many studies to resolve the anomalies (see e.g. \cite{V02}, \cite{GM05}, and \cite{GLZ10}).} In this case, human capital acts as implicit
equity holdings. However, as the investor approaches the endogenously-chosen retirement date, the cointegration 
effect becomes weaker and human capital's role returns to implicit bond holdings.\footnote{Labor income 
is treated as a substitute for bond holdings when income shocks has low correlation with stock 
market returns (see e.g. \cite{HL97} and \cite{JK96}). This is witness to the 
fact that an investor makes more aggressive investment when young (see e.g. \cite{CGM05}, \cite{FP07}, and \cite{DL10}).} 
Consequently, the investor finds it optimal to increase the proportion of wealth invested in stocks as 
he accumulates wealth over a certain threshold.

Second, contrary to the existing retirement studies such as  \cite{FP07} and \cite{DL10}, we find that flexibility in determining 
the retirement time allows the investor to invest less in the stock market than without retirement flexibility. With cointegration, labor income can have a positive beta with the stock market,  causing returns to human capital to be changed significantly with the stock market.
Retirement flexibility strengthens such positive correlation between labor income and
 stock market because working shorter (longer) is optimal when stock market is down
(up). Thus, a more conservative investment strategy rather than an aggressive one is followed by the retirement flexibility.

Third, risk aversion speeds up retirement. The wealth threshold presented here for early retirement decreases
with risk aversion when labor income risks are uninsurable. \cite{FP07} and  \cite{DL10} predict that risk aversion tends to raise 
up the wealth threshold for voluntary retirement when labor income risks are fully diversified, consistent with the standard 
real option analysis (\cite{HM07}). Intuitively, when the stock market can perfectly hedge against  income risks, more risk averse investors are apparently trying 
to avoid the risk of losing the option value of working so that they would like to delay retirement. 
However, this result can be reversed when income risks are not spanned by the stock market. A natural intuition is that the risk averse investors require an additional premium for holding the unspanned income risk. Working shorter becomes 
more attractive, commanding a lower premium for the undiversifiable income risk.
Further, higher risk aversion decreases the amount of future human capital in the presence of uninsurable income risks. With a heavy impact of uninsurable 
income risks on life-cycle strategies and human capital amount, the option for retiring earlier becomes progressively more attractive for more risk averse investors.

Finally, the optimal portfolio strategy  presented here predicts that early retirement is economically plausible,
consistent with empirical observations.\footnote{Individuals quite opted for early retirement,
especially during the stock market booms like those observed in the late 1990's. At that time, the U.S. economy 
experienced a rapid increase in the stock market returns (\cite{GS02} and \cite{GST10}). However, there has been no consensus as to the economic rationality for doing so and even there are 
disagreements about such early retirement during up markets (\cite{BJP16}). Issues in Labor 
Statistics published by Bureau of Labor Statistics in 2000 entitled ``Unemployed Job Leavers: A Meaningful Gauge of 
Confidence in the Job Market?" raised a debate whether or not the increased number of workers who opted for 
voluntarily quitting their jobs is pro-cyclical or counter-cyclical.} There are two main effects of cointegration 
on early retirement. Firstly, our retirement model is able to generate the empirically plausible hump-shaped implicit value of human capital over the life cycle and more importantly, cointegration leads to an earlier peak 
point in the implicit value compared to \cite{FP07} and \cite{DL10} without cointegration.\footnote{The implicit value of human capital 
is captured by the notion of marginal rates of substitution between labor income and financial wealth (\cite{K98}), 
so it can be regarded as an investor's subjective marginal value of human capital. The hump shape of value of human 
capital is the long-standing life-cycle result (\cite{CGM05}; \cite{BDG07}).} 
Therefore, wealth at retirement is lower with cointegration than without cointegration. Secondly, in the presence 
of cointegration, the wealth threshold for voluntary retirement becomes smaller when labor income will be decreased 
than when it will be increased. One sure thing about early retirement is that earlier retirement is favored by 
individuals, especially when wages are expected to decline in the long term. Our main results still hold when considering the mandatory retirement age. 

Mathematically, we give a verification theorem to show that the value function is the solution to the associated HJB equations and the corresponding optimal strategies can be represented as feedback functions of the value function and its derivatives.  All the quantitative analysis are done based on Penalty method \cite{DZ08}.

\section{Organization of the Thesis}
This thesis is organized as follows.  The model of capital structure in banking system, is presented in Chapter \ref{part 1}. In Section \ref{bank_modelsetup}, 
Subsection \ref{fully observed model} presents the fully observed model and Subsection \ref{section with recap} presents the partially observed models.
Section \ref{partI_theoretical analysis} shows the theoretical analysis for fully and partially observed model. The fully observed model has a semi-explicit solution, while in general, the partially observed model does not. Section \ref{comparative statics} shows numerical comparative statics of the partially observed model. Models are calibrated and then tested in Section \ref{calibration}. 

The optimal portfolio selection for early retirement with cointegration between the stock and labor markets, is presented in Chapter \ref{part 2}. In Section \ref{coint_model}, we describe our main
model. In Section \ref{coint_appendix}, we give the theoretical analysis of the optimal strategy and value function.
In Section \ref{coint_numerical}, we carry out quantitative analysis with reasonable parameter as well as the analysis of robustness. 

Section \ref{section conclusion} summarizes this thesis and concludes. Proofs of some propositions and theorems in Chapter \ref{part 1} are given in Appendix. We put these proof in Appendix not because they are unimportant, but to make the main text concise and fluent.

\chapter{The Stochastic Control Problem in Capital Structure}\label{part 1}

\numberwithin{equation}{chapter}

\section{Model}\label{bank_modelsetup}
We consider a single commercial bank that is liquidated by regulators if its equity capital falls too low. First, in Subsection~\ref{fully observed model} we consider our benchmark model, where the shareholders and regulators can fully observe the accounting values of total assets and we call this case as the fully observed model.  After that in Subsection~\ref{section with recap} we discuss our main model, where the accounting asset values are noisy to shareholders and regulators. In this sense the bank is opaque and shareholders and regulators can only partially observe the accounting values of total assets. We call this case as the partially observed model.

\subsection{Fully Observed Model} \label{fully observed model}
In this subsection we introduce our fully observed model, where the shareholders and regulators can fully observe the accounting values of total assets. This situation, where the accounting reports are accurate (i.e. without noise), is a special case of the partially observed model in Subsection~\ref{section with recap}.
The bank's balance sheet is illustrated in Figure~\ref{figure bank structure}.

\begin{figure}[ht]
\centering
  \includegraphics[width=0.5\textwidth]{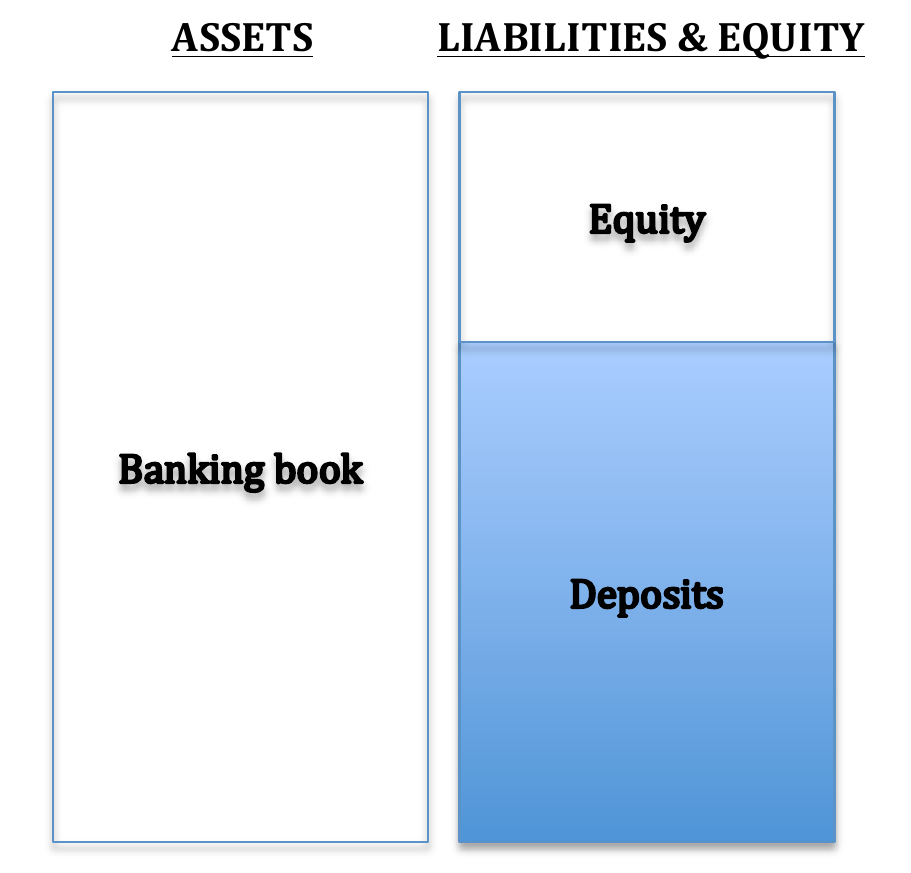}
 \caption[Bank balance sheet]{\textbf{Bank balance sheet.}
 The illiquid banking book equals total assets. Deposits is given by (\ref{debt}), and equity is the total assets minus the deposits (see (\ref{totalassets})).}
 \label{figure bank structure}
\end{figure}
Thus, in our stylized model the total assets equal book equity plus deposits,\footnote{Banks' deposits account for more than 60\% of their liabilities across major economies and regions. For more information on banks' funding structure, see Chapter 3 ``Changes in Banking Funding Patterns and Financial Stability Risks'' in ``Global Financial Stability Report: Transition Challenges to Stability,'' available at \url{https://www.imf.org/External/Pubs/FT/GFSR/2013/02/pdf/c3.pdf}.
The average deposits-to-liabilities ratio among our sample banks in Chapter~\ref{calibration} is 81.82\%.
For instance, \cite{SW16} consider a model with both deposits and subordinated debt.} that is,
\begin{eqnarray}\label{totalassets}
E_t+D_t,
\end{eqnarray}
where  $E_t$ is the book value of equity (equity capital) and $D_t$ is the bank debt which equals deposits.
We assume that the banks' customer base is in a steady state in a sense that its deposits grow steadily at a constant positive rate $\mu$ as follows\footnote{We ignore debt volatility and
debt financing, since they would dramatically increase the complexity of the financial modelling (but with
probably limited impact on the real-side of the model).}
\begin{eqnarray} \label{debt}
dD_t = \mu D_tdt. 
\end{eqnarray}
Total assets equal equity plus bank debt (Figure~\ref{figure bank structure}). 
The total assets are the bank's illiquid loan portfolio, and they are stochastic.
The shareholders control bank capital through dividend payments and equity issuances. Dividends can be paid in continuous time without any frictions, while with the capital issuances, there is a delay and a fixed cost. We focus on the bank's dividends and capital issuance option since they are used more than, for instance, asset sales (see e.g. \cite{BH14} and \cite{BLS16}).

More specifically, outside the dividend and recapitalizations, total assets  grow according to a geometric Brownian motion process.\footnote{For instance, \cite{M74}, \cite{BS08}, \cite{Fl13} model a firm's assets as a geometric Brownian motion process. We include dividends and recapitalization to this process because, by balance sheet, total assets equal equity plus debt (equation $(\ref{totalassets})$).
} Let us define $Y_t$ as the accounting total assets without dividends and recapitalizations.  So it follows
\begin{align}\label{totalassetdyn}
dY_t &= \alpha Y_{t}dt+\sigma Y_{t}dW_t 
\end{align}
where $W_t$ is a standard Wiener process,\footnote{\label{footnote_Ft}
This Wiener process is on a probability space $(\bf{\Omega}, \mathcal{F}, \mathbb{P})$ along
with the standard filtration $\left\{ {\mathcal{F}_t :t \ge 0}
\right\}$. Here $\bf{\Omega}$ is a set, $\mathcal{F}$ is a  $\sigma$-algebra, $\mathbb{P}$ is
a probability measure on $\mathcal{F}$, and $\left( \mathcal{F}_t
\right) _ {0\le t}$ is an increasing family of $\sigma$-algebras: $\mathcal{F}_t:=\sigma\{W_s: s\leq t\}$.} $\alpha$ is the expected return on the assets net of interest rate and other costs, and $\sigma$  is the asset volatility.

 The dynamic process of equity capital without dividends and recapitalzations can be derived from $(\ref{totalassets})$--$(\ref{totalassetdyn})$.  However, it is well known that dividends decrease and recapitalizations increase book equity (and by $(\ref{totalassets})$,  this way they affect the total assets as well). They by including the dividends and recapitalizations, we write the dynamic process for equity capital as\footnote{Equity capital is the accumulated earnings (net income or net loss) and could be negative in some extreme cases. Thus, we do not directly model bank's equity as a geometric Brownian Motion. By the balance sheet, the volatility of equity comes from the volatility of total assets. }
\begin{eqnarray}\label{equitydyn}
E_t =E_0+\int_0^t[\alpha E_{u}+(\alpha-\mu)D_u]dt + \int_0^t (E_{u}+D_u)\sigma dW_u - L_t+\sum_i s_i \mathbf{1}_{\{t_i+\Delta\leq t\}}, \nonumber \\
\end{eqnarray}
where the cumulative dividend process $L_t$ is a nondecreasing right-continuous process adapted to ${\cal{F}}_{t}$ (see footnote \ref{footnote_Ft}) with $L_{0-}^\pi=0$, the sum term $\sum_i s_i \mathbf{1}_{\{t_i+\Delta\leq t\}}$ is the cumulative recapitalization process, $\mathbf{1}_{\{\cdot\}}$ is an indicator function, each $t_i$ is a stopping time of equity issuances adapted to ${\cal{F}}_{t}$, each $s_i$ is the amount of equity issued at $t_i$ and is measurable with respect to ${\cal{F}}_{(t_i^{\pi}+\Delta)-}$, 
and $\Delta$ is the length of the recapitalization process.  The measurability of $s_i^{\pi}$ with respect to ${\cal{F}}_{(t_i^{\pi}+\Delta)-}$ means that owners may decide on the exact amount of capital issuance at time $t_i^{\pi}+\Delta$ based on  information ${\cal{F}}_{(t_i^{\pi}+\Delta)-}$. They do not need to precommit to any quantity of capital at time $t_i$ when they start the recapitalization process. Thus, the new issuance of equity feeds to the shareholders' equity, while dividend payment represents a leakage from it. Further, as in \cite{PK06}, when a new equity is ordered at time $t_i$, there is a delay of $\Delta$ that corresponds to the capital issuance process and uncertainty. We define shareholders' admissible capital control policy by $\pi :=\{L^{\pi}_t,(s_i^{\pi},t_i^{\pi})\}$, 
where $L^{\pi}_t$, $t_i^{\pi}$ and  $s_i^{\pi}$  are defined above.


Regulators decide to liquidate a bank if the equity ratio falls too low. 
The liquidation time under policy $\pi$ is defined as
\begin{eqnarray*}
\tau^{\pi}:=\inf\{t\geq 0: E^{\pi}_t/ D_t < \kappa\},
\end{eqnarray*}
where $\kappa$ is a positive constant given by the banking regulators. In our numerical analysis, the minimum equity-to-debt ratio $\kappa$ corresponds to the Basel minimum capital requirement (see footnote \ref{footnote_kappa} in Section \ref{calibration}). That is, when equity to debt ratio falls too low, the regulators liquidate the bank.

We denote the set of admissible control strategies by $\Pi$  that satisfies:
\begin{eqnarray*}
& &  E_t^{\pi}  \geq \kappa D_t, \  \ \forall t\geq 0,  \\
& &  t_{i+1}^{\pi}-t_i^{\pi}\geq \Delta \text{ and }  s_i^\pi \leq \bar{s} D_{t_i^\pi+\Delta}, \ \ \ i=1,2,\dots, \\
& & dL^{\pi}_t=0, \ \ \forall t\in [t_i^{\pi}, t_i^{\pi}+\Delta],\end{eqnarray*}
where $\bar{s}$ is a positive constant.
First, to avoid liquidation, the dividend payment $dL^\pi_t$ is bounded, which gives $E_t^{\pi}  \geq \kappa D_t$.
Second, a new issue of equity cannot be ordered while previously ordered issuance is still waiting to be completed. 
Third, during the delay, dividends cannot be paid.
This dividend condition has important technical merit but also an economic justification, in that ruling out simultaneous capital issues and dividend payments
is likely to reduce conflicts of incentives between existing and new equity holders. The potential incentive conflicts are not explicitly present in our model, and we do not analyze the division of bank value between existing and new shareholders. We simply think of the dividend constraint as a restriction set by the capital markets. 
Fourth, the bank cannot sell infinite amount of equity, and therefore, the equity issuance $s_i$ is bounded.

The objective function of shareholders under policy $\pi$, given a level of equity $E_0=E$ and a level of debt $D_0=D$, is the value of the bank, which equals the expected discounted present value of dividends less equity issued until liquidation (see e.g. \cite{DMRV011} and  \cite{BCW11}):
\begin{eqnarray} \label{define of value function with recap}
&&\varphi(E,D)\\
&=&\sup_{\pi\in  \Pi}  \mathbb{E}\left[\int_0^{\tau^{\pi}} e^{-\delta t}dL^{\pi}_t-\sum_i e^{-\delta(t^{\pi}_i+\Delta)}
\left(s^{\pi}_i+KD_{\{t^{\pi}_i+\Delta\}}\right)  \mathbf{1}_{\{t^{ \pi}_i+\Delta<\tau^{ \pi}\}} + e^{-\delta \tau^{\pi}} \omega E_{\tau^{\pi}}^+\right], \nonumber
\end{eqnarray}
where $E^+=\max[E,0]$, $\delta$ is the discount rate that is assumed to be higher than $\mu$ and $\alpha$,\footnote{
If our modeling framework is risk neutral then equity premium is zero and the expectations in (\ref{define of value function with recap}) and (\ref{define of value function with recap and partially observed}) are under a risk neutral probability measure and, thus, the drift term in (\ref{equitydyn}) and (\ref{hat E_t}) need not coincide with the observed value (see e.g. \cite{BT04}). However,  \cite{PK06} show that the drift term only has a secondary effect on model capital ratios. In the model calibration we do not use the risk neutral probability measure; Instead we assume that the expectations in (\ref{define of value function with recap}) and (\ref{define of value function with recap and partially observed}) are under the objective probability measure. For an example of this interpretation and model calibration, see e.g. \cite{BM09}. Moreover, higher expected return rate on total assets implies better banking business and in turn attracts more depositors and increases the growth rate of debt. According to our data set in Table \ref{basic information}, the growth rate of debt is smaller than the return rate of total assets. Thus, by (\ref{equitydyn}), Higher growth rate of debt does not mean higher level of recapitalizations.
} 
namely,
\begin{eqnarray}
\label{discounta}
\delta> \max(\mu,\alpha),
\end{eqnarray}
$K$ is the cost of equity issuance, $\omega$ is the proportional liquidation value in terms of book equity (see e.g. \cite{ST92}).
The capital control problem is then to identify the value of an optimally managed bank and an optimal strategy $ \pi^*$ that achieves this supremum.

By $(\ref{define of value function with recap})$, the bank shareholders take the minimum equity-to-debt ratio $\kappa$ as a constraint to the bank's capital structure decision. That is, they decide to hold excess book equity above the minimum level as a hedge against bank liquidation, which means that the bank's capital structure decision is a risk management decision. However, the optimal capital structure is not full equity (actually far from that as we will see in Sections~\ref{comparative statics} and~\ref{calibration}) because future dividends are discounted by $\delta$ that is, by  (\ref{discounta}), higher than the expected growth rate of book equity in $(\ref{equitydyn})$.

\subsection{Partially Observed Model}\label{section with recap}

In this subsection we extend the fully observed model in the previous subsection by modelling the accounting total assets as a partially observed variable due to the opaqueness.  We abstract away the details of why (unintentional or intentional) there is noise in the accounting reports using the partially observed model.\footnote{With this condition we avoid complex signaling games between shareholders and regulators. For instance, if the shareholders were bank insiders and able to use the true assets value as decision criterion then the fact that a bank stops paying dividends
or recapitalizes would convey information on the true assets value to the regulators.}
However, in Section~\ref{calibration} we show that under the estimated parameters bankers clearly have an incentive to raise accounting noise that smooths asset values over time if they have equity-based compensation such as stock options, because the smoothing raises the market value of equity substantially.\footnote{In this case, bankers maximize $F\left(V(m,\rho)\right) - c(m,\rho)$ by selecting optimal $m$ and $\rho$ for them, where $F$ is the compensation, $V$ is the market value of equity and $c$ is the cost of earnings smoothing and they depend on the correlation $\rho$ and signal noise $m$. The cost stems from several sources, such as compliance with regulation and accounting rules, market discipline, risk culture of the bank, and the cost of additional effort. However, we do not model bankers explicitly in our model.}
Further, both the shareholders and the regulators have the noisy information when they make their decisions.
This is consistent e.g. with \cite{G03} and \cite{G07}, where investors have incomplete information on the quality of banks' assets.


Because of the opaqueness, process $Y_t$ in (\ref{totalassetdyn}), the true total assets without dividends and recapitalizations, is observed with noise. That is, the investors and the regulators receive noisy accounting reports,  and then they learn the true asset values from those. We model this learning process  using a Bayesian model (for more about Bayesian learning models see e.g. \cite{MH02},  \cite{DS06},   \cite{S08},  \cite{KMS08},  \cite{MV10}, and  \cite{HL15}). More specifically, the noisy accounting reports correspond to observing the following signal process:
\begin{eqnarray}
dZ_t &=& M_tdt + md\mathcal{B}_t, \hspace{5mm} Z_0 = 0, \label{signalprocess}\\
dM_t & =&\left(\alpha-\frac{\sigma^2}{2}\right)dt+\sigma dW_t, \label{dMt}
\end{eqnarray}
where $M_t :=\log Y_t$ is the log true accounting value of total assets without dividend and recapitalization, $m$ is a positive constant representing the noise level in the accounting reports, $\mathcal{B}_t$ is a standard Wiener process satisfying $dW_t d \mathcal{B}_t=\rho dt$.  Thus, $Z_t$ represents the cumulative observed (noisy) log accounting assets value.

Since $Z_t$ and $\log Y_t$ are correlated with $\rho$, negative correlation $(\rho<0)$ implies asset value smoothing because then a negative shock in $W_t$ is most likely compensated by a positive shock in $\mathcal{B}_t$.
We use this in the identification of asset smoothing in Section~\ref{calibration}.
We define the observed noisy accounting information set at time $t$ as ${\cal{G}}_{t} = \sigma \{Z_s, s\leq t\}$.
Using a similar argument as Chapter 4 in  \cite{B04}, 
we define the conditional expectation and conditional variance of $M_t$ as follows
\begin{eqnarray*}\begin{aligned}
\hat{M}_t &:= \mathbb{E}[M_t|{\cal{G}}_{t}], \\
 S_t & := \mathbb{E}[(M_t-\hat{M}_t)^2 |{\cal{G}}_{t} ], \nonumber
\end{aligned}\end{eqnarray*}
where $\hat{M}_t$ is the expected log assets value out of dividends and recapitalizations, $S_t$ is the variance of $M_t$ at time $t$ and, thus, $1/S_t$ is the precision of $\hat{M}_t$ at time $t$. This means that given the noisy accounting information at time $t$, the correct accounting assets value is unknown, i.e., the assets value is random variables at time $t$.
Further, the uncertainties in the accounting assets value is driven by future shocks in the asset prices and also the accounting uncertainty.
The following proposition gives the expected assets value and its precision (proof is in the Appendix \ref{appendix_theorem of s(t)}).

\begin{proposition}\textbf{(Noisy log asset values)} \label{theorem of s(t)} Variance $S_t=E[(M_t-\hat M_t)^2| \mathcal{G}_t]$ is deterministic and follows a Riccati equation: \begin{eqnarray}  \label{dSt}
dS_t=\left[\sigma^2-\left( \frac{S_t}{m}+\sigma \rho\right) ^2\right] dt.
\end{eqnarray}
The solution is given by\footnote{By the formula of $S(t)$, $\lim_{t\uparrow\infty}S_t=m\sigma(1-\rho)$, which is independent of the initial variance $S_0$. $S_0$ is an exogenous  variable and in the model calibration we choose $S_0 > m\sigma(1-\rho)$.  If $\rho =0$, $S_t$ is the same as in Theorem 2 in \cite{BKS09}.
There is a degenerate case: when $m=0$,  $S_t=0$ which means a fully observed case.}
\begin{equation}\label{S(t)}
S_t=\left\{\begin{aligned}& m\sigma\frac{A\exp(2\sigma t/m)-1}{A\exp(2\sigma t/m)+1}-m\sigma\rho & \text{ if } S_0<m\sigma(1-\rho)\\
& m\sigma(1-\rho)  & \text{ if } S_0=m\sigma(1-\rho)\\
& m\sigma\frac{A\exp(2\sigma t/m)+1}{A\exp(2\sigma t/m)-1}-m\sigma\rho & \text{ if } S_0>m\sigma(1-\rho),
\end{aligned}\right.\end{equation}
where $A=\left|\frac{m\sigma(1+\rho)+S_0}{m\sigma(1-\rho)-S_0}\right|$. Furthermore, the belief $\hat M_t= \mathbb{E}[M_t|{\cal{G}}_{t}]$ is given by Kalman filter:
\begin{eqnarray*}
d\hat M_t=(\alpha-\frac{1}{2}\sigma^2)dt+\left(\frac{S_t}{m}+\sigma\rho \right) \left(\frac{d Z_t}{m}-\frac{\hat M_t}{m}dt\right), \quad \hat M_0=M_0.
\end{eqnarray*}
\end{proposition}

By Proposition~\ref{theorem of s(t)}, the conditional expectation $\hat{M}_t$ follows:
\begin{eqnarray}\label{expected total asset iteration}
d\hat{M}_t = \left(\alpha - \dfrac{1}{2} \sigma^2\right) dt  + \left(\dfrac{S_t}{m}+\sigma\rho\right)d\tilde{\mathcal{B}}_t, \quad \hat M_0=M_0,
\end{eqnarray}
where $\tilde{\mathcal{B}}_t$ is an innovation process defined as
\begin{equation} \label{dWz}
 d\tilde{\mathcal{B}}_t=\frac{d Z_t}{m}-\frac{\hat M_t}{m}dt, \quad \tilde{\mathcal{B}}_0=0.
\end{equation}
By Proposition \ref{theorem of s(t)}, $\tilde{\mathcal{B}}_t$ is a standard Wiener process adapted to the information filtration $\mathcal{G}_t$. 
Given the information $ {\cal{G}}_{t}$, the process $M_t$ defined in $(\ref{signalprocess})$ follows a normal distribution with mean $\hat M_t$ and variance $S_t$. Therefore, $Y_t = e^{M_t}$ is a random variable at time $t$ and follows a log normal distribution:
\begin{align}\label{distribution_Y}
Y_t \sim \operatorname{Log-\mathcal{N}} \left(\hat{M}_t , \ S_t \right).
\end{align}
Define the expected total assets without dividends and recapitalizations by $\hat{Y}_t:= \mathbb{E}[Y_t|{\cal{G}}_{t}]$.
Then we have
\begin{eqnarray}\label{proportion policy}
\hat{Y}_t =  \exp\left\{\hat{M}_t+\dfrac{1}{2}S_t \right\},\quad \text{SD}[ Y_t | \mathcal{G}_t] = \hat Y_t \sqrt{\exp{(S_t)}-1},
\end{eqnarray}
where $\text{SD}$ stands for the stand deviation. 
By  Ito's formula, $(\ref{expected total asset iteration})$, $(\ref{proportion policy})$, $(\ref{totalassets})$ under $\mathcal{G}_t$, and after including dividends and recapitalizations similarly as the fully observed model, we obtain the dynamic process of  $\hat Y_t$ and
the expected equity $\hat E_t$ 
under policy $\pi=\{L_t^\pi, (s_i^\pi,t_i^\pi)\}$:
\begin{eqnarray}
d\hat Y_t &=& \alpha \hat Y_{t}dt+\left(S_t/m+\sigma \rho\right)\hat Y_{t}d\tilde{\mathcal{B}}_t, \label{dY_t}\\
\hat E_t^\pi&=&\hat E_0+\int_0^t \left[\alpha\hat E_{u}^\pi+(\alpha-\mu) D_u\right]du+ \int_0^t(\hat E_{u}^\pi+D_u)\left(S_u/m+\sigma\rho\right) d\tilde{\mathcal{B}}_u \nonumber\\
&&-  L_t^\pi+\sum_i   s^\pi_i \mathbf{1}_{\{t_i^\pi+\Delta\leq t\}}.  \label{hat E_t}
\end{eqnarray}
Figure \ref{figure asset tracing} illustrates the expected equity and the true equity values over time. As can be seen, they stay close to each other all the time.
\begin{figure}[ht!]
\centering
  \includegraphics[height=6cm,width=6.9cm]{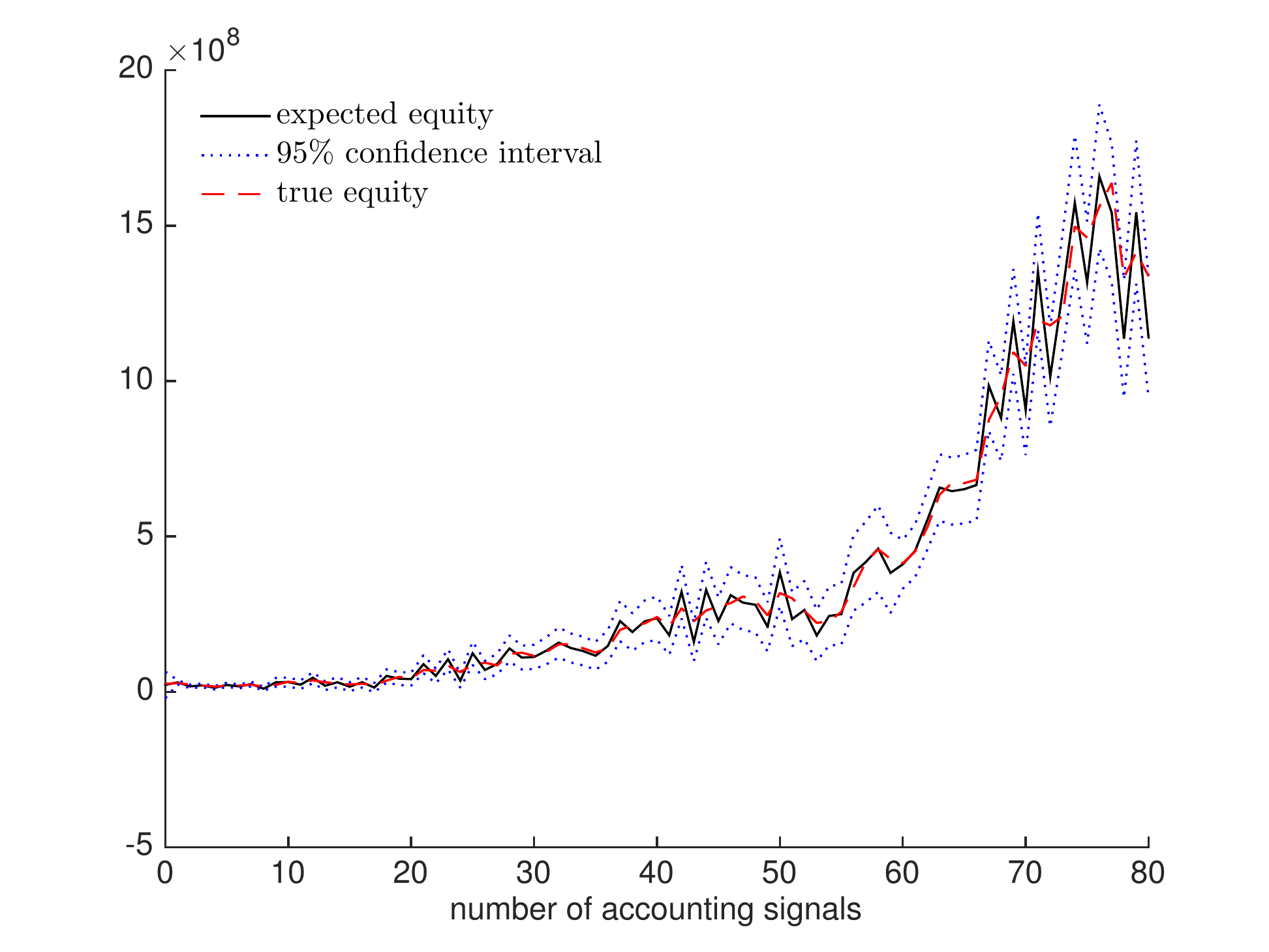}
	\caption[Book equity simulation]{\textbf{Book equity simulation.} This figure illustrates the expected book equity value and the corresponding true book value over the number of quarterly accounting reports (so 80 corresponds to 20 years). The parameter values: $\mu=0.035$, $\alpha=0.04$, $\sigma=0.05$, $m=0.03$, and $\rho=-0.3$. }
 \label{figure asset tracing}
\end{figure}

Regulators decide to liquidate a bank if the equity ratio falls too low. 
Under the filtration $\mathcal{G}_t$, the liquidation time under policy $\pi$ is defined as
\begin{eqnarray}\label{liqtime}\
\hat{\tau}^\pi := \inf\Big\{t\Big|\mathbb{P}\left(E_t^\pi/D_t<   \kappa \big| {\cal{G}}_{t}\right) \geq a\Big\}
\end{eqnarray}
with a positive constant $a$ representing the confidence level of the true unobserved equity-to-debt ratio $E_t^\pi/D_t$ being less than the threshold value $\kappa$. Therefore, if the probability of the equity-to-debt ratio less than $\kappa$  is higher than $a$, then the regulators view the bank insolvent and liquidate that. 

Regulators' liquidation decision can be understood as follows. Regulators face two risks: Not liquidating an insolvent bank and liquidating a solvent bank. Therefore, the regulators minimize the expected cost by selecting no liquidation or liquidation:
\begin{eqnarray*}
\text{expected cost if no liquidation }&=&c_1\mathbb{P}\left(E_t^\pi/D_t<   \kappa \big| {\cal{G}}_{t}\right)\\
\text{expected cost if liquidation}&=&c_2\mathbb{P}\left(E_t^\pi/D_t\geq  \kappa \big| {\cal{G}}_{t}\right),
\end{eqnarray*}
where $c_1$ and $c_2$ are the cost parameters for no liquidation and liquidation, and the probabilities are the probabilities of insolvency and solvency, respectively.
The interpretation of the cost parameters is the cost given the decision (no liquidation or liquidation) and the outcome (insolvency or solvency). 
If the expected cost under liquidation is less than the expected cost for no liquidation, i.e., if 
\[
c_2\mathbb{P}\left(E_t^\pi/D_t\geq  \kappa \big| {\cal{G}}_{t}\right) \leq c_1\mathbb{P}\left(E_t^\pi/D_t<   \kappa \big| {\cal{G}}_{t}\right)
\]
then the regulators decide to liquidate the bank. This condition can be written as $\mathbb{P}\left(E_t^\pi/D_t<   \kappa \big| {\cal{G}}_{t}\right) \geq \tfrac{c_2}{c_1+c_2}$, which means that $a=\tfrac{c_2}{c_1+c_2}$ in $(\ref{liqtime})$.
Thus, if the cost parameter of not liquidating an insolvent bank is less than the cost parameter of liquidating a solvent bank $(c_1<c_2)$ then $a>50\%$. As we will see in Proposition~\ref{special case}, in this case the bank benefits from the noise in the asset values.

In Section~\ref{calibration} we estimate $a$ and it is about 80\%, which means $c_1<c_2$ and that the bank benefits from the accounting noise.
This is consistent with \cite{BW12} and \cite{HL12} who find that regulators may actually exploit financial reporting choices in order to not intervene in troubled banks, often referred to as forbearance  (see footnote~\ref{footnote4} in Chapter \ref{chp: Intro}). 

By (\ref{distribution_Y}) and (\ref{proportion policy}),  $Y_t$ follows a log normal distribution with mean $(\log \hat Y_t-S_t/2)$ and variance $S_t$ under the observed information $\mathcal{G}_t$, similar for the true total assets value $E_t+D_t$.\footnote{During times when there are neither dividends nor recapitalizations, the true total assets without dividends and recapitalizations, $Y_t$, and the true total assets, $E_t+D_t$, both follow geometric Brownian motion processes with the same expected return, volatility, and Wiener process (or alternatively  $\hat Y_t$ and $\hat E_t+D_t$). Then under $\mathcal{G}_t$ $\log(E_t + D_t)$ follows a normal distribution with mean $(\log (\hat E_t+D_t) - S_t/2)$ and variance $S_t$. 
}
We get the following representation for the liquidation time under policy $\pi$: 
\begin{eqnarray}\label{hat tao}
& &\hat{\tau}^\pi = \inf\left\{t \Big | \hat  E_t^\pi/D_t \leq I(S_t)\right\},
\end{eqnarray}
where $I(S_t)=-1+(1+\kappa)e^{\frac{1}{2}S_t-\Phi^{-1}(a)\sqrt{S_t}}$ is the liquidation barrier for $\hat  E^{\pi}_t/D_t$ (proof is in Appendix \ref{Appendix A1}), and $\Phi$ is the cumulative standard normal distribution function. 

Similar with fully observed model, we denote the set of admissible control strategies by ${\Pi}$ that satisfy the following: $L_t^\pi$ is a nondecreasing right-continuous process adapted to ${\cal{G}}_{t}$ and $L_{0-}^\pi=0$; each $t_i^{\pi}$ is a stopping time of the filtration ${\cal{G}}_{t}$; each $s_i^{\pi}$ is measurable with respect to ${\cal{G}}_{(t_i^{\pi}+\Delta)-}$.
The measurability of $s_i^{\pi}$ with respect to ${\cal{G}}_{(t_i^{\pi}+\Delta)-}$ means that owners may decide on the exact amount of capital to be raised at time $t_i^{\pi}+\Delta$ based on all
the available information. Additionally, admissible
controls satisfy:
\begin{eqnarray*}
& & \hat X_t^\pi \geq I(S_t), \  \ \forall t\geq 0,  \quad  t_{i+1}^{\pi}-t_i^{\pi}\geq \Delta \text{ and }  s_i^\pi \leq \bar{s} D_{t_i^\pi+\Delta}, \ \ \  i=1,2,\dots,\\
&&  dL^{\pi}_t=0, \ \ \forall t\in [t_i^{\pi}, t_i^{\pi}+\Delta].
 \end{eqnarray*}

The objective function of shareholders under  information filtration $\mathcal{G}$, given the level of equity $\hat E_0=\hat E$,  the level of debt $D_0=D$, and the initial accounting asset uncertainty level $S_0=S$, is to maximize the expected discounted present value of dividends less equity issuance until liquidation over all admissible strategies:
\begin{equation}\label{define of value function with recap and partially observed}
\begin{aligned}
& \hat \varphi (\hat E,D,S ) \\
 =& \max_{\pi\in\Pi}\mathbb{E}^{\hat E,D,S}\Bigg[ \int_0^{\hat \tau^{\pi}} e^{-\delta t}dL^{\pi}_t 
 -\sum_i e^{-\delta(t^{\pi}_i+\Delta)}
\left(s^{\pi}_i+K D_{\{t^{\pi}_i+\Delta\}}\right)  \mathbf{1}_{\{t^{\pi}_i+\Delta< \hat \tau^{\pi}\}} +e^{-\delta \hat \tau^{\pi} }\omega E_{\hat \tau^\pi}^+ \Bigg].\\
\end{aligned}\end{equation}

Note that, by $(\ref{define of value function with recap and partially observed})$, when the bank is liquidated  the bank shareholders receive $\omega \mathbb{E}[ E_{\hat \tau^\pi}^+|\mathcal{G}_{\hat\tau^\pi}]$, i.e., a proportion of the book equity and, thus, $(1-\omega) \mathbb{E}[ E_{\hat \tau^\pi}^+|\mathcal{G}_{\hat\tau^\pi}]$ is the liquidation cost.  According to our model calibration, $\omega$ is about 32\%.
However, in practice, a violation of the minimum equity-to-debt requirement does not result in immediate liquidation but does generate additional costs and constraints to the bank, due to increased regulatory surveillance (\cite{P97} list the provisions for Prompt Corrective Action specified in the FDICIA). Further, the bank's competitive position is likely to be affected. Therefore, the bank's shareholders are likely to lose a substantial amount of the bank's economic rent. These effects raise our estimated $\omega$.

\section{Theoretical Analysis}\label{partI_theoretical analysis}

\subsection{Fully Observed Model}\label{fully_theorysection}
We show that under the assumption (\ref{discounta}), the value function $\varphi(E,D)$ in (\ref{define of value function with recap}) in fully observed model is finite. The proof is similar to that in Proposition \ref{partial_growth condition} for partially observed model.  Throughout the whole Chapter \ref{part 1}, we always assume that (\ref{discounta}) holds.

The value function (\ref{define of value function with recap}) is associated with the following Hamilton-Jacobi-Bellman equation (\text{HJB}) equation:\footnote{Because our problem is a mixed singular and impulse control with execution delay, we need reestablish an appropriate dynamic programming principle and viscosity property, as given in Proposition \ref{prop_wdpp}  and Theorem \ref{semi-explicit solution with recap}.}
 \begin{eqnarray}\label{fully_hjb2}
 \max\left\{\mathcal{A}^0 \varphi, \quad 1- \frac{\partial }{\partial E} \varphi,  \quad \mathcal{M}^0\varphi- \varphi\right\} =0 
 \end{eqnarray}
in $E/D>\kappa$ with boundary condition: $ \varphi\left(\kappa D, D\right)=\omega \kappa D, \forall D>0$. Here the operators $\mathcal{A}^0$ and $\mathcal{M}^0$ are defined as
{\fontsize{10.5pt}{10.5pt}
\begin{eqnarray*} \begin{aligned}
&\mathcal{A}^0=\frac{1}{2}(E+D)^2\sigma^2\frac{\partial^2}{\partial  E^2}+\left(\alpha E+(\alpha-\mu) D \right) \frac{\partial}{\partial E}+\mu D\frac{\partial}{\partial D}-\delta,\\
&\mathcal{M}^0 \varphi( E, D)=\sup_{s\in (0,\bar sD_{\Delta})} 
\mathbb{E}^{ E, D}\Bigg [ e^{-\delta\Delta} \left( \varphi( E_{\Delta}+s,D_\Delta)-s-KD_\Delta \right) \mathbf{1}_{\{ \tau>\Delta\}} +e^{-\delta \tau} \omega \kappa D_{\tau} \mathbf{1}_{\{ \tau\leq \Delta\}} \Bigg],
\end{aligned}\end{eqnarray*}}
and  $\mathbb{E}^{ E, D}$ is the expectation conditionally on that $E_0= E,\ D_0=D$ with new equity  to be issued at time $\Delta$, and $\tau=\inf\{t\in[0,\Delta]: E_t/D_t  \leq \kappa \}$ is the first stopping time for $E_t/D_t$ hitting $\kappa $ during $[0,\Delta]$.

The value function in $(\ref{define of value function with recap})$ is  two-dimensional and therefore, complicated to solve.
It is easy to verify that the value function $ \varphi(E, D)$ in $(\ref{define of value function with recap})$ is homogeneous in $E$ and $D$: \begin{eqnarray*}\varphi(\gamma E, \gamma D)=\gamma \varphi(E, D), \quad \forall \gamma>0,
 \end{eqnarray*}
which implies the following reduction:
\begin{eqnarray}\label{v(x)}
 V (X):= \varphi(E, D)/D, \  \   X=E/ D.
\end{eqnarray}
The equity-to-debt ratio $X_t:=E_t/D_t$  then satisfies \begin{align}
 \label{X_t}
X_t^\pi=&X+\int_0^t( X_u^\pi+1)(\alpha-\mu)du+\int_0^t \sigma (X_u^\pi+1) d W_u -L_t^\pi+\sum_i s^\pi_i \mathbf{1}_{\{t_i^\pi+\Delta\leq t\}}
\end{align}
with $X_0= X.$ Here for simplicity, we still use $\pi=\{L_t^\pi, (s_i^\pi, t_i^\pi)\}$ to represent an admissible strategy for $X_t$, but $L_t^\pi$ and $s_i^\pi$ above are the sizes of dividends and recapitalizations in terms of debt $D_t$ and $D_{t_i^\pi+\Delta}$. The admissible strategy satisfies $X_t^\pi \geq \kappa$ for all $t\geq 0$ and $s_i^\pi \leq \bar{s}$ for $i\geq 1$.

The corresponding \text{HJB} equation for value function $(\ref{v(x)})$ reduces to 
\begin{eqnarray}\label{fully_hjb3}\max\left\{\mathcal{L}^0 V, \quad 1-  V_X,  \quad \mathcal{P}^0 V- V\right\} =0, \end{eqnarray}
in $ X>\kappa$ with boundary condition $ V(\kappa)=\omega \kappa$, where
\begin{eqnarray*}
\mathcal{L}^0V&=&\frac{1}{2}(1+ X)^2\sigma^2V_{XX}+(\alpha-\mu)(1+ X)V_{X}-(\delta-\mu)V,\\
\mathcal{P}^0 V( X)&=& \sup_{s\in(0,\bar{s})}  \mathbb{E}\left[e^{-(\delta-\mu)\Delta}(V( X_{\Delta}+s)-s-K)\mathbf{1}_{\{{\tau}>\Delta\}}+e^{-(\delta-\mu)\tau}\omega \kappa \mathbf{1}_{\{\tau<\Delta\}} \right],
\end{eqnarray*}
and $ X_t$ follows $d X_t=(\alpha-\mu)( X_t+1)dt+\sigma( X_t+1) d W_t$ with $X_0= X$, and $\tau =\inf\{t\in[0,\Delta]: X_t \leq \kappa \} $ is the first stopping time for $X_t$ hitting $\kappa$ during $[0,\Delta]$.

The following theorem provides a theoretical justification of the linkage between the HJB equation (\ref{fully_hjb3}) and the value function in the fully observed case. Moreover, using the smooth fit principle for impulse control and singular control (see, e.g., \cite{OS02}, \cite{G09}, and \cite{GW09}), we have a semi-explicit solution under some necessary conditions.

\begin{theorem}\label{semi-explicit solution with recap} \textbf{(Fully observed model)}\\
(i) The HJB equation (\ref{fully_hjb3}) has a unique viscosity solution $V=V(X)$.
Define $\varphi(E,D)=DV(E/D)$. Then $\varphi(E,D)$ is the value function of the fully observed model. \\
(ii) Under the conditions (\ref{appendix_fully_assumption}),\footnote{Because of the delay with equity issuance, the conditions in (\ref{appendix_fully_assumption}) are very complex. Numerically we can verify that there is a large set of parameters satisfying (\ref{appendix_fully_assumption}). In particular, the parameters in Table~\ref{table allbankparameter} satisfy (\ref{appendix_fully_assumption}).} $V(X)$ has the following semi-explicit solution:
\begin{equation}\label{value of recap}
 V(X)=\left\{\begin{aligned}&H(X; u_2) & \kappa\leq X\leq u_1 \\
 & f_1(X; u_2)  & u_1 < X < u_2 \\
 & f_2(X; u_2) & u_2 \leq X,
\end{aligned}\right.
\end{equation}
where $f_1, f_2$, and $H$ are  given by (\ref{appendixf1}), (\ref{appendixf2}), and (\ref{fully M}), and the barriers $u_1$ and $u_2$ are determined by (\ref{fully_u1_u2_formula}). Moreover, the optimal strategy $\pi^*=\{L_t^{\pi^*}, (s_i^{\pi^*}, t_i^{\pi^*})\}$ is given by 
\begin{equation*}\begin{aligned}
 t_i^{\pi^*}: = & \inf \{t> t_{i-1}^{\pi^*}+\Delta:  X_t^{\pi^*}\leq u_1\},\quad \forall i=1,2,\dots, \\
 s_i^{\pi^*}: = & \max\left\{u_2-X_{t_i^{\pi^*}+\Delta}^{\pi^*},\   0\right\},  \quad \forall i=1,2,\dots,\\
L_t^{\pi^*}: = & \int_0^t \mathbf{1}_{\{X_t^{\pi^*} = u_2\}} dL_s^{\pi^*},\quad  \forall t\geq 0,\\
\end{aligned}
\end{equation*}
with $t_0^{\pi^*}:=-\Delta$ and $s_0^{\pi^*}: = 0.$\footnote{If liquidation time $\tau^{\pi^*}<\infty$ and $J$ is the last equity issuance time before $\tau^{\pi^*}$ then we define $t_n=\infty$ for all $n>J$.}
\end{theorem}

\begin{proof}
(i). The proof is similar to that in Theorem  \ref{partial_viscosity property} for the partially observed model. Thus we skip it.\\
(ii). We conjecture that  (\ref{fully_hjb3}) with boundary condition $V(\kappa)=\omega\kappa$ is equivalent to the following free boundary problem
\begin{equation}\label{equivalent free boundary p}\left\{\begin{aligned}
&\mathcal{L}^0 V=0  & \mbox{in } (u_1 ,u_2), \\
& V_X=1,\   V_{XX}=0 & \mbox{at } X=u_2, \\
& V= \mathcal{P}^0 V, \   V_X = \frac{d}{dX}\mathcal{P}^0 V  & \mbox{at } X=u_1,
\end{aligned}\right.\end{equation}
where  $u_1$ and $u_2$ are to be determined.  Now we solve problem (\ref{equivalent free boundary p}).

The general solution to $\mathcal{L}^0V=0$ is $V(X)=A_1 (X+1)^{\lambda_-}+A_2(X+1)^{\lambda_+}$, where $A_1$ and $A_2$ are arbitrary constants. By $V_X(u_2)=1$ and $V_{XX}(u_2)=0$, we can solve $A_1$ and $A_2$ as functions of $u_2$. This gives
$V(X)=f_1(X; u_2)$ for $X\in(u_1,u_2)$ and $V(X)=f_2(X;u_2)$ for $X\geq u_2$, where
\begin{eqnarray}
 f_1(X;u_2):&=&\frac{(\lambda_+-1)(1+X)^{\lambda_-}}{\lambda_-(\lambda_+-\lambda_-)(1+u_2)^{\lambda_--1}}-\frac{(\lambda_--1)(1+X)^{\lambda_+}}{\lambda_+(\lambda_+-\lambda_-)(1+u_2)^{\lambda_+-1}},  \label{appendixf1} \\
 f_2(X;u_2) :&=& f_1(u_2;u_2)+(X-u_2),  \label{appendixf2}
\end{eqnarray}\label{fully M}
and $\lambda_{\pm}:=\frac{-(\alpha-\mu-\frac{1}{2}\sigma^2)\pm \sqrt{(\alpha-\mu-\frac{1}{2}\sigma^2)^2+2\sigma^2(\delta-\mu)}}{\sigma^2}$.

Now we derive the formula of $\mathcal{P}^0V$.  By the first order condition, the optimal issuance $s$ equals $(u_2- X_{\Delta})^+$.  If $X_\Delta>u_2$, $V(X_\Delta)=V(u_2)+(u_2-X_\Delta)$ because $V(X)=f_2(X;u_2)$ for all $X>u_2$. Thus, we obtain
 \begin{eqnarray} \label{appendix M(x;u2)}
 \begin{aligned}
&\mathcal{P}^0V(X) \\
  =& \mathbb{E}\big[e^{-(\delta-\mu)\Delta}\left[(V(u_2)-u_2-K-1)+(X_{\Delta}+1)\right]\mathbf{1}_{\{\tau>\Delta\}} +e^{-(\delta-\mu)\tau}\omega \kappa \mathbf{1}_{\{\tau\leq\Delta\}}\big].
\end{aligned}
\end{eqnarray}
Noticing that $\ln(1+X_t)$ follows a Brownian motion during $[0,\Delta]$, we have
\begin{align}\label{p(x,t)}
p(X,t):=&\mathbb{P}\left[\tau\leq t| X_0=X\right] \\
=& 1- \Phi\left(\frac{\ln\left(\frac{X+1}{\kappa+1}\right)+\mu_{-} \Delta}{\sqrt{\Delta}\sigma}\right)+\left(\frac{\kappa+1}{X+1}\right)^{\frac{2\mu_{-}}{\sigma^2}}\Phi\left(\frac{\ln\left(\frac{\kappa+1}{X+1}\right)+\mu_{-} \Delta}{\sqrt{\Delta}\sigma}\right),\nonumber
\end{align}
where $\mu_{-}:= \alpha-\mu-\frac{1}{2}\sigma^2$.
By the reflection principle (see, e.g. Borodin and Salminen (2002)),
 \begin{eqnarray*}\begin{aligned}
&\mathbb{E}\left[(X_\Delta+1)\mathbf{1}_{\{\tau>\Delta\}}\right]\\
=& e^{\mu_{-} \Delta+\ln(x+1)+\frac{1}{2}\Delta \sigma^2}\left[1-\Phi\left( \frac{\ln\left(
\frac{\kappa+1}{X+1}\right)-\mu_{-} \Delta-\Delta \sigma^2}{\sqrt{\Delta}\sigma}\right)\right]\\
& + e^{(\alpha-\mu)\Delta}(1+\kappa)\left(\frac{1+X}{1+\kappa}\right)^{-\frac{2(\alpha-\mu)}{\sigma^2}}\left[1-\Phi\left(\frac{\ln\left(\frac{X+1}{\kappa+1}\right)-\mu_{-} \Delta-\Delta\sigma^2}{\sqrt{\Delta}\sigma}\right)\right].   \end{aligned}\end{eqnarray*}
In addition, by the formula of $p(X,t)$ in (\ref{p(x,t)}), we can calculate $\mathbb{E}[e^{-(\delta-\mu)\tau}\omega \kappa \mathbf{1}_{\{\tau\leq \Delta\}} ]$  in (\ref{appendix M(x;u2)}). Thus, we have $V(X)=H(X;u_2)$ for $X\in(\kappa,u_1]$, where
\begin{eqnarray}
H(X;u_2) : &=& \mathcal{P}^0 V(X) =  e^{-(\delta-\alpha)\Delta}(1+X)\Phi\left(\frac{\ln\left(\frac{1+X}{1+\kappa}\right)+\mu_{+}\Delta}{\sqrt{\Delta}\sigma}\right)  \label{fully M} \\
& & -(1+\kappa)e^{-(\delta-\alpha)\Delta}\left(\frac{1+X}{1+\kappa}\right)^{-\frac{2(\alpha-\mu)}{\sigma^2}}\Phi\left(\frac{\ln\left(\frac{1+\kappa}{1+X}\right)+\mu_{+}\Delta}{\sqrt{\Delta}\sigma}\right) \nonumber \\
& & +[f_1(u_2;u_2)-u_2-K-1]e^{-(\delta-\mu)\Delta}\Phi\left(\frac{\ln\left(\frac{1+X}{1+\kappa}\right)+\mu_{-}\Delta}{\sqrt{\Delta}\sigma}\right)\nonumber \\
& & -[f_1(u_2;u_2)-u_2-K-1]e^{-(\delta-\mu)\Delta} \left(\frac{1+\kappa}{1+X}\right)^{\frac{2\mu_{-}}{\sigma^2}}\Phi\left(\frac{\ln\left(\frac{1+\kappa}{1+X}\right)+\mu_{-}\Delta}{\sqrt{\Delta}\sigma}\right) \nonumber\\
& & + \omega \kappa \int_0^{\Delta} e^{-(\delta-\mu)t}\frac{\partial }{\partial t} p(X,t)dt, \nonumber
\end{eqnarray}
and $\mu_{+}:= \alpha-\mu+\frac{1}{2}\sigma^2$. This leads to the desired result in (\ref{value of recap}).

Now the equations $V(u_1)= \mathcal{P}^0 V(u_1)$ and $V_X(u_1) = \frac{d}{dX}\mathcal{P}^0 V(u_1)$ in (\ref{equivalent free boundary p}) can be rewritten as
\begin{equation}\label{fully_u1_u2_formula}
H(u_1;u_2)=f_1(u_1;u_2),\quad \frac{\partial H(X;u_2)}{\partial x}\Big|_{X=u_1}= \frac{\partial f_1(X;u_2)}{\partial X}\Big|_{X=u_1}.
\end{equation}
We expect that $u_1$ and $u_2$ are determined by (\ref{fully_u1_u2_formula}). We first give a lemma as follows.
\begin{lemma}\label{lemma_u0}
When $\omega\kappa <\frac{\alpha-\mu}{\delta-\mu}(1+\kappa)$,  the equation $f_1(\kappa; \theta)=\omega\kappa$ related to $\theta$, has a unique solution $\theta=u_0$ in domain $[\kappa,+\infty)$, where $f_1$ is given in (\ref{appendixf1}).
\end{lemma}
\noindent\textit{Proof of Lemma \ref{lemma_u0}}:
Because $\lambda_{-}<0<1<\lambda_{+}$, we obtain
\begin{eqnarray*}\begin{aligned}
\frac{\partial f_1(X;\theta)}{\partial \theta}=\frac{(\lambda_+-1)(1-\lambda_-)}{\lambda_-(\lambda_+-\lambda_-)}\left(\frac{1+X}{1+\theta}\right)^{\lambda_-}-\frac{(\lambda_--1)(1-\lambda_+)}{\lambda_+(\lambda_+-\lambda_-)}\left(\frac{1+X}{1+\theta}\right)^{\lambda_+}<0,
\end{aligned}\end{eqnarray*}
for all $X\geq \kappa$. Therefore, function $f_1(\kappa; \theta)$ given in (\ref{appendixf1}) is monotonically decreasing in $\theta$. Moreover, when $\omega\kappa <\frac{\alpha-\mu}{\delta-\mu}(1+\kappa)$, we have $f_1(\kappa;\kappa) = \frac{\alpha-\mu}{\delta-\mu}(1+\kappa) > \omega\kappa$ and $\underset{\theta\rightarrow +\infty}{\lim}f_1(\kappa;\theta) = -\infty<\omega\kappa$.
Hence, there exists a unique solution $\theta = u_0 \in (\kappa, \infty)$ to the equation $f_1(\kappa; \theta)=\omega\kappa$. This completes the proof of Lemma \ref{lemma_u0}. 

Now we give some sufficient conditions for the existence of $u_1$ and $u_2$. We always assume that the upper bound of equity issuance $\bar{s}$ satisfies $\bar s\geq u_0$.
We can show that under the following conditions \begin{eqnarray}\label{appendix_fully_assumption}
\begin{aligned}
&  \frac{\partial H(X;u_0)}{\partial X}\big|_{X=\kappa}> \frac{\partial f_1(X;u_0)}{\partial X}\big|_{X=\kappa},  \  \   \alpha-\mu-\frac{1}{2}\sigma^2\geq 0, \\
&
\omega\kappa <e^{-(\delta-\alpha)\Delta}\frac{\alpha-\mu}{\delta-\mu}(1+\kappa), 
\  \   u_0<\frac{\alpha-\mu}{\delta-\alpha}+\frac{\delta-\mu}{\delta-\alpha}(\kappa-\omega\kappa-K),  \\
& \text{ and } \frac{\delta-\alpha}{\delta-\mu} e^{(\alpha-\mu)\Delta} \geq \bigg[\frac{\delta-\alpha}{\delta-\mu} e^{(\alpha-\mu)\Delta}(1+\kappa) -\frac{\delta-\alpha}{\delta-\mu}(1+u_0)-K\bigg] \frac{\partial p(X,\Delta)}{\partial X}\big|_{X=\kappa},
\end{aligned}
\end{eqnarray}
there exists a pair solution $(u_1,u_2)$ to (\ref{fully_u1_u2_formula}) and  $V(X)$ in (\ref{value of recap}) satisfies HJB equation (\ref{fully_hjb3}). The proof is available in  Appendix \ref{fully_proof_uniqueness_u1_u2} and \ref{verification of conjectured v}.
 Thus $V(X)$ in (\ref{value of recap}) coincides with the value function of the fully observed model by the uniqueness of viscosity solution.

By the construction and proof above, we have $\{X>\kappa: V-\mathcal{P}^0V= 0\} = (\kappa, u_1]$, $\{X>\kappa: \mathcal{L}^0 V=0\}=(u_1,u_2)$, and  $\{X>\kappa: V_X-1=0\} = [u_2,+\infty)$. Therefore,
the strategy $\pi^*$ defined in Theorem \ref{semi-explicit solution with recap}
is the strategy associated with  $V(X)$ in (\ref{value of recap}).
We can verify the optimality of $\pi^*$ as follows. For any admissible strategy $\pi$, define $\tau_\varepsilon^{\pi}=(1/\varepsilon)\wedge\inf\{t\geq 0: X_t^{\pi}<\kappa+\varepsilon\}.$

Remember that $t_0=-\Delta$. Then we have
\begin{eqnarray} \label{discounted increase}
&&e^{-\delta_1 (\tau_\varepsilon^{\pi}\wedge t_n)} V( X^{\pi}_{\tau_\varepsilon^{\pi}\wedge t_n})-V(X) \\
&=&\sum_{i=1}^n \bigg\{ e^{-\delta_{1} (\tau_\varepsilon^{\pi}\wedge t_{i}^{\pi})} V(X^{\pi}_{\tau_\varepsilon^{\pi}\wedge t_{i}^{\pi}}) - e^{-\delta_{1} (\tau_\varepsilon^{\pi}\wedge (t_{i-1}^{\pi}+\Delta))} V(X^{\pi}_{\tau_\varepsilon^{\pi}\wedge (t_{i-1}^{\pi}+\Delta)})\bigg\}\nonumber\\
&&+ \sum_{i=1}^n \mathbf{1}_{\{t_i^{\pi}+\Delta\leq \tau_\varepsilon^{\pi}\}} e^{-\delta_{1}  t_i^{\pi}} \bigg\{ e^{-\delta_{1} \Delta}V( X^{\pi}_{t_i^{\pi}+\Delta})-V(X^{\pi}_{t_i^{\pi}})\bigg\} \nonumber
\end{eqnarray}
for each $n\geq 1$.

Since $X_t$ is a cadlag semimartingale on the stochastic interval $[t_{i-1}+\Delta, t_i)$,
we can apply Ito's formula (see e.g. IV.45 in Rogers and Williams (1987)) to get
\begin{eqnarray} \label{discounted increase1}
& &e^{-\delta_{1} (\tau_\varepsilon^{\pi}\wedge t_i^{\pi})} V( X^{\pi}_{\tau_\varepsilon^{\pi}\wedge t_i^{\pi}}) - e^{-\delta_{1} (\tau_\varepsilon^{\pi}\wedge  (t_{i-1}^{\pi}+\Delta)  )} V(X^{\pi}_{\tau_\varepsilon^{\pi}\wedge (t_{i-1}^{\pi}+\Delta)}) \nonumber \\
=&&\int_{\tau_\varepsilon^{\pi}\wedge (t_{i-1}^{\pi}+\Delta)}^{\tau_\varepsilon^{\pi}\wedge t_i^{\pi}} e^{-\delta_{1} u}\Big\{\underbrace{\mathcal{L}^0V}_{\leq 0}du +\underbrace{(1-V_{X})}_{\leq 0}dL_u^\pi +\sigma(1+X_u^{\pi})V_{X} dW_u-dL_u^{\pi}\Big\} \nonumber\\
\leq& &  \int_{\tau_\varepsilon^{\pi}\wedge (t_{i-1}^{\pi}+\Delta)}^{\tau_\varepsilon^{\pi}\wedge t_i^{\pi}} e^{-\delta_{1} u}\left\{ \sigma(1+ X_u^{\pi})V_{X} dW_u-dL_u^{\pi}\right\}.
\end{eqnarray}
Note that the above inequality holds with equality for the control $\pi^*$.
From $V\geq \mathcal{P}^0V$, we have
  \begin{align} \label{discounted increase2}
 & \sum_{i=1}^n e^{-\delta_{1}  t_i^{\pi}}\left[e^{-\delta_{1} \Delta}V( X^{\pi}_{t_i^{\pi}+\Delta})-V(X^{\pi}_{t_i^{\pi}}) \right] \mathbf{1}_{\{t_i^{\pi}+\Delta\leq \tau_\varepsilon^{\pi}\}}\\
  & \leq \sum_{i=1}^n e^{-\delta_{1} (t_i^{\pi}+\Delta)}(s_i^{\pi}+K) \mathbf{1}_{\{t_i^{\pi}+\Delta\leq \tau_\varepsilon^{\pi}\}} \nonumber
 \end{align}
  with the inequality being tight for the control $\pi^*$.
By the boundedness of $V_X$ (here $V_X(X_t^\pi)\leq V_X(\kappa+\varepsilon)$ for all $t\in [\tau_\varepsilon^{\pi}\wedge (t_{i-1}^{\pi}+\Delta), \tau_\varepsilon^{\pi}\wedge t_i^{\pi})$),
 \begin{eqnarray} \label{0mean}
\mathbb{E}\left [\int_{\tau_\varepsilon^{\pi}\wedge (t_{i-1}^{\pi}+\Delta)}^{\tau_\varepsilon^{\pi}\wedge t_i^{\pi}} e^{-\delta_{1} u}\sigma (1+X_u^{\pi})V_{X}(X_u^{\pi}) dW_u\right ]=0. \end{eqnarray}
Combining (\ref{discounted increase}), (\ref{discounted increase1}), (\ref{discounted increase2}), and (\ref{0mean}),  we have
 \begin{eqnarray} \label{before limit}
 & & V(X) -\mathbb{E}^{ X}\left[e^{-\delta_{1} (\tau_\varepsilon^{\pi}\wedge (t_n^\pi+\Delta))} V(X^{\pi}_{\tau_\varepsilon^{\pi}\wedge (t_n^{\pi}+\Delta)})\right]\\
 &\geq&  \mathbb{E}^{X}\left [ \int_0^{\tau_\varepsilon^{\pi}\wedge t_n^{\pi}}e^{-\delta_{1} u}dL_u^{\pi}- \sum_{i}^n e^{-\delta_{1} (t_i^{\pi}+\Delta)}(s_i^{\pi}+K) \mathbf{1}_{\{t_i^{\pi}+\Delta\leq \tau_\varepsilon^{\pi}\}} \right ] \nonumber
 \end{eqnarray}
 with the inequality being tight for the control $\pi^*$.
By the definition of $\pi$, we can always assume $t_n^{\pi} \rightarrow \infty$ when $n\rightarrow \infty$. That is, if $J$ is is the last equity issuance time before $\tau^{\pi}$, we set $t_{i}^\pi =+\infty$ for all $i>J$.
As such, we have
\begin{eqnarray*}
 \lim_{n\rightarrow \infty} \left\{V(X) -\mathbb{E}^{X}\left[e^{-\delta_{1} (\tau_\varepsilon^{\pi}\wedge (t_n+\Delta))} V(  X^{\pi}_{\tau_\varepsilon^{\pi}\wedge (t_n^{\pi}+\Delta)})\right] \right\}= V(X)-\mathbb{E}^{X}\left[e^{-\delta_{1} \tau_\varepsilon^{\pi}} V(X^{\pi}_{\tau_\varepsilon^{\pi} }) \right].
 \end{eqnarray*}
Sending $\varepsilon\rightarrow 0$ and $n\rightarrow \infty$ in (\ref{before limit}), according to the linear growth property of $V$ and dominant convergence theorem, we have
 \begin{eqnarray*}
 \begin{aligned} 
 V(X)-\mathbb{E}^{X}\left[e^{-\delta_{1} \tau^{\pi}} V(X^{\pi}_{ \tau^{\pi} }) \right] \geq \mathbb{E}^{X}\left [ \int_0^{ \tau^{\pi}}e^{-\delta_{1} u}dL_u^\pi- \sum_{i} e^{-\delta_{1} (t_i^{\pi}+\Delta)}(s_i^{\pi}+K) \mathbf{1}_{\{t_i^{\pi}+\Delta\leq \tau^{\pi}\}} \right ] 
 \end{aligned}\end{eqnarray*}
 with the inequality being tight for the control $\pi^*$. Notice $V(X_{\tau^{\pi}})=\omega \kappa = \omega \mathbb{E}^{X}[X_{\tau^{\pi}}^+]$.  Thus, $\pi^*$ is the optimal strategy. This completes the proof of Theorem \ref{semi-explicit solution with recap}.
\end{proof}

By Theorem \ref{semi-explicit solution with recap}, the optimal strategy of the fully observed model is characterized as follows. Parameter $u_2$ is the dividend barrier and parameter $u_1$ is the recapitalization barrier. The bank issues new equity when the equity-to-debt ratio is equal or below $u_1$. When the equity-to-debt ratio is between $u_1$ and $u_2$, the bank neither pays dividends nor issues equity.  When the equity-to-debt ratio is above $u_2$, the bank pays dividends to decrease the ratio level to $u_2$, as shown in Figure \ref{figure recap_fullymodelvalue}.

The interpretation of $(\ref{value of recap})$ is as follows. The bank issues new equity when equity to debt ratio is smaller than $u_1$. When the equity-to-debt ratio is between $u_1$ and $u_2$, the bank neither pays dividends nor issues equity. When the equity-to-debt ratio is above $u_2$, the bank pays dividends to decrease the ratio level to $u_2$. We also notice the recapitalization control is actually an impulse control with fixed cost so that value function is not twice continuously differentiable at $u_1$.
The numerical solution of $(u_1, u_2)$ and $V(X)$ in equation $(\ref{value of recap})$ is illustrated in Figure \ref{figure recap_fullymodelvalue}. Function $H(X,u_2)$ is more concave than $f_1(X;u_2)$ and they are tangent at the recapitalization barrier $u_1$.  

\begin{figure}[H]
\centering
  \includegraphics[height=7.2cm,width=8.4cm]{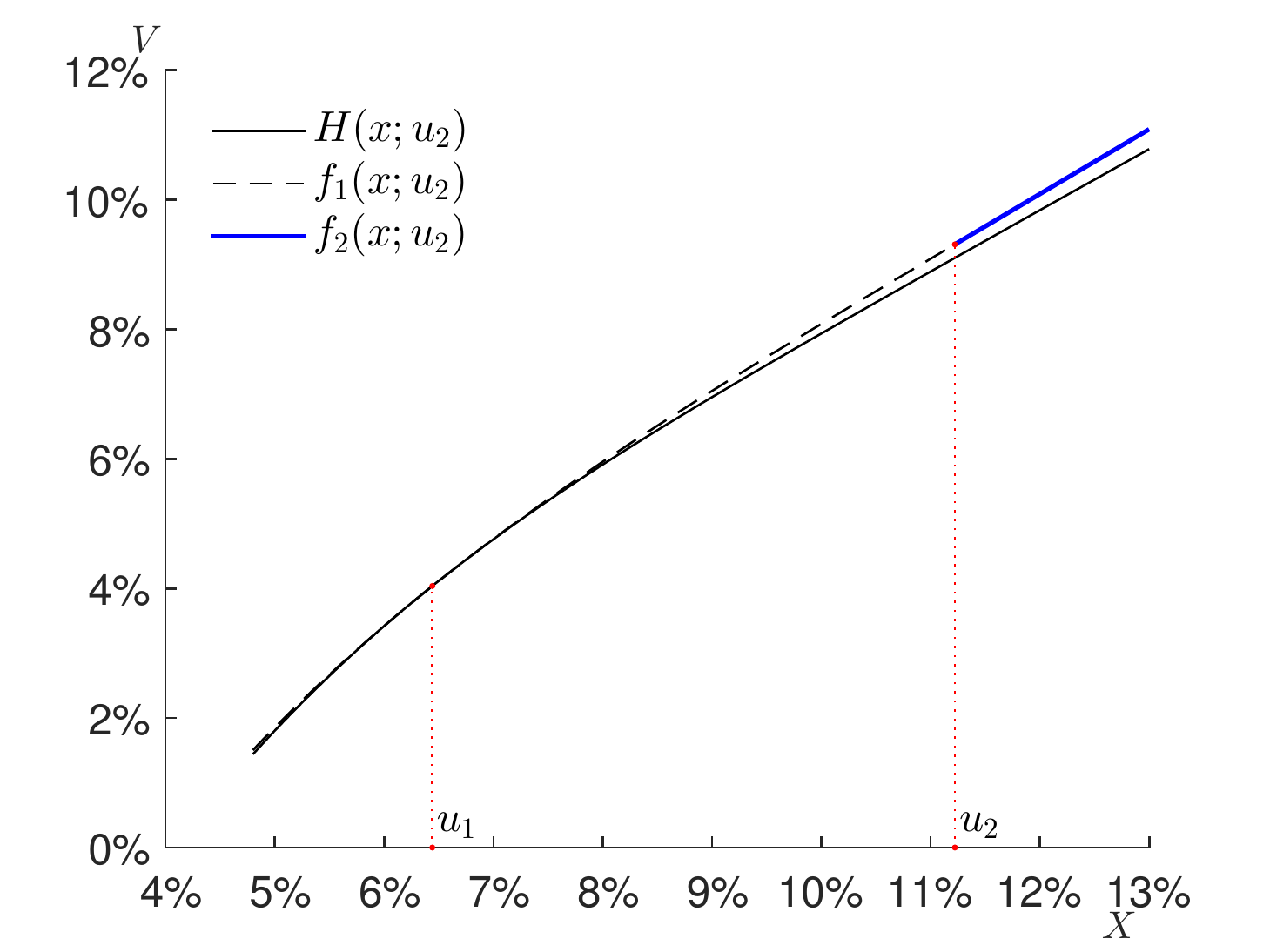}
 \caption[Three regions of the fully observed model]{\textbf{Three regions of the fully observed model.} This figure shows the components of the value function (\ref{value of recap}) under the parameter estimates in Table \ref{table  allbankparameter}. So the parameter values: $\alpha=11.59\%$, $\mu=10.52\%$, $\delta=23.30\%$, $\sigma=3.11\%$, $\kappa=4.80\%$, $\omega=31.50\%$, $\Delta=0.50$, and $K=0.20\%.$ The optimal equity issuance barrier $u_1$ is $6.44\%$, and the optimal dividend barrier $u_2$ is  $11.22\%$.  $V$ is the market value of equity and $X$ is the equity-to-debt ratio.}
 \label{figure recap_fullymodelvalue}
\end{figure}

\subsection{Partially Observed Model}\label{partially_theorysection}
First, we show that under the assumption (\ref{discounta}), $\hat \varphi(\hat E,D,S)$  in (\ref{define of value function with recap and partially observed}) is finite and satisfies the following growth condition (proof is in Appendix \ref{appendix_linear growth condition}).
\begin{proposition}\label{partial_growth condition}
Assume that (\ref{discounta}) holds. The value function $\hat \varphi(\hat E,D,S)$  in (\ref{define of value function with recap and partially observed})  satisfies the following growth condition
\begin{eqnarray*}\begin{aligned}
\hat E - C_0 D \leq \hat \varphi(\hat E,D, S) \leq  \hat E +D(C_1 +C_2S) \quad \text{in } \hat E/D>I(S),\ D>0,\ 0<S<\bar{S}
\end{aligned}\end{eqnarray*}
for some positive constants $C_0$, $C_1$, and $C_2$.
\end{proposition}

Problem $(\ref{define of value function with recap and partially observed})$ is a mixture of singular control and impulse control problems (see, e.g., 
\cite{OS02}, 
\cite{BP09},
and \cite{ARS17}). The corresponding value function is associated with the following HJB equation:
 \begin{eqnarray}\label{hjb2}\max\left\{\left[\mathcal{A}+\left(\sigma^2-\left(S/m+\sigma\rho\right)^2\right)\frac{\partial}{\partial S}\right] \hat \varphi, \quad 1- \frac{\partial }{\partial \hat E}\hat \varphi, \quad \mathcal{M}\hat \varphi-\hat \varphi\right\} =0 \end{eqnarray}
in $\hat E/D>I(S),\ D> 0,\ 0<S< \bar{S}$, with boundary condition: $\hat \varphi\left(I(S)D, D, S\right)=\omega D\psi(I(S), S),\ \forall D>0,0<S<\bar{S}$, where the operators $\mathcal{A}$ and $\mathcal{M}$ are defined as
\begin{eqnarray*} 
&&\mathcal{A}=\frac{1}{2}(\hat E+D)^2\left(S/m+\sigma\rho\right)^2\frac{\partial^2}{\partial \hat E^2}+\left(\alpha\hat E+(\alpha-\mu) D \right) \frac{\partial}{\partial \hat E}+\mu D\frac{\partial}{\partial D}-\delta,\\
&&\mathcal{M} \hat \varphi(\hat E, D, S)=\sup_{s\in(0,\bar{s}D_\Delta)}
\mathbb{E}^{ \hat E, D,S}\Big [ e^{-\delta\Delta}  \left( \hat \varphi(\hat E_{\Delta}+s,D_\Delta, S_\Delta)-s-KD_\Delta \right) \mathbf{1}_{\{\hat \tau>\Delta\}} \\
& &\quad \quad \quad \quad \quad \quad \quad  \quad \quad \quad \quad  \quad +e^{-\delta\hat \tau} \omega D_{\hat \tau} \psi(I(S_{\hat \tau}), S_{\hat \tau})\mathbf{1}_{\{\hat \tau\leq \Delta\}} \Big],
\end{eqnarray*}
where $\psi(x, y):=\frac{1}{\sqrt{2\pi}}\int_{\mathbb{R}} \left[(x+1)e^{-\frac{y}{2}+u \sqrt{y}}-1\right]^{+} e^{-\frac{u^2}{2}} du$ (see Appendix \ref{Appendix A1} for the derivation of $\psi(x,y)$), $\mathbb{E}^{ \hat E, D,S}$ is the expectation conditionally on  $\hat E_0=\hat E,\ D_0=D,\ S_0=S$ with new equity  to be issued at time $\Delta$, and $\hat \tau =\inf\{t\in [0,\Delta]: \hat E_t/D_t \leq I(S_t)\}$ is the first stopping time for $\hat E_t/D_t$ hitting $I(S_t)$ during $[0,\Delta]$.

It is easy to verify that the value function $\hat \varphi(\hat E, D, S)$ in $(\ref{define of value function with recap and partially observed})$ is homogeneous in $\hat{E}$ and $D$: \begin{eqnarray*}\hat \varphi(\gamma\hat E, \gamma D, S)=\gamma \hat \varphi(\hat E, D, S), \quad \forall \gamma>0,
 \end{eqnarray*}
which implies the following reduction:
\begin{eqnarray}\label{v(x,S)}
\hat V (\hat X, S):= \hat \varphi(\hat E, D, S)/D, \  \  \hat X=\hat E/ D.
\end{eqnarray}
By $(\ref{debt})$ and $(\ref{hat E_t})$, $\hat X_t=\mathbb{E}[X_t| \mathcal{G}_t] $ is the expected equity to debt ratio satisfying \begin{align}
 \label{hat X_t}
\hat X_t^\pi=&\hat X+\int_0^t(\hat X_u^\pi+1)(\alpha-\mu)du+\int_0^t (\hat X_u^\pi+1)\left(S_u/m+\sigma\rho\right) d\tilde{\mathcal{B}}_u \nonumber\\
&-L_t^\pi+\sum_i s^\pi_i \mathbf{1}_{\{t_i^\pi+\Delta\leq t\}}
\end{align}
with $\hat X_0=\hat X.$ Thus, the volatility of $\hat X_t$ is $S_t/m+\sigma\rho$.
For simplicity, we still use $\pi=\{L_t^\pi, (s_i^\pi, t_i^\pi)\}$ to represent an admissible strategy for $\hat X_t$, but here $L_t^\pi$ and $s_i^\pi$ are the sizes of cumulative dividends and recapitalizations in terms of debt $D$. The admissible strategy satisfies $\hat X_t^{\pi}\geq I(S_t)$ and $s_i^\pi \leq \bar{s}$ for  $i\geq 1$.

By $(\ref{v(x,S)})$, HJB equation (\ref{hjb2}) can be reduced to \begin{eqnarray}\label{hjb3}\max\left\{\left[\mathcal{L}+\left(\sigma^2-\left(S/m+\sigma\rho\right)^2\right)\frac{\partial}{\partial S}\right] \hat V, \quad 1- \frac{\partial }{\partial \hat X}\hat V, \quad  \mathcal{P}\hat V-\hat V\right\} =0 \end{eqnarray}
in $\Omega:=\{(\hat X, S): \hat X>I(S), S>0 \}$,
 with boundary condition: $\hat V\left(I(S), S\right)=\omega \psi(I(S),S),\ \forall S>0$, where
\begin{eqnarray*}
&&\mathcal{L}=\frac{1}{2}(1+\hat X)^2\left(S/m+\sigma\rho\right)^2\frac{\partial^2}{\partial \hat X^2}+(\alpha-\mu)(1+\hat X)\frac{\partial}{\partial \hat X}-(\delta-\mu),\\
&&\mathcal{P} \hat V(\hat X,S)=\sup_{s\in(0,\bar s)} \mathbb{E}\Big[e^{-(\delta-\mu)\Delta}  [\hat V(\hat X_{\Delta}+s,S_\Delta)-s-K]\mathbf{1}_{\{\hat{\tau}>\Delta\}} \\
&&\quad \quad\quad \quad \quad \quad  +e^{-(\delta-\mu) \hat \tau} \omega  \psi(I(S_{\hat \tau}), S_{\hat \tau})\mathbf{1}_{\{\hat \tau\leq \Delta\}} \Big],
\end{eqnarray*}
and $\hat X_t$ follows $d\hat X_t=(\alpha-\mu)(\hat X_t+1)dt+(S_t/m+\sigma\rho)(\hat X_t+1) d\tilde{\mathcal{B}}_t$ with $X_0=\hat X$, and $\hat \tau=\inf\{t\in [0,\Delta]: \hat X_t \leq I(S_t)\}$ is the first stopping time for $\hat X_t$ hitting $I(S_t)$ during $[0,\Delta]$.

Value functions $(\ref{define of value function with recap and partially observed})$ and $(\ref{v(x,S)})$ are equivalent since $D_t$ is deterministic. Same folds for the HJB equations $(\ref{hjb2})$ and $(\ref{hjb3})$. 
The classical dynamic programming principle requires on the first stage that the value function be measurable. 
 However, our model is nontrivial in the presence of time delay. The value function defined has no priori regularity such as measurability and continuity (which are very difficult to prove in advance).  In the literature, \cite{BT11} and \cite{T12} (Theorem 3.3 in Chapter 3 and Theorem 4.3 in Chapter 4) study the weak dynamic programming principle, which avoids the technical difficulties
related to the measurable selection argument. \cite{ARS17} apply the weak dynamic programming principle for optimal consumption and investment with fixed transaction costs because the value function is known discontinuous at the boundary of the solvency region.  Due to the measurability issue, we give a weak dynamic programming principle (DPP) for value function $\hat V(\hat X,S)$ as defined in (\ref{v(x,S)}). Denote the lower and upper semicontinuous envelope by
\begin{eqnarray}\label{envelope definition}
 \begin{aligned}
\hat V_{*}(\hat X, S) &= \underset{(\hat X', S')\rightarrow (\hat X, S)}{ \lim\inf} \hat V(\hat X', S'), \ \    \hat V^{*}(\hat X, S) &= \underset{(\hat X', S')\rightarrow (\hat X, S)}{ \lim\sup} \hat V(\hat X', S').
\end{aligned}\end{eqnarray}
We have the following weak DPP for $\hat V(\hat X,S)$ (proof is in Appendix \ref{appendix_proof_wdpp}).
\begin{proposition} \label{prop_wdpp}
The value function $\hat V$ satisfies the weak DPP, i.e. for any stopping time $\theta$,
{\fontsize{10.5pt}{10.5pt}
\begin{eqnarray*} \begin{aligned}
&\hat  V (\hat X,S) \leq \sup_{\pi\in\Pi}\mathbb{E}\Bigg[ \int_0^{ \theta} e^{-\delta_1 u }dL^{\pi}_u-\sum_i e^{-\delta_1 (t^{\pi}_i+\Delta)}\left(s^{\pi}_i+K  \right)  \mathbf{1}_{\{ t^{\pi}_i+\Delta< \theta\}} \quad \quad \nonumber\\
&+ \mathbb{E}\Big[e^{-\delta_1  (t^{\pi}_{k_{\theta}^{\pi}}+\Delta)}  \Big( \hat V^*(\hat X_{t^{\pi}_{k_{\theta}^{\pi}}+\Delta} +s^{\pi}_{k_{\theta}^{\pi}},S_{t^{\pi}_{k_{\theta}^{\pi}}+\Delta}) -s^{\pi}_{k_{\theta}^{\pi}}-K\Big)\Big| \mathcal{G}_{\theta} \Big] \mathbf{1}_{\{\theta\geq t_{k_\theta^{\pi}}^{\pi},t_{k_\theta^{\pi}}^{\pi}+\Delta <  \hat \tau^\pi  \}} \\
& +\mathbb{E}\Big[e^{-\delta_1 \hat\tau^\pi} \omega X_{\hat \tau^{\pi}}^+  \Big| \mathcal{G}_{\theta} \Big] \mathbf{1}_{\{\theta\geq t_{k_\theta^{\pi}}^{\pi},t_{k_\theta^{\pi}}^{\pi}+\Delta \geq \hat \tau^\pi \}} + e^{-\delta_1  \theta } \hat V^{*} (\hat X_\theta,S_\theta)\mathbf{1}_{\{\theta<t_{k_\theta^{\pi}}^{\pi}\wedge \hat\tau^\pi \}} + e^{-\delta_1 \hat\tau^\pi }\omega X_{\hat \tau^{\pi}}^+ \mathbf{1}_{\{\hat \tau^\pi \leq \theta<t_{k_\theta^{\pi}}^{\pi}\}}
 \Bigg],  \\
&\hat  V (\hat X,S) \geq \sup_{\pi\in\Pi}\mathbb{E}\Bigg[ \int_0^{ \theta} e^{-\delta_1 u }dL^{\pi}_u-\sum_i e^{-\delta_1 (t^{\pi}_i+\Delta)}
\left(s^{\pi}_i+K  \right)  \mathbf{1}_{\{ t^{\pi}_i+\Delta< \theta\}} \nonumber\\
&+ \mathbb{E}\Big[e^{-\delta_1  (t^{\pi}_{k_{\theta}^{\pi}}+\Delta)}  \Big( \hat V_*(\hat X_{t^{\pi}_{k_{\theta}^{\pi}}+\Delta} +s^{\pi}_{k_{\theta}^{\pi}},S_{t^{\pi}_{k_{\theta}^{\pi}}+\Delta}) -s^{\pi}_{k_{\theta}^{\pi}}-K\Big)\Big| \mathcal{G}_{\theta} \Big] \mathbf{1}_{\{\theta\geq t_{k_\theta^{\pi}}^{\pi},t_{k_\theta^{\pi}}^{\pi}+\Delta <  \hat \tau^\pi  \}} \\
& +\mathbb{E}\Big[e^{-\delta_1 \hat\tau^\pi} \omega X_{\hat \tau^{\pi}}^+  \Big| \mathcal{G}_{\theta} \Big] \mathbf{1}_{\{\theta\geq t_{k_\theta^{\pi}}^{\pi},t_{k_\theta^{\pi}}^{\pi}+\Delta \geq \hat \tau^\pi \}} + e^{-\delta_1  \theta } \hat V_{*} (\hat X_\theta,S_\theta)\mathbf{1}_{\{\theta<t_{k_\theta^{\pi}}^{\pi}\wedge \hat\tau^\pi \}} + e^{-\delta_1 \hat\tau^\pi }\omega X_{\hat \tau^{\pi}}^+ \mathbf{1}_{\{\hat \tau^\pi \leq \theta<t_{k_\theta^{\pi}}^{\pi}\}} \Bigg],
  \end{aligned}\end{eqnarray*}}
where $k_{\theta}^{\pi} : = \min\{i:  \theta \leq t_{i}^{\pi} +\Delta \}$ and $t_i^\pi = +\infty$ if $t_i^\pi\geq \hat \tau^\pi$.
\end{proposition}

Based on the weak DPP, we then give the viscosity characterization for the partially observed model.

\begin{theorem}\label{partial_viscosity property}
Assume $m> 0$ and $\Delta>0$. The HJB equation (\ref{hjb3}) has a unique viscosity solution $\hat V=\hat V(\hat X,S)$. Define $\hat\varphi(\hat E,D,S)= D\hat V(\hat E/D,S)$. Then $\hat\varphi$ is the value function of the partially observed model. 
\end{theorem}

\textit{Remark}. For the partially observed model, to give a similar verification theorem as in Theorem \ref{semi-explicit solution with recap}, we need some regularity results of the value function and the resulting free boundary which, however, are notoriously difficult to prove. Hence, we only describe the optimal strategy in terms of the dividend region $\mathbf{DR}$, continuation region $\mathbf{CR}$, and recapitalization region $\mathbf{RR}$ as defined below:
\begin{eqnarray*}
\mathbf{DR}:&=&\left\{(\hat X,S)\in \Omega: 1- \frac{\partial }{\partial \hat X}\hat V=0\right\},\\
\mathbf{CR}:&=& \left\{(\hat X,S)\in \Omega: \left[\mathcal{L}+\left(\sigma^2-\left(S/m+\sigma\rho\right)^2\right)\frac{\partial}{\partial S}\right] \hat V=0 \right\},\\
\mathbf{RR}:&=& \left\{(\hat X,S)\in \Omega: \mathcal{P}\hat V-\hat V=0 \right\}.
\end{eqnarray*}

\begin{proof} Based on the weak DPP, we now prove that the value function $\hat V(\hat X,S)$ in (\ref{v(x,S)})  is a viscosity solution to HJB (\ref{hjb3}) in the solvency region $\Omega$.
Define
\begin{eqnarray}\label{partial_F}
F(\hat X,S,v,Dv,D^2v): &=
&\min\Big\{F_1(\hat X,S,v,Dv,D^2v), \quad D_{\hat X}v-1 \Big\},\nonumber\\
F_1(\hat X,S,v,Dv,D^2v) : &= & \delta_1 v -(\alpha-\mu)(1+\hat X) D_{\hat X}v+ \left((S/m+\sigma\rho)^2-\sigma^2\right) D_S v \nonumber\\
&&- \frac{1}{2}(1+\hat X)^2\left(S/m+\sigma\rho\right)^2 D_{\hat X\hat X} v,
\end{eqnarray}
where $Dv=[D_{\hat X}v, D_{S}v]^T$ is the Jacobian matrix of $v$ and  $D^2 v=\left(\begin{array}{cc} D_{\hat X\hat X}v & D_{\hat X S} v\\
 D_{S \hat X}v &D_{SS} v\end{array} \right)$ is the Hessian matrix of $v$.

When $S_0=m \sigma (1-\rho)$, we have $S_t\equiv m \sigma (1-\rho)$ for all $t\geq 0$. In this case, by Theorem \ref{semi-explicit solution with recap}, we know $\hat V(\hat X; m \sigma (1-\rho)) = V(\hat X; \kappa_1, \omega_1)$, where the new liquidation barrier $\kappa_1=I(m\sigma(1-\rho))$ and the new liquidation value $\omega_1= \frac{\omega }{\kappa_1} \psi(I(m \sigma (1-\rho)),m \sigma (1-\rho))= \frac{\omega }{\sqrt{2\pi}\kappa_1}\int_{\mathbb{R}} \left[(\kappa_1+1)e^{-\frac{m \sigma (1-\rho)}{2}+u \sqrt{m \sigma (1-\rho)}}-1\right]^{+} e^{-\frac{u^2}{2}} du$.
This implies that we are able to restrict attention to $\Omega_1=\{(\hat X, S): \hat X>I(S), 0<S<m\sigma(1-\rho) \}$ and $\Omega_2=\{(\hat X, S): \hat X>I(S), m\sigma(1-\rho) <S <\bar{S} \}$ respectively. Without loss of generality, we only consider the case in $\Omega_1$, that is,
\begin{eqnarray}\label{p_hjb33}\left\{ \begin{aligned}
&\max\left\{F_1(\hat X,S,\hat V,D\hat V,D^2\hat V), \quad 1- D_{\hat X}\hat V, \quad  \mathcal{P}\hat V-\hat V\right\} =0,  \text{ in } \Omega_1,\\
&\hat V(\hat X,S)=\omega \psi(I(S),S), \text{ when } \hat X=I(S), \\
& \hat V(\hat X,S)= V(\hat X; \kappa_1, \omega_1), \text{ when } S=m\sigma(1-\rho).
\end{aligned}\right.
\end{eqnarray}

$\textbf{Definition of viscostiy solution}$. Throughout, $v_*$ and $v^*$ will denote the lower and upper-semicontinuous envelops of a local bounded function $v$, as defined in (\ref{envelope definition}).  A function $v$ is a viscosity solution to Problem (\ref{p_hjb33}) if $v$ is both a viscosity supersolution and a viscosity subsolution. That is,\\
(a) \textbf{Viscosity supersolution}: for any $(\hat X_0,S_0)\in \Omega_1$, any $\varphi\in C^{2,1}(\Omega_1)$ with $v_*\geq \varphi$ and $v_*-\varphi$ attains its minimum $0$ at $(\hat X_0,S_0)$, i.e $0=(v_*-\varphi)(\hat X_0,S_0)=\underset{(\hat X,S)\in \Omega_1}{\min} (v_*-\varphi)$, we have
\begin{eqnarray*}\left\{\begin{aligned}
&  \min\{F_1(\hat X,S, v_*,D\varphi,D^2\varphi), D_{\hat X}\varphi-1, v_{*}-\mathcal{P}v_{*}\}\geq 0,  \text{ at point }(\hat X_0,S_0) \in \Omega_1,\\
&   v(\hat X_0,S_0)\geq V(\hat X_0; \kappa_1, \omega_1),  \text{ if } S_0=m\sigma(1-\rho), \\
& v(\hat X_0, S_0) \geq \omega  \psi(I(S),S), \text{ if } \hat X_0 = I(S).
\end{aligned}\right.
\end{eqnarray*}
(b)  \textbf{Viscosity subsolution}: for any $(\hat X_0,S_0)\in \Omega_1$, any $\varphi\in C^{2,1}(\Omega_1)$ with $v^*\leq \varphi$ and $v^*-\varphi$ attains its maximum $0$ at $(\hat X_0,S_0)$, i.e $0=(v^*-\varphi)(\hat X_0,S_0)=\underset{(\hat X,S)\in \Omega_1}{\max} (v^*-\varphi)$, we have
\begin{eqnarray*}\left
\{\begin{aligned}
&  \min\{F_1(\hat X,S, v^*,D\varphi,D^2\varphi), D_{\hat X}\varphi-1, v^{*}-\mathcal{P}v^{*}\}\leq  0, \text{ at point }(\hat X_0,S_0) \in \Omega_1,\\
&   v(\hat X_0,S_0)\leq V(\hat X_0; \kappa_1, \omega_1),  \text{ if } S_0=m\sigma(1-\rho),\\
& v(\hat X_0, S_0) \leq \omega  \psi(I(S),S), \text{ if } \hat X_0 = I(S).
\end{aligned}\right.
\end{eqnarray*}

To prove Theorem \ref{partial_viscosity property}, we first prove two lemmas below to verify that $\hat V(\hat X,S)$ in (\ref{v(x,S)}) is a viscosity supersolution and suobsolution to Problem (\ref{p_hjb33}).
\begin{lemma} \label{supersolution} The value function $\hat V$ in (\ref{v(x,S)}) is a viscosity supersolution to Problem (\ref{p_hjb33}).
\end{lemma}
\noindent\textit{Proof of Lemma \ref{supersolution}}:
Choose any $(\hat X_0,S_0)\in \Omega_1, \varphi\in C^{2,1}(\Omega_1)$ such that $(\hat V_*-\varphi)(\hat X_0,S_0)= 0$ and $\hat V_* \geq \varphi$ in $\Omega_1$.  We take any admissible control $L_0=l$ with $l>0$ at time 0: \begin{eqnarray*}\varphi(\hat X_0,S_0)=\hat V(\hat X_0,S_0)\geq   \hat V(\hat X_0-l,S_0)+l\geq \varphi(\hat X_0-l,S_0)+l.
\end{eqnarray*}
Sending $l\rightarrow 0$, hence, $D_{\hat X}\varphi(\hat X_0,S_0)-1\geq0 $.

By the weak dynamic programming principle, if the shareholders choose to issue equity at initial time, then we have $\hat V(\hat X_0,S_0)\geq \mathcal{P} \hat V_{*}(\hat X_0,S_0)$. Denote by $\hat X_t^{\hat X_n,S_n}$ and $S_t^{S_n}$ the process of expected equity-to-debt ratio and uncertainty level starting from initial point $(\hat X_n,S_n)$.
By Fatou's Lemma, tower property, and the property of lower semi-continuous envelop $\hat V_{*}$, i.e.,
$\underset{(\hat X,S)\to (\hat X_0,S_0)}{\underline{\lim}} \hat V_{*}(\hat X,S) \geq \hat V_{*}(\hat X_0,S_0),$
we have\\
{\fontsize{10.5pt}{10.5pt}
\begin{equation*}\begin{aligned}
& \hat V_{*}(\hat X_0,S_0)\geq  \underset{n\rightarrow \infty}{\underline{\lim}} \mathcal{P} \hat V_{*}(\hat X_n,S_n) \\
= & \underset{n\rightarrow \infty}{\underline{\lim}} \sup_{s\in(0,\bar{s})} \mathbb{E}\Big[e^{-\delta_1\Delta} [\hat V_*(\hat X_{\Delta}^{\hat X_n,S_n}+s,S_\Delta^{S_n})-s-K]\mathbf{1}_{\{\hat{\tau}>\Delta\}} +e^{-\delta_1 \hat \tau} \omega  \psi(I(S_{\hat \tau}^{S_n}), S_{\hat \tau}^{S_n})\mathbf{1}_{\{\hat \tau\leq \Delta\}} \Big]\\
 \geq &\sup_{s\in(0,\bar{s})}  \underset{n\rightarrow \infty}{\underline{\lim}}  \mathbb{E}\Big[e^{-\delta_1\Delta} [\hat V_*(\hat X_{\Delta}^{\hat X_n,S_n}+s,S_\Delta^{S_n})-s-K]\mathbf{1}_{\{\hat{\tau}>\Delta\}} +e^{-\delta_1 \hat \tau} \omega  \psi(I(S_{\hat \tau}^{S_n}), S_{\hat \tau}^{S_n})\mathbf{1}_{\{\hat \tau\leq \Delta\}} \Big]\\
  \geq &\sup_{s\in(0,\bar{s})}   \mathbb{E}\Big[e^{-\delta_1\Delta} \mathbb{E}\big[\underset{n\rightarrow \infty}{\underline{\lim}} [\hat V_*(\hat X_{\Delta}^{\hat X_n,S_n}+s,S_\Delta^{S_n})-s-K]\mathbf{1}_{\{\hat{\tau}>\Delta\}} \big|\mathcal{G}_{\Delta}\big] +e^{-\delta_1 \hat \tau} \omega  \psi(I(S_{\hat \tau}^{S_0}), S_{\hat \tau}^{S_0})\mathbf{1}_{\{\hat \tau\leq \Delta\}} \Big],
 \end{aligned}
\end{equation*}}
for any sequence $(\hat X_n,S_n)$ converging to $(\hat X_0,S_0)$. This gives $ \hat V_{*}(\hat X_0,S_0)  \geq  \mathcal{P} \hat V_{*}(\hat X_0,S_0)$.

It remains to show that $F_1(\hat X_0,S_0, \hat V_*,D\varphi,D^2\varphi) \geq 0$. By the definition of $\hat V_*$,
we choose the sequence $\{\gamma_n\}_n$ and $\{h_n\}_n$ where
\begin{eqnarray*}
\gamma_n: = \hat V(\hat X_n, S_n) - \varphi(\hat X_n,S_n) \rightarrow 0, \quad h_n>0, \quad h_n, \frac{\gamma_n}{h_n} \rightarrow 0, \quad \hat X_n \rightarrow \hat X_0, S_n \rightarrow S_0.
\end{eqnarray*}
Now choose small $\varepsilon$ to guarantee $B(\hat X_n,S_n; \varepsilon):=\{(\hat X,S): |\hat X-\hat X_n|<\varepsilon, |S-S_n|<\varepsilon\}\subseteq \Omega_1$  and a sequence of admissible strategy $\pi_n \in \Pi((\hat X_n, S_n))$ with liquidation time $\hat \tau^{\pi_n}$ satisfying that
there is no dividend payment and equity issuance before time $\tau_n$, where
\begin{eqnarray*}\tau_n:=h_n \wedge \inf\{t\geq0: (\hat X_t^{\hat X_n,S_n,\pi_n}, S_t^{S_n})\notin B(\hat X_n,S_n; \varepsilon) \},
\end{eqnarray*}
 and $\hat X_t^{\hat X_n,S_n,\pi_n}$ denotes the expected equity-to-debt ratio under policy $\pi_n$ starting from $(\hat X_n,S_n)$.
Applying Ito's formula and Dynkin's formula to smooth function $\varphi$  and noting that the integrand in the stochastic integral is bounded during $[0,\tau_n]$, we get
\begin{eqnarray*}\begin{aligned}
&\hat V(\hat X_n, S_n)-\gamma_n=\varphi(\hat X_n,S_n) \\
=& \mathbb{E}\left[e^{-\delta_1\tau_n}\varphi(\hat X_{\tau_n}^{\hat X_n,S_n,\pi_n},S_{\tau_n}^{S_n}) +\int_0^{\tau_n} e^{-\delta_1t}F_1(\hat X_t^{\hat X_n,S_n,\pi_n},S_t^{S_n}, \varphi,D\varphi,D^2\varphi) dt\right].
\end{aligned}\end{eqnarray*}
Therefore, by the weak dynamic programming principle, $\hat V_{*}\geq \varphi$,  and $\delta_1>0$,
 \begin{eqnarray*}\begin{aligned}
\hat V(\hat X_n, S_n)
\geq & \mathbb{E}\left[ e^{-\delta_1\tau_n}\hat V_{*}(\hat X_{\tau_n}^{\hat X_n,S_n,\pi_n},S_{\tau_n}^{S_n})\right]
\geq  \mathbb{E}\left[e^{-\delta_1\tau_n}\varphi(\hat X_{\tau_n}^{\hat X_n,S_n,\pi_n},S_{\tau_n}^{S_n})\right]\\
=&\hat V(\hat X_n, S_n)-\gamma_n-\mathbb{E}\left[\int_0^{\tau_n} e^{-\delta_1t} F_1(\hat X_t^{\hat X_n,S_n,\pi_n},S_t^{S_n}, \varphi,D\varphi,D^2\varphi) dt\right].
\end{aligned}
\end{eqnarray*}
It follows that \begin{equation*}
\frac{1}{h_n}\mathbb{E}\left[\int_0^{\tau_n} e^{-\delta_1t} F_1(\hat X_t^{\hat X_n,S_n,\pi_n},S_t^{S_n}, \varphi, D\varphi,D^2\varphi) dt \right]+\frac{\gamma_n}{h_n}\geq 0.
\end{equation*}
Notice that $\tau_n = h_n$ when $n$ is large enough.
Sending $n\rightarrow\infty$ and using the mean-value theorem and the dominated convergence theorem, we conclude that $F_1(\hat X_0,S_0, \varphi,D\varphi,D^2\varphi) \geq0.$ Thanks to $\varphi(\hat X_0,S_0) = \hat V_*(\hat X_0,S_0)$, this completes the proof of Lemma \ref{supersolution}.

\begin{lemma}\label{subsolution} The value function $\hat V(\hat X,S)$ in (\ref{v(x,S)}) is a viscosity subsolution to Problem (\ref{p_hjb33}).
\end{lemma}
\noindent\textit{Proof of Lemma \ref{subsolution}}: Choose any $(\hat X_0,S_0)\in \Omega_1, \varphi\in C^{2,1}(\Omega_1)$ such that $(\hat V^*-\varphi)(\hat X_0,S_0)= 0$ and $\hat V^*(\hat X_0,S_0) \leq\varphi(\hat X_0,S_0)$ on $\Omega_1$.  If $\hat V^*(\hat X_0,S_0)\leq \mathcal{P}\hat V^*(\hat X_0,S_0)$, the subsolution inequality apparently holds. If $\hat V^*(\hat X_0,S_0) > \mathcal{P}\hat V^*(\hat X_0,S_0)$, we prove it by contradiction. Assume, to the contrary, that
\begin{equation}\label{partial_viscosity_contradiction0}
\lambda:=\min\{ F_1(\hat X_0,S_0, \hat V^*,D\varphi,D^2\varphi), \hat V^*(\hat X_0,S_0)-\mathcal{P} \hat V^{*} (\hat X_0,S_0), D_{\hat X}\varphi(\hat X_0,S_0)-1\} >0.
\end{equation}
Define
\begin{equation}
\begin{aligned}\label{partial_Ps}
 &\mathcal{P}_{s} V(\hat X,S): = \mathbb{E}\Big[e^{-\delta_1\Delta} [  V(\hat X_{\Delta}^{\hat X,S}+s,S_\Delta^{S})-s-K]\mathbf{1}_{\{\hat{\tau}>\Delta\}}  +e^{-\delta_1 \hat \tau} \omega  \psi(I(S_{\hat \tau}^{S}), S_{\hat \tau}^{S})\mathbf{1}_{\{\hat \tau\leq \Delta\}} \Big] \\
 & = \mathbb{E}\Big[ \mathbb{E}\big[e^{-\delta_1\Delta} [ V(\hat X_{\Delta}^{\hat X,S}+s,S_\Delta^{S})-s-K]\mathbf{1}_{\{\hat{\tau}>\Delta\}} \big|\mathcal{G}_{\Delta}\big]  +e^{-\delta_1 \hat \tau} \omega  \psi(I(S_{\hat \tau}^{S}), S_{\hat \tau}^{S})\mathbf{1}_{\{\hat \tau\leq \Delta\}} \Big]. 
\end{aligned}\end{equation}
Now we show that for any admissible issuance $s \in \mathcal{G}_{\Delta}$, $\mathcal{P}_{s}\hat V^*$ is upper semicontinuous.
Let $(\hat X_n,S_n)\rightarrow (\hat X_0,S_0)$, then we have $(\hat X_\Delta^{\hat X_n,S_n},  S_\Delta^{S_n})\rightarrow (\hat X_\Delta^{\hat X_0,S_0},  S_\Delta^{S_0})$ under $\mathcal{G}_\Delta$. By the linear growth property of $\hat V$, the formula in expectation in (\ref{partial_Ps}) has an upper bound $e^{-\delta_1\Delta} \hat X_\Delta^{\hat X_n,S_n}+C_0+C_1S_\Delta^{S_n}$.
Note that 
\begin{eqnarray*}
e^{-\delta_1\Delta}\hat X_{\Delta}^{\hat X_n,S_n} = (\hat X_n+1)e^{-(\delta-\alpha)\Delta - \frac{1}{2}\int_0^\Delta \left(\frac{S_t^{S_n}}{m}+\sigma\rho\right)^2dt + \int_0^\Delta \left(\frac{S_t^{S_n}}{m}+\sigma\rho\right)d\tilde{\mathcal{B}}_t }-e^{-\delta_1\Delta}
\end{eqnarray*}
and the bounded sequence $S_t^{S_n}/m+\sigma\rho \rightarrow S_t^{S_0}/m+\sigma\rho$ for all $t\in[0,\Delta]$ when $n\rightarrow \infty$.
Hence we have $\underset{n\rightarrow\infty}{\lim}\mathbb{E}[e^{-\delta_1\Delta} \hat X_\Delta^{\hat X_n,S_n}+C_0+C_1S_\Delta^{S_n}] = \mathbb{E}[e^{-\delta_1\Delta} \hat X_\Delta^{\hat X_0,S_0}+C_0+C_1S_\Delta^{S_0}]$.
Applying the Fatou's Lemma, we then obtain
$\mathcal{P}_s \hat V^*(\hat X_0,S_0)\geq \mathcal{P}_s ( \underset{n\rightarrow \infty}{\lim\sup}  \hat V^*(\hat X_n,S_n)) \geq \underset{n\rightarrow \infty}{\lim\sup}  \mathcal{P}_s \hat V^*(\hat X_n,S_n)$. Thus $\mathcal{P}_{s}\hat V^*$ is upper semicontinuous.

For each $(\hat X_n,S_n)$ and any small $\varepsilon>0$, there exists $s_n \in \mathcal{G}_{\Delta}$ depending on $\hat X_\Delta^{\hat X_n,S_n}$ such that $\mathcal{P}\hat V^*(\hat X_n,S_n) < \mathcal{P}_{s_n} \hat V^*(\hat X_n,S_n)+\varepsilon$.
Using the upper semicontinuity property of $\mathcal{P}_{s}\hat V^*$, we obtain
\begin{equation*}
\begin{aligned}
\mathcal{P}\hat V^*(\hat X_0,S_0) &=\sup_{s\in(0,\bar{s})} \mathcal{P}_s \hat V^*(\hat X_0,S_0) \geq \sup_{s\in(0,\bar{s})} \underset{n\rightarrow \infty}{\lim\sup}  \mathcal{P}_s \hat V^*(\hat X_n,S_n)\\
 &\geq \underset{n\rightarrow \infty}{\lim\sup}  \mathcal{P}_{s_n} \hat V^*(\hat X_n,S_n) \geq \underset{n\rightarrow \infty}{\lim\sup}  \mathcal{P} \hat V^*(\hat X_n,S_n) -\varepsilon.
\end{aligned}
\end{equation*}
Therefore, $\mathcal{P}\hat V^*$ is also upper semicontinuous.
Because $\hat V^*$ and $\mathcal{P}\hat V^*$ are upper-semicontinuous and $\varphi\geq \hat V^*$, we claim that there exists $\varepsilon>0$ such that
\begin{equation}\label{partial_viscosity_contradiction}
\min\{ F_1(\hat X,S, \varphi,D\varphi,D^2\varphi), \quad \varphi(\hat X,S)-\mathcal{P} \hat V^{*}(\hat X,S), \quad D_{\hat X}\varphi(\hat X,S)-1\} \geq \frac{\lambda}{2},
\end{equation}
for all $(\hat X,S) \in B(\hat X_0,S_0; \varepsilon)\subseteq\Omega_1.$
Otherwise, if no such $B(\hat X_0,S_0; \varepsilon)$ exists, then there exists a sequence $(\hat X_n, S_n)\rightarrow (\hat X_0,S_0)$ and $\varphi(\hat X_n,S_n)-\mathcal{P}\hat V^*(\hat X_n,S_n)<\frac{\lambda}{2}$ for large $n$. Sending $n\rightarrow\infty$, we have $\varphi(\hat X_0,S_0)\leq \underset{n\rightarrow \infty}{\lim\sup}   \mathcal{P}\hat V^*(\hat X_n,S_n) +\frac{\lambda}{2}\leq \mathcal{P}\hat V^*(\hat X_0,S_0) +\frac{\lambda}{2}$, which contradicts (\ref{partial_viscosity_contradiction0}).
Now choose $(\hat X_n,S_n)\rightarrow (\hat X_0,S_0)$ in $\Omega_1$, such that $\hat V(\hat X_n,S_n)\rightarrow \hat V^*(\hat X_0,S_0)$. We have
\begin{eqnarray}\label{subsolution4}
 \varphi(\hat X_n,S_n)-\hat V(\hat X_n,S_n)\rightarrow 0.\end{eqnarray}
We can choose $\varepsilon$ small enough to gaurantee $B(\hat X_n; \varepsilon/2) : =\{\hat X: |\hat X -\hat X_n |< \varepsilon \} \subseteq (I(S),+\infty)$ for all $n$. For any admissible strategy $\pi_n\in  \Pi((\hat X_n,S_n))$, define a sequence of stopping time:
\begin{equation*}\begin{aligned}& \bar{\tau}_n:= R\wedge \inf\{t>0: \hat X_t^{\hat X_n,S_n,\pi_n} \notin B(\hat X_n; \varepsilon/2) \},\\
& \theta_n:=t_1^{\pi_n} \wedge \bar{\tau}_n,
\end{aligned}\end{equation*}
where $R>0$ is a fixed constant, and we set $t_1^{\pi_n}=+\infty$ if there is no equity issuance time determined before liquidation time $\hat\tau^{\pi_n}$.

By the definition of $\bar\tau_n$, we know $(\hat X_t^{\hat X_n,S_n,\pi_n},S_t^{S_n})$ is inside $\Omega_1$ during $[0,\theta_n)$ so that there is no liquidation before $\theta_n$.
By the  weak dynamic programming principle in Proposition \ref{prop_wdpp}, 
\begin{eqnarray}\label{subsolution1}
\begin{aligned}
&\hat V(\hat X_n,S_n) \leq \sup_{\pi_n\in \Pi((\hat X_n,S_n))}
\mathbb{E}\Bigg[\int_0^{ t_1^{\pi_n} \wedge \bar{\tau}_n} e^{-\delta_1t} dL_t^{\pi_n} +e^{-\delta_1\bar \tau_n}\hat V^*(\hat X_{\bar \tau_n}^{\hat X_n,S_n,\pi_n},S_{\bar \tau_n}^{S_n})\mathbf{1}_{\{\bar \tau_n <  t^{\pi_n}_1 \}} \\
& + e^{-\delta_1(t^{\pi_n}_1+\Delta)} \mathbb{E}\left[\hat V^*\big(\hat X_{t^{\pi_n}_1+\Delta}^{\hat X_n,S_n,\pi_n}+s_1^{\pi_n},S_{t_1^{\pi_n}+\Delta}^{S_n}\big)-s_{1}^{\pi_n} -K \Big| \mathcal{G}_{t^{\pi_n}_{1}}\right]\mathbf{1}_{\{ \bar \tau_n\geq t^{\pi_n}_{1}, t^{\pi_n}_{1}+\Delta <\hat\tau^{\pi_n} \}}\\
&+ e^{-\delta_1 \hat\tau^{\pi_n}} \mathbb{E}\Big[ \omega \left(X_{\hat\tau^{\pi_n}}^{\hat X_n,S_n,\pi_n} \right)^+\Big| \mathcal{G}_{t^{\pi_n}_{1}}  \Big] \mathbf{1}_{\{\bar \tau_n\geq  t^{\pi_n}_{1}, t^{\pi_n}_{1}+\Delta \geq \hat\tau^{\pi_n} \}}
\Bigg].
\end{aligned}\end{eqnarray}

Applying Ito's formula and Dynkin's formula to smooth function $\varphi$  and noting that the integrand in the stochastic integral is bounded during $[0,\theta_n)$, we get
\begin{eqnarray*}\begin{aligned}
\varphi(\hat X_n,S_n)
= & \mathbb{E}\Big[e^{-\delta_1\theta_n}\varphi(\hat X_{\theta_n-}^{\hat X_n,S_n,\pi_n},S_{\theta_n-}^{S_n})\mathbf{1}_{\{\theta_n <t_1^{\pi_n}\}}+e^{-\delta_1\theta_n}\varphi(\hat X_{\theta_n-}^{\hat X_n,S_n,\pi_n},S_{\theta_n-}^{S_n})
 \mathbf{1}_{\{\theta_n=t_1^{\pi_n}\}}\\
 &+\int_0^{\theta_n-}e^{-\delta_1t} D_{\hat X}\varphi d(L_t^{\pi_n})^c +\int_{0}^{\theta_n-}e^{-\delta_1 t}F_1(\hat X^{\hat X_n,S_n,\pi_n}_t,S^{S_n}_t, \varphi,D\varphi,D^2\varphi) dt \\
 &-\sum_{0\leq u< \theta_n} e^{-\delta_1 u}[\varphi(\hat X^{\hat X_n,S_n,\pi_n}_u,S^{S_n}_u)-\varphi(\hat X^{\hat X_n,S_n,\pi_n}_{u-},S^{S_n}_{u-}) ] \Big].
 \end{aligned}
\end{eqnarray*}
Notice that $\varphi_{\hat X}\geq 1+\frac{\lambda}{2}$ in $B(\hat X_0,S_0; \varepsilon)$, $(\hat X_{\theta_n}^{\hat X_n,S_n,\pi_n},S_{\theta_n}^{S_n})$ locates in $B(\hat X_0,S_0; \varepsilon)$ for large $n$, 
and
$\hat X_{t_1^{\pi_n}}^{\hat X_n,S_n,\pi_n} \in \bar{B}(\hat X_n; \varepsilon/2) \subseteq B(\hat X_0; \varepsilon)$ for large $n$, where $\bar{B}(\hat X_n; \varepsilon/2)$ is the closure of $B(\hat X_n; \varepsilon/2)$, and we have used $\Delta L_{t_1^{\pi_n}}^{\pi_n}=0$.
We can derive
\begin{eqnarray}\label{subsolution2}
\begin{aligned}
&\varphi(\hat X_n,S_n)
=   \mathbb{E}\Big[e^{-\delta_1 \bar \tau_n}\varphi(\hat X_{\bar \tau_n-}^{\hat X_n,S_n,\pi_n},S_{\bar \tau_n-}^{S_n})\mathbf{1}_{\{\bar \tau_n < t_1^{\pi_n}\}}+e^{-\delta_1t_1^{\pi_n}}\varphi(\hat X_{t_1^{\pi_n}-}^{\hat X_n,S_n,\pi_n},S_{t_1^{\pi_n}-}^{S_n})
 \mathbf{1}_{\{\bar\tau_n \geq  t_1^{\pi_n}\}}\\
 &+\int_0^{\theta_n-}e^{-\delta_1t} D_{\hat X}\varphi d(L_t^{\pi_n})^c +\int_{0}^{\theta_n-}e^{-\delta_1 t}F_1(\hat X^{\hat X_n,S_n,\pi_n}_t,S^{S_n}_t, \varphi,D\varphi,D^2\varphi) dt \\
 &-\sum_{0\leq u < \bar\tau_n\wedge t_1^{\pi_n} } e^{-\delta_1 u}[\varphi(\hat X^{\hat X_n,S_n,\pi_n}_u,S^{S_n}_u)-\varphi(\hat X^{\hat X_n,S_n,\pi_n}_{u-},S^{S_n}_{u-}) ]\Big]\\
\geq &  \mathbb{E}\Big[e^{-\delta_1\bar \tau_n}\varphi(\hat X_{\bar \tau_n-}^{\hat X_n,S_n,\pi_n},S_{\bar \tau_n-}^{S_n})\mathbf{1}_{\{\bar \tau_n < t_1^{\pi_n}\}}+e^{-\delta_1t_1^{\pi_n}}\varphi(\hat X_{t_1^{\pi_n}}^{\hat X_n,S_n,\pi_n},S_{t_1^{\pi_n}}^{S_n})
 \mathbf{1}_{\{\bar\tau_n \geq   t_1^{\pi_n}\}}\\
 &+\int_0^{\theta_n-}e^{-\delta_1t} D_{\hat X}\varphi d (L_t^{\pi_n})^c+\int_{0}^{\theta_n-} e^{-\delta_1 t}\underbrace{ F_1(\hat X^{\hat X_n,S_n,\pi_n}_t,S^{S_n}_t, \varphi,D\varphi,D^2\varphi) }_{\geq \lambda/2 \text{ for large } n} dt \\
 &+\sum_{0\leq u<\bar\tau_n\wedge t_1^{\pi_n}} e^{-\delta_1 u}(1+\frac{\lambda}{2}) (L_{u}^{\pi_n} -L_{u-}^{\pi_n}) \Big]. 
 \end{aligned}
\end{eqnarray}
Thanks to (\ref{partial_viscosity_contradiction}),
\begin{eqnarray*}\begin{aligned}
&\varphi(\hat X_{t_1^{\pi_n}}^{\hat X_n,S_n,\pi_n},S_{t_1^{\pi_n}}^{S_n})
 \mathbf{1}_{\{\bar\tau_n \geq t_1^{\pi_n} \}}\geq \left(\mathcal{P}\hat V^*(\hat X_{t_1^{\pi_n}}^{\hat X_n,S_n,\pi_n},S_{t_1^{\pi_n}}^{S_n}) +\frac{\lambda}{2}\right) \mathbf{1}_{\{\bar\tau_n \geq t_1^{\pi_n}\}}\\
\geq \mathbb{E} \Bigg[&e^{-\delta_1(t^{\pi_n}_1+\Delta-t_1^{\pi_n})} \mathbb{E}\left[\hat V^*\big(\hat X_{t^{\pi_n}_1+\Delta}^{\hat X_n,S_n,\pi_n}+s_1^{\pi_n},S_{t_1^{\pi_n}+\Delta}^{S_n}\big)-s_{1}^{\pi_n} -K \Big| \mathcal{G}_{t^{\pi_n}_{1}}\right]\mathbf{1}_{\{\bar\tau_n \geq t_1^{\pi_n}, t^{\pi_n}_{1}+\Delta  < \hat\tau^{\pi_n}\}}  \\
 &+ e^{-\delta_1 (\hat\tau^{\pi_n} -t_1^{\pi_n})} \mathbb{E}\Big[ \omega \left( X_{\hat\tau^{\pi_n}}^{\hat X_n,S_n,\pi_n} \right)^+\Big| \mathcal{G}_{t^{\pi_n}_{1}}  \Big] \mathbf{1}_{\{\bar\tau_n \geq t_1^{\pi_n}, t^{\pi_n}_{1}+\Delta \geq \hat\tau^{\pi_n} \}}\Bigg] + \frac{\lambda}{2}\mathbf{1}_{\{\bar\tau_n \geq t_1^{\pi_n}\}}.
 \end{aligned}
\end{eqnarray*}
Combining with  (\ref{subsolution2}) and $D_{\hat X}\varphi -1\geq \frac{\lambda}{2}$ for large $n$, we can derive
\begin{eqnarray}\label{subsolution3}\begin{aligned}
&\varphi(\hat X_n,S_n) \geq  \mathbb{E}\Bigg[ \int_0^{\theta_n-}e^{-\delta_1t}dL_t^{\pi_n}+e^{-\delta_1\bar \tau_n}\varphi(\hat X_{\bar \tau_n-}^{\hat X_n,S_n,\pi_n},S_{\bar \tau_n-}^{S_n})\mathbf{1}_{\{\bar \tau_n < t_1^{\pi_n}\}} \\
 &  + e^{-\delta_1(t^{\pi_n}_1+\Delta)} \mathbb{E}\left[\hat V^*\big(\hat X_{t^{\pi_n}_1+\Delta}^{\hat X_n,S_n,\pi_n}+s_1^{\pi_n},S_{t_1^{\pi_n}+\Delta}^{S_n}\big)-s_{1}^{\pi_n} -K \Big|\mathcal{G}_{t^{\pi_n}_{1}}\right]\mathbf{1}_{\{\bar\tau_n \geq t_1^{\pi_n}, t^{\pi_n}_{1}+\Delta <\hat\tau^{\pi_n}\}}  \\
 &+ e^{-\delta_1 \hat\tau^{\pi_n}} \mathbb{E}\Big[ \omega \left(X_{\hat\tau^{\pi_n}}^{\hat X_n,S_n,\pi_n} \right)^+\Big| \mathcal{G}_{t^{\pi_n}_{1}}  \Big] \mathbf{1}_{\{\bar\tau_n \geq t_1^{\pi_n}, t^{\pi_n}_{1}+\Delta \geq \hat\tau^{\pi_n} \}}\\
 &+ \int_0^{\theta_n-}e^{-\delta_1t}\frac{\lambda}{2}dL_t^{\pi_n}+\int_{0}^{\theta_n} e^{-\delta_1t}\frac{\lambda}{2} dt + e^{-\delta_1 t_1^{\pi_n}} \frac{\lambda}{2} \mathbf{1}_{\{\bar\tau_n \geq   t_1^{\pi_n}\}}  \Bigg].
\end{aligned}
\end{eqnarray}
Notice that  $S_t^{S_n}$ is always continuous  and
$\hat X_{\bar\tau_n}^{\hat X_n,S_n,\pi_n}$ may not locate in $B(\hat X_n;\varepsilon/2)$.
However, similar as the proof of Theorem 5.1 in Chapter 8 in \cite{FS06}, we can find $\zeta_n\in[0,1]$ such that $\tilde X_{n} = \hat X_{\bar \tau_n-}^{\hat X_n,S_n,\pi} -\zeta_n\Delta L_{\bar\tau_n}^{\pi_n}$ with $\Delta L_{\bar\tau_n}^{\pi_n}: = L_{\bar\tau_n}^{\pi_n}-L_{\bar\tau_n-}^{\pi_n} = \hat X_{\bar \tau_n-}^{\hat X_n,S_n,\pi_n} - \hat X_{\bar \tau_n}^{\hat X_n,S_n,\pi_n}$ satisfying
$\tilde X_{n} \in \partial B(\hat X_n;\varepsilon/2)$ if $\bar\tau_n\leq R$ and $\tilde X_{n} =\hat X_{\bar \tau_n}^{\hat X_n,S_n,\pi_n} \in B(\hat X_n;\varepsilon/2)$ if $\bar\tau_n>R$. Since $\varphi_{\hat X}\geq 1+\lambda/2$, we have
\begin{eqnarray}\label{partial_varphi_dis}
\varphi(\tilde X_n, S_{\bar \tau_n}^{S_n}) -\varphi(\hat X_{\tau_n-}^{\hat X_n,S_n,\pi_n}, S_{\bar \tau_n-}^{S_n}) \leq -\zeta_n(1+\lambda/2) \Delta L_{\bar\tau_n}^{\pi_n}.
\end{eqnarray}

We claim that $\hat V^*(\hat X,S)\geq \hat V^*(\hat X',S) + (\hat X-\hat X')$  for $\hat X>\hat X'>I(S)$.  This is because apparently we have $\hat V(\hat X,S)\geq \hat V(\hat X',S) + (\hat X-\hat X')$, and there exists a sequence $(\hat X_m,S_m)$ with $\hat X_m\in (I(S),\hat X)$ converging to $(\hat X',S)$ such that $\hat V(\hat X_m,S_m)\rightarrow \hat V^*(\hat X',S)$. Sending $m\rightarrow \infty$ in  $\hat V(\hat X,S_m)\geq \hat V(\hat X_m,S_m) + (\hat X-\hat X_m)$, we can derive $\hat V^*(\hat X,S)\geq \hat V^*(\hat X',S) + (\hat X-\hat X')$. Using this inequality and $\hat X_{\bar\tau_n}^{\hat X_n,S_n,\pi_n}= \tilde X_n -(1-\zeta_n)\Delta L_{\bar\tau_n}^{\pi_n}$, we obtain
$\hat V^*(\tilde X_n, S_{\bar \tau_n}^{S_n}) \geq \hat V^*(\hat X_{\bar\tau_n}^{\hat X_n,S_n,\pi_n}, S_{\bar \tau_n}^{S_n}) + (1-\zeta_n)\Delta L_{\bar\tau_n}^{\pi_n}.$
Recalling that $\varphi(\tilde X_n, S_{\bar \tau_n}^{S_n})\geq \hat V^*(\tilde X_n, S_{\bar \tau_n}^{S_n})$, from (\ref{partial_varphi_dis}) we have
\begin{eqnarray}\label{partial_v_varphi}
\varphi (\hat X_{\bar\tau_n-}^{\hat X_n,S_n,\pi_n}, S_{\bar \tau_n-}^{S_n}) \geq \hat V^*(\hat X_{\bar\tau_n}^{\hat X_n,S_n,\pi_n}, S_{\bar \tau_n}^{S_n}) + (1+\zeta_n\lambda/2)\Delta L_{\bar\tau_n}^{\pi_n}.
\end{eqnarray}
Substituting (\ref{partial_v_varphi}) into (\ref{subsolution3}) gives
\begin{equation}\label{subsolution3_3}
\begin{aligned}
&\varphi(\hat X_n,S_n) \geq  \mathbb{E}\Bigg[ \int_0^{\bar\tau_n\wedge t_1^{\pi_n}}e^{-\delta_1t}dL_t^{\pi_n}+e^{-\delta_1\bar\tau_n}\hat V^*(\hat X_{\bar\tau_n}^{\hat X_n,S_n,\pi_n},S_{\bar\tau_n}^{S_n})\mathbf{1}_{\{\bar\tau_n<t^{\pi_n}_1 \}} \\
&  + e^{-\delta_1(t^{\pi_n}_1+\Delta)} \mathbb{E}\left[\hat V^*\big(\hat X_{t^{\pi_n}_1+\Delta}^{\hat X_n,S_n,\pi_n}+s_1^{\pi_n},S_{t_1^{\pi_n}+\Delta}^{S_n}\big)-s_{1}^{\pi_n} -K \Big|\mathcal{G}_{t^{\pi_n}_{1}}\right]\mathbf{1}_{\{\bar\tau_n \geq t_1^{\pi_n}, t^{\pi_n}_{1}+\Delta \leq \hat\tau^{\pi_n}\}}  \\
 &+ e^{-\delta_1 \hat\tau^{\pi_n}} \mathbb{E}\Big[ \omega \left(X_{\hat\tau^{\pi_n}}^{\hat X_n,S_n,\pi_n} \right)^+\Big| \mathcal{G}_{t^{\pi_n}_{1}}  \Big] \mathbf{1}_{\{\bar\tau_n \geq t_1^{\pi_n}, t^{\pi_n}_{1}+\Delta > \hat\tau^{\pi_n} \}}\\
 &+\frac{\lambda}{2}\bigg[\int_0^{\theta_n-}e^{-\delta_1t} dL_t^{\pi_n}+\int_{0}^{\theta_n} e^{-\delta_1t} dt + e^{-\delta_1 t_1^{\pi_n}}  \mathbf{1}_{\{\bar\tau_n \geq   t_1^{\pi_n}\}}+\zeta_ne^{-\delta_1\bar\tau_n}\Delta L_{\bar\tau_n}^{\pi_n} \mathbf{1}_{\{\bar\tau_n<t^{\pi_n}_1 \}} \bigg] \Bigg].
\end{aligned}
\end{equation}
To lead to a contradiction, we want to prove that
{\fontsize{10pt}{10pt}
\begin{equation*}
\underset{n\rightarrow \infty}{\underline{\lim}}  \sup_{\pi_n\in\Pi((\hat X_n,S_n))} \bigg[\int_0^{\theta_n-}e^{-\delta_1t} dL_t^{\pi_n}+\int_{0}^{\theta_n} e^{-\delta_1t} dt + e^{-\delta_1 t_1^{\pi_n}}  \mathbf{1}_{\{\bar\tau_n \geq   t_1^{\pi_n}\}}+\zeta_ne^{-\delta_1\bar\tau_n}\Delta L_{\bar\tau_n}^{\pi_n} \mathbf{1}_{\{\bar\tau_n<t^{\pi_n}_1 \}} \bigg] >0.
\end{equation*}}
We consider two cases as follows.\\
(1). If $\bar{\tau}_n \leq R$, we have $\tilde X_{n} \in \partial B(\hat X_n;\varepsilon/2)$. Note that the $C^2$ function
$
g_n(\hat X,S) := G_n[ (\hat X-\hat X_n)^2 -\varepsilon^2/4]$
satisfies
\begin{equation*}\left\{ \begin{aligned}
&\min\{F_1(\hat X,S,g_n,Dg_n,D^2g_n) +1, \quad D_{\hat X}g_n +1, \quad g_n+1\}\geq 0, \ \  \text{ in } B(\hat X_n,S_n;\varepsilon/2),\\
& g_n = 0, \  \  \text{ on } \partial \bar{B}(\hat X_n, S_n;\varepsilon/2),
\end{aligned}\right.\end{equation*}
where \begin{equation*} \begin{aligned}
G_n = \min \left\{\frac{1}{\delta_1\varepsilon^2/4+(\alpha-\mu)(1+\hat X_n+\varepsilon/2)\varepsilon+(1+\hat X_n+\varepsilon/2)^2(\bar{S}/m+\sigma\rho)^2} , \frac{1}{\varepsilon}, \frac{4}{\varepsilon^2}\right\}.
\end{aligned}\end{equation*}
By  Ito's formula and the nice property of $g_n$ above, we infer that for any admissible strategy $\pi_n\in \Pi((\hat X_n,S_n))$,
{\fontsize{10.5pt}{10.5pt}
\begin{eqnarray*}\begin{aligned}
&g_n(\hat X_n,S_n)
=  \mathbb{E}\Big[e^{-\delta_1\theta_n}g_n(\hat X_{\theta_n-}^{\hat X_n,S_n,\pi_n},S_{\theta_n-}^{S_n})\mathbf{1}_{\{\theta_n<t_1^{\pi_n}\}}+e^{-\delta_1\theta_n}g_n(\hat X_{\theta_n-}^{\hat X_n,S_n,\pi_n},S_{\theta_n-}^{S_n})
 \mathbf{1}_{\{\theta_n=t_1^{\pi_n}\}}\\
 &+\int_0^{\theta_n-}e^{-\delta_1t} D_{\hat X}g_n d(L_t^{\pi_n})^c +\int_{0}^{\theta_n-}e^{-\delta_1 t}F_1(\hat X^{\hat X_n,S_n,\pi_n}_t,S^{S_n}_t, g_n,Dg_n,D^2g_n) dt \\
 &-\sum_{0\leq u <  \theta_n} e^{-\delta_1 u}[g_n(\hat X^{\hat X_n,S_n,\pi_n}_u,S^{S_n}_u)-g_n(\hat X^{\hat X_n,S_n,\pi_n}_{u-},S^{S_n}_{u-}) ] \Big] \\
 \geq & \mathbb{E}\Big[e^{-\delta_1\bar\tau_n}g_n(\hat X_{\bar\tau_n-}^{\hat X_n,S_n,\pi_n},S_{\bar\tau_n-}^{S_n})\mathbf{1}_{\{\bar\tau_n < t_1^{\pi_n}\}}+e^{-\delta_1t_1^{\pi_n}}(\underbrace{g_n(\hat X_{t_1^{\pi_n}}^{\hat X_n,S_n,\pi_n},S_{t_1^{\pi_n}}^{S_n})+1}_{\geq 0})
 \mathbf{1}_{\{\bar \tau_n\geq t_1^{\pi_n}\}}\\
 &+\int_0^{\theta_n-}e^{-\delta_1t} (\underbrace{D_{\hat X}g_n +1}_{\geq 0})d(L_t^{\pi_n})^c +\int_{0}^{\theta_n-}e^{-\delta_1 t}(\underbrace{F_1(\hat X^{\hat X_n,S_n,\pi_n}_t,S^{S_n}_t, g_n,Dg_n,D^2g_n) +1}_{\geq 0})dt \\
&-\sum_{0\leq u < \bar\tau_n \wedge t_1^{\pi_n}} e^{-\delta_1 u} \Delta L_{u}^{\pi_n}-
 \int_0^{\theta_n-}e^{-\delta_1t} d(L_t^{\pi_n})^c  - \int_{0}^{\theta_n-}e^{-\delta_1 t} dt-e^{-\delta_1t_1^{\pi_n}} \mathbf{1}_{\{\bar \tau_n\geq t_1^{\pi_n}\}} \Big].
 \end{aligned}
\end{eqnarray*}}
Since $D_{\hat X}g_n\geq -1$, we have
\begin{eqnarray*}\begin{aligned}
g_n(\hat X_{\bar\tau_n-}^{\hat X_n,S_n,\pi_n},S_{\bar\tau_n-}^{S_n}) -\underbrace{g_n(\tilde X_n,S_{\bar\tau_n}^{S_n})}_{=0 }\geq -1(\hat X_{\bar\tau_n-}^{\hat X_n,S_n,\pi_n} -\tilde X_n) =-\zeta_n \Delta L_{\bar\tau_n}^{\pi_n}.
 \end{aligned}
\end{eqnarray*}
The two inequalities above imply that for any admissible strategy $\pi_n\in \Pi((\hat X_n,S_n))$,
\begin{equation*}\begin{aligned}
 &\mathbb{E}\Big[
 \int_0^{\theta_n-}e^{-\delta_1t} dL_t^{\pi_n}  + \int_{0}^{\theta_n}e^{-\delta_1 t} dt+ e^{-\delta_1t_1^{\pi_n}} \mathbf{1}_{\{\bar \tau_n\geq t_1^{\pi_n}\}} +e^{-\delta_1\theta_n}\zeta_n \Delta L_{\bar\tau_n}^{\pi_n} \mathbf{1}_{\{\bar \tau_n< t_1^{\pi_n}\}}\Big]\\
 &\geq - g_n(\hat X_n,S_n) =G_n\varepsilon^2/4.
 \end{aligned}
\end{equation*}
 (2). If $\bar\tau_n>R$, we have $\zeta_n=0$ and $\tilde X_{n} =\hat X_{\bar \tau_n}^{\hat X_n,S_n,\pi_n}$. In this case, we find that $\int_{0}^{\theta_n} e^{-\delta_1t} dt + e^{-\delta_1 t_1^{\pi_n}}  \mathbf{1}_{\{\bar\tau_n \geq   t_1^{\pi_n}\}} \geq e^{-\delta_1 R}$ if $t_1^{\pi_n}\leq R$ and $\int_{0}^{\theta_n} e^{-\delta_1t} dt + e^{-\delta_1 t_1^{\pi_n}}  \mathbf{1}_{\{\bar\tau_n \geq   t_1^{\pi_n}\}} \geq \frac{1}{\delta_1}(1-e^{-\delta_1R})$ if $t_1^{\pi_n} > R$. Thus,
\begin{align*}
\bigg[\int_{0}^{\theta_n} e^{-\delta_1t} dt + e^{-\delta_1 t_1^{\pi_n}}  \mathbf{1}_{\{\bar\tau_n \geq   t_1^{\pi_n}\}} \bigg] \geq \min\left\{(1-e^{-\delta_1R})/\delta_1, e^{-\delta_1 R}\right\}.
\end{align*}
Therefore, combining two cases above, taking supremum for $\pi_n\in \Pi((\hat X_n,S_n))$ in (\ref{subsolution3_3}), and by the weak DPP formula (\ref{subsolution1}), we obtain
 \begin{equation*}\begin{aligned}
\varphi(\hat X_n,S_n) \geq \hat V(\hat X_n,S_n) + \frac{\lambda}{2}\min\left\{ G_n \varepsilon^2/4,(1-e^{-\delta_1R})/\delta_1, e^{-\delta_1 R} \right\},
\end{aligned}\end{equation*}
 which contradicts (\ref{subsolution4}) when $n$ goes to infinity as $\underset{n\rightarrow \infty}{\underline{\lim}} G_n >0$. This completes the proof of Lemma \ref{subsolution}.

Therefore, a combination of Lemma \ref{supersolution} and \ref{subsolution} yields that the value function $\hat V$ in (\ref{v(x,S)}) is a viscosity solution to Problem  (\ref{p_hjb33}) on $\Omega_1$.

In order to prove the uniqueness of viscosity solution, we give a comparison principle as follows.
\begin{lemma} \label{partial_comparison}
\textbf{(Comparison principle)}: Let $v,u$ be respectively viscosity subsolution and supersolution of Problem (\ref{p_hjb33}), and both are of at most linear growth. Then $v^*\leq u_*$ on $\Omega_1$.
\end{lemma}
\noindent\textit{Proof of Lemma \ref{partial_comparison}}:
We will prove the comparison principle in several steps.\\
\textit{Step 1}. Let us show that $\eta-$strict supersolution always exists, that is, for any $\eta\in (0,1)$, there exists a $\eta-$strict supersolution $u^\eta$ to
\begin{eqnarray}\label{partial_etauper}
\min\{F_1(\hat X,S, v,D v,D^2 v), v-\mathcal{P}v,  D_{\hat X}v-1\}= \eta.
\end{eqnarray}
Consider the function $u^\eta(\hat X,S) = u(\hat X,S)  +\eta\hat X +cS+C^\eta$, where $u(\hat X,S)$ is a viscosity supersolution to (\ref{p_hjb33}), and $C^\eta$ and $c$ are two nonnegative constants to be determined later.  Note that $D_{\hat X}u^\eta = D_{\hat X} u+\eta\geq 1+\eta$, and
\begin{eqnarray*}\begin{aligned}
&
F_1(\hat X,S, u^\eta,D u^\eta,D^2 u^\eta)\\
& = F_1(\hat X,S, u,D u,D^2 u)+\left((S/m+\sigma\rho)^2-\sigma^2\right) (\eta\hat X+cS+C^\eta)_S-\mathcal{L} (\eta\hat X+cS+C^\eta)\\
&\geq 0 -(\alpha-\mu)(1+\hat X)\eta +(\delta-\mu) ( \eta\hat X+cS+C^\eta)+\left((S/m+\sigma\rho)^2-\sigma^2\right)c
\geq \eta
\end{aligned}
\end{eqnarray*}
if $C^\eta \geq  \frac{ \sigma^2c+(1+\alpha-\mu)\eta-(\delta-\alpha)\eta (I(\bar{S}\wedge (\Phi^{-1}(a))^2) \wedge \kappa )}{\delta-\mu}$. Moreover, as  $u_*\geq \mathcal{P}u_*$, we have
\begin{eqnarray*}\begin{aligned}
\mathcal{P} u^\eta_*(\hat X,S)&\leq \mathcal{P} u_* + \sup_{s\geq 0}\mathbb{E}\Big[e^{-\delta_1\Delta}  [(\eta(\hat X_{\Delta}+s)+cS_\Delta+C^\eta]\mathbf{1}_{\{\hat{\tau}>\Delta\}}   \Big],\\
& \leq \mathcal{P} u_*  + e^{-\delta_1\Delta}(\eta \bar{s} +cS +cm\sigma(1-\rho)+C^\eta)+  \mathbb{E}[e^{-\delta_1\Delta}\eta\hat X_{\Delta}\mathbf{1}_{\{\hat{\tau}>\Delta\}} ]\\
&\leq u_* + \eta\hat X +cS+ C^\eta  -\eta = u^\eta_* -\eta, \quad \text{for large $C^\eta$},
\end{aligned}
\end{eqnarray*}
where $\mathbb{E}[e^{-\delta_1\Delta}\hat X_{\Delta}\mathbf{1}_{\{\hat{\tau}>\Delta\}} ]\leq \mathbb{E}[e^{-\delta_1\Delta}\hat X_{\Delta}] = e^{-(\delta-\alpha)\Delta}(1+\hat X) - e^{-\delta_1\Delta}$ and $\Delta>0$ are used in the last inequality. 
Therefore, we can always find a large constant $C^\eta$  such that $u^\eta$ is a supersolution to (\ref{partial_etauper}). This indicates the existence of $\eta-$strict supersolution that is of at most linear growth.\\
\textit{Step 2}. Now we claim that if $v$ and $u^\eta$ are respectively viscosity subsolution and $\eta-$strict supersolution of (\ref{p_hjb33}) that are of at most linear growth, then
$v^*\leq u^\eta_*$.  We employ the technique in \cite{CIL92} and use a method of contradiction. Assume, to the contrary, that there exists some point $(\hat X^0, S^0)\in \Omega_1$ such that
 \begin{eqnarray}\label{partial_comparison_contra}
 \gamma:=v^*(\hat  X^0,S^0) - u_*^\eta(\hat X^0,S^0)>0.\end{eqnarray}
Let \begin{eqnarray*}
\begin{aligned} 
& M_{\theta} : = \underset{(\hat X,S),(\hat X',S')\in \Omega_1}{\sup} \left\{v^*(\hat X,S)-u_*^\eta(\hat X', S') - \psi_\theta(\hat X,S,\hat X',  S')\right\}, \\
& \psi_\theta(\hat X,S, \hat X', S') := \frac{\theta}{2} \Big(|\hat X-\hat X'|^2+|S-S'|^2 \Big).
\end{aligned}
\end{eqnarray*}
Since $u^\eta_*$ and $v$ grow at most linearly at infinity and $\psi_\theta$ grows quadratically at infinity, there is a maximizer $(\hat X_\theta, \hat X_\theta', S_\theta, S_\theta')$ such that
\begin{eqnarray*}M_\theta= v^*(\hat X_\theta,S_\theta) - u_*^\eta(\hat X'_\theta, S'_\theta) - \psi_\theta(\hat X_\theta,S_\theta, \hat X'_\theta,  S'_\theta),
\end{eqnarray*} and there is a sequence $\theta_n\rightarrow \infty$ such that
\begin{eqnarray*}
(\hat X_n,S_n,  \hat X_n', S_n') := (\hat X_{\theta_n},S_{\theta_n}, \hat X'_{\theta_n}, S'_{\theta_n}) \rightarrow (\hat X_0,  S_0, \hat X'_0, S'_0), \quad n\rightarrow\infty.
\end{eqnarray*}
We claim that $\hat X_0 = \hat X'_0$, $S_0=S'_0$, and
\begin{eqnarray}\label{partial_doublepoints}
\begin{aligned}
&\theta_n|\hat X_n-\hat X'_n|^2, \quad \theta_n|S_n-S'_n|^2 \rightarrow 0, \\
& M_{\theta_n} \rightarrow M_{\infty} : = \sup_{(\hat X,S) \in \Omega_1}\big\{ (v^*-u^\eta_*)(\hat X,S)\big\}.
\end{aligned}
\end{eqnarray}
By the linear growth property proved in Appendix \ref{appendix_linear growth condition}, we have $v^*(\hat X,S) \leq \hat X+C_1S+C_2$ and $u^{\eta}_* (\hat X,S) \geq (1+\eta)\hat X +cS+C^\eta-C_0$ for some positive constants $C_0$, $C_1$, and $C_2$. Therefore, the boundedness of $S$ implies
\begin{eqnarray*}\begin{aligned}
\underset{(\hat X,S)\in\Omega_1}{\sup} \big\{ (v^* - u^{\eta}_*)(\hat X,S) \big\}<+\infty.
\end{aligned}
\end{eqnarray*}
Hence there exists a bounded maximizer $(\hat X^*, S^*)\in \Omega_1$ of $(v^*-u^\eta_*)(\hat X,S)$. Then
\begin{eqnarray*}\begin{aligned}
(v^*-u^\eta_*)(\hat X^*,S^*)  \leq M_{\theta_n}=v^*(\hat X_n,S_n) - u_*^\eta(\hat X'_n, S'_n) - \psi_{\theta_n}(\hat X_n, S_n, \hat X'_n, S'_n).
\end{aligned}
\end{eqnarray*}
Sending $n\rightarrow \infty$, we obtain
\begin{eqnarray*}\begin{aligned}
&\frac{1}{2}\underset{n\rightarrow \infty}{\lim\sup} \left \{\theta_n |\hat X_n -\hat X'_n|^2+\theta_n|S_n-S'_n|^2\right\} \\
\leq &\underset{n\rightarrow \infty}{\lim\sup}\Big\{v^*(\hat X_n,S_n) - u_*^\eta(\hat X'_n, S'_n) \Big\}  - (v^*-u^\eta_*)(\hat X^*,S^*)  \\
\leq & v^*(\hat X_0,S_0) - u_*^\eta(\hat X'_0, S'_0) - (v^*-u^\eta_*)(\hat X^*,S^*)  <\infty.
\end{aligned}
\end{eqnarray*}
This indicates $\underset{n\rightarrow \infty}{\lim\sup} \left \{\theta_n |\hat X_n -\hat X'_n|^2 +\theta_n |S_n-S'_n|^2\right\}<\infty$. Noticing $\theta_n\rightarrow \infty$, we conclude $\hat X_0 = \hat X'_0$ and $S_0=S'_0$. Moreover, by the definition of $(\hat X^*, S^*)$ as a maximizer of $(v^*-u^\eta_*)(\hat X,S)$, we have
\begin{eqnarray*}\begin{aligned}
0\leq  (v^*-u_*^\eta)(\hat X_0,S_0) -(v^*-u^\eta_*)(\hat X^*,S^*)  \leq 0.
\end{aligned}
\end{eqnarray*}
Thus $(\hat X_0,S_0)$ is the maximizer of $(v^*-u^\eta_*)(\hat X,S)$ and (\ref{partial_doublepoints}) holds.

By the semi-continuity property of $u^\eta_*$ and $v^*$, we have
\begin{eqnarray}\label{partial_vmu}\begin{aligned}
\underset{n\rightarrow \infty}{\lim\sup} M_{\theta_n} &= \underset{n\rightarrow \infty}{\lim\sup} \left\{v^*(\hat X_n,S_n) - u_*^\eta(\hat X'_n, S'_n) - \psi_{\theta_n}(\hat X_n,S_n, \hat X'_n,  S'_n)\right\}\\
& \leq \underset{n\rightarrow \infty}{\lim\sup} \left\{v^*(\hat X_n,S_n) - u_*^\eta(\hat X'_n, S'_n) \right\}\\
& \leq \underset{n\rightarrow \infty}{\lim\sup} v^*(\hat X_n,S_n) - \underset{n\rightarrow \infty}{\lim\inf} u_*^\eta(\hat X'_n, S'_n) \\
& \leq (v^*-u^\eta_*)(\hat X_0,S_0).
\end{aligned}
\end{eqnarray}
By (\ref{partial_comparison_contra}), $M_{\theta_n} \geq \gamma -\psi_{\theta_n}(\hat X^0,S^0,  \hat X^0, S^0) = \gamma>0$. Hence, the limit $M_\infty>0$ and $(v^*-u^\eta_*)(\hat X_0,S_0)>0$.
Since $v^*\leq u^\eta_*$ on $\partial \Omega_1$, we deduce that $(\hat X_0,S_0)\notin \partial \Omega_1$ and therefore $(\hat X_n, S_n, \hat X_n', S_n')$ is a local maximizer of $v^*(\hat X,S)-u_*^\eta(\hat X', S') - \psi_{\theta_n}(\hat X,S, \hat X', S')$ in $\Omega_1$. \\
\textit{Step 3}. 
Applying Theorem 3.2 in \cite{CIL92},
there exist
\begin{eqnarray*}\begin{aligned}
&(D_{\hat X,S}\psi_{\theta_n}(\hat X_n,  S_n, \hat X_n', S_n'),A_n) \in \bar{J}^{2,+} v^*(\hat X_n,S_n), \\
& (-D_{\hat X',S'}\psi_{\theta_n}(\hat X_n, S_n, \hat X_n',  S_n'),B_n) \in \bar{J}^{2,-} u^\eta_*(\hat X_n', S_n'),
\end{aligned}\end{eqnarray*}
 where $D_{\hat X,S}\psi_{\theta_n}:= [D_{\hat X}\psi_{\theta_n}, D_{S}\psi_{\theta_n} ]^{T},$
$D_{\hat X',S'}\psi_{\theta_n} := [D_{\hat X'}\psi_{\theta_n}, D_{S'}\psi_{\theta_n} ]^{T}$,
and  $A_n, B_n \in \mathbb{R}^{2\times 2}$ satisfy
\begin{eqnarray*}
\left(\begin{array}{cc} A_n& 0 \\ 0& -B_n\end{array}\right)& \leq& D^2 \psi_{\theta_n}(\hat X_n, S_n, \hat X_n', S_n') + \frac{1}{\theta_n} (D^2 \psi_{\theta_n}(\hat X_n, S_n,  \hat X_n', S_n'))^2.
\end{eqnarray*}
By calculation,
\begin{eqnarray*}\begin{aligned}
D_{\hat X,S}\psi_{\theta_n}(\hat X_n, S_n, \hat X_n', S_n') =
& [\theta_n (\hat X_n -\hat X'_n), \quad \theta_n(S_n-S'_n) ]^{T},\\
-D_{\hat X',S'}\psi_{\theta_n}(\hat X_n, S_n,\hat X_n',  S_n') =& [ \theta_n (\hat X_n -\hat X'_n),\quad  \theta_n(S_n-S'_n) ]^{T},
\end{aligned}\end{eqnarray*}
and
$D^2 \psi_{\theta_n}(\hat X_n, S_n,\hat X_n',  S_n')=
\theta_n\left(\begin{array}{cc}I_2& -I_2 \\ -I_2& I_2\end{array}\right),$
where $I_2$ is $2\times 2$ identity matrix. 
Therefore, we have
\begin{eqnarray}\label{partial_matrixmode}
\left(\begin{array}{cc} A_n& 0 \\ 0& -B_n\end{array}\right) \leq
3\theta_n \left(\begin{array}{cc} I_2& -I_2 \\ -I_2& I_2\end{array}\right).
\end{eqnarray}
It follows
{\fontsize{10.5pt}{10.5pt}
\begin{equation}\label{partial_twomin}
\begin{aligned}
&\min\{F(\hat X_n, S_n, v^*(\hat X_n,S_n), D_{\hat X,S}\psi_{\theta_n}(\hat X_n, S_n, \hat X_n',  S_n'), A_n), v^*(\hat X_n,S_n) - \mathcal{P}v^*(\hat X_n,S_n)\}\leq 0,\\
&\min\{F(\hat X'_n, S'_n, u^\eta_*(\hat X'_n,S'_n), -D_{\hat X',S'}\psi_{\theta_n}(\hat X_n,  S_n, \hat X_n', S_n'), B_n), u^\eta_*(\hat X'_n,S'_n) -\mathcal{P}u^\eta_*(\hat X'_n,S'_n) \}\geq \eta.
\end{aligned}\end{equation}}
Based on the first inequality in (\ref{partial_twomin}), we consider the following three cases.\\
(1).  Case 1: If $v^*(\hat X_n,S_n)-\mathcal{P}v^*(\hat X_n,S_n)\leq 0$ in (\ref{partial_twomin}), then there exists a  sequence  $s_{\theta_n} \in \mathcal{G}_{\Delta}$ such that
\begin{eqnarray*}\begin{aligned}
v^*(\hat X_n,S_n) \leq \mathcal{P}_{s_{\theta_n}} v^*(\hat X_n,S_n)+\frac{\eta}{2},  \quad u^\eta_*(\hat X'_n,S'_n) \geq \mathcal{P}_{s_{\theta_n}}u^\eta_*(\hat X'_n,S'_n)  +\eta,
\end{aligned}\end{eqnarray*}
where $\mathcal{P}_{s} V(\hat X,S)$ is defined in (\ref{partial_Ps}).
Because $|\hat X_n -\hat X'_n|, |S_n-S'_n|\rightarrow 0$, $\hat X_n \rightarrow \hat X_0$, $S_n \rightarrow S_0$ with $|X_0|+|S_0|<\infty$, $\psi(I(S),S)\leq I(S)+1$, $S$ is bounded in $[0,\bar{S}]$, and $\hat X_t+1$ follows geometric Brownian motion during time $[0,\Delta)$, we can obtain
\begin{eqnarray*}\begin{aligned}
&\mathcal{P}_{s_{\theta_n}} v^*(\hat X_n,S_n) - \mathcal{P}_{s_{\theta_n}}u^\eta_*(\hat X'_n,S'_n) -\psi_{\theta_n}(\hat X_n,S_n,  \hat X_n', S_n')\\
 \leq  &e^{-\delta_1 \Delta} M_\infty + O(1)\Big |\mathbb{P}(\hat\tau>\Delta | \hat X_0 = \hat X_n,S_0=S_n) - \mathbb{P}(\hat\tau>\Delta | \hat X_0 = \hat X_n',S_0=S'_n) \Big | \\
 &+ O(|S_n-S_n'|) + O(1)\Big |\mathbb{P}(\hat\tau\leq \Delta | \hat X_0 = \hat X_n,S_0=S_n) - \mathbb{P}(\hat\tau\leq \Delta | \hat X_0 = \hat X_n',S_0=S'_n) \Big |\\
 & +\mathbb{E}\left[e^{-\delta_1\Delta}\psi_{\theta_n}(\hat X_\Delta^{\hat X_n,S_n}+s_{\theta_n},  S_\Delta^{S_n}, \hat X_\Delta^{\hat X'_n,S'_n}+s_{\theta_n}, S_\Delta^{S_n'})\right]-\psi_{\theta_n}(\hat X_n,  S_n, \hat X_n', S_n')\\
 \leq & e^{-\delta_1 \Delta} M_\infty +o(1) + O(1)\psi_{\theta_n}(\hat X_n, S_n, \hat X_n',  S_n'), \text{ for large $n$}.
\end{aligned}\end{eqnarray*}
Therefore, combining with  $\psi_{\theta_n}(\hat X_n,  S_n,\hat X_n', S_n')=o(1)$ for large $n$, we have
\begin{eqnarray*}\begin{aligned}
M_{\theta_n} &= v^*(\hat X_n,S_n) - u^\eta_*(\hat X'_n,S'_n) -\psi_{\theta_n}(\hat X_n,  S_n,\hat X_n', S_n') \\
& \leq \mathcal{P}_{s_{\theta_n}} v^*(\hat X_n,S_n) - \mathcal{P}_{s_{\theta_n}}u^\eta_*(\hat X'_n,S'_n) -\psi_{\theta_n}(\hat X_n, S_n, \hat X_n',  S_n') -\frac{\eta}{2}\\
& \leq e^{-\delta_1 \Delta} M_\infty +o(1) -\frac{\eta}{2}, \text{ for large $n$},
\end{aligned}\end{eqnarray*}
which contradicts (\ref{partial_doublepoints}) for large $n$.\\
(2). Case 2: If $D_{\hat X_n}\psi_{\theta_n}(\hat X_n, S_n, \hat X_n',  S_n')=\theta_n (\hat X_n -\hat X'_n) \leq 1$ in (\ref{partial_twomin}), noticing that  $-D_{\hat X'_n}\psi_{\theta_n}(\hat X_n, S_n, \hat X_n', S_n')=\theta_n (\hat X_n -\hat X'_n) \geq 1+\eta$, we have $0 \leq -\eta,$
which yields a contradiction.\\
(3). Case 3: If $F_1(\hat X_n, S_n, v^*(\hat X_n,S_n), D_{\hat X,S}\psi_{\theta_n}(\hat X_n,  S_n,\hat X_n',  S_n'), A_n)\leq 0$ in (\ref{partial_twomin}), noticing that $F_1(\hat X'_n, S'_n, u^\eta_*(\hat X'_n,S'_n), -D_{\hat X',S'}\psi_{\theta_n}(\hat X_n, S_n, \hat X_n',  S_n'), B_n)\geq \eta$, we have
\begin{eqnarray*}\begin{aligned}
-\eta \geq &  \delta_1(v^*(\hat X_n,S_n) -u^\eta_*(\hat X'_n,S'_n))
- (\alpha-\mu)\theta_n |\hat X_n-\hat X'_n|^2 \\
&+\frac{\theta_n}{m}|S_n-S'_n|^2{(S_n/m+S'_n/m+2\sigma\rho)}  \\
&  -\frac{1}{2} \text{Tr} \left\{\left(\begin{array}{cc} C_nC_n^T& C_nG_n^T \\ G_nC_n^T& G_nG_n^T \end{array}\right) \left(\begin{array}{cc} A_n& 0 \\ 0& -B_n\end{array}\right)\right\},
\end{aligned}\end{eqnarray*}
where $C_n = \left(\begin{array}{cc} (1+\hat X_n)(S_n/m+\sigma\rho)& 0 \\ 0& 0\end{array}\right)$ and $G_n=\left(\begin{array}{cc} (1+\hat X'_n)(S'_n/m+\sigma\rho)& 0 \\ 0& 0\end{array}\right)$.
Notice that $\hat X_n \rightarrow \hat X_0, \hat X'_n \rightarrow \hat X_0$ with $|\hat X_0|, |S_0|<\infty$ and $(v^*-u^\eta_*)(\hat X_0,S_0)>0$. From (\ref{partial_doublepoints}), (\ref{partial_vmu}),  and (\ref{partial_matrixmode}),  we have
\begin{eqnarray*}\begin{aligned}
0>& -\eta \geq  \delta_1(v^*(\hat X_n,S_n) -u^\eta_*(\hat X'_n,S'_n))
 +o(1)  \\
 & -\frac{1}{2} \text{Tr} \left\{\left(\begin{array}{cc} C_nC_n^T& C_nG_n^T \\ G_nC_n^T& G_nG_n^T \end{array}\right) 3\theta_n \left(\begin{array}{cc} I_2& -I_2 \\ -I_2& I_2\end{array}\right)  \right\} \\
= &  \delta_1(v^*(\hat X_n,S_n) -u^\eta_*(\hat X'_n,S'_n))
 +o(1)   \\
 &-\frac{3\theta_n}{2} \bigg((S_n/m+\sigma\rho)(1+\hat X_n) - (S'_n/m+\sigma\rho)(1+\hat X'_n)\bigg)^2 \\
\geq &  \delta_1(v^*(\hat X_n,S_n) -u^\eta_*(\hat X'_n,S'_n))
 +o(1)  \\
 &-\frac{3\theta_n}{2}  \bigg(2(S_n/m+\sigma\rho)^2|\hat X_n-\hat X_n|^2 + 2 |S_n/m-S'_n/m|^2)(1+\hat X'_n)^2\bigg)\\
\geq &  \delta_1(v^*(\hat X_n,S_n) -u^\eta_*(\hat X'_n,S'_n))
 +o(1) -C\big(\theta_n |\hat X_n-\hat X'_n|+\theta_n|S_n-S'_n|^2\big) \\
 \geq &  \frac{\delta_1}{2} (v^*(\hat X_n,S_n) -u^\eta_*(\hat X'_n,S'_n))+o(1) >0,  \text { when $n$ is sufficiently large},
\end{aligned}\end{eqnarray*}
which yields a contradiction.
Therefore, according to the three cases above, we have $v^*\leq u^\eta_*$, which leads to the desired result by sending $\eta\rightarrow 0$. This completes the proof of Lemma  \ref{partial_comparison}.

The uniqueness of viscosity solution is a straightforward corollary of the comparison principle. Therefore, combining  Lemma \ref{supersolution}, \ref{subsolution}, \ref{partial_comparison}, and Proposition \ref{partial_growth condition}, we complete the proof of Theorem \ref{partial_viscosity property}.
\end{proof}

\begin{proposition}\label{special case} 
Assume $m> 0$. Then we have:\\
(i) If $a> 50\%$, the liquidation barrier  $I(S)$ defined in $(\ref{hat tao})$ falls in $S$ when $S\leq (\Phi^{-1}(a))^2$ and rises in $S$ when $S>(\Phi^{-1}(a))^2$. If $a\leq 50\%$,  $I(S)$ rises in $S$.\\
(ii)  When  $\Delta=0$, it is never optimal to issue equity when $\hat X_t>I(S_t)$.
\end{proposition}

\begin{proof}
(i). It is easy to verify
\begin{align*}
I'(S) = \frac{1+\kappa}{2\sqrt{S}}e^{\frac{1}{2}S-\Phi^{-1}(a)\sqrt{S}}\left(\sqrt{S}-\Phi^{-1}(a)\right).
\end{align*}
When $a<50\%$, we have $\Phi^{-1}(a)<0$ and $I'(S)\geq 0$ so that $I(S)$ decreases in $S$. When $a>50\%$, we have $\Phi^{-1}(a)>0$ so that $I'(S)\leq 0$ if $S\in[0,(\Phi^{-1}(a))^2]$ and $I'(S)>0$ if $S>(\Phi^{-1}(a))^2$, Hence, $I(S)$ first decreases in $S$ when $S\in[0,(\Phi^{-1}(a))^2]$ and then increases in $S$ when $S>(\Phi^{-1}(a))^2$. \\
(ii). Suppose that there is an admissible strategy $\pi=\{L_t^{\pi}, (t_i^{\pi}, s_i^{\pi})\}$ such that the bank issues equity before the liquidation barrier is reached.  
We need to show that $\pi$ is suboptimal. Indeed, we can construct another admissible strategy $\tilde \pi$ such that the new strategy yields a higher value function.
We define  $t^*\equiv t_{i^*}:=\min\{t_i: X_{t_i}^{\pi}>I(S_{t_i^{\pi}})\}$.  As $X_{t^*}^{\pi}>I(S_{t^*})$,  we can always find a small $\varepsilon>0$  such that $X_t^{\tilde \pi}>I(S_t)$ for all $t\in[t^*, t^*+\varepsilon]$. The new strategy $\tilde \pi:=\{L_t^{\tilde \pi}, (t_i^{\tilde \pi}, s_i^{\tilde \pi} )\}$ is defined as $L_t^{\tilde \pi}:=L_t^{\pi}$, and
\begin{equation*}
 t_i^{\tilde \pi}:= \left\{
\begin{aligned}
&t_i^{\pi}, & i& \leq  i^*,\\
&t_{i^*}^{\pi}+\varepsilon & i&= i^*+1, \\
& t_{i-1}^{\pi} & i &> i^*+1,
\end{aligned} \right.
\quad s_i^{\tilde \pi}:=\left\{
\begin{aligned}
&s_i^{\pi}, & i&< i^*,\\
&s_{i^*}^{\pi}-\xi , & i&= i^*,\\
&\xi e^{(\alpha-\mu)\varepsilon+C(t^*,\varepsilon)} & i&= i^*+1, \\
& s_{i-1}^{\pi} & i &> i^*+1,
\end{aligned} \right.
\end{equation*}
where $C(t^*, \varepsilon)=(S_{t^*+\varepsilon}/m+\sigma\rho)\tilde B_{t^*+\varepsilon}-(S_{t^*}/m+\sigma\rho)\tilde{\mathcal{B}}_{t^*}$, and $\xi$ is to be determined later.
  Obviously, the new strategy $\tilde \pi$ is admissible, and $X_t^{\tilde \pi}$ can completely replicate $X_t^{\pi}$ before $t^*$ and after $t^*+\varepsilon$. Therefore, the difference of value functions at time $t^*$ under these two strategies is given as
\begin{eqnarray*}\label{value difference of replication}
& & \hat V_{\tilde \pi}(\hat X_{t^*}, S_{t^*}; t)-\hat V_{\pi}(\hat X_{t^*}, S_{t^*}; t) \nonumber \\
&&=-e^{-(\delta-\mu)\varepsilon}\mathbb{E}\left[\xi e^{(\alpha-\mu)\varepsilon+(S_{t^*+\varepsilon}/m+\sigma\rho)\tilde B_{t^*+\varepsilon}-(S_{t^*}/m+\sigma\rho)\tilde B_{t^*}}+K|\mathcal{G}_{t^*}\right]+ \xi+K\nonumber \\
&&= \xi\left(1-e^{-(\delta-\alpha)\varepsilon +\frac{1}{2}\int_{t^*}^{t^*+\varepsilon}(S_u/m+\sigma\rho)^2du}\right)+K(1-e^{-(\delta-\mu)\varepsilon}).
\end{eqnarray*}
To make the RHS of the above equation positive, we can choose $\xi$ as
\begin{equation*}
\xi=\left\{\begin{aligned} &s_i^{\pi}  &\text{if }&  \delta - \alpha \geq A_1,\\
&\frac{(\delta-\mu)K}{2(A_1-\delta+\alpha)} &\text{if }& \delta - \alpha \leq A_1\end{aligned}\right.
 \end{equation*}
 where $A_1=\max(\sigma^2/2, (S/m+\sigma\rho)^2/2)$.
 Then we have $\hat V_{\tilde \pi}(\hat X_{t^*}, S_{t^*}; t)>\hat V_{\pi}(\hat X_{t^*}, S_{t^*}; t) $. This implies that $\tilde \pi$ is better than  $\pi$. This yields the desired result.
\end{proof}

Note that the liquidation barrier $I$ depends on the variance $S$. 
If regulators weight more the risk of liquidating a solvent bank than the risk of not liquidating an insolvent bank (i.e., if $c_2/(c_1+c_2)=a>50\%$) and if the accounting asset uncertainty level $S<(\Phi^{-1}(a))^2$, then the regulators  prefer  to postpone  bank liquidation.
That is, if the regulators liquidate a bank, then they want to be sure that the bank is insolvent or the probability for that is high.
 However, under our model parameters ($a=80\%$ and $S<(\Phi^{-1}(a))^2$), by the result above, the bank has an incentive to increase the uncertainty level $S$ since the uncertainty hides the insolvency risk and the liquidation barrier $I(S)$ falls, and thus, the bank is liquidated later.

Note that if there is no uncertainty in the capital issuance ($\Delta$ is zero) then the bank does not face default risk since it can always order new equity just before its default time. This conflicts real financial markets as many banks issue equity during some special periods even when their Capital ratio stays above the minimum requirement.  In reality, it takes time for banks to sell equity in open markets.

Figure \ref{figure region depiction} depicts the dividend region, continuation region, and recapitalization region in the $\hat X$-$S$ plane. Intuitively, for any given accounting asset uncertainty level $S$, banks should pay dividends if the equity-to-debt ratio is sufficiently high, and issue new equity if the ratio is low enough. In the continuation region, banks take no action.
\begin{figure}[ht]
\centering
  \includegraphics[height=7.2cm,width=8.4cm]{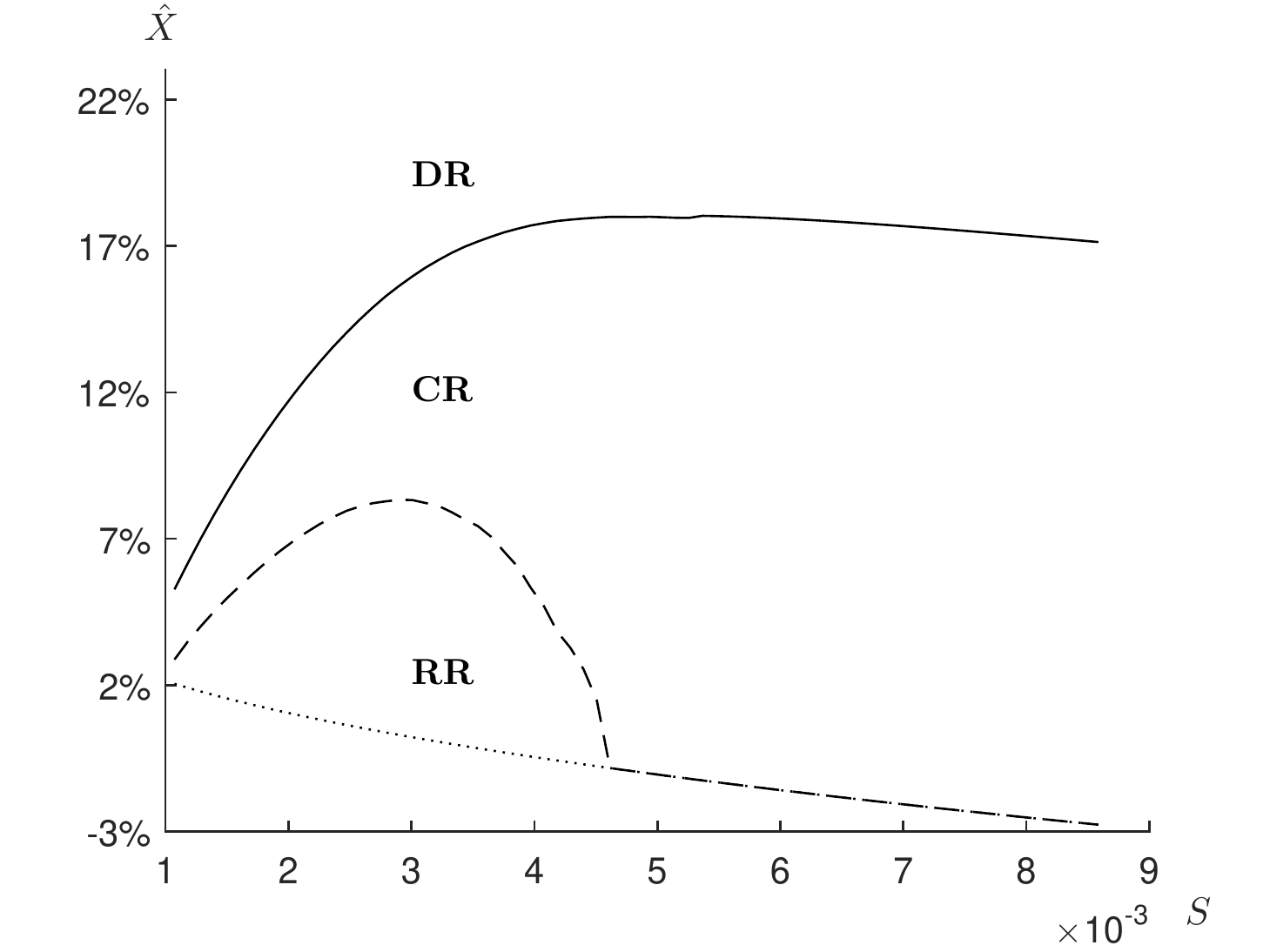}
	\caption[Three regions of the partially observed model]{\textbf{Three regions of the partially observed model.} This figure depicts the three regions, $\textbf{RR}, \textbf{CR}$, and $\textbf{DR}$ for the partially observed model. They represent the recapitalization region, continuation region, and dividend region, respectively. The solid line is the dividend boundary (between $\textbf{DR}$ and $\textbf{CR}$), the dashed line is the recapitalization boundary (between $\textbf{CR}$ and $\textbf{RR}$), and the dotted line is the liquidation barrier $I(S)$. $\hat{X}$ is the expected equity-to-debt ratio and $S$ is the accounting asset uncertainty level. The parameters are chosen from Panel B in Table \ref{table  allbankparameter}. $\alpha=12.85\%$, $\mu=10.52\%$, $\delta=25.70\%$, $\sigma=5.21\%$, $\kappa=4.80\%$, $\hat \kappa=1.15\%$, $a=79.93\%$, $\omega=25.10\%$, $m=2.85\%$, $\rho=-26.71\%$, $\Delta=0.50$, and $K=0.20\%.$ }
 \label{figure region depiction}
\end{figure}

\section{Comparative Statics}\label{comparative statics}
In this section we analyze how different parameters affect the value function and optimal strategy in Subsection \ref{partially_theorysection}. As we cannot solve the HJB (\ref{hjb3}) analytically, we solve it numerically.
When $S=m \sigma (1-\rho)$, we have $S_t\equiv m \sigma (1-\rho)$ for all $t\geq0$. In this special case, by the definition of value function in (\ref{define of value function with recap and partially observed}) and (\ref{v(x,S)}), $\hat V(\hat X,m \sigma (1-\rho))$ degenerates to a fully observed case with new liquidation barrier $\kappa_1=I(m\sigma(1-\rho))$ and new liquidation value $\omega_1$, where 
\begin{eqnarray*}\begin{aligned}
\omega_1: &= \frac{\omega }{\kappa_1} \psi(I(m \sigma (1-\rho)),m \sigma (1-\rho))\\
&= \frac{\omega }{\sqrt{2\pi}\kappa_1}\int_{\mathbb{R}} \left[(\kappa_1+1)e^{-\frac{m \sigma (1-\rho)}{2}+u \sqrt{m \sigma (1-\rho)}}-1\right]^{+} e^{-\frac{u^2}{2}} du.
\end{aligned}\end{eqnarray*}
 By Theorem \ref{semi-explicit solution with recap}, we know $\hat V(\hat X; m \sigma (1-\rho)) = V(\hat X; \kappa_1, \omega_1)$, which is the boundary condition of $\hat V(\hat X,S)$ at $S=m\sigma (1-\rho)$.
Using the penalty method in \cite{DZ08}, we  numerically solve the HJB equation (\ref{hjb3}) separately in regions
\[
\{(\hat X,S): \hat X\geq I(S), S\leq m\sigma(1-\rho)\} \text{ and } \{(\hat X,S): \hat X\geq I(S),  m\sigma(1-\rho) \leq S <\bar{S}\}
\]
with boundary conditions
 $ \hat V(I(S),S)
 =\frac{\omega }{\sqrt{2\pi}}\int_{\mathbb{R}} \left[(I(S)+1)e^{-\frac{S}{2}+u \sqrt{S}}-1\right]^{+} e^{-\frac{u^2}{2}} du$ and $ \hat V(\hat X, m\sigma(1-\rho))=V(\hat X; \kappa_1, \omega_1).$

By $(\ref{S(t)})$, when $t$ is high, that is, when the shareholders and regulators have followed the noisy accounting information for a long time, then $S=m\sigma(1-\rho)$, and it is constant.
In this section, we analyze how the regulators' liquidation parameter $a$ (which equals $c_2/(c_1+c_2)$), minimum equity-to-debt ratio $ \kappa$, asset return volatility $\sigma$, accounting information noise $m$, and signal and asset return correlation $\rho$ affect the bank market value and the optimal dividend and recapitalization policy under the long-term accounting information noise level $(S=m\sigma(1-\rho))$ and the model parameters estimated in Section~\ref{calibration} (note in particular that the regulators' liquidation parameter $a$ is about 80\%).
Figure \ref{figure st} illustrates that after 1.5 years, the noise level is close to the long-term level even though initially, the noise level is twice the long-term level. We use $(\ref{hat X_t})$ to understand some of the comparative statics discussed below.

By the dynamic process (\ref{hat X_t}), the volatility of $\hat X_t$ is $S_t/m+\sigma\rho$, and we use this to understand some of the comparative statics discussed below.

\begin{figure}[H]
\centering
  \includegraphics[height=6.8cm,width=7.8cm]{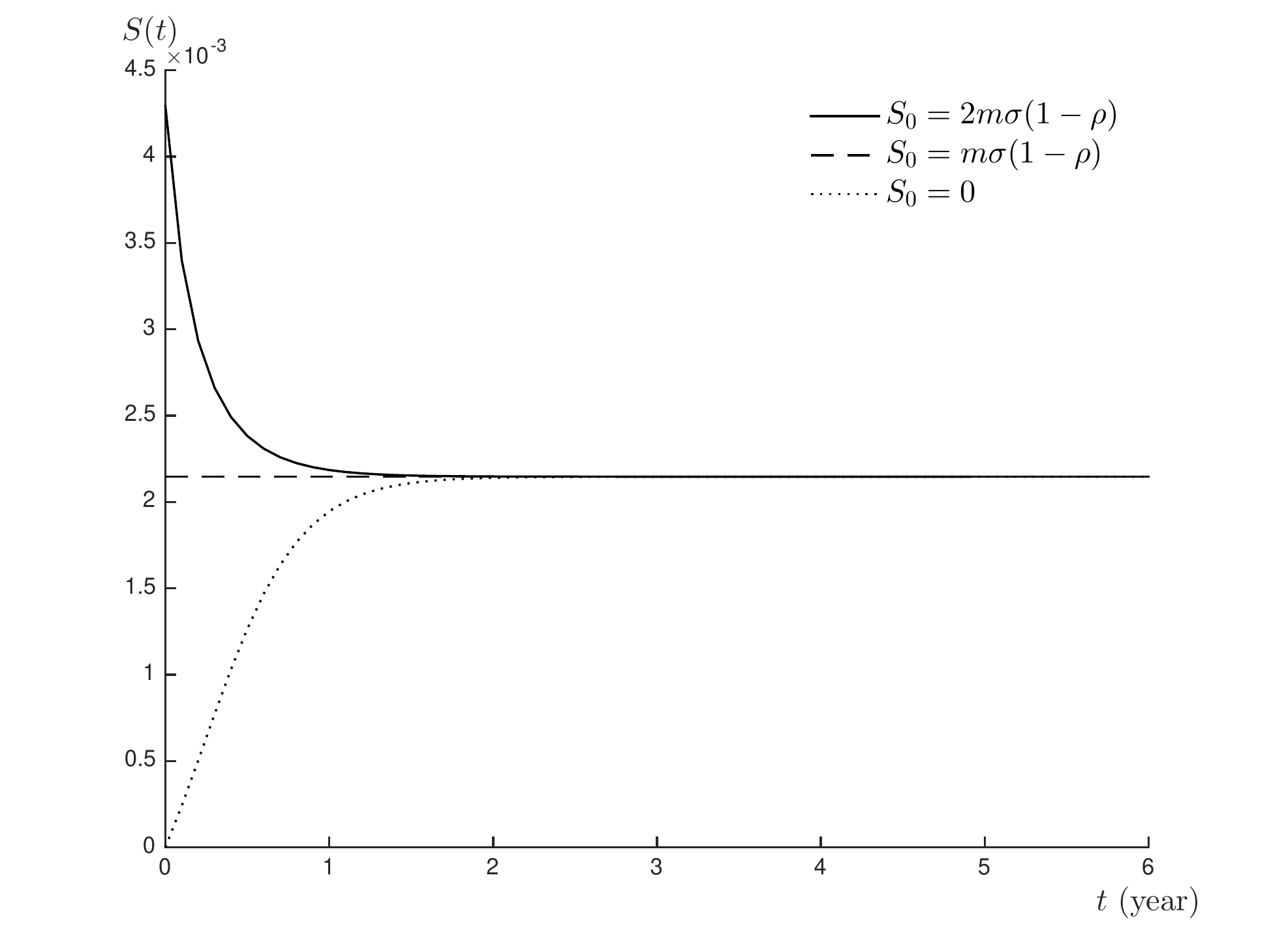}
 \caption[Accounting asset noise $S_t$ over time]{\textbf{Accounting asset noise $S_t$ over time.} This figure shows $S_t$ with respect to time $t$ under the parameter estimates in Table \ref{table  allbankparameter}. Thus, the parameter values: $\sigma=5.19\%$, $m=2.96\%$, and $\rho=-39.86\%$. We observe the accounting asset noise converges to its long-term limit $m\sigma(1-\rho)$.}
 \label{figure st}
\end{figure}

Figure \ref{figure st} illustrates that after 1.5 years the noise level is close to the long-term level even though initially the noise level is twice the long-term level.

Figure~\ref{recap_partiallymodelvalue}(i) shows how the regulators' liquidation parameter $a$ affects the bank's market equity value in terms of bank debt, that is, $\hat V(\hat X,S)$.
If the regulators weight more the risk of liquidating a solvent bank than the risk of not liquidating an insolvent bank (i.e., if $c_2>c_1$), then they liquidate the bank only when there is more than 50\% probability that the bank is insolvent.
 Under our parameter estimates in Table~\ref{table allbankparameter}, the liquidation parameter $a = 80\%$ and the partially observed model gives a higher market equity value than the corresponding fully observed model.
Thus, in this case the bank benefits from noisy accounting information.
On the other hand, if the regulators can liquidate the bank even if they are less than 50\% certain that the bank is insolvent, then the partially observed model gives a lower value than the fully observed model.
Figure~\ref{recapbarrier2}(i) shows that when the liquidation parameter $a$ rises, then the liquidation barrier $I$ falls. This decreases both the recapitalization barrier $u_1$ and the dividend barrier $u_2$.

Figure~\ref{recap_partiallymodelvalue}(ii) shows that, as expected, the higher the minimum equity-to-debt ratio $\kappa$ is, the lower the bank value.
Figure~\ref{recapbarrier2}(ii) shows that when $\kappa$ rises, the bank hedges the liquidation risk by raising all the barriers ($I$, $u_1$, and $u_2$).

Figure~\ref{recap_partiallymodelvalue}(iii) shows that when $a=80\%$ then the liquidation risk falls in uncertainty level S and, therefore, the bank equity value rises.  
By Proposition \ref{special case}(i), the opposite is true when $a=20\%$ (this is true for all $a<50\%$).
Figure
~\ref{figure region depiction}  shows first that, by Proposition \ref{special case}(i), when $a=80\%$, then the liquidation barrier $I$ falls in $S$ when $S$ is below $(\Phi^{-1}(0.8))^2 = 0.71$, which is greater than the limit $m\sigma(1-\rho)=0.0019$.
The opposite is true when $a=20\%$.
When $S$ is low in Figure
~\ref{figure region depiction},
then the diffusion term of the expected equity-to-debt ratio in $(\ref{hat X_t})$ is close to zero. Therefore, at low values of $S$, $\hat X_t$ is close to a deterministic process, and the dividend barrier $u_2$ and the recapitalization barrier $u_1$ are close to each other.
However, when $S$ rises, then $\hat X_t$ fluctuates more, and to minimize frequent recapitalization costs, the barriers depart from each other. Further, when $S$ is high, the bank finds it optimal not to use the recapitalization option at all. This is also illustrated  in Figure~\ref{u2 and u1 wrt delay}.

Since we have $S=m\sigma(1-\rho)$, $S$ increases in $\sigma$ and $m$ and decreases in $\rho$. Therefore, $\sigma$ and $m$ have a similar, but weaker, effect on the bank value as $S$ has, and $\rho$ has an opposite effect (compare Figure~\ref{recap_partiallymodelvalue}(iv), Figure~\ref{recap_partiallymodelvalue}(v), and  \ref{recap_partiallymodelvalue}(vi) with Figure~\ref{recap_partiallymodelvalue}(iii)).
By $(\ref{hat X_t})$ and $S=m\sigma(1-\rho)$, when $\sigma$ is low in Figure~\ref{recapbarrier2}(iii), then the diffusion term of the expected equity-to-debt ratio is almost zero, and as with $S$ in Figure~\ref{figure region depiction}, the dividend barrier $u_2$ and the recapitalization barrier $u_1$ are close to each other. Further, when $\sigma$ rises, the liquidation, recapitalization, and dividend barriers behave similarly as with respect to $S$.
Given $(\ref{hat tao})$ and Proposition \ref{special case}(i), the liquidation barrier falls in $m$ in Figure~\ref{recapbarrier2}(iv) and rises in $\rho$ in Figure~\ref{recapbarrier2}(v).
The liquidation barrier changes the dividend and recapitalization barriers, respectively.

\begin{figure}[H]
\centering
\begin{tabular}{cc}
\includegraphics[width=0.46\textwidth]{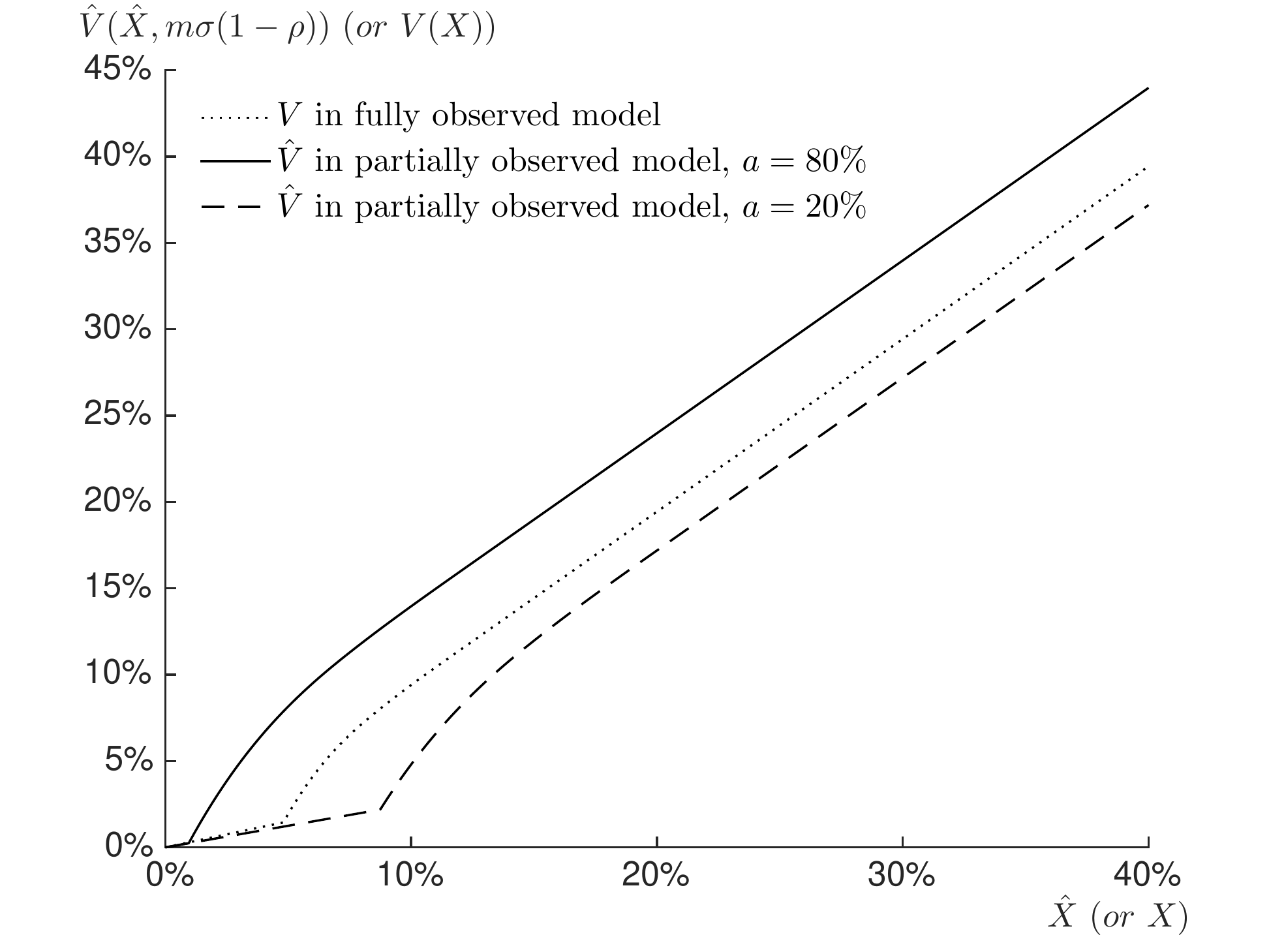} & \hspace{-.5cm}
\includegraphics[width=0.46\textwidth]{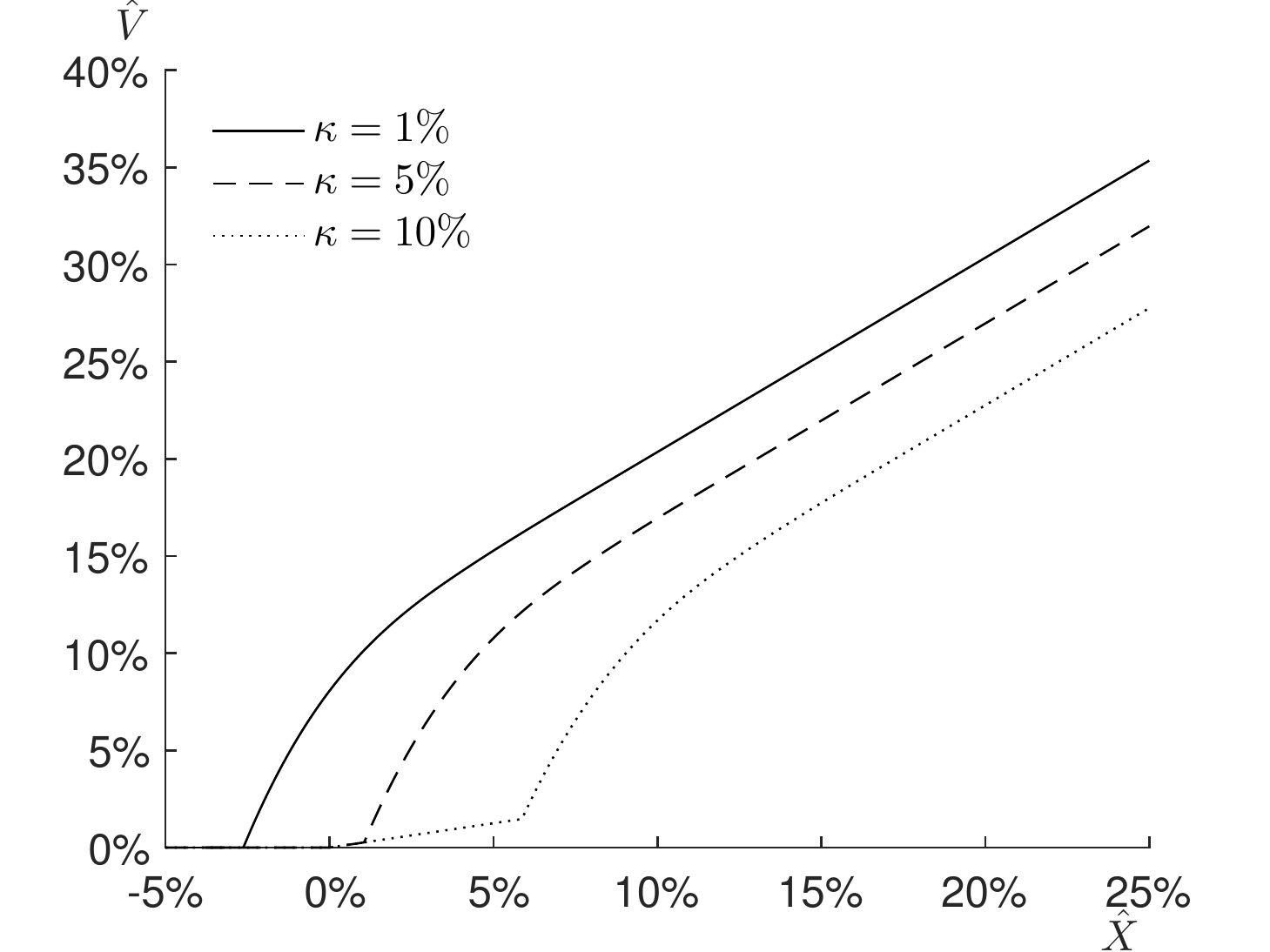} \\
{\small (i) regulators' confidence level $a$} & {\small (ii) minimum capital level $\kappa$}  \\
\includegraphics[width=0.46\textwidth]{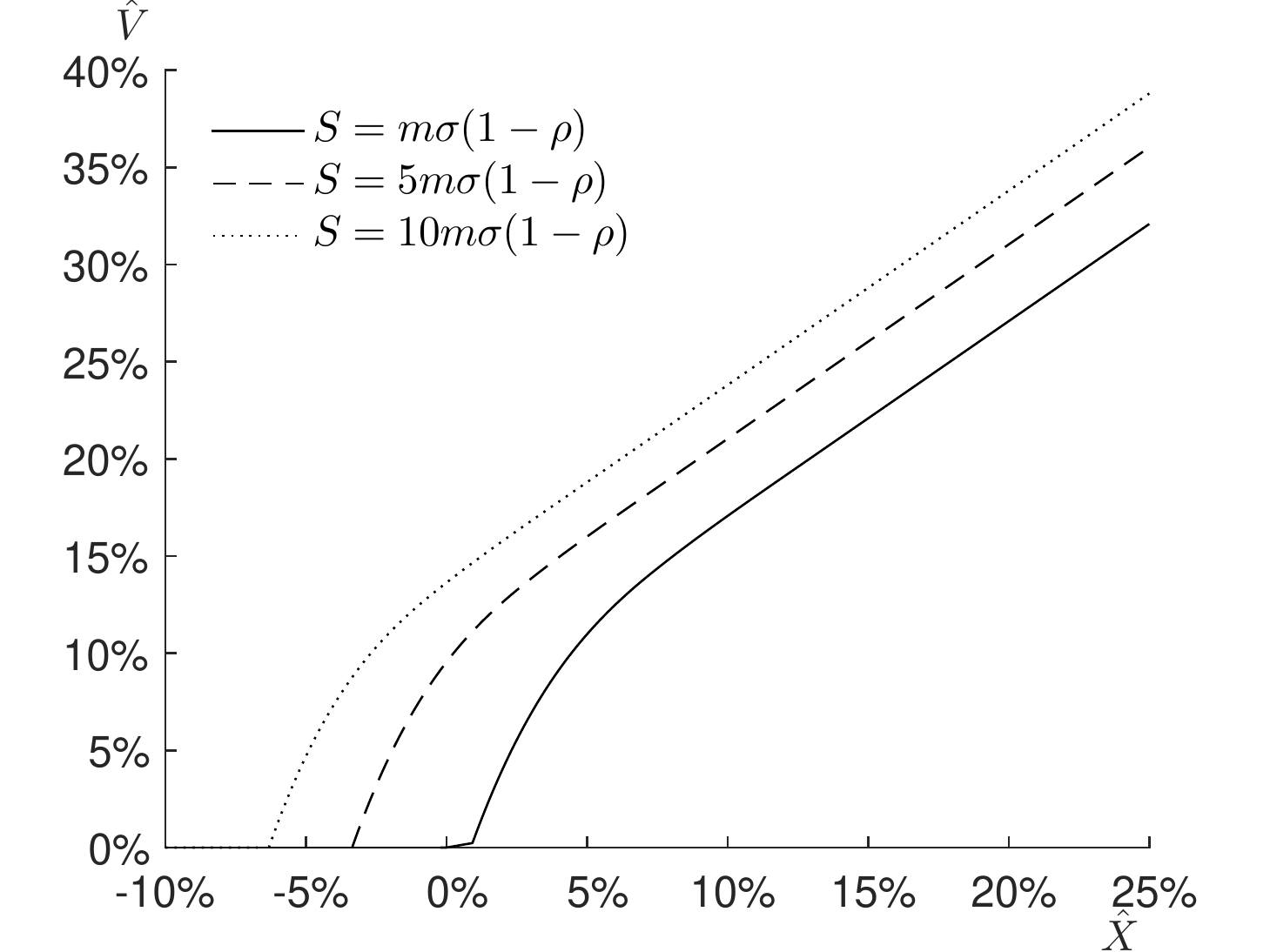} &\hspace{-.5cm}
\includegraphics[width=0.46\textwidth]{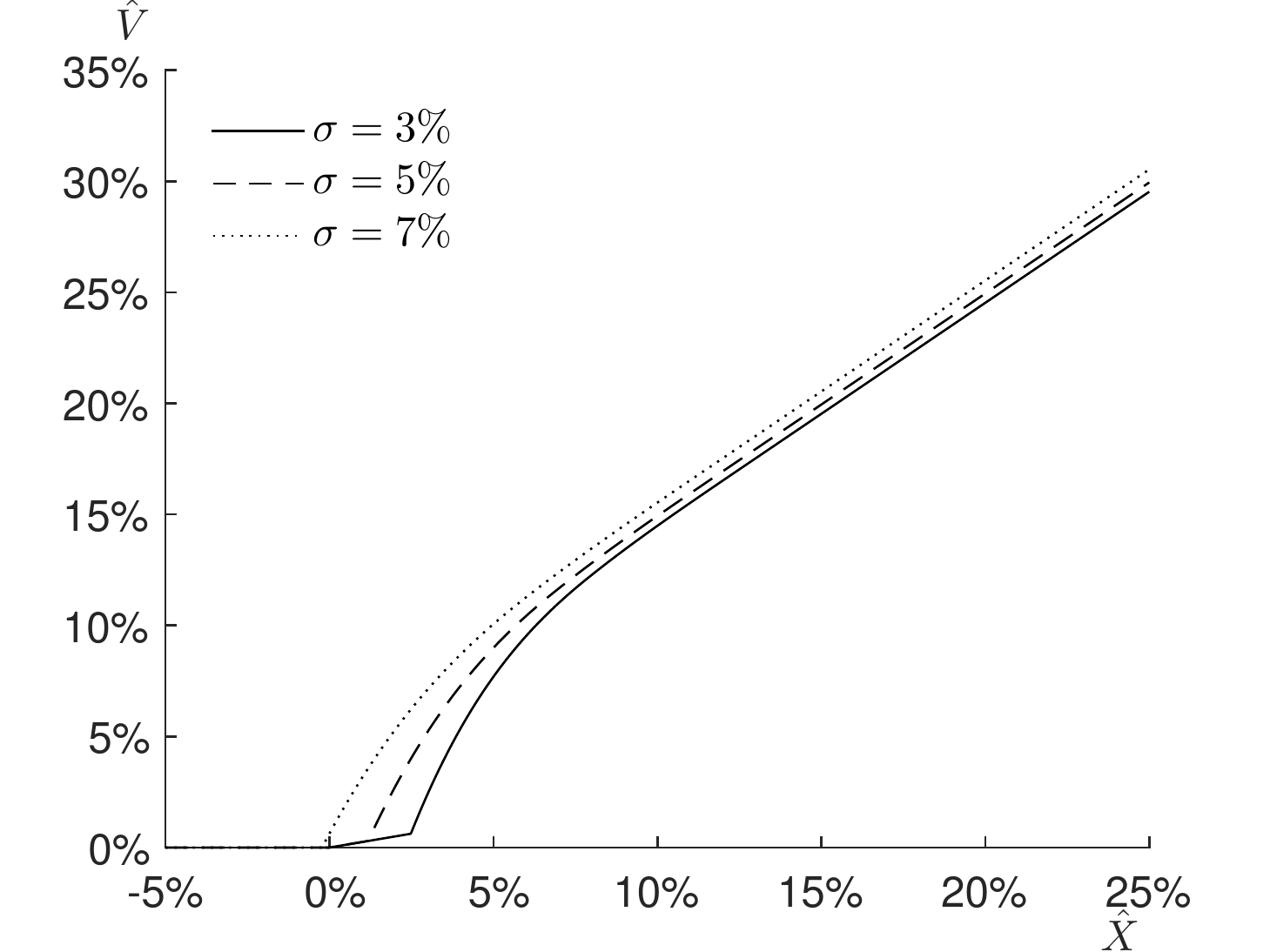} \\
{\small (iii) accounting asset uncertainty $S$} & {\small (iv) asset return volatility $\sigma$} \\
\includegraphics[width=0.46\textwidth]{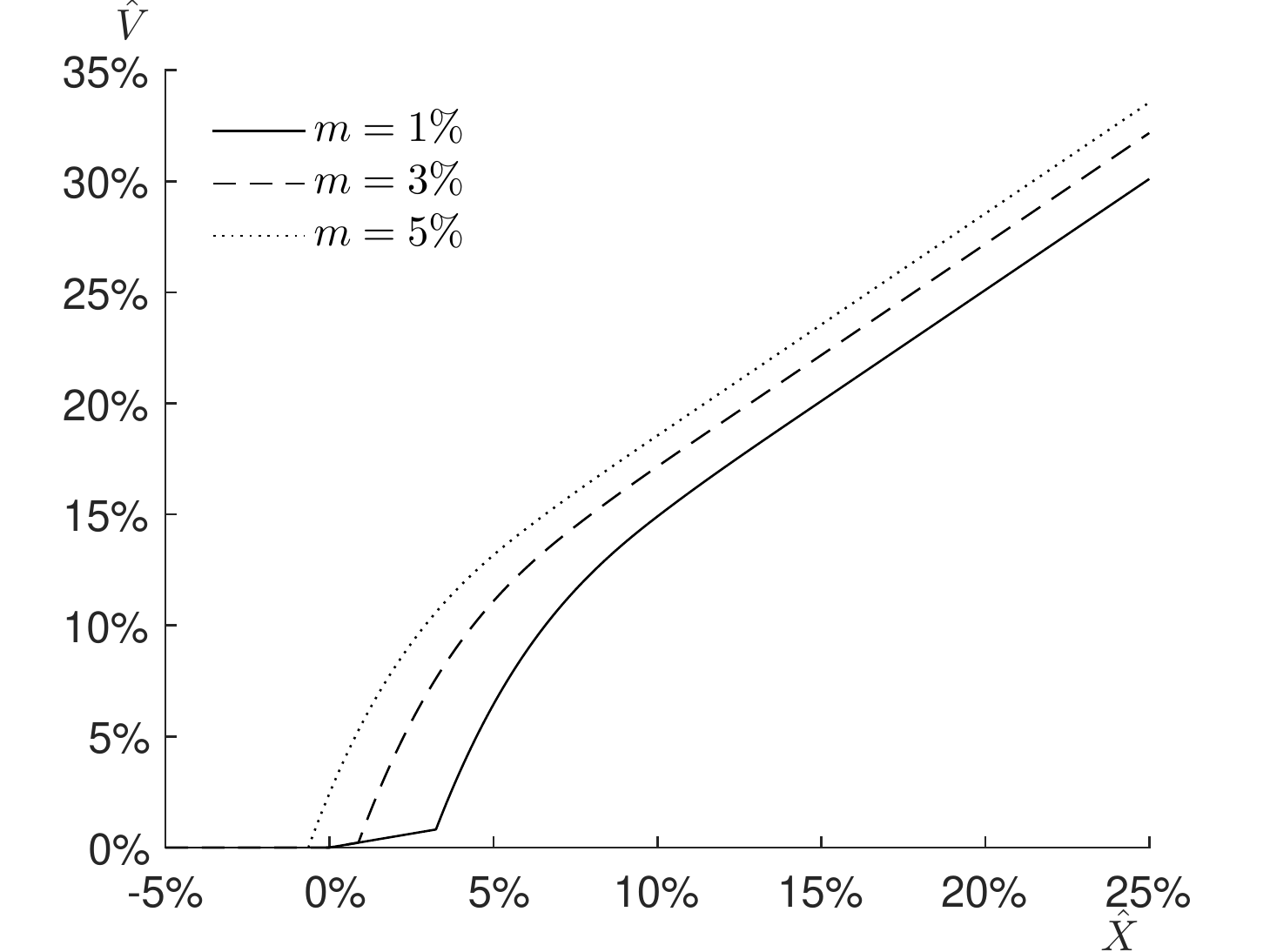} &\hspace{-.5cm}
\includegraphics[width=0.46\textwidth]{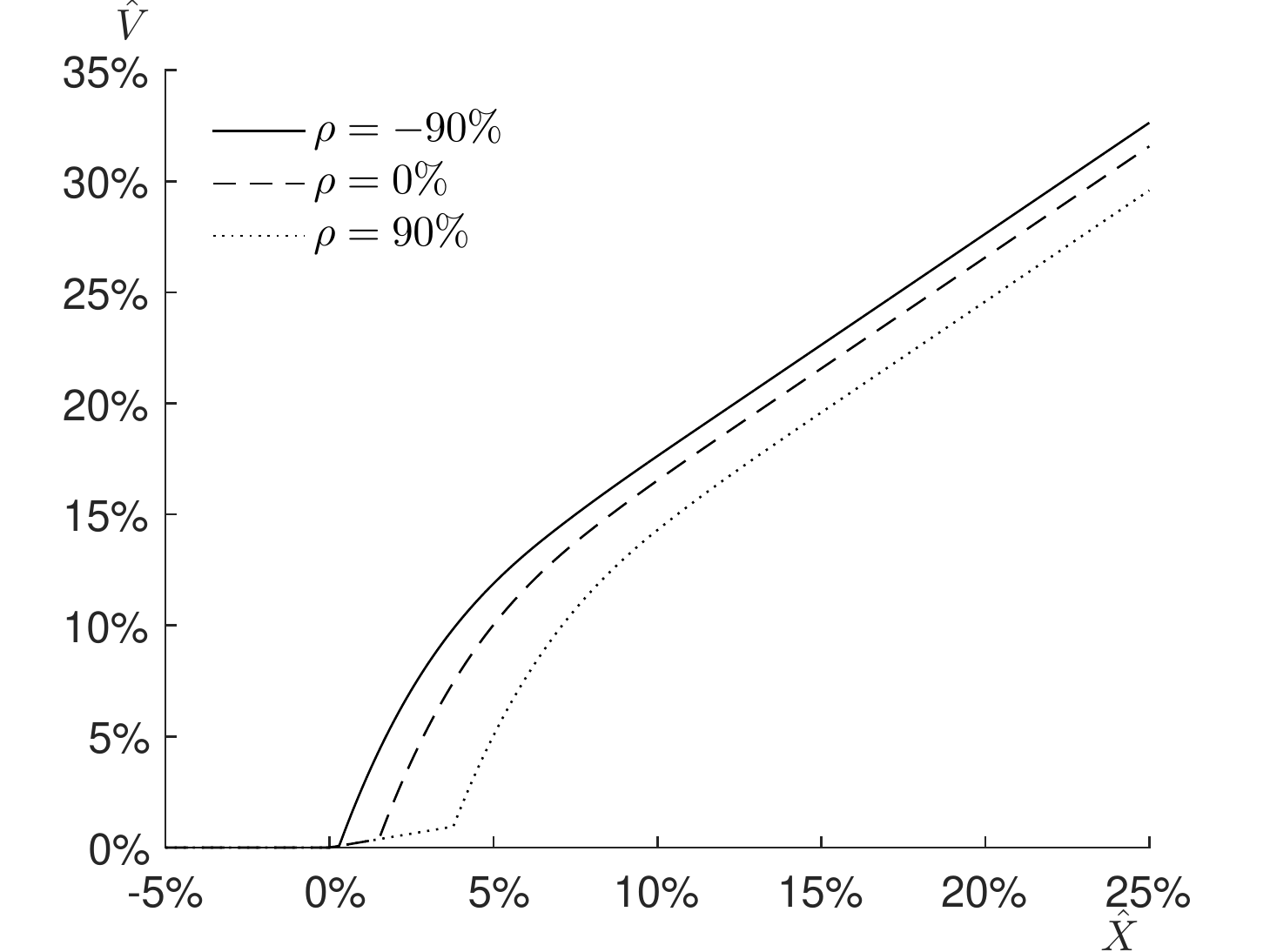}\\
{\small (v) accounting information noise $m$} & {\small (vi) signal and asset return correlation $\rho$}
\end{tabular}
\caption[The value function with different parameters]{\textbf{The value function with different parameters.}
  The parameter values are from Table \ref{table allbankparameter}: $\alpha=12.85\%$, $\mu=10.52\%$, $\delta=25.70\%$, $\sigma=5.21\%$, $\kappa=4.80\%$, $\hat \kappa=1.15\%$, $a=79.93\%$, $\omega=25.10\%$, $m=2.85\%$, $\rho=-26.71\%$, $\Delta=0.50$, and $K=0.20\%.$ In all the panels except in (iii), the accounting asset uncertainty $S$ equals its long-term level $m\sigma(1-\rho)$.}
  \label{recap_partiallymodelvalue}
\end{figure}

\begin{figure}[H]
\centering
\begin{tabular}{cc}
\includegraphics[width=0.46\textwidth]{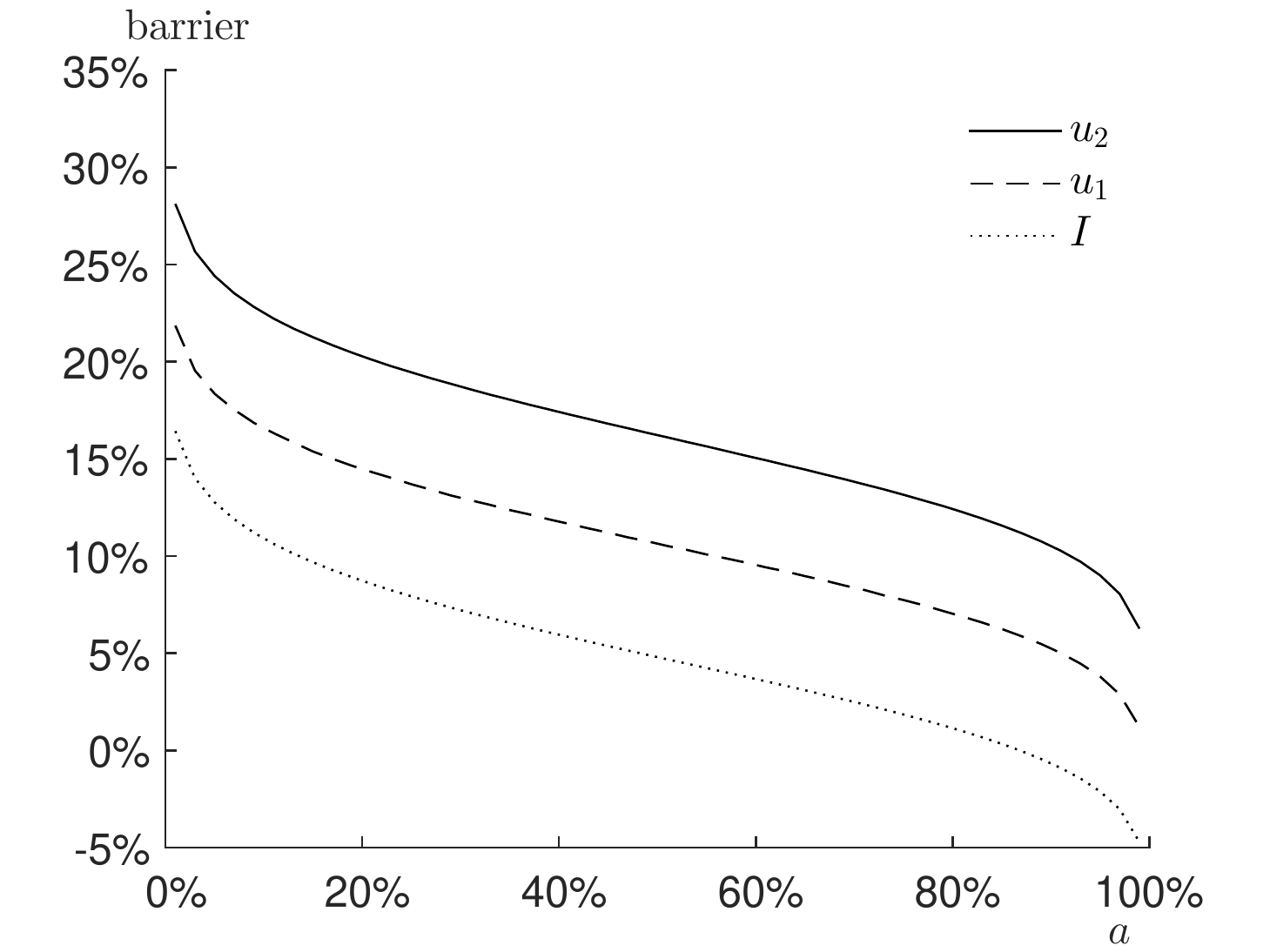} & \hspace{-.5cm}
\includegraphics[width=0.46\textwidth]{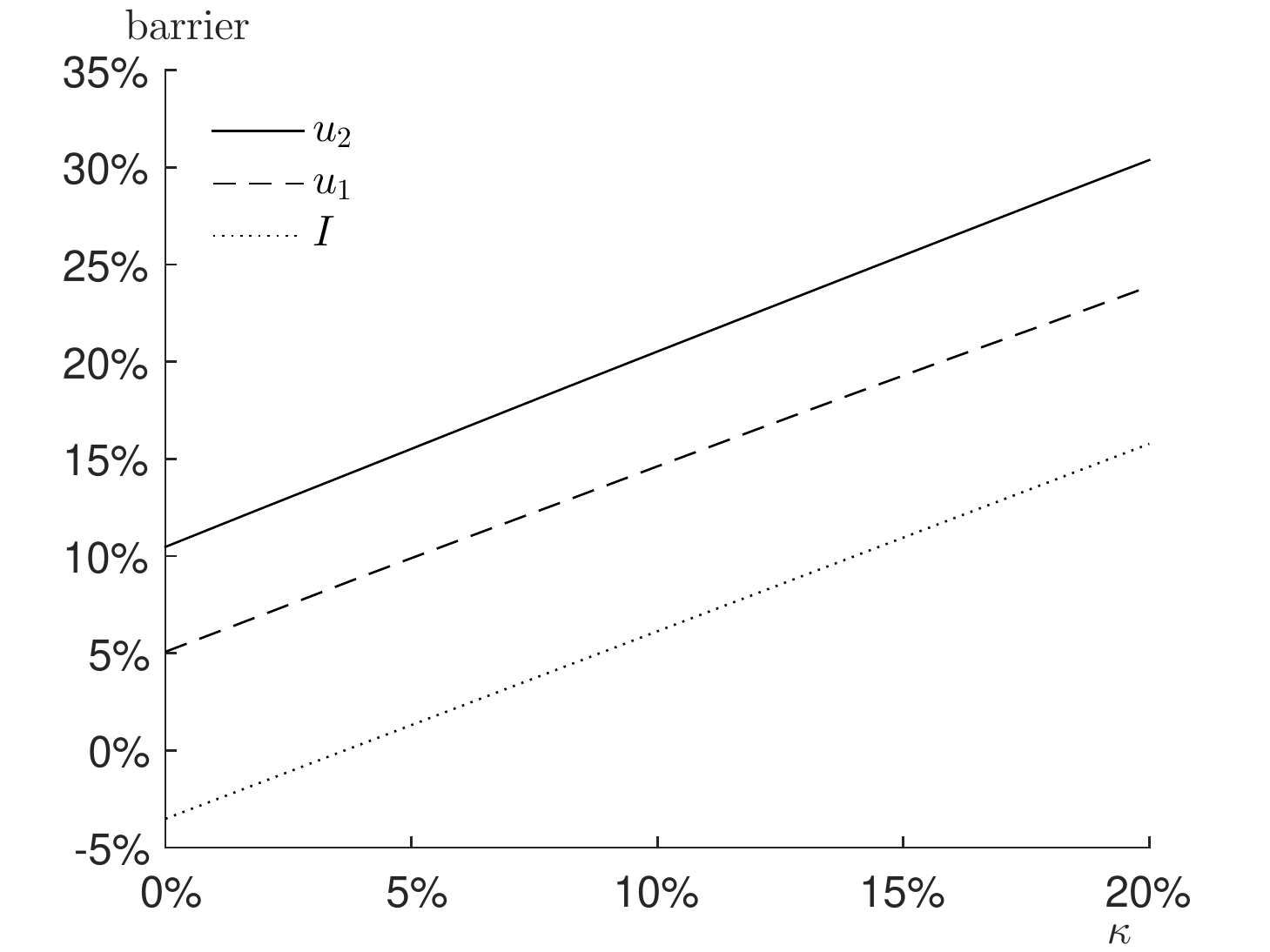} \\
{\small (i) barriers with different $a$} & {\small (ii) barriers with different $\kappa$}  \\
\includegraphics[width=0.46\textwidth]{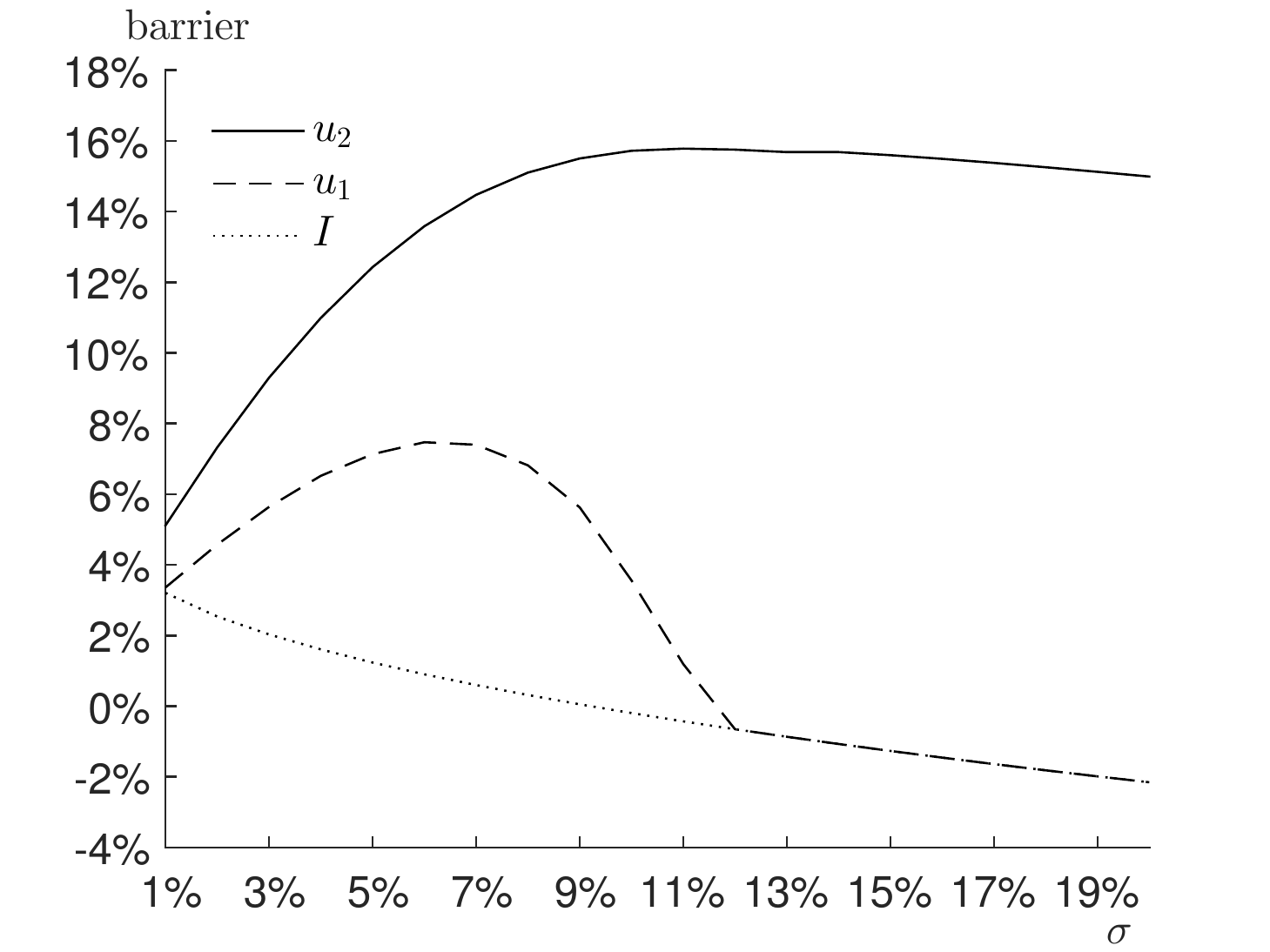} &\hspace{-.5cm}
\includegraphics[width=0.46\textwidth]{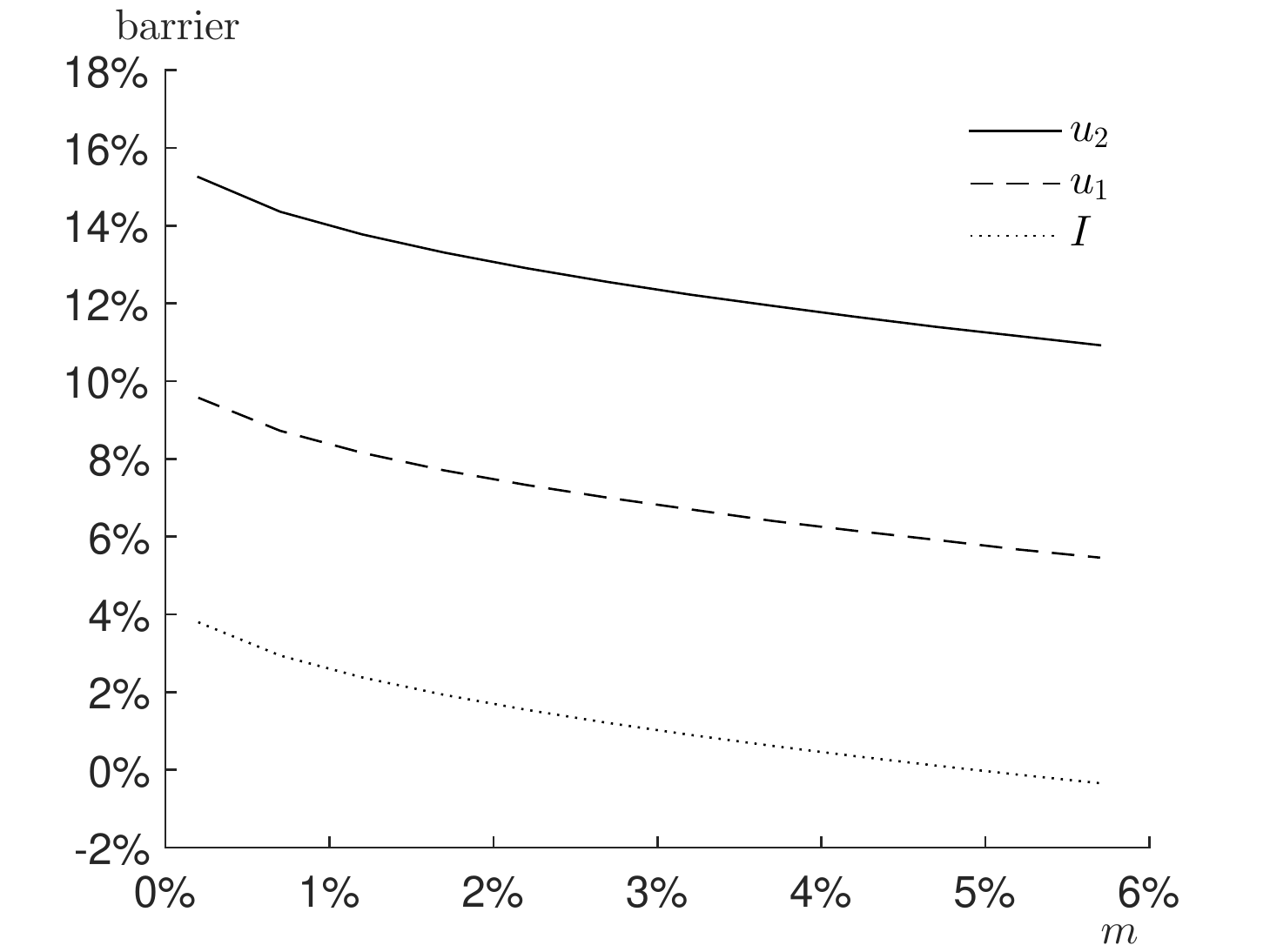} \\
{\small (iii) barriers with different $\sigma$} & {\small (iv) barriers with different $m$} \\
\includegraphics[width=0.46\textwidth]{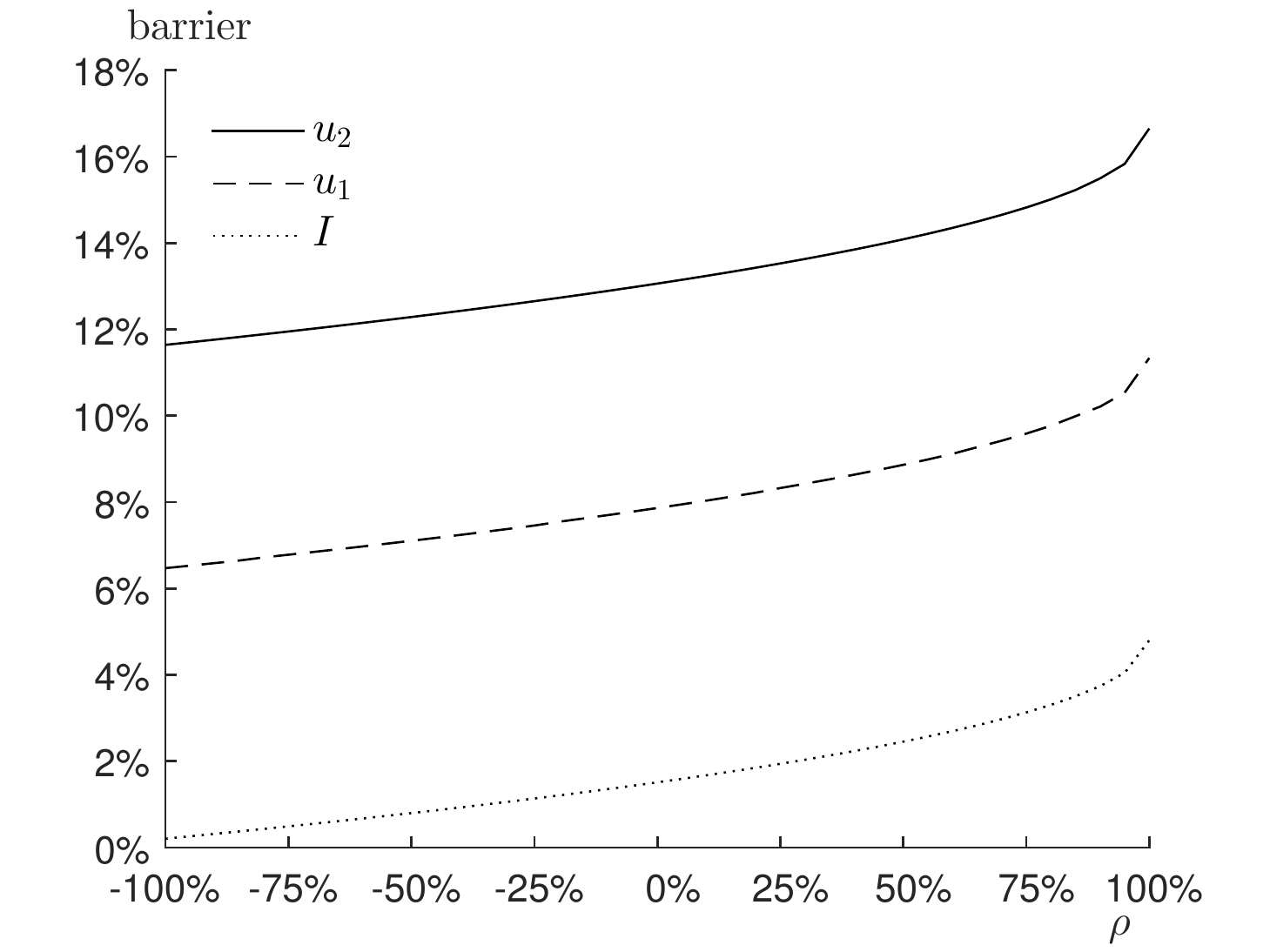} &\hspace{-.5cm}
\includegraphics[width=0.46\textwidth]{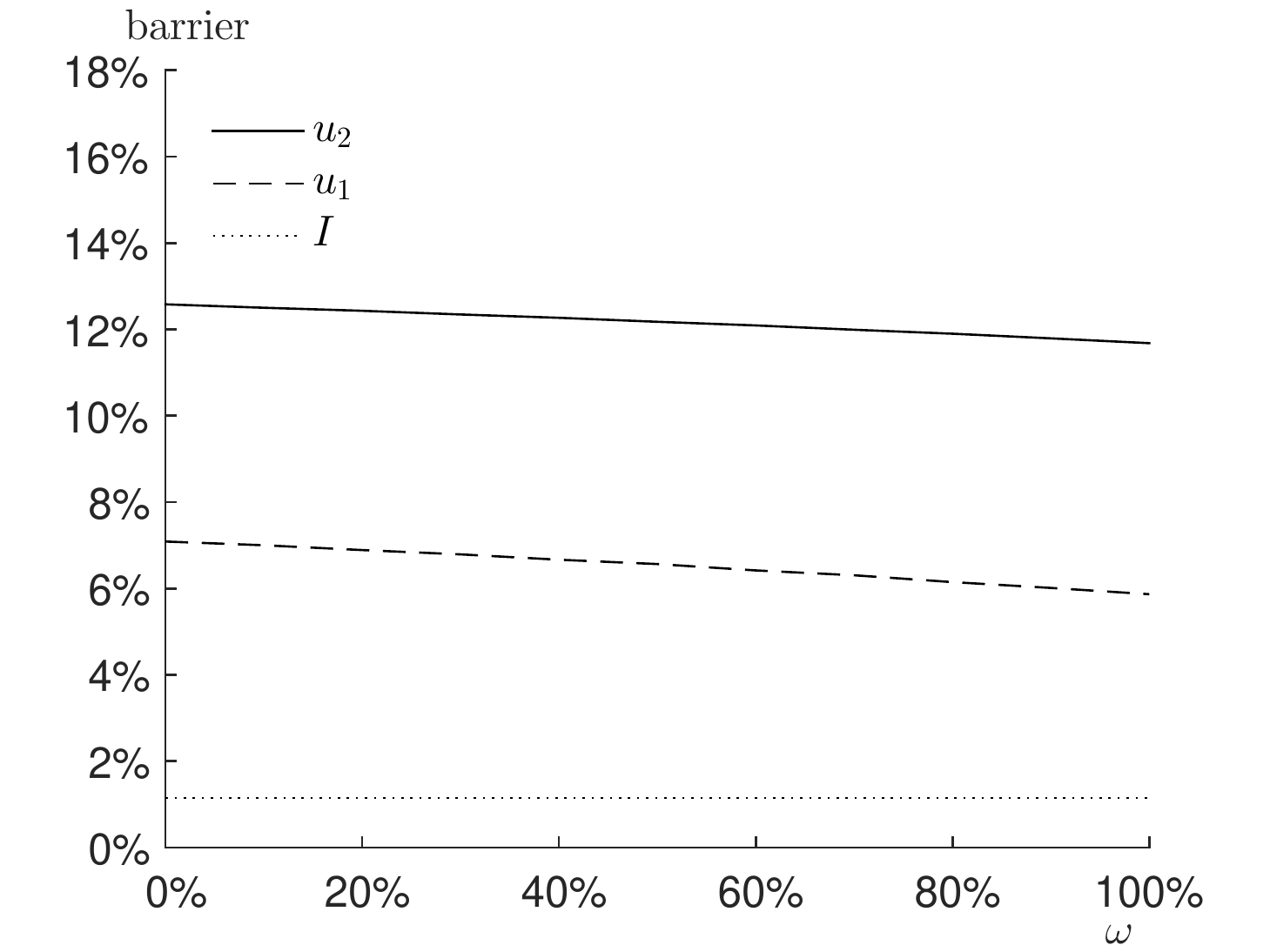}\\
{\small (v) barriers with different $\rho$} & {\small (vi) barriers with different $\omega$}
\end{tabular}
\caption[Optimal barriers $(u_2$, $u_1$, and $I)$ under different parameter values]{\textbf{Optimal barriers $(u_2$, $u_1$, and $I)$ under different parameter values.}
  The parameter values are from Table \ref{table allbankparameter}: $\alpha=12.85\%$, $\mu=10.52\%$, $\delta=25.70\%$, $\sigma=5.21\%$, $\kappa=4.80\%$, $\hat \kappa=1.15\%$, $a=79.93\%$, $\omega=25.10\%$, $m=2.85\%$, $\rho=-26.71\%$, $\Delta=0.50$, and $K=0.20\%.$ In all the panels, the accounting asset uncertainty $S$ equals its long term level $m\sigma(1-\rho)$.}
\label{recapbarrier2}
\end{figure}

\begin{figure}[tbh]
\centering
\begin{tabular}{cc}
\includegraphics[width=0.46\textwidth]{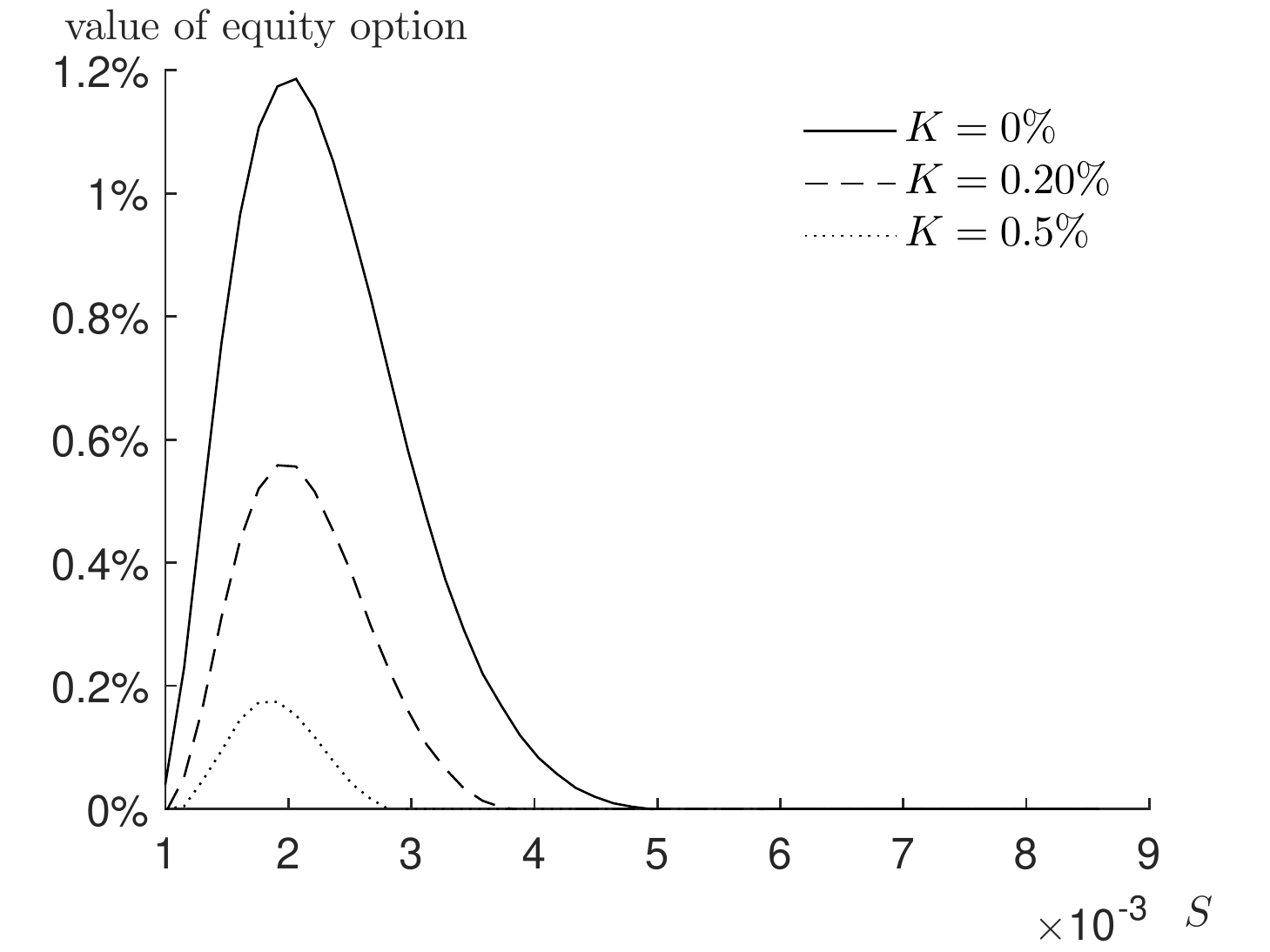} & \hspace{-.5cm}
\includegraphics[width=0.46\textwidth]{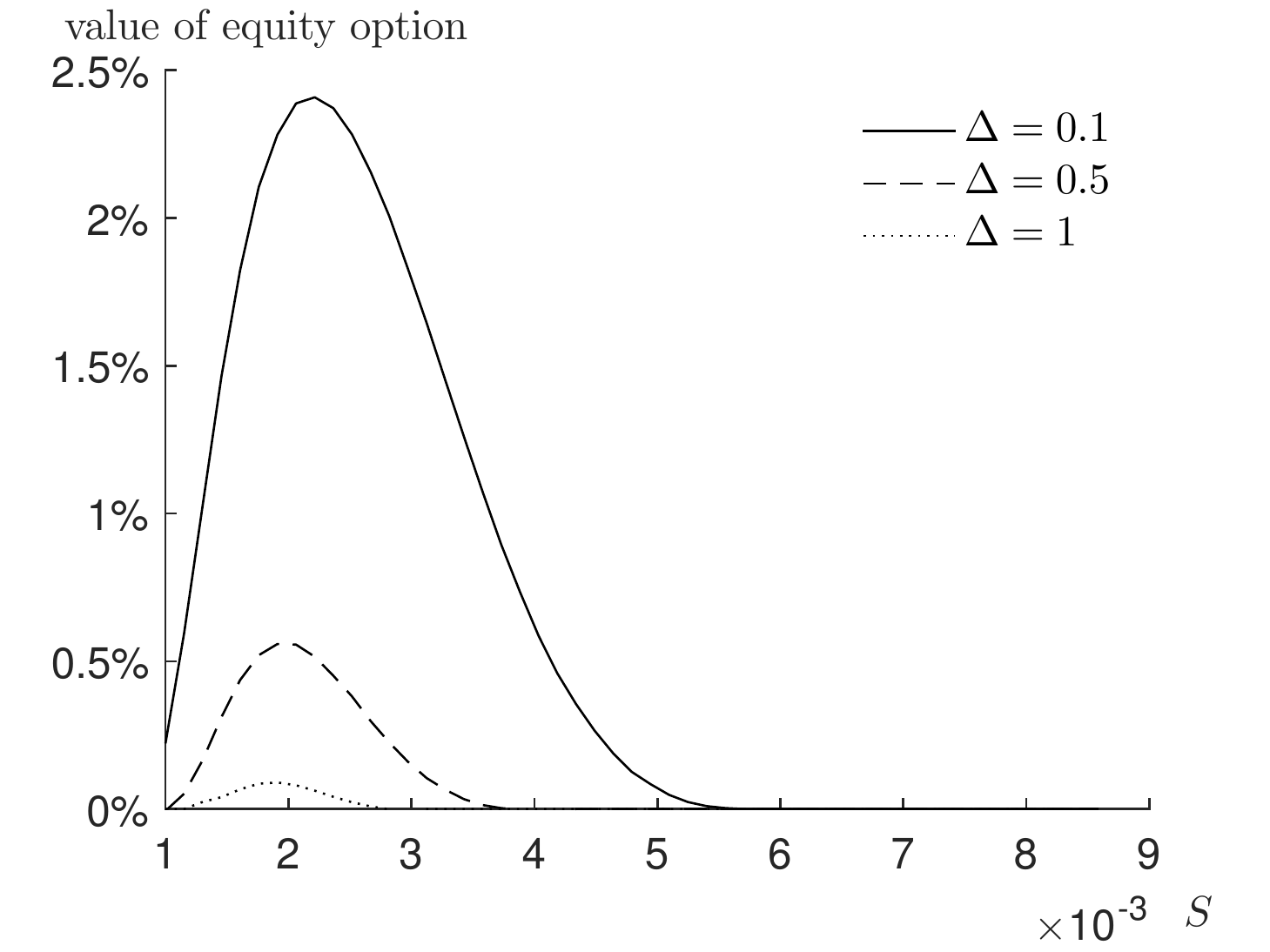} \\
{\small (i) } & {\small (ii) }
\end{tabular}
\caption[Three barriers  and the value of equity option with different ${K}$ and ${\Delta}$]{\textbf{Three barriers  and the value of equity option with different $\mathbf{K}$ and $\mathbf{\Delta}$.} Value of equity option is plotted against static uncertainty $S$ with different cost $K$ and delay time $\Delta$ for partially observed model. The value of equity issue option is the difference between value function with and without recapitalization. This value goes to $0$ when static uncertainty $S$ becomes too large.  Parameter values are from Table \ref{table allbankparameter}: $\alpha=13.25\%$, $\mu=10.42\%$, $\delta=26.95\%$, $\sigma=5.19\%$, $\hat \kappa_0=0.92\%$, $m=2.96\%$, $\rho=-39.86\%$, $\Delta=0.5$, and $K=0.20\%.$ }
\label{u2 and u1 wrt delay}
\end{figure}

By $(\ref{define of value function with recap and partially observed})$, the bank value rises in $w$, the proportional liquidation value in terms of book equity (not reported for brevity). Figure~\ref{recapbarrier2}(vi) shows that when $w$ rises then the bank takes more risk by taking more leverage, and thus, the dividend barrier $u_2$ and the recapitalization barrier $u_1$ fall. 

To analyze further the sensitivity of the optimal policy and bank value with respect to the accounting asset uncertainty level $S$, volatility $\sigma$, accounting information noise $m$, proportional liquidation value $\omega$, correlation $\rho$, equity issuance time delay $\Delta$, and fixed cost $K$, we calculate the elasticities for the dividend barrier, recapitalization barrier, and liquidation barrier, as well as the model market value with respect to these model parameters. For example, the elasticity with respect to $\sigma$ for the dividend barrier $u_2$ is given by
\begin{eqnarray*}
\text{Average}\left\{\frac{d(\text{dividend barrier }u_2)/\text{ dividend barrier }u_2}{d(\text{volatility }\sigma)/\text{ volatility }\sigma}\right\},
\end{eqnarray*}
where the differential $dx$ represents an infinitely small change in the variable $x$, and the average is over the values in Figure~\ref{recapbarrier2}(iii).
When we calculate the elasticity with the liquidation barrier $I$ or correlation $\rho$, they might be negative. In that case, we use absolute value in the denominators above.
Note that with other variables than $S$, we have $S = m \sigma (1-\rho)$.

The elasticity results are in Table \ref{elasticities}.
The table implies that the uncertainty parameters $S$, $\sigma$, $m$, and $\rho$ affect the dividend and recapitalization policy and bank value more than the recapitalization frictions $\Delta$ and $K$.
The dividend barrier $u_2$ is most sensitive to the accounting asset uncertainty $S$ and accounting information noise $m$, the recapitalization barrier $u_1$ is most sensitive to the accounting information noise $m$ and volatility $\sigma$, and the liquidation barrier $I$ is most sensitive to the  accounting asset uncertainty $S$ and  volatility $\sigma$.
The signs of the elasticities are as explained above.
For instance, a $1\%$ increase in the accounting asset uncertainty decreases the liquidation barrier by about $3.8\%$.
The bank value is most sensitive to the accounting information noise $m$ and the correlation $\rho$.
That is, asset smoothing raises the bank value. For instance, a $1\%$ increase in the accounting information noise raises the bank value by about $0.04\%$.
Thus, if bankers have equity-based compensation, they have an incentive to create and increase the noise.

\begin{table}[ht]
\caption[Elasticities]{\textbf{Elasticities.} This table uses parameter values from Table \ref{table allbankparameter}. 
 In all the rows except in the first, the accounting asset uncertainty level $S$ equals its long-term level $m\sigma(1-\rho)$.}
\label{elasticities}
\begin{center}
\begin{tabular}{ crrr r }
\hline \hline
\multicolumn{1}{ c }{Variable }&  $u_2$ & $u_1$ & $I$  & $\hat V$  \\
\hline
elasticity w.r.t \\
$S$ & $0.2038$ & $-0.1417$ & $-3.756$ & $0.0245$\\
$\sigma$ & $0.0802$ & $-0.1856$ & $-3.5408$ & $0.0171$\\
$m$ & $-0.1348$ & $-0.2307$ & $-2.3859$ &  $0.0436$\\
$\omega$ & $-0.0094$ & $-0.0034$ & $0$ & $0.0002$ \\
$\rho$ & $0.1049$ & $0.1682$ & $0.4961$ & $-0.0302$\\
$\Delta$ & $0.0808$ & $-0.1444$ &$0$ & $-0.0146$ \\
$K$ &  $0.0017$ & $-0.0543$ & $0$ & $-0.0004$\\
\hline \hline
\end{tabular} 
\end{center}
\end{table}

\section{Model Calibration}\label{calibration}

\subsection{Dataset and Summary Statistics}
To calibrate our model, we use a sample of 292 publicly traded U.S. banks between the first quarter of 1993 and the fourth quarter of 2015. The dataset is from the CRSP/Compustat Merged Database.
We use these banks since they have market equity prices and all the needed accounting variables for our analysis during the sample period (total assets, debt, tier 1 equity, number of common shares outstanding, stock price, dividend payments, income before extraodinary items, net cash flows from operating activities, 
retained earnings, convertible debt, total liability, preferred stock, deferred taxes, loan loss provisions, and real estate loan).
Consistent with Figure~\ref{figure bank structure}, we model bank debt in $(\ref{debt})$ as the total assets minus tier 1 equity.\footnote{The average debt-to-assets ratio among the sample banks is $87.8\%$ and the corresponding average deposits-to-assets ratio is $75.0\%$.}
Table \ref{basic information}  provides the summary statistics of the data.

\begin{table}[ht]
\caption[Summary statistics]{\textbf{Summary statistics.} This table reports variable names, means, standard deviations, minimum and maximum values from the dataset. The dataset contains $292$ publicly traded U.S. banks from Q1 1993 to Q4 2015. $m$ corresponds to million U.S. dollars. }
\label{basic information}
\begin{center}
\resizebox{\linewidth}{!}{%
\begin{tabular}{ lr rrrr }
\hline \hline
\multicolumn{1}{ c }{Variable } & \multicolumn{1}{ c }{Mean} & \multicolumn{1}{ c }{(Std. Dev.)}  &  \multicolumn{1}{ c }{Max} &  \multicolumn{1}{ c }{Min}    \\
\hline
Total assets (million U.S. dollars)& $30,102 $ & $(169,244)$ & $2,577,148$ & $23$   \\
Debt  (million U.S. dollars)& $17,869 $  & $(89,011 )$   & $2,098,061 $ &  $17$  \\
Tier 1 equity (million U.S. dollars) & $3,181 $ & $(18,826)$ & $319.059$  &$8$   \\
Number of common shares outstanding (million)& $166$  & $(628)$ & $10,822$ & $3$  \\
Stock price (U.S. dollars) & $30.63$  & $(38.97)$ & $189.00$ & $0.26$ \\
Dividend (million U.S. dollars)& $141$ & $(711)$  & $12,437$ & $0$  \\
Income before extraordinary items & $217$ & $(1,274)$ & $12,762$ & $-5,584$\\
Net cash flow from operating activities & $669$ & $(3,825)$ & $61,385$ & $-1,024$\\
Stock Return (Q2 in 2007 to Q4 in 2008) & $-13.91\%$ &$(32.65\%)$ & $66.53\%$ & $-93.71\%$   \\
Preferred stocks (million U.S. dollars)& $8 m$ & $(25)$ & $701 $ & $0$ \\
Deferred taxes (million U.S. dollars) & $19 m$ & $(64)$ & $2433 $ & $1000$ \\
Retained Earnings (million U.S. dollars)& $153$  & $(1,325)$ & $30,803$ & $0$ \\
Convertible debt  (million U.S. dollars) & $105$ & $(240)$ & $834$ & $0$\\
Total liability (million U.S. dollars) &$370$  &$(2,145)$ & $359,250$& $2$\\
Preferred stock  (million U.S. dollars) &$33$ &$(225)$  & $302$ &  $-4,049$\\
Deferred taxes (million U.S. dollars)  & $78$ & $(467)$ & $6,804$  &  $-8,826$\\
Loan loss provisions (million U.S. dollars)  & $45$ & $(402)$ & $13,380$ & $-1,567$\\
Real estate loan (million U.S. dollars)  & $8$ & $(38)$ & $573$ & $0$ \\
\hline \hline
\end{tabular} }
\end{center}
\end{table}

\subsection{Calibration of the Fully Observed Model}\label{calibration_fully}
By the fully observed model in Subsection \ref{fully observed model} and \ref{fully_theorysection}, the total assets follow a geometric Brownian motion process outside the dividend and recapitalization times.
To estimate the expected proportional change $\alpha$ and asset volatility $\sigma$, we calculate first the quarterly return time series of the total assets, and after that, the estimates for $\alpha$ and $\sigma$ are the mean and the standard deviation of the time series.
%
We estimate the proportional change of bank debt, $\mu$, in a similar way from the time series of bank debt $D$.
Each quarter $D$ is the total assets minus tier 1 capital.

According to  \cite{PK06}, issuance time delay $\Delta$ equals $0.5$ year, and the fixed cost of equity issuance in terms of risk weighted assets is $0.25\%$.
These correspond to our model's cost parameter $K =0.20\%$ (in terms of debt $D$) and recapitalization delay $\Delta=0.5$.\footnote{Since our issuance cost is in terms of bank debt, by the average debt-to-asset ratio among the sample banks ($87.8\%$) and the risk weighted assets in terms of total assets ($70\%$), the cost parameter $K =0.20\%$. Risk weighted assets are a weighted sum of the bank's nominal exposures, where the weights depend on product type and counterparty sector. For large banks, risk weighted assets are typically between $65$ and $70$ percent of total assets. Here we use $70\%$ (see e.g. \cite{PK06}).}
We get the minimum  equity-to-debt ratio $\kappa$ as follows.
The smallest equity-to-debt ratio in our dataset is about $4.9\%$. Since there are no bank liquidations in our dataset, we set $\kappa = 4.8\%$ (in terms of debt $D$), and it also corresponds to the Basel minimum capital requirement.\footnote{\label{footnote_kappa}The Basel minimum tier 1 capital level is $6\%$ of the risk weighted assets (see ``Basel III: A global regulatory framework for more resilient banks and banking systems'' available at \url{http://www.bis.org/publ/bcbs189.pdf}). The risk-weighted assets are about 70\% of total assets, so the minimum equity-to-assets ratio requirement is $4.2\%$.  In our sample the average debt-to-assets ratio is $87.8\%$, which gives the minimum equity-to-debt ratio requirement of $4.8\%$ (which equals $4.2\%/87.8\%$). }
We calibrate the discount rate $\delta$ and the proportional liquidation value $\omega$ by fitting the fully observed value function $(\ref{define of value function with recap and partially observed})$ with $m=0$ to the realized market equity value in the least square sense.
All the calibrated parameters are shown in Table~\ref{table allbankparameter}.

\subsection{Calibration of the Partially Observed Model}\label{calibration_partially}
Because of the opaqueness, banks'  true asset value process $Y_t$ is partially observed. The observed quarterly noisy accounting asset value is denoted by $Y_t^{ac}$.  Let us  define $M_t=\log Y_t$ and $M_t^{ac}=\log Y_t^{ac}$.
Since we have quarterly data, let us introduce quarterly time index $k$.
Then the signal process $Z_t$ in $(\ref{signalprocess})$ gives $M_{k}^{ac}=Z_{k}-Z_{k-1}$ and,
by the Euler scheme (see e.g. \cite{B96}), we have 
\begin{eqnarray}\label{discrete model} \begin{aligned}
M_{k}&=M_{k-1}+\left(\alpha-\frac{1}{2}\sigma^2\right)+\sigma \varepsilon_k,\quad M_{k}^{ac}&=M_{k} +m e_k,
  \end{aligned}
\end{eqnarray}
where  $\varepsilon_k$ and $e_k$ are standard normal random variables with the correlation $\rho$. 
Note that here we use quarter as a unit of time; however, in the end, we transform the quarterly estimates to annual estimates.

The state space model $(\ref{discrete model})$ is a linear state space model with correlated error terms.  To estimate the model parameters, let us define
$
\varepsilon_k=\rho e_k + \sqrt{1-\rho^2}\eta_k
$
where $\varepsilon_k$ and $e_k$ are the standard normal random variables with the correlation $\rho$ in $(\ref{discrete model})$, $\eta_k$ is a standard normal random variable, and $e_k$ and $\eta_k$ are independent. Now $(\ref{discrete model})$ can be written as\footnote{Note that there is no restriction on the sign of the correlation $\rho$ for stability because our model is linear and time-invariant (all the coefficients are constant) and satisfies the condition $\lim_{k\uparrow \infty}\frac{m^2}{(m^2+P_{k})(1+\sigma\rho/m)}<1$. Therefore, the Kalman filtering is a stable dynamic system, and the steady state is independent on the initial distribution (see e.g. \cite{K15}). If $1+\sigma\rho/m=0$, we use $e_k=\rho \varepsilon_k+\sqrt{1-\rho^2}\eta_k$ instead.
 }
\begin{eqnarray*}\begin{aligned}\label{discrete model change}
 & M_{k}= \frac{M_{k-1}}{1+\sigma\rho/m}+\frac{\alpha-\sigma^2/2}{1+\sigma\rho/m}+\frac{\sigma \rho M_k^{ac}/m}{1+\sigma\rho/m}+\frac{\sigma\sqrt{(1-\rho^2)}\eta_k}{1+\sigma\rho/m},\\
 & M_{k}^{ac} = M_{k}+me_{k}.
 \end{aligned}\end{eqnarray*}
By \cite{H94}, we write the Kalman filter recursion as
\begin{eqnarray*}
& & a_{k|k}= a_{k}+ P_{k} F_{k}^{-1}v_{k},\quad  P_{k|k}=P_{k}-P_{k}F_{k}^{-1}P_{k},\\
& &  a_{k+1}=a_{k|k}/(1+\sigma\rho/m)+\sigma \rho  M_{k}^{ac}/(m+\sigma\rho)+(\alpha-\sigma^2/2)/(1+\sigma\rho/m),\\
&& P_{k+1}=P_{k|k}/(1+\sigma\rho/m)^2 +\sigma^2(1-\rho^2)/(1+\sigma\rho/m)^2,
\end{eqnarray*}
\begin{eqnarray*}
& &F_{k}=\text{var}[v_{k}| \mathcal{G}_{k-1}]=P_{k}+m^2,
\end{eqnarray*}
where $a_{k+1}=\mathbb{E}[M_{k+1}|\mathcal{G}_t]$, $a_{k|k}=\mathbb{E}[M_{k}|\mathcal{G}_{k}]$,  $v_{k}=M_{k}^{ac}-\mathbb{E}[M_{k}^{ac}|\mathcal{G}_{k-1}]=M_{k}^{ac}-a_{k}$, $P_{k+1}=\text{var}[M_{k+1}| \mathcal{G}_{k}],$ and $P_{k |k}=\text{var}[M_{k}| \mathcal{G}_{k}]$ for $k=1,2,\dots, n$, and $n$ is the number of quarterly data used in the estimation.
We assume that $M_0$ is normally distributed with the mean $ a_0$ and variance $P_0$, and therefore, the log-likelihood is given by
\begin{eqnarray}\label{log likelihood}
\log L( M_0^{ac}, M_{1}^{ac},...,M_{n}^{ac})&=&\log p(M_0^{ac})+\sum_{k=1}^n \log p(M_{k}^{ac}|\mathcal{G}_{k-1})\\
&=&-\frac{n+1}{2}\log(2\pi)-\frac{1}{2}\sum_{t=0}^n \left(\log|F_{k}|+v_{k} F_{k}^{-1} v_{k}\right),\nonumber
\end{eqnarray}
where 
 $p(\cdot| \mathcal{G}_{k-1})$ is the probability density of a normal distribution with the mean $a_{k}$ and variance $F_{k}$, and $v_{k}=M_{k}^{ac}- a_{k}$.
Given the parameters, the expectation $ \hat M_{k}:=\mathbb{E}[M_{k}|\mathcal{G}_{k}]$ and variance $P_{k}$ are updated  as follows:
\begin{eqnarray}\label{kalman update}\begin{aligned}
 \hat M_{k+1}=& \hat M_{k}+\frac{P_{k+1}}{P_{k+1}+m^2}(M_{k+1}^{ac}- \hat M_{k})\\
 & +\frac{m\sigma\rho (M_{k}^{ac} - \hat M_{k})}{(P_{k+1}+m^2)(1+\sigma\rho/m)}+\frac{(\alpha-\sigma^2/2)m^2}{(P_{k+1}+m^2)(1+\sigma\rho/m)},\\
P_{k+1}=&\frac{P_{k} m^2}{(P_{k}+m^2)(1+\sigma\rho/m)^2}+\frac{\sigma^2(1-\rho^2)}{(1+\sigma\rho/m)^2}.
\end{aligned}
\end{eqnarray}
In theory, we only need to maximize the log-likelihood function (\ref{log likelihood}) to find the optimal parameters $\alpha$, $m$, $\sigma$, and $\rho$. However, the function is a four-dimensional nonconcave function and is sensitive with respect to the correlation $\rho$, dynamic volatility $\sigma$, and signal volatility $m$. Therefore, we use first the Kalman iteration step (\ref{kalman update}) to get the mean and variance estimates of the state variable $M_{k}$ and then simulate the augmented  state $\{M_{k}, \alpha, \sigma, m,\rho\}_{k}$ using sequential importance resampling (SIR) particle filtering (see e.g. \cite{D00}).\footnote{Particle filtering is  a simulation tool, where the probability distribution of the state is formulated by a cloud of weight particles.  We use systematic resampling to avoid the degeneration of particles and jittering approach to avoid a fixed parameter estimation problem (see e.g. \cite{LW01}).  Since our discrete time model is time-invariant and Gaussian, the approximated distribution formulated by the particles converges to the true distribution of the unknown state almost surely when the number of particles goes to infinity (see. e.g. \cite{CD02}). More variational inference methods can be found in \cite{han2017stein}, \cite{han2018stein}, \cite{han2020stein}).} Information filtering and retiring of time-series/images/videos can be found in \cite{lombardo2019deep}, \cite{jun2015numerical}, \cite{han2016bootstrap} and \cite{han2017high}.

We first generate a sample process for $M_k^{ac}$ under the above parameter values by using $(\ref{discrete model})$, and define $\Theta=\{\alpha, \sigma, m, \rho\}$. The particle filtering approximates the distribution of $M_k|M_k^{ac}$ and $\Theta | M_k^{ac}$ by sufficiently large set of $N$ particles $M_k^{(1)}, \Theta_{k}^{(1)},$ $M_k^{(2)}, \Theta_{k}^{(2)},\dots, M_k^{(N)}, \Theta_{0}^{(N)}$ with discrete probability mass of $w_k^{(1)}, w_k^{(2)},\dots,w_k^{(N)}$.
At time step $k=0$, we  give a initial guess for $\{ M_0^{(i)}, \Theta_{0}^{(i)}\}_i$ with weight $\{w_0^{(i)}\}_i$ and sample mean $\{M_0, \Theta_0\}$, where $i$ is the index for  $N$ particles.  $\{w_0^{(i)}\}_i$ indicates the distribution of  parameters (this initial distribution can be arbitrary, like uniform distribution, or normal distribution). After that, at each time step $k$ $(k\geq 1)$,  we execute following steps:

\begin{itemize}
\item[(1)] Derive $\hat M_{k}^{(i)}$ by Kalman iteration (\ref{kalman update}) with $\Theta_{k-1}^{(i)}$.
\item[$(2)$] Sample $ M_{k}^{(i)} \sim  \mathcal{N}(\cdot |\hat M_{k}^{(i)}, P_{k}^{(i)} )$.
\item[$(3)$] By $(\ref{discrete model})$ and (\ref{discrete model change}),  we update weight and normalise the weight
\begin{eqnarray*}\label{weight of particle}
w_k^{(i)}=w_{k-1}^{(i)}\frac{p(M_{k}^{ac}| M_{k}^{(i)})p(M_k^{(i)}| M_{k-1}^{(i)}, M^{ac}_{k-1})}{p(M_k^{(i)}| M_{k-1}^{(i)}, M^{ac}_{k})}, \ \   \tilde w_k^{(i)}= w_k^{(i)}/\sum_{i=1}^N w_k^{(i)}
\end{eqnarray*}
where the density $ p( \cdot | M_{k}^{(i)})= \phi(\cdot | M_{k}^{(i)},m^2)$, $p(\cdot| M_{k-1}^{(i)}, M^{ac}_{k-1})=\phi(\cdot | M_{k-1}^{(i)}+\alpha-\sigma^2/2, \sigma^2)$, and
$p(\cdot| M_{k-1}^{(i)}, M^{ac}_{k})=\phi(\frac{M_{k-1}^{(i)}}{1+\sigma\rho/m}+\frac{\alpha-\sigma^2/2}{1+\sigma\rho/m}+\frac{\sigma \rho M_k^{ac}/m}{1+\sigma\rho/m}, \frac{\sigma^2(1-\rho^2)}{(1+\sigma\rho/m)^2})$.
\item[$(4)$] Jittering:  sample $\Theta_{k}^{(i)}   \sim \mathcal{N}( \cdot |a \Theta_{k-1}^{(i)}+(1-a)  \bar \Theta_{k-1}, \   Q_{k-1}) $,  where $\bar \Theta_{k-1}$ and $Q_{k-1}$ are the mean and variance of $\{\Theta_{k-1}^{(i)}\}_i$ with weight $\{\tilde w_{k}^{(i)}\}_i$.
\item[$(5)$] Resample If the effective sample size $1/\sum_{i}^N \tilde w_{k}^{(i)}$ falls below $2N/3$: First we define $\tilde M_{k}^{(i)}:=M_{k}^{(i)}$, $\tilde \Theta_{k}^{(i)}:=\Theta_{k}^{(i)}$, and sample $u_1\sim U[0,1/n]$.  We then find the minimum index $j(i)$ for each $i$, such that  $\sum_{i=1}^{j(i)} w_k^{(i)}>(j-1)/N+u_1$.
Replace the  state variables $M_{k}^{(i)} = \tilde M_{k}^{(j(i))}$,  $\Theta_{k}^{(i)}=\tilde \Theta_{k}^{(j(i))}$, and $w_k^{(i)}=1/N$ for all $i=1,2,\dots, N$.
\item[$(6)$] Update state $\{ M_{k}, \Theta_{k} \}$ by the sample mean of $\{ M_{k}^{(i)}, \Theta_{k}^{(i)} \}_i$ with weight $\{w_k^{(i)}\}_i$.

\end{itemize}

We test our simulation method by setting 
$\alpha=0.04,\sigma=0.05, m=0.03,$ and $\rho=-0.30.$
Here we use $N = 2000$ particles for each augment state  to match the generated samples. The test result is shown in Figure \ref{test}.
Figure  \ref{test} illustrates the convergence of the parameter estimates to the true values. 

\begin{figure}[H]
\centering
\begin{tabular}{cc}
\includegraphics[width=0.46\textwidth]{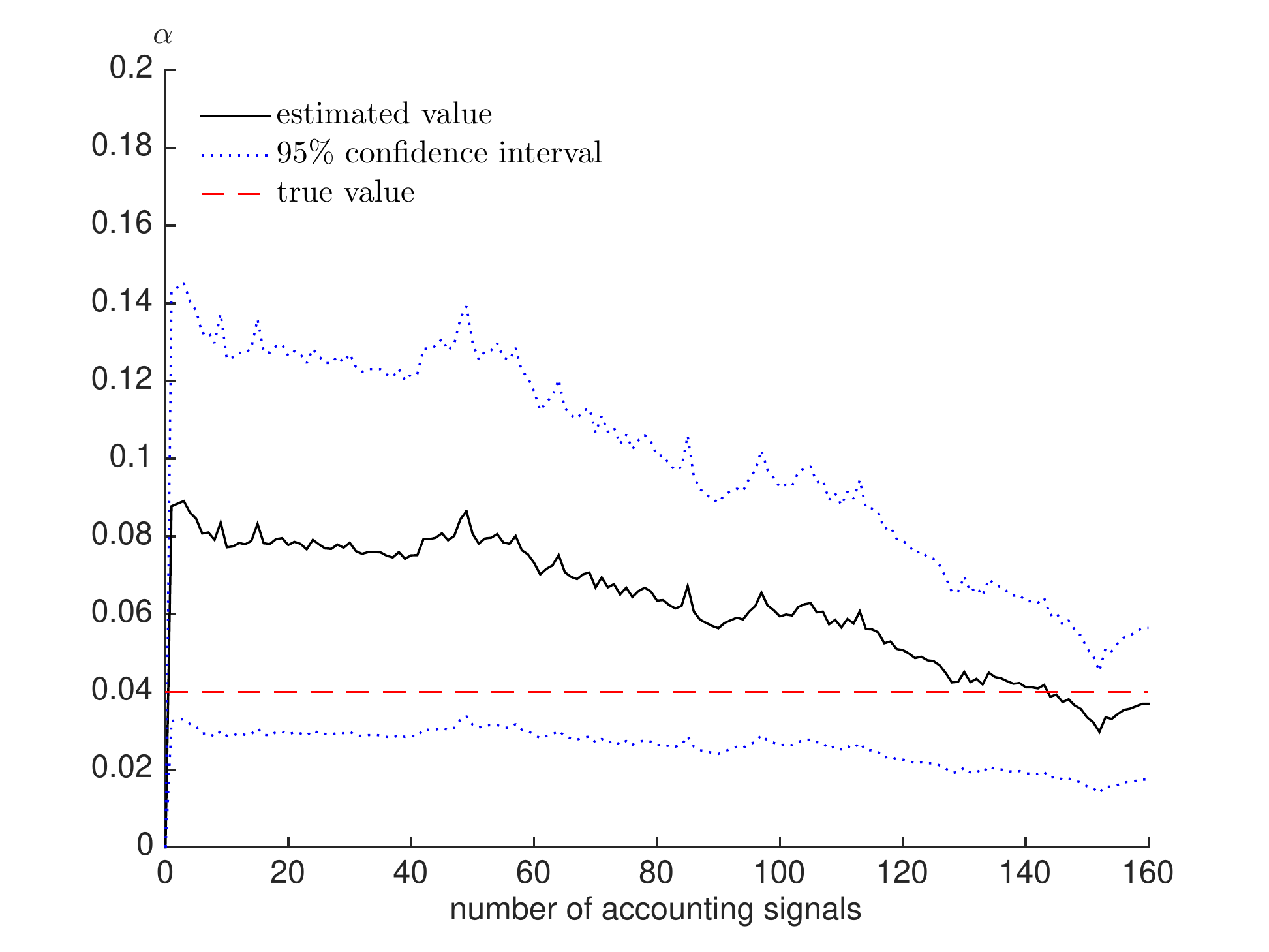} & \hspace{-.5cm}
\includegraphics[width=0.46\textwidth]{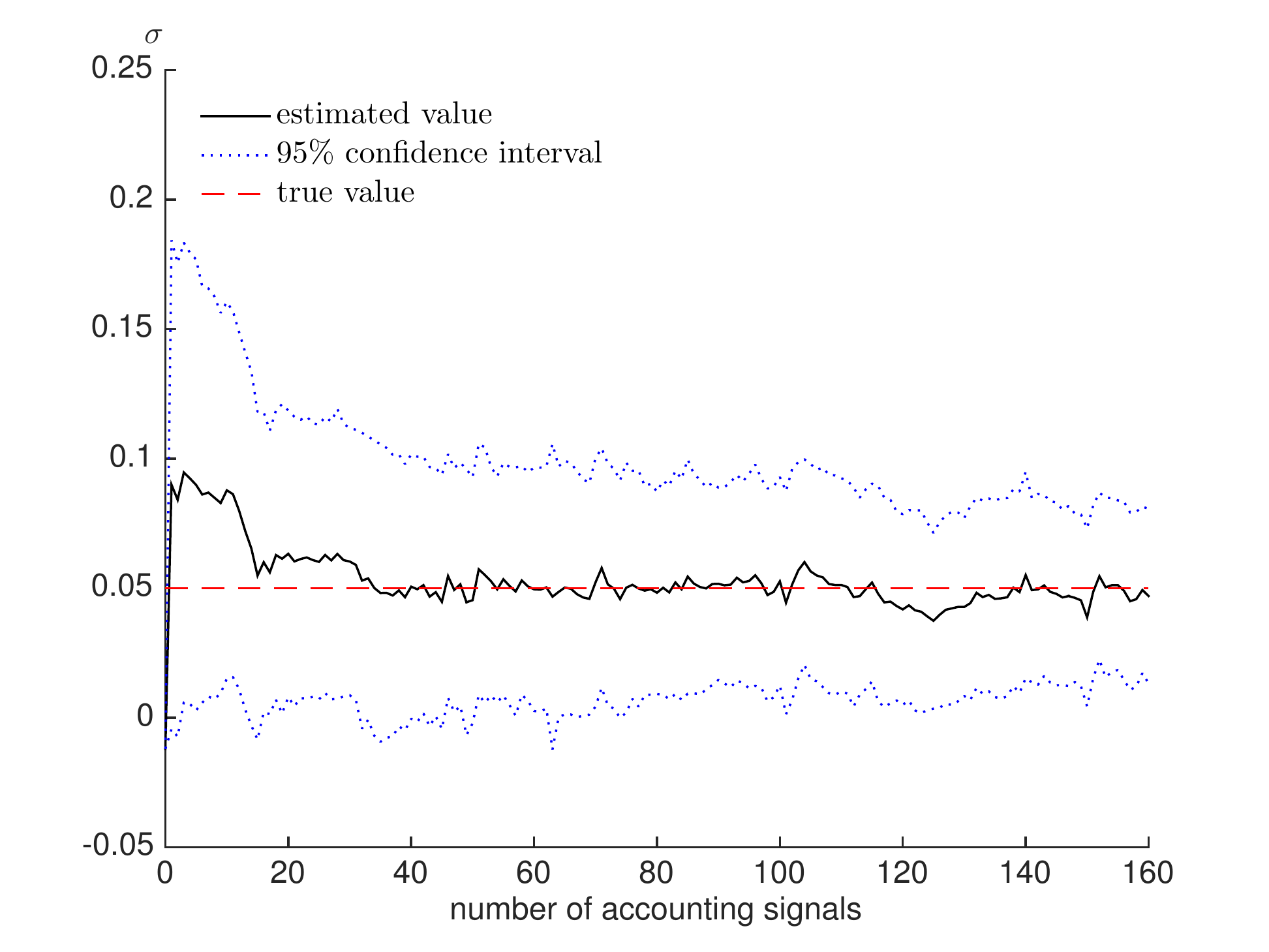} \\
{\small (i) parameter  $\alpha$} & {\small (ii) parameter $\sigma$}  \\
\includegraphics[width=0.46\textwidth]{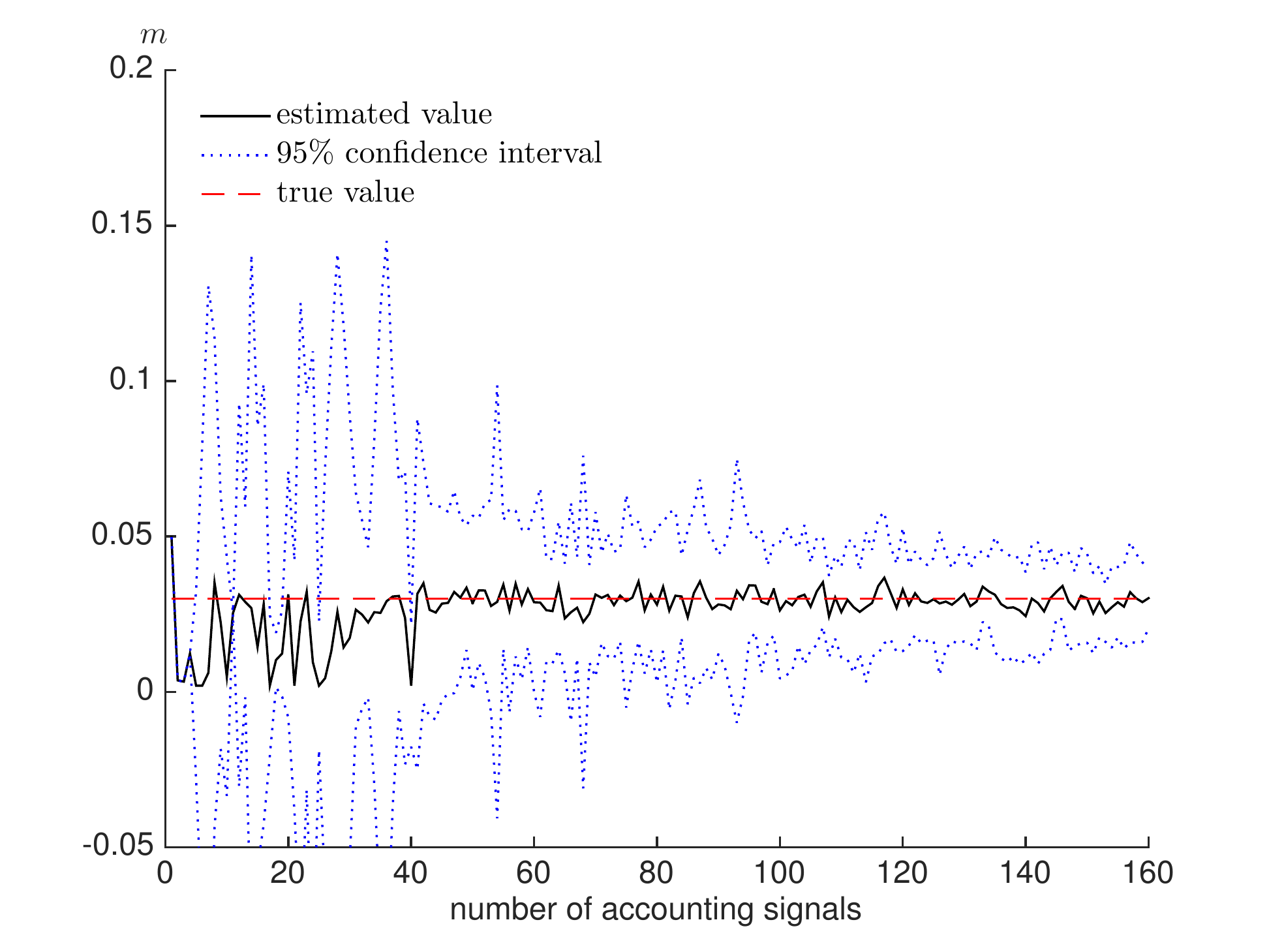} &\hspace{-.5cm}
\includegraphics[width=0.46\textwidth]{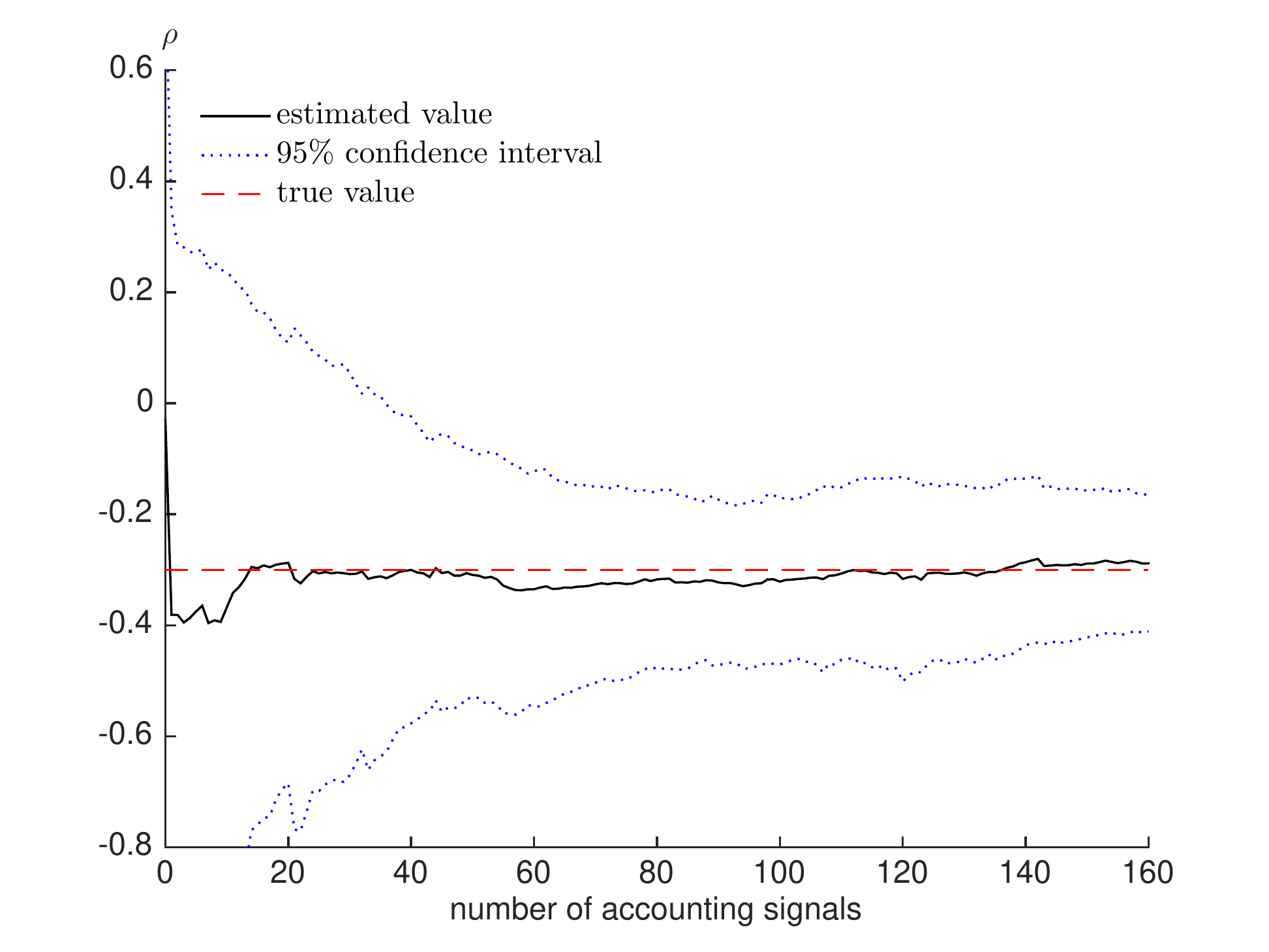} \\
{\small (iii) parameter  $m$} & {\small (iv) parameter  $\rho$} 
\end{tabular}
\caption[Parameter estimation algorithm]{ \textbf{Parameter estimation algorithm.} 
This figure illustrates the parameter estimation via SIR particle filtering with jittering. The true parameter values: $\alpha=0.04$, $\sigma=0.05$, $m=0.03$, and $\rho=-0.30$.}
\label{test}
\end{figure}

The parameter estimates results for $\alpha$, $\sigma$, $m$, and $\rho$ are reported in Table \ref{table allbankparameter}. As in Subsection~\ref{calibration_fully}, the parameter $\mu$ is estimated directly from the banks' debt time series. Further, $K=0.20\%$ and $\Delta=0.5$ as in the fully observed model. By $(\ref{dSt})$ and $(\ref{hat tao})$, we define the liquidation barrier corresponding to $S_t=m\sigma(1-\rho)$ as follows:
\begin{eqnarray}\label{hat kappa}
\hat \kappa:=\lim_{t\rightarrow \infty} I(S_t)
=-1+(1+\kappa)e^{\frac{1}{2}m\sigma(1-\rho)-\Phi^{-1}(a)\sqrt{m\sigma(1-\rho)} },
\end{eqnarray}
where $\kappa$ is the minimum equity-to-debt ratio of the fully observed model, $a$ is the liquidation parameter of bank regulators and it equals $c_2/(c_1+c_2)$ as explained in Section~\ref{section with recap}.
The liquidation barrier $\hat \kappa$, proportional liquidation value $\omega$,  and discount factor $\delta$ are calibrated by fitting the partially observed value function $(\ref{define of value function with recap and partially observed})$ to the realized market equity values in the least square sense.
These parameters are common for all the banks, and they are reported in Table \ref{table allbankparameter}.
The values of $\omega$ and $\delta$ are close to the corresponding parameter values of the fully observed model; however, due to the accounting noise, $\hat \kappa$ is substantially lower than $\kappa$ of the fully observed model in Section~\ref{calibration_fully}.  Note that in the fully observed model $\kappa = 4.8\%$ and it corresponds to the Basel minimum
tier 1 capital ratio.

Given the values of $\kappa$, $\hat \kappa$, $m$, $\sigma$, and $\rho$, we solve for the regulators' liquidation parameter $a$ from $(\ref{hat kappa})$.
The estimate of parameter $a$ is reported in Table \ref{table allbankparameter}.
As can be seen, the estimate is greater than $50\%$, which, by Proposition \ref{special case}(i), implies that the liquidation barrier falls in uncertainty $S$, and thus, the banks benefit from noisy accounting information.
Our estimate of $a$ is consistent with the papers discussed in footnote \ref{footnote4} in Chapter \ref{chp: Intro}.


Table~\ref{table allbankparameter} also shows that, because of the signal noise $m$ and the negative correlation between the accounting noise and the assets' dynamic uncertainty ($\rho=-26.71\%$), the volatility of asset value in the partially observed model is $5.21\%$, which is significantly greater than the volatility in the fully observed model ($3.11\%$).
Thus, the banks smooth their asset value dynamics  using the accounting noise, and this way, they are able to hide about one-third of the true underlying asset volatility. On average, the noise in the reported accounting asset values raises the banks' market
equity value by $7.8\%$,  which is the relative difference between the value function  under $m=2.85\%$ and $m=0$ at the average dividend barrier $12.58\%$ (see $u_2$ in Table~\ref{table allbankparameter}). In this way, banks hide their solvency risk from the banking regulators. Because of the volatility difference, the fully observed model gives an average dividend barrier of $11.22\%$, and the corresponding number for the partially observed model is $12.58\%$, which is closer to the sample average equity to debt ratio of $13.51\%$. The partially observed model also gives a higher recapitalization barrier. Hence, under the noisy accounting values, bank owners know that the true asset uncertainty is higher, and they hedge that by paying less dividends and issuing more equity.


\begin{table}[ht]
\caption[Model calibration]{\textbf{Model calibration.} This table reports the estimated model parameters and the barriers of the optimal policy for the sample banks in Table \ref{basic information} (full sample, Q1 1993--Q4 2015).  Here $\hat \kappa$  is the long-term liquidation barrier defined in $(\ref{hat kappa})$ in the partially observed model and $\kappa$ is the  liquidation barrier in the fully observed model, 
$u_1$ is the recapitalization barrier, $u_2$ is the dividend barrier, and $u_0$ is the dividend barrier without recapitalization. 
$t$-statistics are in the parentheses.
Significance
levels are indicated by ${}^{***} p < 0.01$, ${}^{**} p < 0.05$, and ${}^{*} p < 0.1$.}
\label{table allbankparameter}
\begin{center}
\resizebox{\linewidth}{!}{%
\begin{tabular}{rrrr rrrr rrrr r}
\hline 
\hline
\multicolumn{13}{ c }{ Fully observed model}\\
\hline
All banks & $\mu$ & $\alpha$ & $\sigma$ &  $\delta$ & $\omega$&  $u_1$& $u_2$ & $u_0$ \\
Mean & $10.52\%^{***}$  & $11.59\%^{***}$ & $3.45\%^{***}$ & $23.30\%^{***}$ & $31.50\%^{***}$   &  $6.44\%^{***}$ & $11.22\%^{***}$  & $11.27\%^{***}$ \\
& $(16.70)$ & $(15.43)$ & $(15.68)$ & $(6.93)$ & $(8.50)$  & $(12.90)$ & $(16.60)$ & $(12.86)$\\
\hline
correlation \\
$\mu$ & &$79.24\%$ & $7.13\%$ & $25.91\%$ & $16.50\%$  & $31.25\%$ & $35.77\%$ & $39.37\%$ \\
$\alpha$ & & & $12.80\%$ &$24.56\%$ & $15.05\%$  & $38.07\%$ & $44.89\%$ & $37.18\%$ \\
$\sigma$ &  & &  & $-16.55\%$ & $-32.14\%$  & $13.66\%$ & $41.78\%$ & $34.46\%$ \\
$\delta$ & & & & & $65.10\%$   & $6.39\%$ &$-17.04\%$ & $-16.86\%$\\
$\omega$ & & & & &  &  $38.05\%$ & $24.25\%$ & $2.66\%$\\
$u_1$ & & & & & & &   $83.93\%$ & $79.65\%$ \\
$u_2$ & & & & & & & &   $97.74\%$\\

\hline \hline
\multicolumn{13}{ c }{ Partially observed model}\\
\hline
 All banks & $\mu$ & $\alpha$ & $\sigma$ & $m $  & $\rho$ & $\delta$ & $\hat \kappa$ & $a$ & $\omega$ & $u_1$ & $u_2$ & $u_0$\\
Mean & $10.52\%^{***}$ & $12.85\%^{***}$ &$5.21\%^{***}$ &$2.85\%^{***}$ &$-26.71\%^{***}$ & $25.70\%^{***}$ & $1.15\%^{***}$  & $79.93\%^{***}$  & $25.10\%^{***}$ & $7.03\%^{***}$ & $12.58\%^{***}$ & $13.40\%^{***}$ \\
& $(16.70)$ & $(16.47)$ & $(15.32)$ & $(15.00)$ & $(-8.56)$ & $(11.17)$ & $(13.14)$ & $(12.97)$  & $(11.83)$ & $(18.89)$ & $(18.26)$ & $(17.71)$\\
\hline
correlation \\
$\mu$ & & $76.89\%$ & $1.87\%$ & $7.67\%$ & $14.05\%$ & $60.54\%$ & $-2.00\%$ & $2.06\%$ & $14.30\%$ &$25.34\%$  & $28.18\%$ & $25.65\%$ \\
$\alpha$ & &  & $5.60\%$ & $9.32\% $ & $0.69\%$ & $64.32\%$ & $-5.96\%$ & $-3.14\%$ &$7.05\%$ &$23.44\%$ & $27.02\%$ & $25.05\%$ \\
$\sigma$ & & & & $51.41\%$ & $4.48\%$ & $28.06\%$ & $-15.48\%$ & $-25.14\%$ &  $20.17\%$& $-20.46\%$ & $5.25\%$ & $8.59\%$\\
$m$ & & & & & $2.38\%$ & $28.06\%$ & $-12.49\%$ & $-29.29\%$ & $13.91\%$ & $-8.40\%$ & $0.95\%$ & $6.15\%$\\
$\rho$ & &&&&& $-13.45\%$ & $12.47\%$ & $9.88\%$  & $7.22\%$ &$0.20\%$ &$2.53\%$ &$0.97\%$\\
$\delta$ &&&&&&& $5.49\%$ &  $-17.95\%$  & $2.05\%$ &  $-22.62\%$ & $-3.16\%$ & $-4.38\%$\\
$\hat \kappa$ &&&& &&&&   $-32.17\%$ & $-23.33\%$   & $8.89\%$ & $9.39\%$ & $-0.59\%$\\
$a $ &&&& &&&& &  $32.25\%$ & $1.41\%$ & $-8.71\%$ & $-7.12\%$   \\
$\omega$ &&&& &&&& && $8.43\%$ & $10.57\%$ & $15.74\%$\\
$u_1$ &&&& &&&& &&&  $81.49\%$ & $84.41\%$\\
$u_2$ &&&& &&&& &&&& $96.62\%$\\
\hline\hline
\end{tabular}}
\end{center}
\end{table}

\subsection{Out-of-Sample Tests}
In this subection, we report the out-of-sample tests of the fully observed and partially observed models. The out-of-sample period is 2004--2015, which includes the latest financial crisis in 2007--2009. We analyze each out-of-sample year separately, and select the in-sample period as the time period before each out-of-sample year.
For instance, for the year 2010, the in-sample period is 1993--2009.
For this analysis, we calculate the equity issuance as (see e.g. \cite{B02})\footnote{We also calculated the equity issuance as follows: equity issuance $=$ sale of common and preferred stock $-$ purchase of common and preferred stock, or  equity issuance $=$ sale of common and preferred stock. The results are similar as reported in this subsection (not reported for brevity).}
\begin{eqnarray*}\label{Baker issue}
\text{ Equity issuance}=\Delta \text{ book equity}  - \Delta \text{ balance sheet retained earnings}, 
\end{eqnarray*}
where $\text{ book equity}= \text{total assets } - \text{ total liabilities } - \text{ preferred stock }  + \text{ deferred taxes}$  $+ \text{ convertible debt }$,
$\Delta$book equity is the current book equity minus the previous book equity, and $\Delta$balance sheet retained earnings is calculated correspondingly.

First, we analyze the asset smoothing of the partially observed model. Table \ref{difference between expectation and observation in recession} shows that the partially observed model predicts significantly lower total assets (and, therefore, also significantly lower book equity value) during the financial crisis 2007--2009 than the corresponding reported accounting values. This is due to the asset value smoothing, which is more obvious during the financial crisis than outside that. Before and after the crisis, the situation is opposite; the model expected total assets are higher than the reported asset values, but this difference is insignificant.

Second, in Table \ref{difference between expectation and observation in recession},  we test the out-of-sample performances of the model market equity value and the dividend and recapitalization policy of the partially observed and fully observed models. 
The model predicted dividend at year $k$ equals $\mathbb{P}(u_2<\hat X_k)(\hat X_k-\mathbb{E}[u_2])$ and the model predicted recapitalization equals $\mathbb{P}(u_1>\hat X_k)(\mathbb{E}[u_2]-\mathbb{E}[u_1])$, where the distributions of $u_1$ and $u_2$ are the distributions of the bank-level $u_1$ and $u_2$ estimates (see Table $\ref{table allbankparameter}$). Overall, the partially observed model explains the out-of-sample variations substantially better than the fully observed model.
More specifically, the average $R^2$ values over all the out-of-sample years for the partially observed model's market equity value, dividends, and recapitalization are  $85\%$, $52\%$, and $14\%$, while the corresponding $R^2$ values for the fully observed model are $79\%$, $35\%$, and $4\%$, respectively.
This suggests that the investors and bank regulators view accounting reports as noisy and act accordingly.
That is, regulators consider the accounting noise when they decide to take actions or not with a possible insolvent bank.
Further, when shareholders decide the banks' dividend and recapitalization policy, they also consider the accounting reports as noisy and this way take the asset smoothing and the regulators' actions into account.
By our estimate of $a$ in Table~\ref{table allbankparameter} and Proposition \ref{special case}(i), the regulators postpone the actions with a possible insolvent bank the higher the accounting uncertainty, and this raises the bank value if the bank does asset smoothing, as we showed in Section~\ref{comparative statics}.

In Table \ref{correlation of volatility difference}, we analyze which kind of banks use asset smoothing. For that, we use two asset smoothing measures: $(i)$ the difference between the asset return volatilities of partially and fully observed models and $(ii)$ the product of signal noise and the correlation between the accounting noise and the assets' dynamic uncertainty.
Note that, by Table~\ref{table allbankparameter}, metric $(i)$ is significantly greater than zero and metric $(ii)$ is significantly less than zero.
That is, the higher the metric $(i)$ and the lower the metric $(ii)$ are, the more the bank uses asset smoothing.

In Panel A of Table~\ref{correlation of volatility difference}, we first consider the correlation between asset smoothing metrics and different bank characteristics.
As can be seen, loan loss provision, stock return, nonperforming assets, real estate loans, the volatility of reported asset values, and the observed equity-to-debt ratio have significant correlations with the asset smoothing metrics.
Therefore, in Panel B of Table~\ref{correlation of volatility difference}, we use these variables as independent variables in regression models, where dependent variables are the smoothing metrics $(i)$ and $(ii)$.
The regression models imply that banks with a high level of loan loss provisions (in both the models), nonperforming assets (only the volatility difference model),  real estate loans (only the volatility difference model),  and with a low asset return volatility (in both the models) use more asset smoothing.
This is consistent with \cite{A99}, \cite{LR95}, and \cite{Fl13}, who find that loan loss provisions have a negative correlation with the earnings process and that the balance sheet composition of a bank affects its opacity.
The $R^2$ of these regressions are $66.6\%$ and $47.8\%$, indicating that the bank characteristics explain substantial variation in the asset smoothing metrics.

\begin{table}[ht]
\center
\caption[Out-of-sample tests]{\textbf{Out-of-sample tests}. This table first reports the relative difference between the expected total assets value $\hat Y$ of the partially observed model and the corresponding reported accounting assets value $Y^{ac}$ in out-of-sample from 2004 to 2015, and the same for the predicted book equity value $\hat E$ and the reported book equity value $E^{ac}$.
Second, this table reports  $R^2$ of a regression model between the realized and model predicted variables.
$t$-statistics are in the parentheses.
Significance levels are indicated by ${}^{***} p < 0.01$, ${}^{**} p < 0.05$, and ${}^{*} p < 0.1$.}
\label{difference between expectation and observation in recession}
\begin{center}
\resizebox{\linewidth}{!}{%
\begin{tabular}{ crrr rrrr rr  rrr}
\hline \hline
\multicolumn{1}{ c }{Out of sample / Year} &  2004 & 2005 &  2006 & 2007 &  2008 &   2009  &   2010 &  2011  &  2012 &  2013  &   2014  &  2015  \\
\hline
\multicolumn{1}{ c }{$(\hat Y-Y^{ac})/Y^{ac}$ } \\
\multicolumn{1}{ c }{Mean}  & $0.16\%$ & $0.34\%$ & $0.31\%$ & $-0.45\%^{*}$
 & $-0.99\%^{***}$ & $-0.63\%^{**}$ & $0.20\%$ & $0.20\%$ & $0.39\%$ & $0.38\%$ & $0.37\%$ & $0.12\%$
\\
\multicolumn{1}{ c }{} & $(0.67)$ & $(1.42)$ & $(1.29)$ & $(-1.88)$ & $(-4.30)$ & $(-2.63)$ & $(0.83)$ & $(0.83)$ & $(1.63)$ & $(1.58)$ & $(1.54)$ & $(0.50)$\\
\multicolumn{1}{ c }{Max}  & $5.78\%$  & $5.78\%$ & $3.90\%$
&  $3.27\%$ & $1.15\%$ & $1.66\%$  & $2.15\%$ & $4.45\%$ & $5.38\%$ & $4.15\%$  & $4.61\%$ & $5.42\%$
\\
\multicolumn{1}{ c }{Min}  & $-2.54\%$ & $-3.74\%$ & $-6.75\%$
 &  $-11.33\%$ & $-18.73\%$ & $-18.75\%$ & $-13.41\%$ & $-9.42\%$ & $-6.24\%$ &$-7.07\%$ & $-4.84\%$ & $-6.22\%$
\\
\hline
\multicolumn{1}{ c }{$(\hat E-E^{ac})/E^{ac}$ }     \\
\multicolumn{1}{ c }{Mean}  & $1.47\%$ & $2.78\%$ & $2.70\%$
& $-3.93\%^{*}$ & $-8.66\%^{***}$ & $-5.27\%^{**}$ & $1.91\%$  & $3.63\%$ &$3.37\%$ & $3.21\%$ & $3.34\%$ & $2.58\%$  \\
\multicolumn{1}{ c }{} & $(0.67)$ & $(1.26)$ & $(1.27)$ & $(-1.82)$ & $(-4.05)$ & $(-2.45)$ & $(0.88)$ & $(1.64)$ & $(1.52)$ & $(1.45)$ & $(1.51)$ & $(1.16)$\\
\multicolumn{1}{ c }{Max}  & $41.31\%$  & $46.29\%$ & $30.85\%$
& $20.54\%$ & $8.81\%$ & $14.86\%$ & $25.01\%$ & $47.74\%$ & $47.23\%$ & $35.25\%$ & $45.23\%$ & $30.23\%$\\
\multicolumn{1}{ c }{Min} & $-23.38\%$ & $-17.20\%$ & $-51.45\%$
 & $-101.59\%$ & $-125.25\%$ & $-98.03\%$ & $-77.44\%$ & $-45.31\%$ & $-42.79\%$ &$-28.51\%$ & $-32.56\%$ & $-45.43\%$\\
\hline
\multicolumn{13}{ l }{ \it Model market equity value vs.  realized  market equity value:  $y=a+b x +\text{ error}$,
where $y$ is the log of realized value and $x$ is the log of model value. } \\
\multicolumn{13}{ l }{ \it{\underline{Fully observed model}} } \\
 \multirow{2}{*}{$a$} & $-1.49$ & $-1.55$  & $-2.60$  & $-2.65$ & $-2.35$ & $-1.12^{*}$  & $1.21^{*}$ & $-1.34$ & $-0.07$ & $-0.04$ & $-0.06$
 & $0.35$\\
    & $(-1.27)$ & $(-1.29)$ & $(-1.11)$ & $(-0.98)$ & $(1.42)$ & $(1.88)$ & $(1.83)$ & $(0.22)$ & $(-0.04)$ & $(0.04)$ & $(0.11)$ & $(0.77)$\\
\multirow{2}{*}{$b$} & $1.04^{***}$ & $1.04^{***}$ & $1.04^{***}$ & $1.02^{***}$  & $1.03^{***}$ & $1.02^{***}$  & $0.87^{***}$ & $0.98^{***}$ & $0.98^{***}$ & $0.93^{***}$ & $0.96^{***}$ & $0.96^{***}$ \\
    & $(32.67)$ & $(32.67)$ & $(32.67)$ & $(32.00)$ & $(31.67)$ & $(24.50)$ & $(22.75)$ & $(29.00)$ & $(21.00)$ & $(22.50)$ & $(25.00)$ & $(24.00)$\\
 \multicolumn{1}{ c }{ $R^2$} &  $75.70\%$ & $75.30\%$ & $75.32\%$  & $75.76\%$ & $75.81\%$  & $76.10\%$  & $73.00\%$ & $82.78\%$  & $86.74\%$  & $82.38\%$ & $83.91\%$ & $83.90\%$\\
\multicolumn{13}{ l }{ \it{\underline{Partially observed model}} } \\
 \multirow{2}{*}{$a$} & $0.63$ & $0.47$  & $0.33$  & $1.12$ & $1.17$  & $1.41^{**}$ & $1.38^{**}$ & $1.07$ & $0.61$ & $1.15$ & $0.86$
 & $0.88$\\
    & $(0.90)$ & $(0.68)$ & $(0.47)$ &$(1.57)$ & $(1.60)$ & $(2.31)$ & $(2.12)$ & $(1.64)$ & $(0.83)$ & $(1.55)$ & $(1.12)$ & $(1.14)$\\
\multirow{2}{*}{$b$} & $0.99^{***}$ & $0.99^{***}$ & $1.00^{***}$ & $0.97^{***}$  & $0.96^{***}$ & $0.91^{***}$  & $0.96^{***}$ & $0.94^{***}$ & $0.94^{***}$ & $0.96^{***}$ & $0.98^{***}$ & $0.97^{***}$ \\
     & $(33.00)$ & $(33.00)$ & $(33.33)$ & $(32.33)$ & $(32.00)$ & $(30.33)$ & $(32.00)$ & $(31.33)$ & $(23.50)$ & $(24.00)$ & $(24.50)$ & $(24.25)$\\
\multicolumn{1}{ c }{ $R^2$} &  $83.74\%$ & $84.20\%$ & $84.45\%$  & $78.45\%$ & $86.48\%$  & $86.71\%$  & $86.10\%$ & $87.62\%$  & $86.87\%$  & $86.70\%$ & $86.73\%$ & $86.32\%$\\
\hline

\multicolumn{13}{ l }{ \it Model dividends vs.  realized dividends:  $y=a+b x +\text{ error}$, where $y$ is the log of realized value and $x$ is the log of model value. } \\
\multicolumn{13}{ l }{ \it{\underline{Fully observed model}} } \\
 \multirow{2}{*}{$a$} & $-1.61$ & $-0.66$  & $-1.32$  & $-1.16$ & $1.11$  & $1.52$ & $1.53$ & $1.17$ & $0.34$ & $1.43$ & $0.99$
 & $-1.88$\\
    & $(-0.76)$ & $(-0.70)$ & $(-0.67)$ & $(0.57)$ & $(-0.58)$ & $(-0.60)$ & $(-0.02)$ & $(-0.53)$ & $(-0.80)$ & $(-0.84)$ & $(-1.13)$ & $(-1.19)$\\
\multirow{2}{*}{ $b$} & $0.91^{***}$ & $0.87^{***}$ & $0.93^{***}$ & $0.89^{***}$  & $0.88^{***}$ & $0.87^{***}$  & $0.84^{***}$ & $0.86^{***}$ & $0.86^{***}$ & $0.88^{***}$ & $0.93^{***}$ & $0.91^{***}$ \\
    & $(11.37)$ & $(11.37)$ &$(11.38)$ & $(14.83)$ & $(15.00)$ & $(14.50)$ & $(14.00)$ &$(10.75)$ &$(9.77)$ &$(9.77)$ &$(9.10)$ &$(9.10)$\\
\multicolumn{1}{ c }{$R^2$}  & $35.33\%$ & $34.75\%$ & $35.58\%$ &  $33.45\%$ & $32.22\%$ & $36.89\%$  & $39.65\%$ & $38.70\%$  & $37.21\%$  & $32.23\%$ & $33.31\%$  & $34.78\%$ \\

\multicolumn{13}{ l }{ \it{\underline{Partially observed model}} } \\
 \multirow{2}{*}{ $a$} & $-3.55^{**}$ & $-2.83^{*}$  & $-2.17$  & $0.17$ & $-0.12$ & $-1.11$  & $-0.01$ & $-1.41$ & $-2.65$ & $-2.90$ & $-2.95$
 & $-2.53$\\
   & $(-2.46)$ & $(-1.91)$ & $(-1.42)$ & $(0.16)$ & $(-0.10)$ & $(-0.85)$ &$(-0.01)$ & $(-0.92)$ & $(-1.45)$ & $(-1.60)$ & $(-1.62)$ & $(-1.45)$\\
\multirow{2}{*}{ $b$} & $1.13^{***}$ & $1.08^{***}$ & $1.02^{***}$ & $0.97^{***}$  & $0.95^{***}$ & $0.93^{***}$  & $0.96^{***}$ & $0.98^{***}$ & $1.06^{***}$ & $1.07^{***}$ & $1.06^{***}$ & $1.03^{***}$ \\
    & $(16.14)$ & $(15.42)$ & $(12.75)$ & $(12.13)$ & $(15.83)$ & $(15.50)$ & $(16.00)$ & $(14.00)$ & $(11.77)$ & $(11.89)$& $(11.78)$ & $(11.44)$ \\
\multicolumn{1}{ c }{ $R^2$} &  $57.18\%$ & $53.14\%$ & $50.53\%$  & $48.40\%$ & $53.37\%$  & $55.22\%$  & $51.93\%$ & $51.02\%$  & $49.81\%$  & $52.46\%$ & $48.73\%$ & $50.12\%$\\
\hline

\multicolumn{13}{ l }{ \it Model recapitalization vs.  realized  recapitalization:  $y=a+b x +\text{ error}$, where $y$ is realized value and $x$ is the model value.} \\
\multicolumn{13}{ l }{ \it{\underline{Fully observed model}} } \\
 \multirow{2}{*}{ $a$ (million)} & $1025.77^{*}$ & $4341.08^{*}$  & $2650.41^{**}$  & $2307.01^{**}$  & $3187.50$  & $4318.52^{*}$ & $3108.97^{**}$ & $3202.43^{*}$ & $5028.32$ & $6208.00^{*}$ & $62114.34$ & $6222.90$\\
   & $(1.74)$ & $(1.69)$ & $(2.79)$ & $(2.86)$ & $(1.59)$ & $(1.77)$ & $(2.34)$ & $(1.51)$ & $(1.89)$ & $(2.11)$ & $(1.61)$ &$(1.47)$\\
\multirow{2}{*}{ $b$ } & $-0.26$ & $-0.50$ & $-0.31$ & $-0.21$  & $-0.19$ & $-0.32$  & $-0.13$ & $-0.16$ & $-0.24$ & $-0.23$ & $-0.16$ & $-0.47$ \\
   & $(-0.44)$ & $(-0.17)$ &$(-0.29)$ & $(-0.31)$ & $(-0.13)$ & $(-0.21)$ & $(-0.20)$ & $(-0.29)$ & $(-0.23)$ & $(-0.24)$ & $(-0.15)$ & $(-0.45)$\\
\multicolumn{1}{ c }{ $R^2$} &  $3.43\%$ & $3.71\%$ & $5.21\%$  & $6.73\%$ & $3.25\%$  & $2.47\%$  & $4.00\%$ & $6.17\%$  & $3.89\%$  & $4.86\%$ & $3.13\%$ & $4.27\%$\\

\multicolumn{13}{ l }{ \it{\underline{Partially observed model}} } \\
 \multirow{2}{*}{ $a$ (million)} & $867.38^{*}$ & $935.67^{*}$  & $835.21$  & $311.14$ & $-850.15$  & $839.15$ & $1374.16$ & $1071.63$ & $570.46$ & $238.32$ & $461.25$ & $378.37$\\
    & $(1.71)$ & $(1.73)$ & $(1.63)$ & $(1.07)$ & $(-0.64)$ & $(0.47)$ & $(1.34)$ & $(1.31)$ & $(0.74)$ & $(0.12)$ & $(1.15)$ & $(1.09)$\\
\multirow{2}{*}{$b$} & $1.05^{***}$ & $3.32^{***}$ & $1.29^{***}$ & $1.99^{***}$  & $4.02^{***}$ & $8.01^{***}$  & $5.54^{***}$ & $13.99^{***}$ & $25.63^{***}$ & $19.44^{***}$ & $6.17^{***}$ & $8.03^{***}$ \\
    & $(10.50)$ & $(6.91)$ & $(5.86)$ & $(13.16)$ & $(14.35)$ & $(11.95)$ & $(9.08)$ & $(12.16)$ & $(16.32)$ & $(10.12)$ & $(3.37)$ & $(3.43)$\\
\multicolumn{1}{ c }{$R^2$} &  $9.63\%$ & $13.51\%$ & $8.22\%$  & $13.17\%$ & $17.35\%$  & $17.74\%$  & $10.02\%$ & $13.18\%$  & $17.61\%$  & $11.63\%$ & $15.37\%$ & $15.70\%$\\
\hline
\hline
\end{tabular}}
\end{center}
\end{table}

\begin{table}[ht]
\caption[Bank characteristics and asset smoothing]{\textbf{Bank characteristics and asset smoothing.} This table gives the bank characteristics that explain asset smoothing in terms of $(i)$ the difference between the asset return volatilities of partially and fully observed models and $(ii)$ the product of signal noise and the correlation between the accounting noise and the assets' dynamic uncertainty $(m\rho)$.
Panel A reports the correlation between the asset smoothing metrics and different bank characteristics. The volatility of asset returns is the asset return volatility of the fully observed model.
In Panel B, we use the variables that have significant correlations in Panel A as independent variables in regression models, where dependent variables are the smoothing metrics.
The p-values are in the parentheses. Significance
levels are indicated by ${}^{***} p < 0.01$, ${}^{**} p < 0.05$, and ${}^{*} p < 0.1$.}
\label{correlation of volatility difference}
\begin{center}
\resizebox{\linewidth}{!}{%
\begin{tabular}{ lllll  }
\hline \hline
  \multicolumn{5}{ c }{ Panel A: Correlation with volatility difference and $m\rho$ }\\
\hline
Variables &  \multicolumn{1}{ l }{Volatility difference} &  \multicolumn{1}{ l }{ (p-value)}&  \multicolumn{1}{ l }{$m\rho$} &  \multicolumn{1}{ l }{ (p-value)}\\
\hline
 \multicolumn{1}{ l }{Mean of loan loss provisions}& $10.14\%^{**}$& $(0.042)$
 & $15.63\%^{***}$ & $(0.004)$  \\
  \multicolumn{1}{ l }{Std. Dev. of loan loss provisions}  & $10.63\%^{**}$ & $(0.043)$
  & $-12.62\%^{**}$ & $(0.016)$ \\
 \multicolumn{1}{ l }{Mean of total assets} & $6.98\%$ &$(0.117)$
 & $-5.60\%$ & $(0.170)$  \\
  \multicolumn{1}{ l }{Std. Dev. of total assets} & $5.34\%$ & $(0.182)$
  & $4.89\%$ & $(0.203)$ \\
   \multicolumn{1}{ l }{Mean of cash / earnings}  & $-6.83\%$  & $(0.122)$
   & $4.06\%$ & $(0.245)$\\
   \multicolumn{1}{ l }{Std. Dev. of cash /  earnings}  & $5.37\%$ & $(0.180)$
   & $-5.01\%$ & $(0.197)$\\
   \multicolumn{1}{ l }{Mean of book equity / market equity value}  & $5.44\%$  & $(0.177)$
   & $6.53\%$ & $(0.133)$\\
  \multicolumn{1}{ l }{Std. Dev. of book equity / market equity value} & $-6.48\%$  & $(0.135)$
  & $7.47\%$ & $(0.102)$\\
   \multicolumn{1}{ l }{Mean of stock returns } & $-16.02\%^{***}$ & $(0.003)$
   & $12.49\%^{**}$ & $(0.016)$ \\
   \multicolumn{1}{ l }{Std. Dev. of stock returns }& $6.82\%$   & $(0.123)$
   & $-3.50\%$ & $(0.275)$ \\
   \multicolumn{1}{ l }{Mean of nonperforming assets / total assets} & $23.25\%^{***}$ & $(<0.001)$ & $11.05\%^{**}$ & $(0.030)$
\\
  \multicolumn{1}{ l}{Std. Dev. of nonperforming assets /  total assets} & $12.37\%^{**}$
  & $(0.017)$  & $-0.33\%$ & $(0.478)$\\
   \multicolumn{1}{ l }{Mean of real estate loans /  total assets} & $11.45\%^{**}$ & $(0.025)$
 & $10.49\%^{**}$ & $(0.037)$ \\
  \multicolumn{1}{ l}{Std. Dev. of real estate loans /  total assets} & $-12.22\%^{**}$
  & $(0.018)$ & $10.14\%^{**}$ & $(0.042)$ \\
  \multicolumn{1}{ l}{Volatility of asset returns} & $-51.24\%^{***}$ & $(<0.001)$ & $23.80\%^{***}$ & $(<0.001)$ \\
 \multicolumn{1}{ l}{Observed equity to debt ratio} & $-12.39\%^{**}$ & $(0.017)$ & $13.50\%^{**}$ & $(0.011)$  \\
\hline \hline
  \multicolumn{5}{ c }{ Panel B: Regression model }\\
\hline
Dependent variable &  \multicolumn{1}{ c }{ Volatility difference} &  \multicolumn{1}{ c }{ (p-value)}&  \multicolumn{1}{ l }{$m\rho$}&  \multicolumn{1}{ l }{ (p-value)}\\
\hline
 \multirow{1}{*}{Mean of loan loss provisions}
 & $8.67\times {10^{-11}}^{***}$  & $(0.004)$ & $-1.47\times {10^{-11}}^{**}$ & $(0.035)$\\
  \multirow{1}{*}{Std. Dev. of loan loss provisions}
 & $4.52 \times 10^{-11}$  & $(0.463)$ & $-1.99\times {10^{-11}}^{*}$ & $(0.084)$\\
    \multirow{1}{*}{Mean of stock returns}
  & $2.87 \times 10^{-3}$  & $(0.591)$ & $-4.17\times 10^{-4}$ & $(0.662)$\\
 \multirow{1}{*}{Mean of nonperforming assets / total assets}
  & $0.10^{***}$ & $(<0.001)$ & $0.01$ & $(0.457)$ \\
 \multirow{1}{*}{Std. Dev. of nonperforming assets /  total assets}
  & $0.13$  & $(0.221)$ \\
 \multirow{1}{*}{Mean of real estate loans /  total assets}
 & $25.53^{**}$  & $(0.028)$ & $2.13$ & $(0.494)$\\
 \multirow{1}{*}{Std. Dev. of real estate loans /  total assets}
   & $-0.24$  & $(0.532)$ & $0.17$ & $(0.397)$\\
 \multirow{1}{*}{Volatility of asset returns}
  & $-0.87^{***}$   & $(<0.001)$ & $0.14^{***}$ & $(0.009)$ \\
  \multirow{1}{*}{Observed equity to debt ratio}
 & $-0.06$  & $(0.443)$ & $0.04$ & $(0.263)$\\
 \multirow{1}{*}{Constant}
 & $0.03^{***}$   & $(0.005)$ & $-0.02^{***}$ & $(0.004)$\\
Observations &   \multicolumn{1}{ l }{ $292$} & & \multicolumn{1}{ l }{ $292$}\\
  R-squared
 & \multicolumn{1}{ l }{ $66.6\%$ }  & & \multicolumn{1}{ l }{ $47.8\%$ } \\
   \hline \hline
\end{tabular}}
\end{center}
\end{table}

\chapter{The Stochastic Control Problem in Portfolio Choice}\label{part 2}
\numberwithin{equation}{chapter}
\section{Problem Formulation}\label{coint_model}

\subsection{Labor and Stock Markets}
It is important to incorporate empirically and quantitatively important features of labor income process
(\cite{WWY16}): (1) diffusive and continuous shock, and (2) discrete and jump shock. Following 
the conventional labor market setting as in \cite{V01}, \cite{LT11}, and \cite{BUV14}, we rule out temporary (or transitory) income shocks, because the temporary feature of income risks 
has a eventually negligible impact on the optimal strategies. We specify the labor 
income process $I_t$ with only permanent components of income risks. We consider the following widely 
used geometric Brownian motion process with an exogenously-driven Poisson jump:
\begin{equation}\label{lip}
dI_{t}=\mu_{I}I_{t}dt+\sigma_{I}I_{t}d\tilde{\mathcal{B}}_{t}-(1-\kappa)I_{t-}d\mathcal{N}_{t},~~I_{0}=I>0,
\end{equation}
where $\mu_{I}>0$ is the expected income growth rate, $\sigma_{I}>0$ measures the volatility on income growth, $\tilde{\mathcal{B}}_{t}$ is a
standard one-dimensional Brownian motion, $\kappa$ follows a power distribution over $[0,1]$ with parameter $\nu>0$,\footnote{The probability density function for $\kappa$ with parameter $\nu$ is given by $P_{\kappa}(z)=\nu z^{\nu-1}$, where $0\leq z\leq 1$. This specification is appropriate when adopting a well-behaved distribution for $\kappa$. An expected income loss decreases with respect to an increase of $\nu$ due to the relationship of
$\mathbb{E}[1-\kappa]=1/(\nu+1)$.
}  and $\mathcal{N}_{t}$ is a pure
Poisson jump process with intensity $\delta_{D}$. We assume that the diffusive and continuous shocks (represented by  Brownian motion $\tilde{\mathcal{B}}_t$) and the discrete and jump shocks (represented by  Poisson 
jump process $\mathcal{N}_t$) are independent for technical convenience.\footnote{Some correlation between those can be possibly considered via an additional stochastic process
for the probability distribution of jump size $\kappa$.}

The income volatility $\sigma_{I}$ represents the permanent and uninsurable income risks. The diffusive and continuous income shocks
captured by the Brownian motion $\tilde{\mathcal{B}}_t$ are independently and identically distributed for the income growth rate $dI_{t}/I_{t}$. In the
absence of the Poisson jump process $\mathcal{N}_{t}$, if we rewrite the income process (\ref{lip}) as the dynamics for logarithmic income $\ln{I_t}$  using
Ito's formula, then the logarithmic income follows an arithmetic Brownian motion and is a unit-root process. The change of logarithmic
income has mean $(\mu_{I}-\sigma^{2}_{I}/2)$ and volatility $\sigma_{I}$ per unit of time. As a result, the income shocks, i.e., the
income fluctuations by $\tilde{\mathcal{B}}_t$ have a permanent impact on the levels of income $I_t$.

Now we turn to the modeling and interpreting of discrete and jump shocks in income process. We include the possibility of large income 
shocks into the conventional setup.\footnote{\cite{C92} introduces this type of disastrous labor income shock and \cite{CGM05} also explicitly allow for the possibility of the income shock by following \cite{C92}.} The crucial exogenous shocks arrive 
with a constant probability (\cite{V01}). More specifically, for time $t\geq 0$,
$$
\mathbb{P}\{\tau_{D}\leq t\}=1-e^{-\delta_{D}t},
$$
where $\tau_{D}$ is the time at which the large income shocks arrive and $\delta_{D}>0$ is the intensity for the exogenous Poisson arrivals. As
a result of this specification, an investor's income is exposed to unexpected, exogenous, and permanent reductions from $I_{\tau_{D}-}$ to 
$\kappa I_{\tau_{D}-}$ $(\kappa\in[0,1])$. In other words, the investor receives only $100 \kappa\%$ of labor income immediately before the income 
shocks occur. This positive income assumption after the income shock happens has been well adopted in the previous important life-cycle 
articles.\footnote{\cite{C92} reports that households may have zero income and $10\%$ of permanent income in the low-income state with 
some probabilities equal to $0.5\%$ and $1.3\%$, respectively. \cite{P07} assesses the importance of specific events causing major 
income shocks such as unemployment and disability by considering the low-income state reported by \cite{C92}. \cite{LT11} also 
investigate the impact of a transitory unemployment state on an individual's asset allocation and assume that the individual can obtain $10\%$ 
of labor income in the unemployment state. 
}


An investor can invest his savings in  riskless bonds and  risky stocks. The bond price grows at  a constant rate $r>0$. The 
stock price, $S_{t}$, follows a geometric Brownian motion:\footnote{In the literature on investment, the 
assumption of a geometric Brownian motion for the stock price is widely used combined with the assumption of 
a constant investment opportunity set. However, in reality, the distribution of stock price does not follow the 
geometric Brownian motion. The empirical distributions of stock returns are given as follows: the skewed 
Student-$t$ distribution, the generalized lambda distribution, Johnson system of distributions, the normal inverse 
Gaussian distribution, and the $g$-and-$h$ distribution. One parsimonious way to capture the stochastic nature of 
investment opportunity set is to introduce a two-state Markov regime switching model for fundamental parameters.}
\begin{equation*}\label{labor_stockprocess}
dS_{t}=\mu S_{t}dt+\sigma S_{t}d\mathcal{B}^{1}_{t},
\end{equation*}
where $\mu>r$ is the expected  stock return rate, $\sigma>0$ is the volatility of the return on the stock,
and $\mathcal{B}^{1}_{t}$ is a standard one-dimensional Brownian motion with an instantaneous correlation $\rho\in[-1,1]$ with
the labor income process given in (\ref{lip}),  i.e., $d\mathcal{B}^1_td\tilde{\mathcal{B}}_t = \rho dt$. Here, $\mu$ and $\sigma$ are the stock returns' mean and standard 
deviation, respectively. They summarize an investment opportunity set provided by the stock, i.e., they represent the 
expected return and the risk in the stock market. We assume that $r$, $\mu$, $\sigma$ are constants, i.e., the 
investment opportunity is constant.

Note that an investor partially hedge against his permanent labor income shocks by investing in the stock market. As soon as the investor 
has positive stock investment, the uninsurable component $\sigma_{I}$ of permanent income shocks is reduced to 
$\sqrt{1-\rho^{2}}\sigma_{I}$ due to the correlation assumption between the stock and labor markets. The hedging 
demand rises with the income volatility $\sigma_{I}$.

\subsection{Cointegration Between the Stock and Labor Markets}
In this subsection, we add one more state variable to capture the well-known empirical evidence that the returns to human
capital and stock market returns are highly correlated, i.e., cointegration between the stock and labor markets exist
(see e.g. \cite{BJ97}, \cite{MSV04}, \cite{SV06}, and  \cite{BDG07}). 
Following  \cite{DL10}, 
we assume that labor income 
process $I_t$ is governed by
$$
I_{t}=S_{t}e^{Z_{t}},
$$
where
$
dS_{t}=\mu S_{t}dt+\sigma S_{t}d\mathcal{B}^{1}_{t}.
$
Here, $Z_{t}$ is the difference between the log labor income and log stock price, and is assumed to follow a mean
reverting process
\begin{eqnarray}\label{labor_zprocess}
dZ_{t}=-\alpha(Z_{t}-\overline{z})dt-\sigma_{z}d\mathcal{B}^{1}_{t}+\sigma_{I}d\mathcal{B}^{2}_{t},
\end{eqnarray}
where $\alpha>0$ measures the degree of mean reversion, $\overline{z}$ denotes the long-term mean, $\sigma_{z}$ and 
$\sigma_{I}$ measure the conditional volatilities of the difference change $dZ_{t}$, and $\mathcal{B}^{2}_{t}$ is a 
standard one-dimensional Brownian motion independent of $\mathcal{B}^{1}_{t}$. Given these specifications, we obtain that
\begin{equation}\label{lip2}
dI_{t}/I_{t-}=\{\mu_{I}-\alpha(Z_{t}-\overline{z})\}dt+(\sigma-\sigma_{z})d\mathcal{B}^{1}_{t}+\sigma_{I}d\mathcal{B}^{2}_{t}-(1-\kappa)d\mathcal{N}_{t},~~I_{0}=I>0,
\end{equation}
where
$
\mu_{I}=\mu+\df{1}{2}\sigma^{2}_{z}+\df{1}{2}\sigma^{2}_{I}-\sigma\sigma_{z}.
$

The labor income process (\ref{lip2}) can be calibrated to the empirically observed low contemporaneous
correlations between changes to labor income level and market returns (\cite{CGM05} and \cite{DW13}). 
Concretely, we allow for zero correlation between income shocks and market returns by assuming that 
$\sigma=\sigma_{z}$, so that market risk exposure of labor income dynamics becomes  zero. Instead, our specification 
(\ref{lip2}) for the individual labor income process reflects long-run cointegration between the stock and labor markets. 
The drift term in (\ref{lip2}) captures the notion of long-run dependence between these two markets. When $Z_{t}-\overline{z}<0$, 
the labor income will be increased in the long term, whereas when $Z_{t}-\overline{z}>0$, the labor income will be decreased. 

\subsection{Short Sale and Borrowing Constraints}
Many life-cycle studies explicitly lay out short sale and borrowing constraints to incorporate the realistic U.S.
equity markets (see e.g. \cite{GM05}, \cite{CGM05},  \cite{P07}, \cite{GM08}, \cite{MS10}, \cite{WY10}, and  \cite{LT11}). Indeed, there are a variety of legal and institutional 
constraints that preclude investors from short selling stocks at no cost. 
Further, the presence of borrowing constraints is also consistent 
with the realistic ramifications present in capital markets: many investors are constrained from borrowing against future income, partly because of some realistic market frictions such as informational asymmetry, agency conflicts, and limited enforcement. 

It is important to keep in mind that borrowing constraints with which investors cannot borrow against their future labor income have 
significant effects on their consumption and portfolio decisions over the life cycle (\cite{DL10}). In addition, it is not 
realistic for investors to capitalize their risky investment. In this regard, we try to move the canonical consumption and asset allocation problem with retirement to a more realistic problem by allowing for not only cointegration between the stock and labor markets, but 
also short sale and borrowing constraints.

In the presence of short sale and borrowing constraints, we require that both bond investment $x_{t}$ and stock investment $y_{t}$
are  nonnegative, as a result, financial wealth $W_{t}$ that is the sum of $x_{t}$ and $y_{t}$ is also nonnegative:\footnote{Although the labor income stream is stochastic, allowing for borrowing up to the net present
value of the lowest possible labor income is possible.  \cite{P07} considers this extension
in the robustness check for his model.}
\begin{eqnarray}\label{constraints}
x_{t}\geq 0,~~y_{t}\geq 0,~~W_{t}\equiv x_{t}+y_{t}\geq 0,
\end{eqnarray}
which gives the following dynamics:
\begin{eqnarray}\label{labor_wealthprocess}
dW_{t}=(rW_{t}-c_{t}+I_{t})dt+y_{t}\sigma(d\mathcal{B}^{1}_{t}+\theta dt),~~W_{0}=w\geq 0, 
\end{eqnarray}
where $c_{t}$ is the per-period consumption and $\theta=(\mu-r)/\sigma$ is the Shapre ratio. Investors accumulate wealth by investing
their savings in the riskless and risky assets. Investment in the riskless assets yields a risk-free return $r$, while investment in the risky assets
yields a positive risk premium $y_{t}\sigma\theta=y_{t}(\mu-r)$ in return for bearing stock market risk, which involves
stochastically time-varying term $y_{t}\sigma d\mathcal{B}^{1}_{t}$. Let $\mathcal{A}(w,I,z)$ denote the set of admissible policies such that short sale and borrowing constraints given in (\ref{constraints}) are satisfied.

\subsection{A Portfolio Choice Problem for Voluntary Retirement}
Based on  (\ref{labor_zprocess}), (\ref{lip2}), and (\ref{labor_wealthprocess}), an investor's portfolio selection problem for early retirement with cointegration between the stock and labor markets, and short sale and borrowing 
constraints is to maximize his CRRA utility preference by optimally controlling per-period consumption $c$, investment $y$ in the stock market, and 
the retirement time $\tau$:
\begin{eqnarray}\label{valfcosts}
\begin{aligned}
V(w,I,z) \equiv &  \sup_{(c,y,\tau)\in\mathcal{A}(w,I,z)}\mathbb{E}\Big[\int^{\tau\wedge \tau_D}_{0} e^{-\beta t} \df{c^{1-\gamma}_{t}}{1-\gamma} dt + e^{-\beta \tau\wedge \tau_D} \int_{\tau\wedge \tau_D}^{\tau} e^{-\beta (t-\tau\wedge \tau_D)} \df{c^{1-\gamma}_{t}}{1-\gamma} dt      \\
& \quad \quad \quad \quad \quad \quad \quad  +e^{-\beta \tau}\int^{\infty}_{\tau}e^{- \beta (t-\tau)}\df{(Bc_{t})^{1-\gamma}}{1-\gamma}dt\Big], \\
= & \sup_{(c,y,\tau)\in\mathcal{A}(w,I,z)}\mathbb{E}\Big[\int^{\tau}_{0} e^{-(\beta+\delta_{D})t}
\Big\{\df{c^{1-\gamma}_{t}}{1-\gamma}+\delta_{D}V(W_{t},\kappa I_{t},Z_{t})\Big\}dt\\
&\quad \quad \quad  \quad\quad \quad \quad   +e^{-(\beta+\delta_{D})\tau}\int^{\infty}_{\tau}e^{-(\beta+\delta_{D})(t-\tau)}\df{(Bc_{t})^{1-\gamma}}{1-\gamma}dt\Big],
\end{aligned}
\end{eqnarray}
where $\mathbb{E}$ is the expectation taken at time $0$, and $\beta>0$ is the subjective discount rate. The parameter $B>1$ stands for the leisure preference 
after retirement and implies that the marginal utility of consumption after retirement is larger than that before retirement. Following  \cite{FP07} and \cite{DL10}, we assume 
that the investor does not have any income source after retirement, i.e., labor income $I_{t}$ becomes zero for $t\geq\tau$. 
For calculation simplicity, we assume that the individual also faces disastrous shocks after retirement (This assumption makes no difference because there is zero income after retirement).  As the disastrous shock $\tau_D$ follows a Poisson process with intensity $\delta_D$, we obtain the second equality in (\ref{valfcosts}) via integrating out $\tau_D$. 
We further assume that the investor 
has no bequest motive to simplify our analysis. It is straightforward to extend our problem to include the positive post-retirement income and the bequest 
motive.

Notice that the last integral term in (\ref{valfcosts}) is to be maximized after retirement. This maximization problem is the canonical Merton's (\cite{M69}, \cite{M71}) optimal 
consumption and investment problem without labor income. We denote by $G(W_{\tau})$ the maximal utility value after retirement that has the form of
$$
G(w)=\df{B^{1-\gamma}\overline{K}^{-\gamma}}{1-\gamma}w^{1-\gamma},~~ \text{ where }\overline{K}=\df{\beta}{\gamma}-\df{1-\gamma}{\gamma}
\Big(r+\df{\theta^{2}}{2\gamma}\Big).
$$
Although an investor does not receive any income after retirement, the Merton solution $G(W_{\tau})$ does have an impact on the
pre-retirement strategies via a non-linear option-type component in the retirement decision that gives rise to the optimal characterizations of work and retirement regions. The principle of dynamic programming suggests that the value function (\ref{valfcosts}) 
satisfies:
\begin{equation}\label{valfcosts2}
\begin{aligned}
V(w,I,z)\equiv\sup_{(c,y,\tau)\in\mathcal{A}(w,I,z)}\mathbb{E}\Big[\int^{\tau}_{0}e^{-(\beta+\delta_{D})t}
\Big\{\df{c^{1-\gamma}_{t}}{1-\gamma}+\delta_{D}&V(W_{t},\kappa I_{t},Z_{t})\Big\}dt\\
&+e^{-(\beta+\delta_{D})\tau}G(W_{\tau})\Big].
\end{aligned}
\end{equation}

Moreover, we call DL's model, which is aimed to maximize the investor's utility preference by optimally controlling $(c,y,\tau)$ subjected to the basic dynamic process (\ref{lip}) and (\ref{labor_wealthprocess}) with the instantaneous correlation $\rho=0$, as the benchmark model. Hence, in the benchmark model, there is no cointegration effect between stock and labor markets. Indeed, the benchmark model can be treated as a special case of our cointegration model when $\alpha=0$.  We will compare the benchmark model with our model quantitively in  Section \ref{coint_numerical}.

\section{Theoretical  Analysis}\label{coint_appendix}
\subsection{The HJB Equation}
The value function (\ref{valfcosts2}) should satisfy the following HJB equation (see e.g. Chapter 4 in \cite{BL82} and Chapter 11 in  \cite{O07}): 
for any $w\geq 0$, $I\geq 0$, $z\in\mathbb{R}$,
\begin{equation}\label{vari}
\max_{(c,y)\in\mathcal{A}(w,I,z)}\Big\{\mathcal{L}V(w,I,z), \quad G(w)-V(w,I,z)\Big\}=0.
\end{equation}
where the differential operator $\mathcal{L}$ is given by
$$
\begin{aligned}
\mathcal{L}V=&\df{c^{1-\gamma}}{1-\gamma}-cV_{w}+\df{1}{2}\sigma^{2}y^{2}V_{ww}+\df{1}{2}[\sigma^{2}_{I}+(\sigma-\sigma_{z})^{2}]I^{2}V_{II}
+\df{1}{2}(\alpha^{2}_{z}+\sigma^{2}_{I})V_{zz}\\
&+\sigma(\sigma-\sigma_{z})yIV_{wI}-\sigma\sigma_{z}yV_{wz}+[\sigma^{2}_{I}-(\sigma-\sigma_{z})\sigma_{z}]IV_{zI}\\
&+[y(\mu-r)+rw+I]V_{w}+[\mu_{I}-\alpha(z-\overline{z})]IV_{I}-\alpha(z-\overline{z})V_{z} -\beta V\\
&+\delta_{D}\Big(\mathbb{E}_\kappa[V(w,\kappa I,z)]- V(w,I,z)\Big),
\end{aligned}
$$ 
$\mathbb{E}_{\kappa}$ is the expectation with respect to $\kappa$.
Here, the subscripts of $V$ denote its partial derivatives. The HJB equation (\ref{vari}) implies that there are two regions: one is the work region in which an investor is optimal to work to receive labor income and the other one is the retirement region in which 
he is optimal to exit from the workforce, i.e., to enter retirement. The first variational inequality in (\ref{vari}) becomes zero when the investor stays in the work region, whereas it becomes negative when he enters into the retirement region. 
The presence of a labor income jump shock  is captured by the last expectation term involving $\delta_{D}$ in the differential operator $\mathcal{L}$.

The second variational inequality in (\ref{vari}) measures the difference between value functions after and before retirement. As long as this variational inequality stays negative, i.e.,  the value function $V$ with an unexercised retirement option is larger than the value function $G$ after retirement, an investor is optimal to continue to work. 
As the investor accumulates wealth and thus, once his value function prior to retirement approaches the post-retirement value function, i.e., this variational inequality becomes 
an equality, he enters into the retirement region and consequently the retirement option is exercised. Since the work region and the retirement region cannot be overlapped, 
we should jointly consider those two variational inequalities in (\ref{vari}), to allow the maximum of these two terms to be equal to zero.

The work region and retirement region can be represented as follows:
\begin{eqnarray*}
\textbf{Work Region } &=& \{(w,I,z): \mathcal{L}V(w,I,z) = 0, w\geq 0, I\geq 0, z \in \mathbb{R}\}, \\
\textbf{Retirement Region} &=& \{(w,I,z):  G(w) - V(w,I,z) = 0,  w\geq 0, I\geq 0, z \in \mathbb{R}\}.
\end{eqnarray*}


Using the homogeneity property of the value function, we can reduce the dimensionality
of the problem by the following transformation:
\begin{eqnarray*}
V(w,I,z)=\frac{\bar K^{-\gamma}}{1-\gamma}\left(w+\frac{I}{r}\right)^{1-\gamma} e^{(1-\gamma)u(\xi, z)}, \quad \xi=\frac{I/r}{w+I/r}\in [0, 1].
\end{eqnarray*}
After retirement, we know $G(w)$ satisfies $G(w) = \frac{B^{1-\gamma}\bar K^{-\gamma}}{1-\gamma}w^{1-\gamma}$.
Therefore, the associated HJB equation for new function $u(\xi, z)$ becomes
\begin{eqnarray}\label{vari2}
\max_{\bar y, \bar c}\left\{ \mathcal{L}_1 u(\xi, z),\quad  \mathcal{R} u(\xi)\right\} =0,
\end{eqnarray}
on $\{(\xi, z): \xi\in[0,1], z\in \mathbb{R}\}$, where 
\begin{eqnarray*} \begin{aligned}
\mathcal{L}_1 u =& \left[\frac{1}{2}\sigma^2 \bar y^2 \xi^2 + \frac{1}{2}(\sigma_I^2+(\sigma-\sigma_z)^2)\xi^2(1-\xi)^2 -\sigma(\sigma-\sigma_z)\bar y \xi^2(1-\xi)\right][u_{
\xi\xi}+(1-\gamma) u_{\xi}^2] \\
& +\left[\sigma\sigma_z \bar y \xi + (\sigma_I^2 -\sigma_z(\sigma-\sigma_z)) \xi(1-\xi)\right] [u_{\xi z} +(1-\gamma)u_{\xi}u_z] \\
&+\frac{1}{2}(\sigma_I^2+\sigma_z^2)[u_{zz}+(1-\gamma)u_z^2] +\bigg[\gamma\sigma^2 \bar y^2 +\gamma\sigma(\sigma-\sigma_z)(2\xi-1)\bar y-(
\mu-r)\bar y \\
&-\gamma(\sigma_I^2+(\sigma-\sigma_z)^2)\xi(1-\xi) +(\mu_I-\alpha(z-\bar z))(1-\xi) -r \bigg]\xi u_{\xi}\\
&+ \left[-(1-\gamma)\sigma\sigma_z \bar y +(1-\gamma)(\sigma_I^2-\sigma_z(\sigma-\sigma_z))\xi -\alpha(z-\bar z)\right] u_z\\
& +(\mu-r -\gamma\sigma(\sigma-\sigma_z)\xi)\bar y -\frac{1}{2}\gamma\sigma^2 \bar y^2 -\frac{1}{2}(\sigma_I^2+(\sigma-\sigma_z)^2)\gamma \xi^2 \\
&+(\mu_I-\alpha(z-\bar z))\xi +r -\frac{\beta +\delta_D}{1-\gamma}+\delta_D\mathbb{E}_{\kappa}\left[\frac{(1+(\kappa-1)\xi)^{1-\gamma}}{1-\gamma} e^{(1-\gamma) \left(u\left(\frac{\kappa \xi}{1+(\kappa-1)\xi}\right) - u\right)}  \right]\\
& + \frac{\bar{K}^\gamma}{1-\gamma} e^{-(1-\gamma) u} \bar c^{1-\gamma} - \bar c(1-\xi u_\xi),\\
\mathcal{R} u =& \ln(1+(\kappa-1)x) +\ln B -u,
\end{aligned}\end{eqnarray*}
where $\mathbb{E}_{\kappa}$ is the expectation with respect to $\kappa$, $\bar y := \frac{y}{w+I/r}$, and $\bar c:=\frac{c}{w+I/r}$. 

At boundary $\xi=0$, i.e., when $w=\infty$, the HJB equation is degenerated and the solution approximates Merton case.
At boundary $\xi=1$, i.e., when $w=0$, it is known that the investor should not invest in stock\footnote{It is well documented that optimal investment in the stock market should be zero as wealth approaches zero. This condition is exactly same with the borrowing constraint against future labor income (\cite{DL10}).} and his consumption should not exceed the current labor income. Thus we have
$$
\bar y^*\big|_{\xi = 1}=0, \quad \bar c^*\big |_{\xi=1}=\min\left\{\bar K e^{(1-1/\gamma) u(1,z)}(1-u_\xi(1,z))^{-1/\gamma},  r \right\}.
$$
For $\xi \in(0,1)$, the optimal investment and consumption in the presence of constrained borrowing and short selling are determined by the first order condition:
\begin{eqnarray*}
\bar y^* = \min\left\{ \max\{h(\xi,z), 0\}, 1-\xi\right\}, \quad  \bar c^* = \bar K e^{(1-1/\gamma)u}(1-\xi u_{\xi})^{-1/\gamma},
\end{eqnarray*}
where $\bar y^*\geq 0$ comes from the short selling constraint, $\bar y^* \leq 1-\xi$ comes from the borrowing constraint against future labor income, and 
$$
\begin{aligned}
h(\xi,z)=-\Big[&
\mu-r -\gamma\sigma(\sigma-\sigma_z)\xi+[\gamma\sigma(\sigma-\sigma_z)(2\xi-1)
+r-\mu ]\xi u_\xi -(1-\gamma)\sigma\sigma_z u_z\\ 
&+\sigma\sigma_z\xi[u_{\xi z}+(1-\gamma)u_\xi u_z] 
-\sigma(\sigma-\sigma_z) \xi^2(1-\xi)[u_{\xi\xi}+(1-\gamma)u_{\xi}^2]\Big]\\
&\Big/\Big[ 
{\sigma^2\xi^2[u_{\xi\xi}+(1-\gamma)u_{\xi}^2]+2\gamma\sigma^2\xi u_\xi -\gamma\sigma^2}\Big].
\end{aligned}
$$

\subsection{Verification of the Value Function and Optimal Policy}
We show that the solution of HJB equation (\ref{vari}) coincides with the original utility function defined in (\ref{valfcosts}). 
\begin{proposition}\textbf{(Verification Theorem).} Let $V(w,I,z)$ be a smooth solution to HJB equation 
(\ref{vari}). Assume that for any admissible controls, we have the following transversality condition\begin{eqnarray*}
\lim_{t\rightarrow \infty} \mathbb{E}\bigg[e^{-\beta t} V(W_t, I_t, Z_t)\bigg]  =0. 
\end{eqnarray*}
Then $V(w,I,z)$ equals the value function defined in  (\ref{valfcosts}), and  the optimal strategy is given by  \begin{eqnarray}\label{labor_optimalstrategy}\begin{aligned}
c^*_t =& V_w(W_t^*, I_t, Z_t)^{-1/\gamma}, \\
y^*_t =&\max\left\{\min\bigg\{ \frac{\big[\sigma\sigma_z V_{wz} -\sigma(\sigma-\sigma_z)I_t V_{wI}-(\mu-r)V_w\big](W_t^*, I_t, Z_t)}{\sigma^2 V_{ww}(W_t^*, I_t, Z_t)}, 1 \bigg\},  0  \right\},\\
\tau^* =& \inf\{t\geq 0:  V(W_t^*, I_t, Z_t)\geq G(W_t^*)\}.
\end{aligned} \end{eqnarray}
\end{proposition}


\begin{proof}
We want to prove the solution $V(w,I,z)$ is not less than the value function defined in (\ref{valfcosts}) and the equality achieves under the strategy defined in (\ref{labor_optimalstrategy}).

For any admissible strategy $\{c_t,y_t,\tau\}$, let us define 
\begin{eqnarray*}
M_t = \int_0^t e^{-\beta s} \bigg[(1-R_s)U(c_s) ds+ G(W_s)dR_s \bigg]+ e^{-\beta t} (1-R_t) V(W_t, I_t, Z_t),
\end{eqnarray*}
where $R_t:=\mathbf{1}_{\{t>\tau\}}$. 
Without loss of generality, we assume $R_0=0$. Otherwise if $R_0=1$, then we have $V(w, I, z) \geq \mathbb{E}[\int_0^\infty e^{-\beta t}U(Bc_t) dt ]$  with equality achieved when $c_t=c_t^*, y_t=y_t^*$, and $R_t=R_t^*$.
By the generalized Ito's formula,
\begingroup\makeatletter\def\f@size{11}\check@mathfonts
\begin{eqnarray*}\begin{aligned}
dM_t =& e^{-\beta t}(1-R_t)U(c_t) dt +  e^{-\beta t} G(W_t)dR_t- e^{-\beta t}V(W_t, I_t, Z_t) dR_t  \\
&+ e^{-\beta t}(1-R_t)[V_w(W_t,I_t, Z_t)dW_t+V_I(W_t,I_t, Z_t)dI_t+V_z(W_t,I_t, Z_t)dZ_t]\\
& + e^{-\beta t}(1-R_t)[V_{ww}(W_t,I_t, Z_t)dW_tdW_t+V_II(W_t,I_t, Z_t)dI_tdI_t+V_{zz}(W_t,I_t, Z_t)dZ_tdZ_t ]\\
& + e^{-\beta t}(1-R_t)[V_{wI}(W_t,I_t, Z_t)dW_tdI_t+V_{wz}(W_t,I_t, Z_t)dW_tdZ_t+V_{Iz}(W_t,I_t, Z_t)dI_tdZ_t ]\\
& + (1-R_t)\frac{\partial }{\partial t}\left(e^{-\beta t} V(W_t, I_t, Z_t)\right)dt  \\
=& e^{-\beta t}(1-R_t)U(c_t) dt  +(1-R_t)\left[ e^{-\beta t}(\mathcal{L}V(W_t,I_t,Z_t)-U(c_t))\right]dt \\
& +e^{-\beta t} G(W_t)dR_t- e^{-\beta t}V(W_t, I_t, Z_t) dR_t \\
&+ e^{-\beta t}(1-R_t)\big[\sigma V_w(W_t,I_t, Z_t) y_t  + (\sigma-\sigma_z)V_I(W_t, I_t, Z_t)I_t -\sigma_z V_z(W_t, I_t, Z_t) \big]d\mathcal{B}_t^1\\
& + e^{-\beta t}(1-R_t)\big[\sigma_I  V_I(W_t, I_t, Z_t)I_t+\sigma_I V_z(W_t, I_t, Z_t)  \big]d\mathcal{B}_t^2.
\end{aligned}
\end{eqnarray*}\endgroup
Define $\mathcal{O}_n: = \{(w,I,z): \frac{1}{2n} \leq w\leq n, |z|<n, \frac{1}{2n} \leq I\leq n \}$ and a sequence of stopping time $\theta_n: = n\wedge \inf\{t\geq0: (W_t,I_t,Z_t)\notin \mathcal{O}_n \}$.
We then integrate the above equation from $0$ to $\theta_n$:
\begingroup\makeatletter\def\f@size{10}\check@mathfonts
\begin{eqnarray*}\begin{aligned}
M_{\theta_n} =& M_0 +\int_0^{\tau\wedge\theta_n} (1-R_s)e^{-\beta s} \mathcal{L}V(W_s,I_s,Z_s) ds +\int_{\tau\wedge\theta_n}^{\theta_n} (1-R_s)e^{-\beta s} \mathcal{L}V(W_s,I_s,Z_s) ds\\
&+\int_0^{\theta_n} e^{-\beta s}\bigg[ G(W_s) -V(W_s, I_s, Z_s)\bigg]dR_s \\
&+\int_0^{\theta_n} e^{-\beta s} (1-R_s)\bigg[\sigma V_w(W_s,I_s, Z_s) y_s + (\sigma-\sigma_z)V_I(W_s, I_s, Z_s)I_s -\sigma_z V_z(W_s, I_s, Z_s) \bigg]d\mathcal{B}_s^1 \\
& + \int_0^{\theta_n} e^{-\beta s} (1-R_s)\bigg[\sigma_I  V_I(W_s, I_s, Z_s)I_s+\sigma_I V_z(W_s, I_s, Z_s)  \bigg]d\mathcal{B}_s^2. 
\end{aligned}\end{eqnarray*}\endgroup
By the form of (\ref{vari}) and the definition of $\{c^*_t, y^*_t, R_t^*\}$ in (\ref{labor_optimalstrategy}), we obtain that the first integral is always non-positive for any feasible strategy $\{c_t, y_t, R_t\}$ and is equal to zero for the claimed optimal policy $\{c^*_t, y^*_t, R_t^*\}$ in (\ref{labor_optimalstrategy}). That is because if $(\tau^*\wedge \theta_n)\geq (\tau\wedge\theta_n)$, the solution function $V$ satisfies $\mathcal{L}V\leq 0$ by (\ref{vari}) and the equality achieves under claimed optimal strategy $\{c_t^*, y_t^*\}$, and if $(\tau^*\wedge \theta_n)<(\tau\wedge\theta_n)$, we have $V=G(w)$ and $\mathcal{L} = U(c) - c\partial_w +\frac{1}{2}\sigma^2 y^2 \partial_{ww}+[rw+y(\mu-r)]\partial_w-\beta$ during $[\tau^*\wedge \theta_n, \tau\wedge\theta_n]$ so that $\mathcal{L}V<0$ as $B>1$. Therefore, the first non-positive integral equals to zero only when  $c_t=c^*_t$, $y_t=y^*_t$, and $R_t=R_t^*$. The second integral equals zero for both $\theta_n\leq\tau$ and $\theta_n>\tau$ (in this case, $1-R_t=0$ during $[\tau,\theta_n]$).
The third integral is always non-positive for every feasible policy $\{c_t,y_t,R_t\}$ because $V(W_t, I_t, Z_t)\geq G(W_t)$ and is equal to zero only when $c_t=c_t^*$, $y_t=y_t^*$, and $\tau\geq \tau^*$ as $V(W_t, I_t, Z_t) =G(W_t)$ for $t\geq \tau^*$. The last  two stochastic integrals under expectation equals zero as $V_w(W_t, I_t, Z_t)$, $V_z(W_t, I_t, Z_t)$, and $V_I(W_t, I_t, Z_t)$ are bounded when $(W_t,I_t, Z_t)$ is in a bounded domain during $[0, \theta_n]$. 

Noticing that $M_0 = V(W_0, I_0, Z_0)$, we then take expectation in above equation to get 
\begin{eqnarray*}
V(W_0, I_0, Z_0) &\geq& \mathbb{E} \int_0^{\theta_n} e^{-\beta s} \bigg[(1-R_s)U(c_s)ds + G(W_s)dR_s\bigg] \\
&&+ \mathbb{E}\big[e^{-\beta \theta_n} (1-R_{\theta_n}) V(W_{\theta_n}, I_{\theta_n}, Z_{\theta_n})\big]. 
\end{eqnarray*}
As analyzed above, the equality above holds only for the claimed optimal strategy $\{c_t^*, y_t^*, R_t^*\}$ defined in (\ref{labor_optimalstrategy}).
As $n\rightarrow \infty$, $\theta_n$ increases to infinity with probability 1. By the transversality condition of $V$ and dominant convergence theorem, 
the first expectation above converges to the original utility function $\mathbb{E} [\int_0^\tau e^{-\beta s}U(c_s)ds$ $+e^{-\beta \tau}G(W_{\tau})]$ and the second expectation goes to zero.  Equality holds for the claimed optimal policy $\{c^*_t, y^*_t, R_t^*\}$ and this completes the proof.
\end{proof}

\section{Quantitative Analysis}\label{coint_numerical}
In order to obtain economically plausible implications on optimal portfolio choice
for early retirement with cointegration between the stock and labor markets, we
carry out quantitative analysis with reasonable parameter values. Penalty method (\cite{DZ08}) is used to solve the new HJB (\ref{vari2}) numerically. We set a lower bound of solvency domain at $z_{\text{min}}=\bar z -8\sigma_z$ and an upper bound at $z_{\text{max}} =\bar z+ 8\sigma_z$. We add boundary condition $V_z\big |_{z=z_{\text{min}}}=V_z\big |_{z=z_{\text{max}}}=0$, or equivalently, $u_z\big |_{z=z_{\text{min}}}=u_z\big |_{z=z_{\text{max}}}=0$.  We note that our numerical approach is robust to the choice of $z_{\text{min}}$ and $z_{\text{max}}$ in $z$ direction. 

\subsection{Baseline Parameters}
For the benchmark case, we set asset returns according to  \cite{DL10}.  We set the expected stock return 
$\mu$ to $5\%$ and the risk-free return to $1\%$ and hence the mean equity premium to 
$4\%$. We choose the stock volatility $\sigma=18\%$. The annual subjective 
discount rate $\beta$ is considered to be $4\%$, which is a common value adopted in 
the existing life-cycle literature (see e.g. \cite{CGM05}, \cite{GM05}, \cite{WY10}, and \cite{WWY16}). The higher value of subjective 
discount rate than the risk-free return incorporates a constant mortality rate and makes 
an individual relatively impatient compared with the bond market, so that incentives of 
the wealthier people to save in the form of riskless assets are weakened.

The coefficient $\gamma$ of relative risk aversion is assumed to have a moderate value 
of $3$, which is significantly lower than the upper bound $10$ for relative risk aversion 
suggested by \cite{MP85}. The parameter $B$ of leisure preference after 
voluntary retirement is chosen as $2$ for the baseline value.

As regards to the parameters in labor income dynamics, we set the annual expected rate
$\mu_{I}$ and volatility $\sigma_{I}$ of income growth to $0.5\%$ and
$10\%$, respectively, which are very similar values compared to  \cite{Dn91}, \cite{C92}, \cite{DL10}, and \cite{WWY16}. Accordingly, under these parameter values 
the implied expected change $(\mu_{I}-\sigma^{2}_{I}/2)$ of logarithmic income level 
becomes  zero. We normalize the initial annual rate of labor income $I$ while working 
as  $1$. The recovery parameter $\kappa$ in the probability density function for jump size  is assumed to be constant and set to $80\%$.\footnote{We check the robustness of
our main results when $\kappa$ is assumed to follow a random variable according to a
power distribution (not reported).} Under these specifications, labor income will be 
reduced from $1$ to $0.8$ when a discrete and jump shock takes place. \cite{C92},  \cite{P07}, and \cite{LT11} all consider this kind of low-income state 
in which income level would be even zero or have a significantly lower proportion of 
permanent labor income.\footnote{In line with this, \cite{CDK03} 
assume that the low-income state can be thought of as $20\%$ of permanent income after 
accounting for some roles in a safety net taken by formal and informal insurance markets.}

The labor income process and its risk features are well estimated using the Panel Study of Income Dynamics data set (\cite{C92} and \cite{CGM05}), while the probabilities of a disastrous labor income shock that results in an 
unexpected, permanent, and exogenous labor income drop are very difficult to be estimated 
from the panel data. Unlike the common transitory income shocks (see e.g. \cite{C92}, \cite{CGM05}, and \cite{P07}), our income shocks are regarded as a permanent 
low-income state scenario, which seems to induce observationally similar effects of the 
random horizon.\footnote{For the treatment of random horizon, see  \cite{V01} who 
considers a finite expected lifetime occurring with constant probabilities.}  In order to 
fully reflect some realistic episodes for the calibration of disastrous labor income shocks, 
we consider an \text{extreme} event of firm default, which induces a negative and large 
income reduction followed by unemployment. We relate our disastrous 
labor income shock intensity $\delta_{D}$ to the permanent disastrous event arising from 
firm default. We calibrate the intensity $\delta_{D}$ to Moody's historical data of average 
cumulative issuer-weighted global rates by rating categories from 1983 to 2011 
(\cite{JPR13}). More precisely, we select four rating categories of $Aaa$, $Aa$, 
$A$, and $B$. The calibration results for $\delta_{D}$ are as follows: 0.0001 (Aaa), 0.0012 
(Aa), 0.0030 (A), and 0.0526 (B). The baseline parameter for $\delta_{D}$ is set to $5\%$, 
which is an approximation for the calibration result of rating category $B$ and the same 
value used in \cite{WWY16}).

Many existing life-cycle models (see e.g., \cite{BMS92}, \cite{JK96}, \cite{V01},  \cite{FP07}, and  \cite{DL10}) are not compatible with empirical stylized facts such as stock market 
non-participation and household portfolio share that rises in wealth, without resorting to high 
correlations between labor income shocks and market returns. However, income shocks and stock 
returns are not highly correlated, consistent with the data (\cite{CCGM01} and \cite{GM05}). For instance, \cite{CCGM01} and \cite{GM05} estimate 
the contemporaneous correlation between the changes to labor income and market returns as $15\%$. Most 
importantly, we take the empirically plausible assumption that the contemporaneous correlation 
should be zero, i.e., we set $\rho=0$ in the benchmark model and 
 $\sigma=\sigma_{z}$ in our cointegration model. Under this assumption, market risk exposure of labor 
income dynamics (\ref{lip2}) involveing $(\sigma-\sigma_{z})$ becomes zero. Labor income is active in its response to shocks in the stock market in the long term. This reflects 
the long-run cointegration effect. More precisely, in our labor income dynamics, when $Z_{t}-\overline{z}<0$, labor  income will be expected to rise, whereas when $Z_{t}-\overline{z}>0$, labor income will be expected to 
fall. For the baseline parameter value, we set the degree of mean reversion $\alpha$ and long-term mean 
$\overline{z}$ to $15\%$ and zero, respectively. Further, we specify the initial condition for $Z_{t}$ 
as $Z_{0}=0$, which is highly likely to be steady state.

The Table \ref{coint_table_summary} summarizes our baseline parameter values.

\begin{table}[t]\small
\centering
$$
\begin{array}{l|c|c}
\hline\hline
\mbox{Parameters} & \mbox{Symbol} & \mbox{Value}\\
\hline
\mbox{Risk-free interest rate} & r & 1\%\\
\mbox{Expected rate of stock return} & \mu & 5\%\\
\mbox{Stock volatility} & \sigma & 18\%\\
\mbox{Relative risk aversion} & \gamma & 3\\
\mbox{Leisure preference after voluntary retirement} & B & 2\\
\mbox{Subjective discount rate} & \beta & 4\%\\
\mbox{Expected rate of income growth} & \mu_{I} & 0.5\%\\
\mbox{Volatility on income growth} & \sigma_{I} & 10\%\\
\mbox{Annual rate of labor income} & I & 1\\
\mbox{Recovery parameter} & \kappa & 80$\%$\\
\mbox{Disastrous labor income shock intensity} & \delta_{D} & 5\%\\
\mbox{Volatility on difference between the log stock price and log income} & \sigma_{z} & \sigma\\
\mbox{Degree of mean reversion} & \alpha & 15\%\\
\mbox{Long-term mean} & \overline{z} & 0\\\mbox{Contemporaneous correlation between stock and income returns (DL)} & \rho & 0\\
\hline\hline
\end{array}
$$
\caption[Summary of baseline parameters]{\textbf{Summary of baseline parameters.}}
\label{coint_table_summary}
\end{table}

\subsection{Optimal Consumption}
Figure \ref{coint_MPC} states that the marginal propensities to consume (MPC) out of financial wealth
decreases as wealth increases, consistent with \cite{FP07} and  \cite{DL10}, with and without cointegration between 
the stock and labor markets. This implies the concavity of consumption function (\cite{CK96}). More interestingly, in our model the MPC is much lower than \cite{DL10}, which is a strong indicator
of investors' sentiment toward early retirement. As shown later, wealth at retirement is lower
with cointegration than without cointegration. Because wealth can be used to not only finance future
consumption, but also control the irreversible time of voluntary retirement, early retirement encourages investors to save more compared to the case without conintegration. As a result, investors value a unit of consumption less than a unit of financial wealth in order to accumulate wealth for optimal retirement, implying lower the MPCs than DL (\cite{DL10}).

\begin{figure}[H]
\centering
\begin{tabular}{cc}
\includegraphics[width=0.46\textwidth]{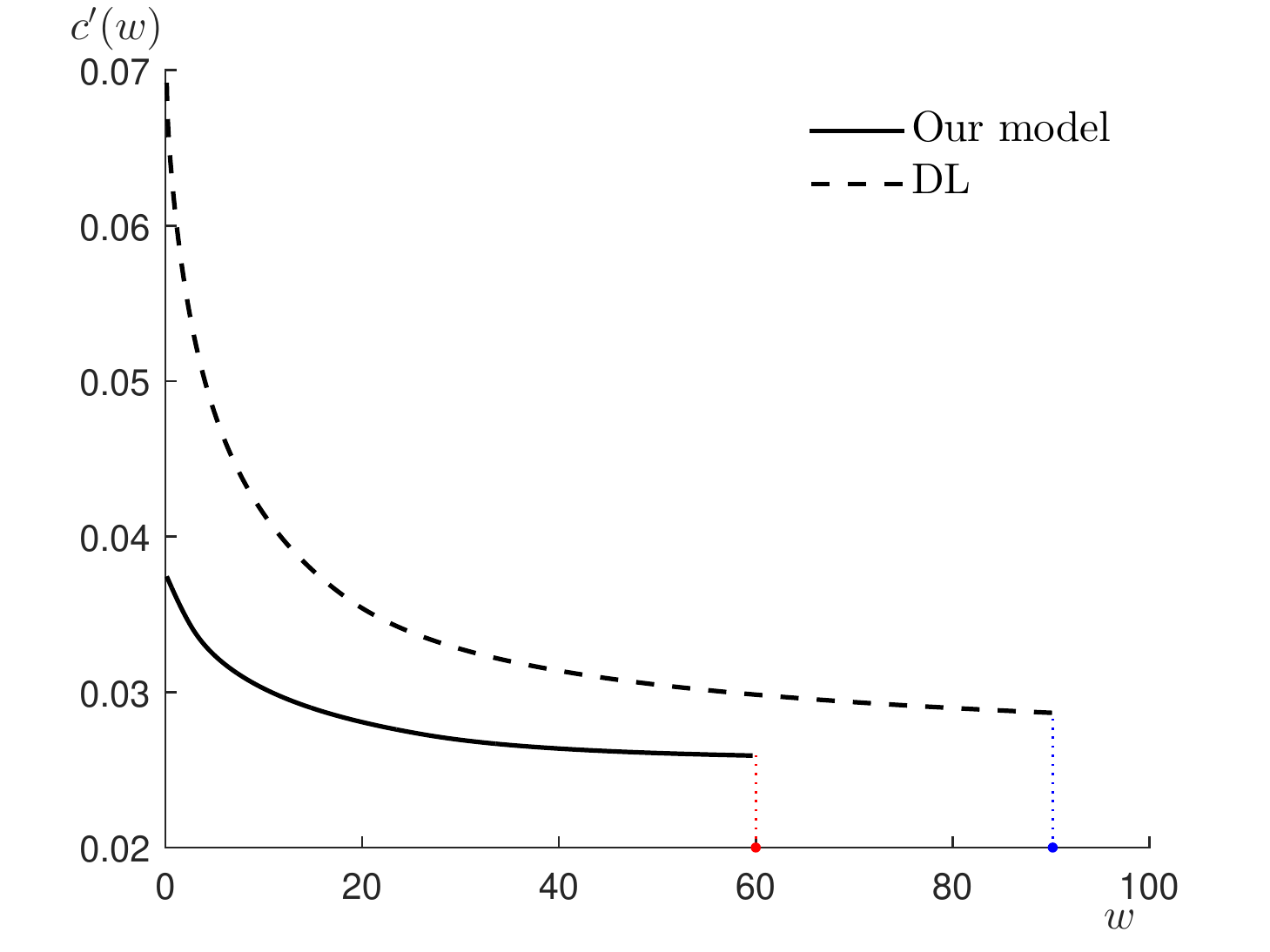} & \hspace{-.5cm}
\includegraphics[width=0.46\textwidth]{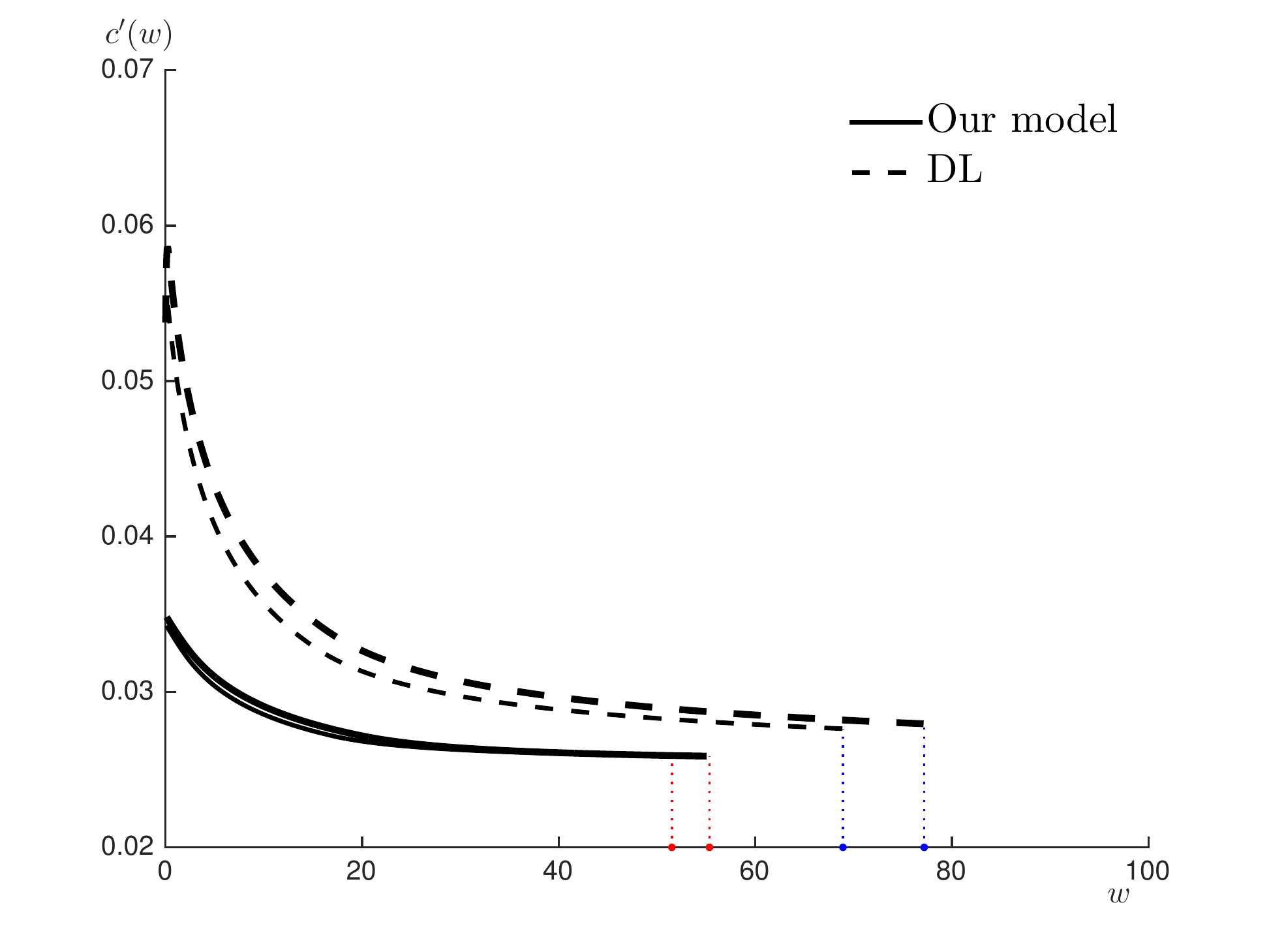} \\
{\small (i) $\delta_{D}=0, z=0$} & {\small (ii) $z=0$, $\delta_{D}=5\%$ for thick line; $\delta_{D}=3\%$ for thin line}
\end{tabular}
\caption[Marginal propensities to consume out of financial wealth (MPC)]{\textbf{Marginal propensities to consume out of financial wealth (MPC).}  Since we normalize  annual rate of labor income as one, i.e., $I=1$, wealth-to-income ratio $w/I$ reduces to financial wealth $w$. The 
dotted line and the solid line represent the \cite{DL10} result and the our result, 
respectively. The projected end points of DL and our model on wealth-to-income ratio horizon (or x-axis) denote wealth thresholds for voluntary retirement over which an investor's optimal choice is to retire 
permanently. In our model, the MPC is much lower than DL. Basic parameters are chosen from Table \ref{coint_table_summary}.}
\label{coint_MPC}
\end{figure}

\subsection{Optimal Investment}
Figure \ref{coint_portfolio share} represents the proportion of financial wealth invested in the stock
market (or the portfolio share) as a function of wealth-to-income ratio. Since we normalize annual
rate of labor income as one, wealth-to-income ratio reduces to financial wealth. While in DL (\cite{DL10}),
the portfolio share decreases in financial wealth, in our model we find that there exists a
target wealth-to-income ratio under which an investor does not participate in the stock market
at all (This is  resolution to the non-participation puzzle), and above which the investor
increases the portfolio share as he  accumulates wealth. As a result of cointegration between
the stock and labor markets, i.e., due to the long-run dependence between those two markets, returns
to human capital and stock market returns become highly positively correlated (\cite{BDG07}) and hence, an investor with little wealth who is away from retirement
does not make any investments in stocks, even when the market risk premium is positive. Put
differently, human capital's implicit equity holdings resulting from cointegration significantly
lower risky asset investments and even lead to zero stock holdings. 

As the investor accumulates wealth, the cointegration effect becomes attenuated and human capital 
starts to act as implicit bond holdings (see e.g. \cite{HL97}, \cite{JK96}, \cite{CGM05}, \cite{FP07}, and \cite{DL10}). Consequently, the investor finds it optimal to increase 
the portfolio share as he  accumulates wealth over a certain threshold, in the interest of striking 
a balance of his optimal portfolio by tilting it toward stocks. 
The presence of downward jumps in labor income, at the most basic level, reduces 
the target wealth-to-income ratio. As a result, the presence of downward jumps in labor income induces an investor to participate in the stock market earlier than without downward jumps, which is caused by additional  hedging component in portfolio against human capital risk.

\begin{figure}[H]
\centering
\begin{tabular}{cc}
\includegraphics[width=0.46\textwidth]{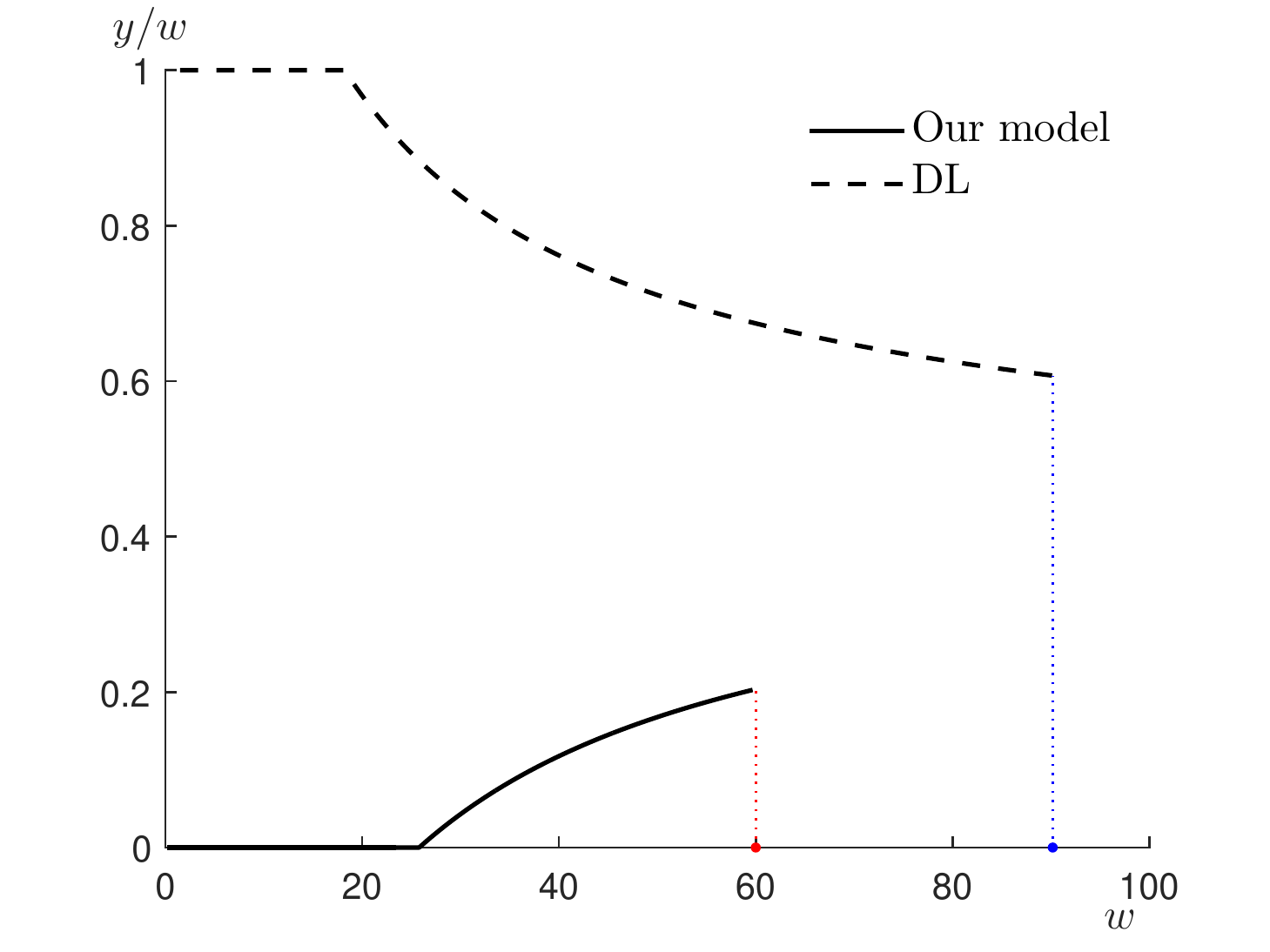} & \hspace{-.5cm}
\includegraphics[width=0.46\textwidth]{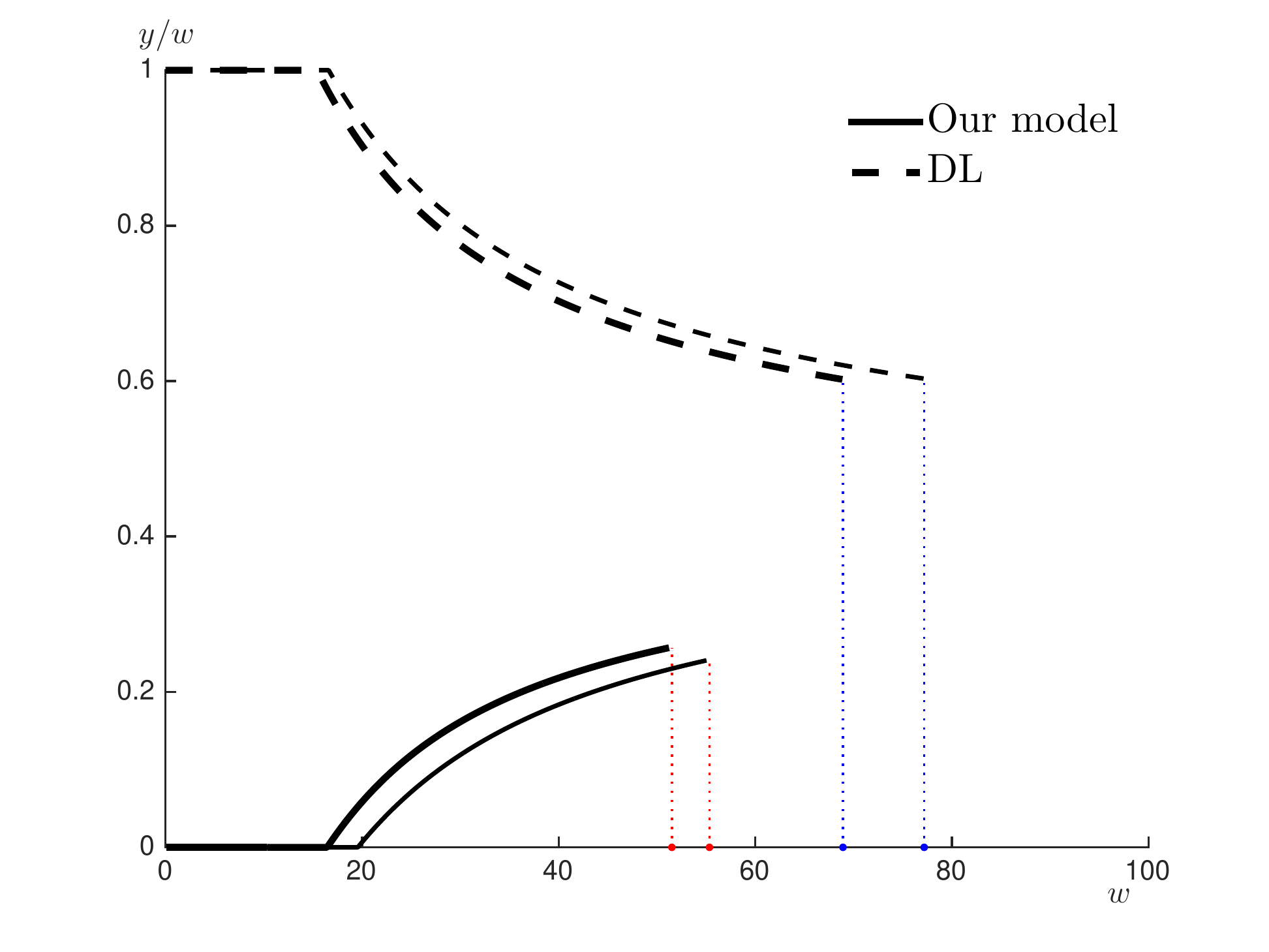} \\
{\small (i) $\delta_{D}=0, z=0$} & {\small (ii) $z=0$, $\delta_{D}=5\%$ for thick line; $\delta_{D}=3\%$ for thin line}
\end{tabular}
\caption[Portfolio share 
as a function of wealth-to-income ratio]{\textbf{Portfolio share 
as a function of wealth-to-income ratio.}  Since we normalize annual rate of labor income as one, i.e., $I=1$, wealth-to-income ratio $w/I$ reduces to financial wealth $w$.  The dotted line and the solid line represent DL (\cite{DL10}) and our
model, respectively. The left panel denotes the case without downward jumps in labor income, whereas 
the right panel stands for the case with downward jumps in labor income.  Basic parameters are chosen from Table \ref{coint_table_summary}. }
\label{coint_portfolio share}
\end{figure}

\noindent\textbf{Changes in investment opportunity and risk aversion}

In relation to the effects of changing investment opportunity set and risk aversion on portfolio 
share, Figure \ref{coint_portfolio share_investment_opportunity} and Figure \ref{coint_portfolio share_gamma}
demonstrate that in DL (\cite{DL10}) and our model, better investment opportunity set and lower risk aversion 
(equivalently, higher expected rate of stock return $\mu$, lower stock volatility $\sigma$, or
lower risk aversion $\gamma$) raises the portfolio share, which can be inferred from the standard mean-variance 
effect rule in optimal portfolio choice (\cite{M69}, \cite{M71}).

\begin{figure}[H]
\centering
\begin{tabular}{cc}
\includegraphics[width=0.46\textwidth]{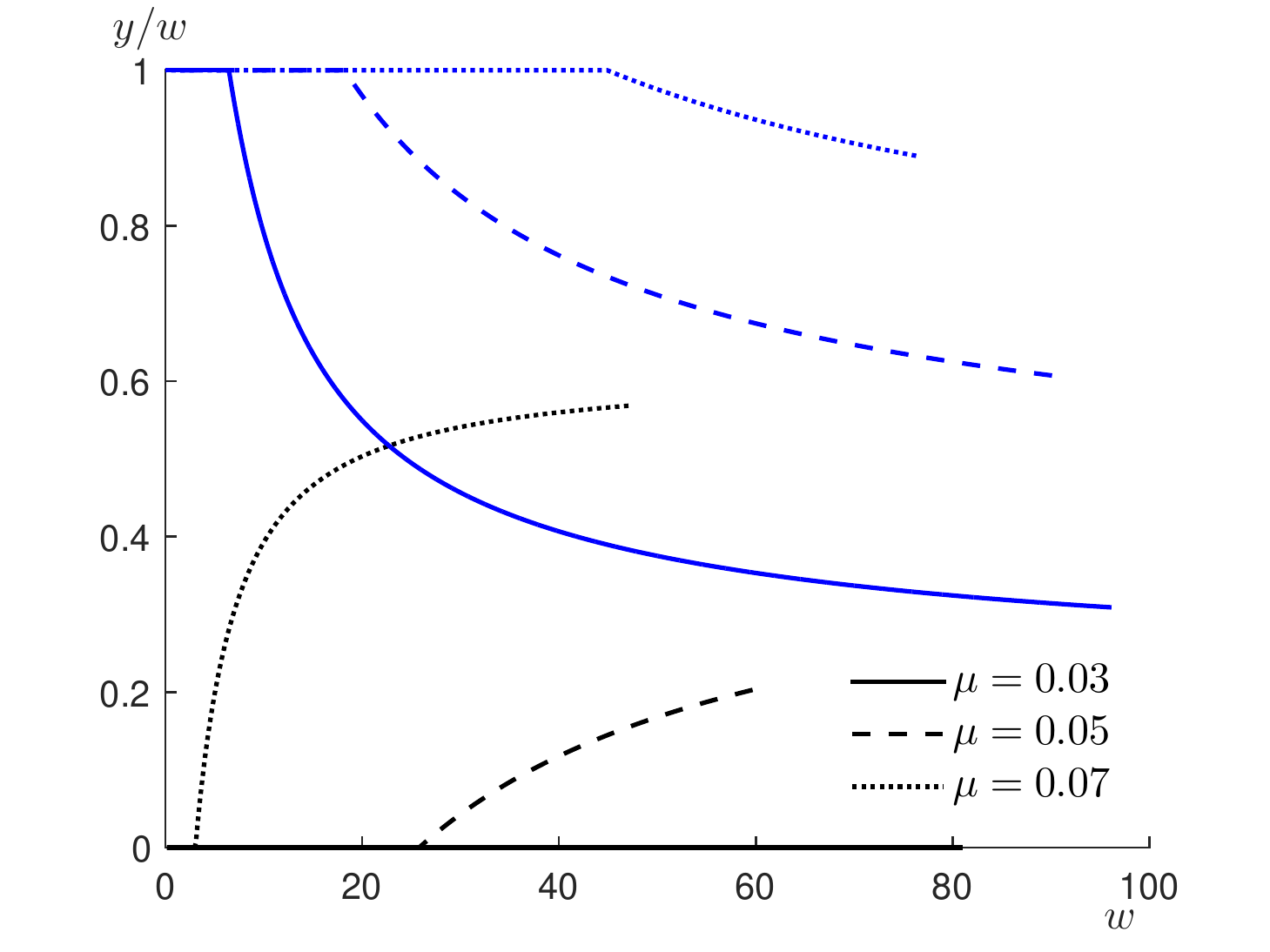} & \hspace{-.5cm}
\includegraphics[width=0.46\textwidth]{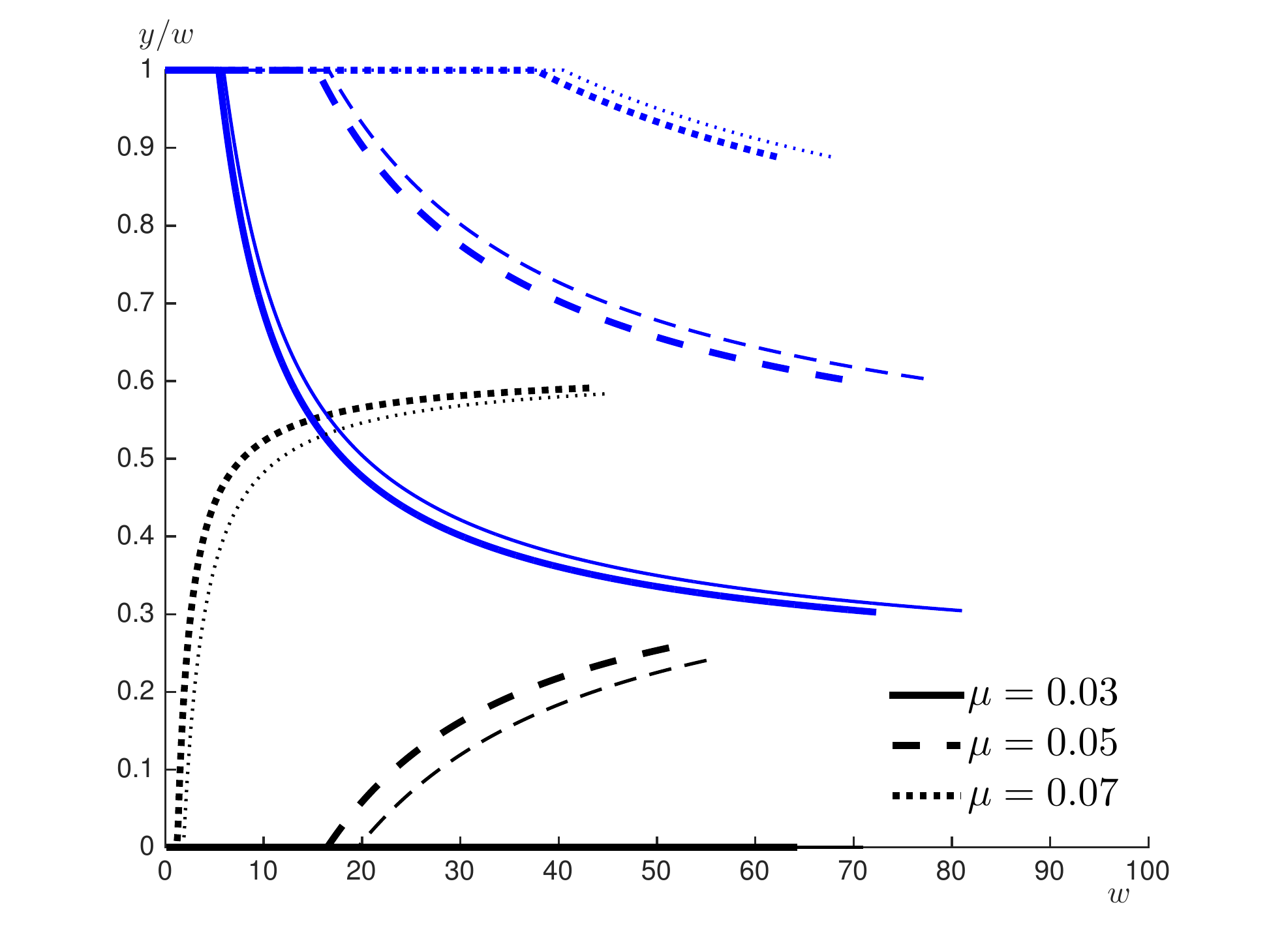} \\
{\small (i) $\delta_{D}=0, z=0$} & {\small (ii) $z=0$, $\delta_{D}=5\%$ for thick line; $\delta_{D}=3\%$ for thin line} \\
\includegraphics[width=0.46\textwidth]{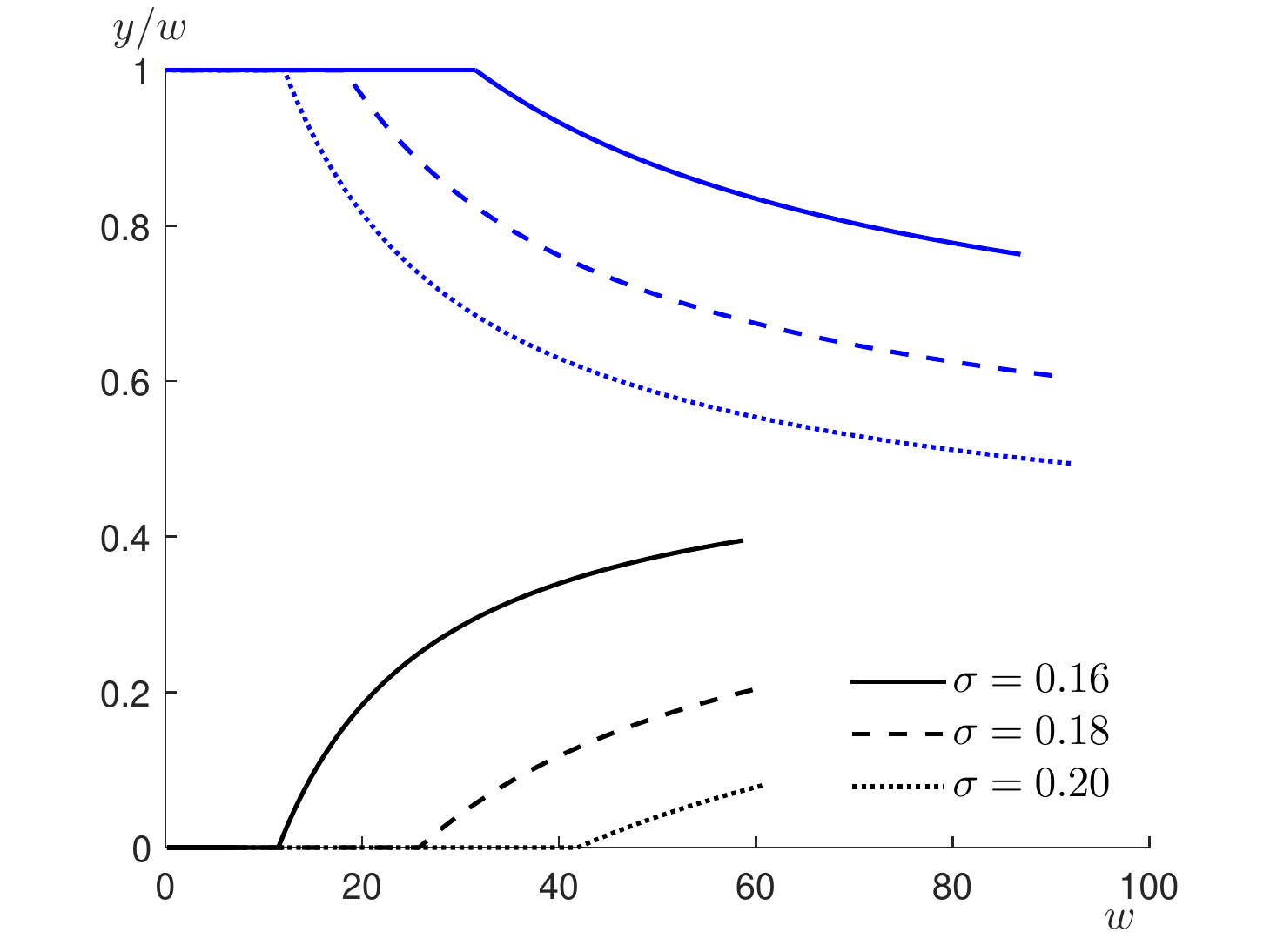} & \hspace{-.5cm}
\includegraphics[width=0.46\textwidth]{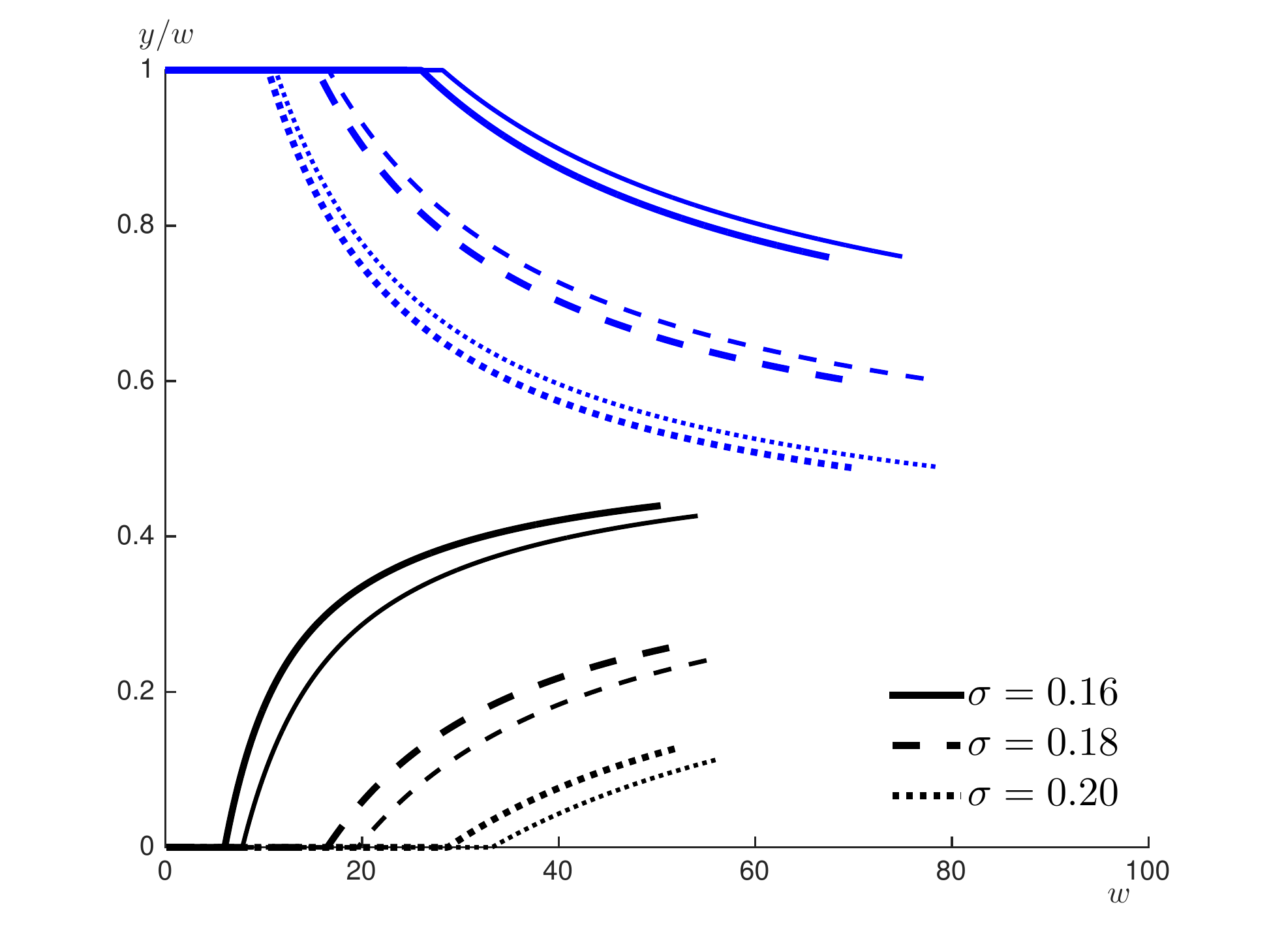} \\
{\small (iii) $\delta_{D}=0, z=0$} & {\small (iv) $z=0$, $\delta_{D}=5\%$ for thick line; $\delta_{D}=3\%$ for thin line}
\end{tabular}
\caption[Sensitivity analysis of portfolio share with respect to investment opportunity set]{\textbf{Sensitivity analysis of portfolio share with respect to investment opportunity set.}  The blue lines (upper three lines) and the
black lines (lower three lines) represent DL (\cite{DL10}) and our model, respectively. The left panels denote the case without downward  jumps in labor income, whereas the right panels stand for the case with downward jumps in labor income. Basic parameters are chosen from Table \ref{coint_table_summary}. }
\label{coint_portfolio share_investment_opportunity}
\end{figure}

\begin{figure}[H]
\centering
\begin{tabular}{cc}
\includegraphics[width=0.46\textwidth]{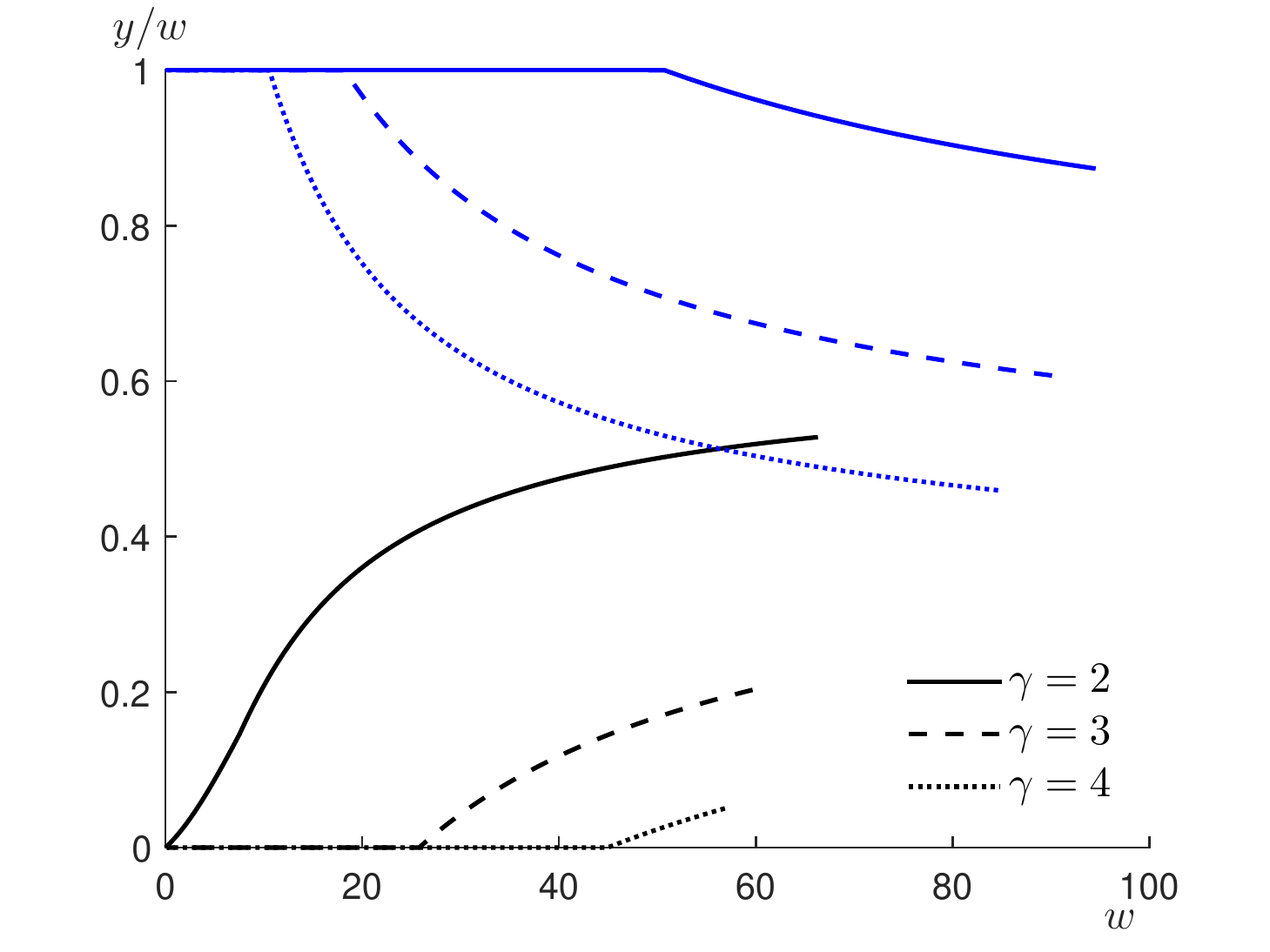} & \hspace{-.5cm}
\includegraphics[width=0.46\textwidth]{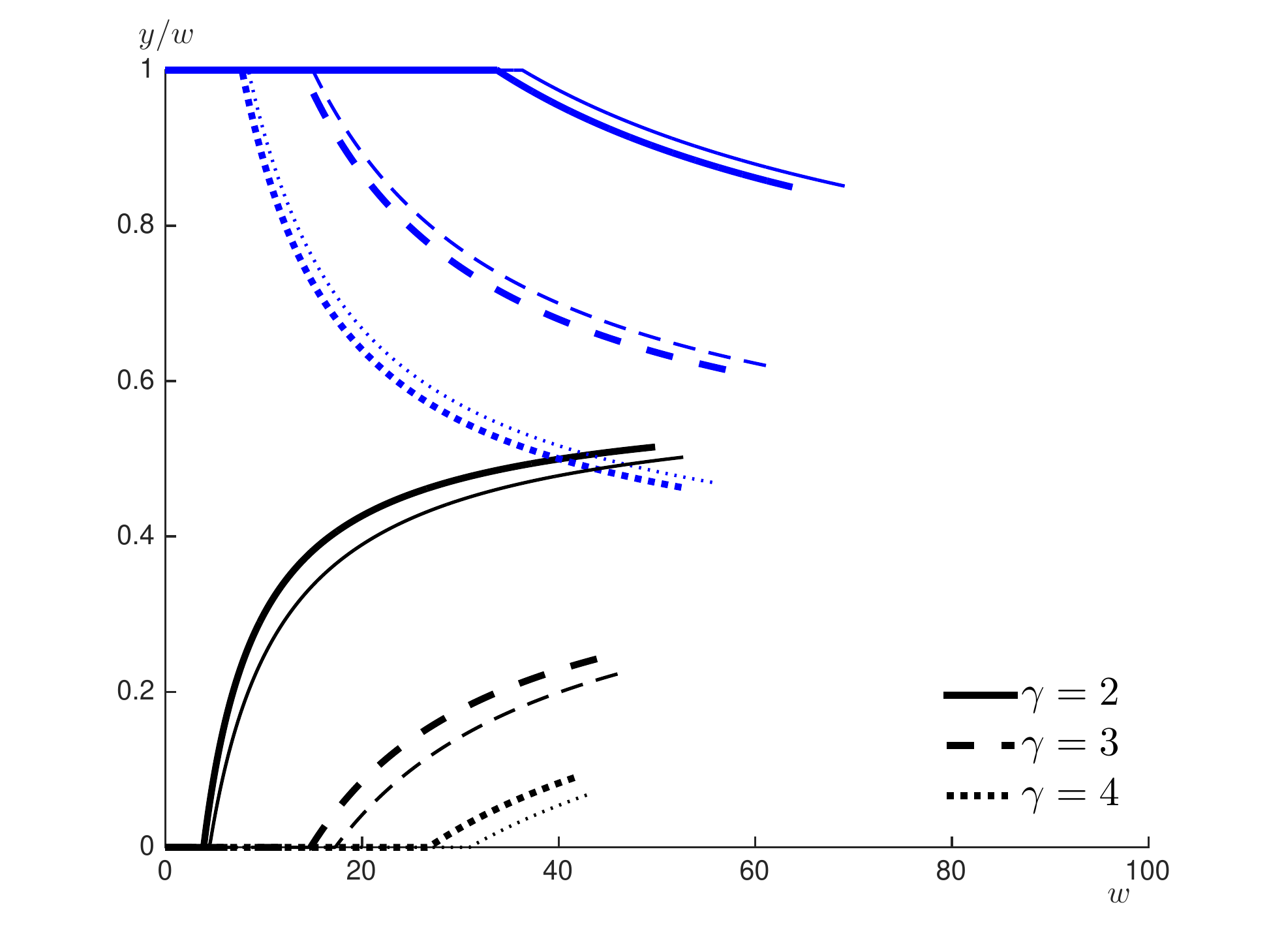} \\
{\small (i) $\delta_{D}=0, z=0$} & {\small (ii) $z=0$, $\delta_{D}=5\%$ for thick line; $\delta_{D}=3\%$ for thin line}
\end{tabular}
\caption[Sensitivity analysis of portfolio share with respect to risk aversion $\gamma$]{\textbf{Sensitivity analysis of portfolio share with respect to risk aversion $\gamma$.}  The blue lines (upper three lines) and the
black lines (lower three lines) represent DL (\cite{DL10}) and our model, respectively. The left panel denotes 
the case without downward jumps in labor income, whereas the right panel stands for the case with downward jumps in labor income. Basic parameters are chosen from Table \ref{coint_table_summary}. }
\label{coint_portfolio share_gamma}
\end{figure}

\noindent\textbf{Changes in volatility on income growth}

When an investor trades a risky asset, he  can achieve partial hedging effects against uninsurable
human capital risk, especially in the presence of cointegration. Figure \ref{coint_portfolio share_income_volatility}
shows that while in DL (or without cointegration) a riskier stream of future labor income (or higher $\sigma_{I}$) lowers the 
portfolio share, which is an accurate reflection of risk 
diversification,\footnote{Theoretical or empirical evidence for pointing out the effect of labor income 
risks has been provided by \cite{BMS92}, \cite{K93}, \cite{K98}, \cite{V01}, \cite{CGM05}, \cite{GM05}, \cite{BDG07}, \cite{WY10}, \cite{MS10}, \cite{LT11}, and \cite{CS14}.} in our model (or with cointegration) the result is reversed. Further, when labor income risk is greater (or
higher $\sigma_{I}$), the target wealth-to-income threshold will be more decreased and consequently, 
the investor participates in the stock market earlier than scheduled. In addition to the standard Merton 
mean-variance rule, the investor has additional hedging demand against uninsurable labor income risk, 
because his labor income is likely to be highly correlated with the stock market in the long run.

\begin{figure}[H]
\centering
\begin{tabular}{cc}
\includegraphics[width=0.46\textwidth]{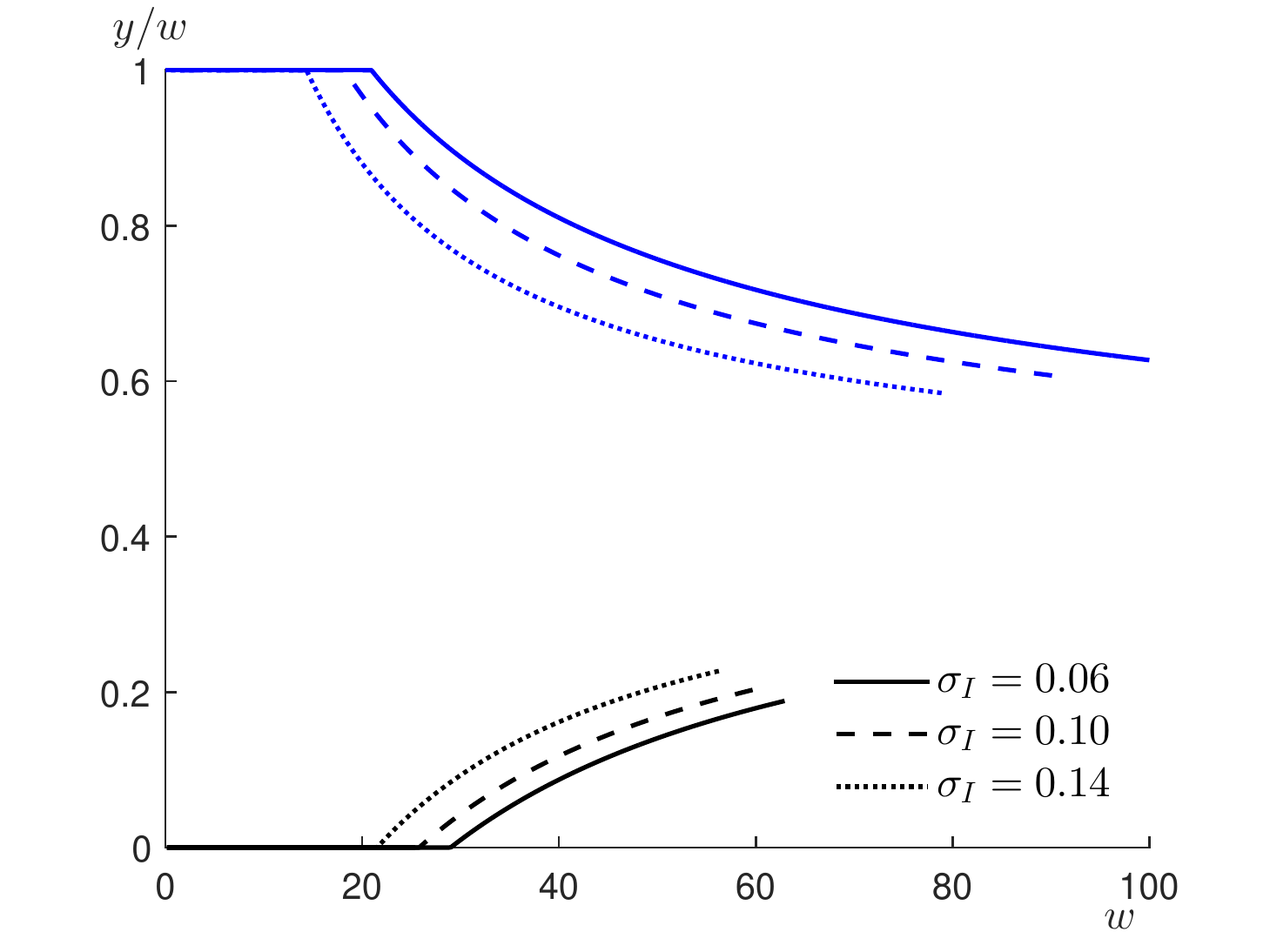} & \hspace{-.5cm}
\includegraphics[width=0.46\textwidth]{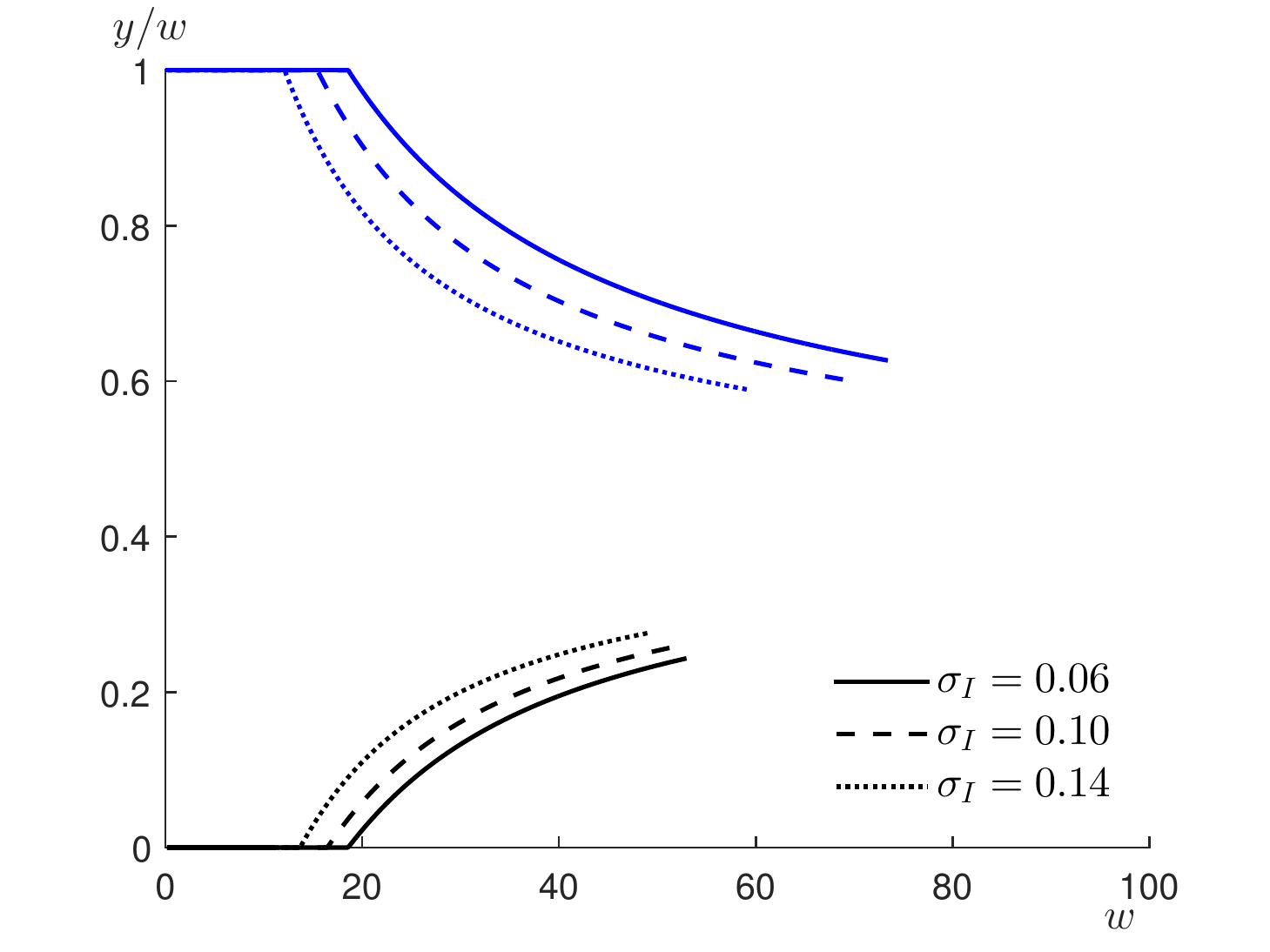} \\
{\small (i) $\delta_{D}=0, z=0$} & {\small (ii) $\delta_{D}=5\%, z=0$}
\end{tabular}
\caption[Sensitivity analysis of portfolio share with respect to  income
growth volatility $\sigma_{I}$]{\textbf{Sensitivity analysis of portfolio share with respect to  income growth volatility $\sigma_{I}$.}  The blue lines (upper three lines) and the
black lines (lower three lines) represent DL (\cite{DL10}) and our model, respectively. The left panel denotes the case without downward jumps in labor income, whereas the right panel  stands for the case with downward jumps in labor income. Basic parameters are chosen from Table \ref{coint_table_summary}. }
\label{coint_portfolio share_income_volatility}
\end{figure}

\noindent\textbf{Changes in degree of mean reversion}

Let's look into further details of the effects of cointegration on portfolio share. Due to the 
mean-reverting process $z_t$ that represents cointegration, when the difference
between current income and its long-run mean is positive, income will fall, whereas when the 
difference is negative, income will increase in the long term. A larger degree of mean reversion 
(or higher $\alpha$) is equivalent to a larger difference between current income and its long-run
mean and hence, it results in more savings in the form of riskless assets to finance future consumption
in anticipation of larger income fluctuations, \text{ceteris paribus} (Figure 
\ref{coint_portfolio share_alpha}). This result can be witness to the fact that an investor delays his 
optimal stock market participation when labor income is more transitory (or higher $\alpha$).

\begin{figure}[H]
\centering
\begin{tabular}{cc}
\includegraphics[width=0.46\textwidth]{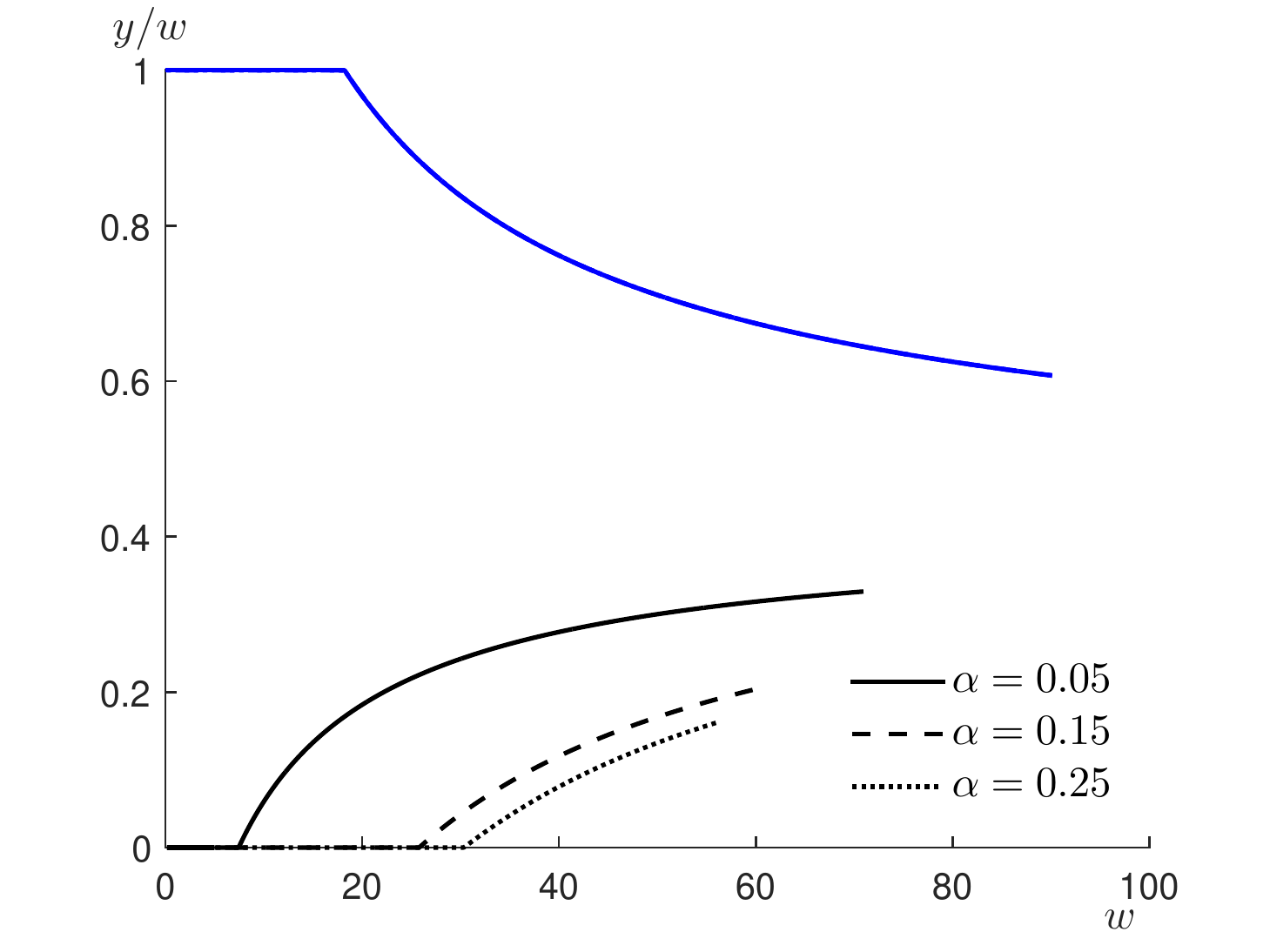} & \hspace{-.5cm}
\includegraphics[width=0.46\textwidth]{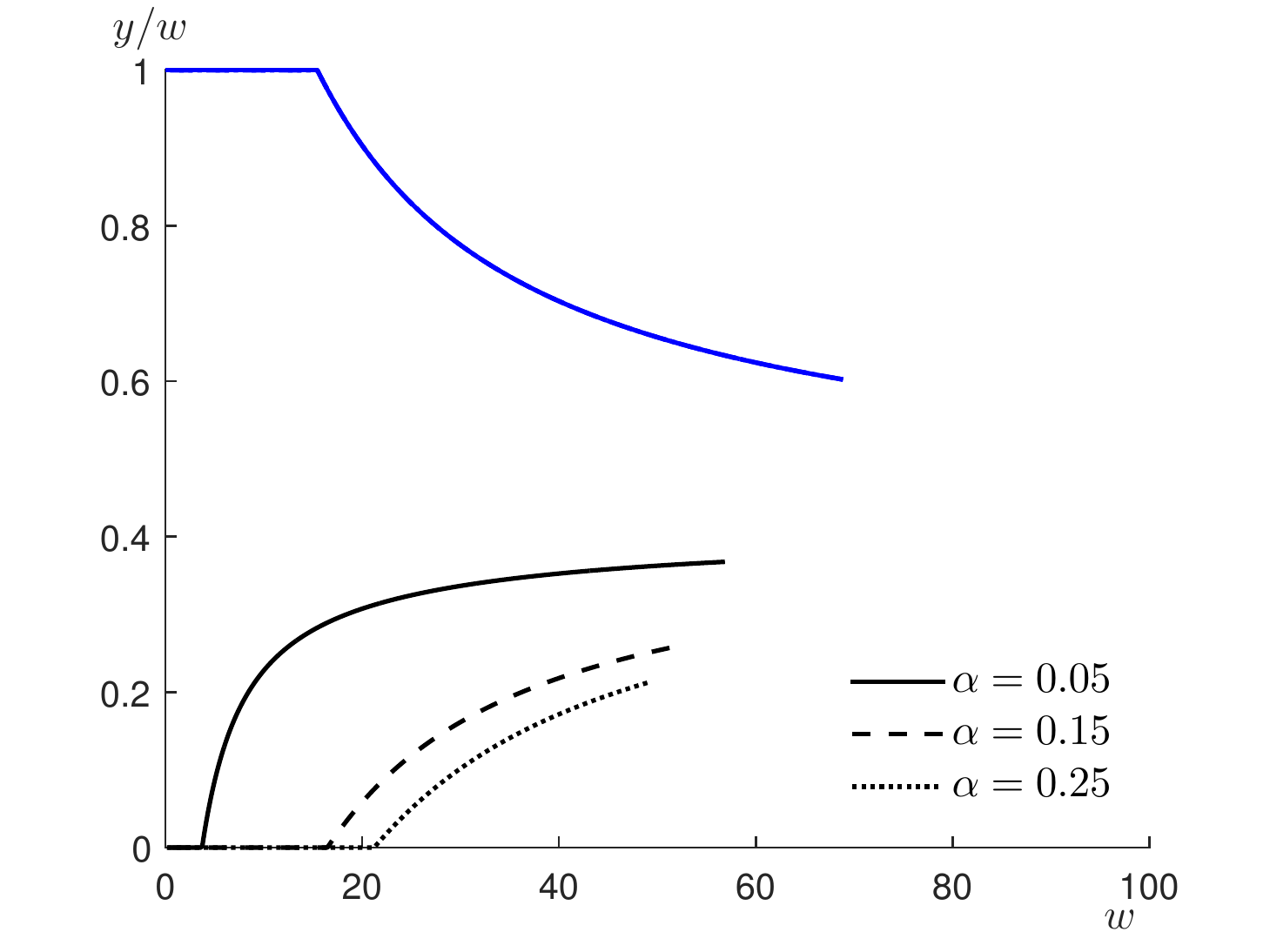} \\
{\small (i) $\delta_{D}=0, z=0$} & {\small (ii) $\delta_{D}=5\%, z=0$}
\end{tabular}
\caption[Sensitivity analysis of portfolio share with respect to mean reversion $\alpha$]{\textbf{Sensitivity analysis of portfolio share with respect to mean reversion $\alpha$.} The blue line (upper line) and the black lines (lower three lines) represent DL (\cite{DL10}) and our model, respectively. The left panel denotes the case without downward jumps in labor income, whereas the right panel stands for the case with downward jumps in labor income.  Basic parameters are chosen from Table \ref{coint_table_summary}. }
\label{coint_portfolio share_alpha}
\end{figure}

\noindent\textbf{Effects of retirement flexibility}

Consistent with \cite{FP07} and \cite{DL10}, we also consider the labor supply along the extensive margin, i.e.,
the flexibility in choosing the irreversible optimal time of retirement.\footnote{\cite{BMS92} and  \cite{LN02} explore some implications of flexible labor 
supply in a continuous fashion on life-cycle strategies. However, working-hours are, if nothing 
else, not flexible, rather irreversible labor supply is consistent with empirical evidence.}
Retirement flexibility has turned out to be a crucial element when deciding optimal portfolio 
choice in the sense that it increases equity holdings when the correlation between labor income shocks and stock returns is set to be very low, consistent with the data. The intuition behind this is that retirement flexibility can
be liable for making labor income's beta negative when labor income and stock returns are not
highly correlated and subsequently, leads to more investment in stocks due to its effective
hedging role against labor income risks (left panel in Figure \ref{coint_portfolio share_retirement_flexibility}). 
However, in the presence of cointegration, i.e., with long-run dependence between the labor and 
stock markets, the result can be reversed. We find that retirement flexibility pushes an investor's
portfolio to be geared toward relatively safe assets. The reason 
is that as a result of cointegration, that is, when returns to human capital vary significantly
with market returns, retirement flexibility strengthens such cointegration effects and 
makes labor income's beta more positive rather than negative. On account of cointegration, a
more conservative investment policy shows up with retirement flexibility than without
it (right panel in Figure \ref{coint_portfolio share_retirement_flexibility}).

\begin{figure}[H]
\centering
\begin{tabular}{cc}
\includegraphics[width=0.46\textwidth]{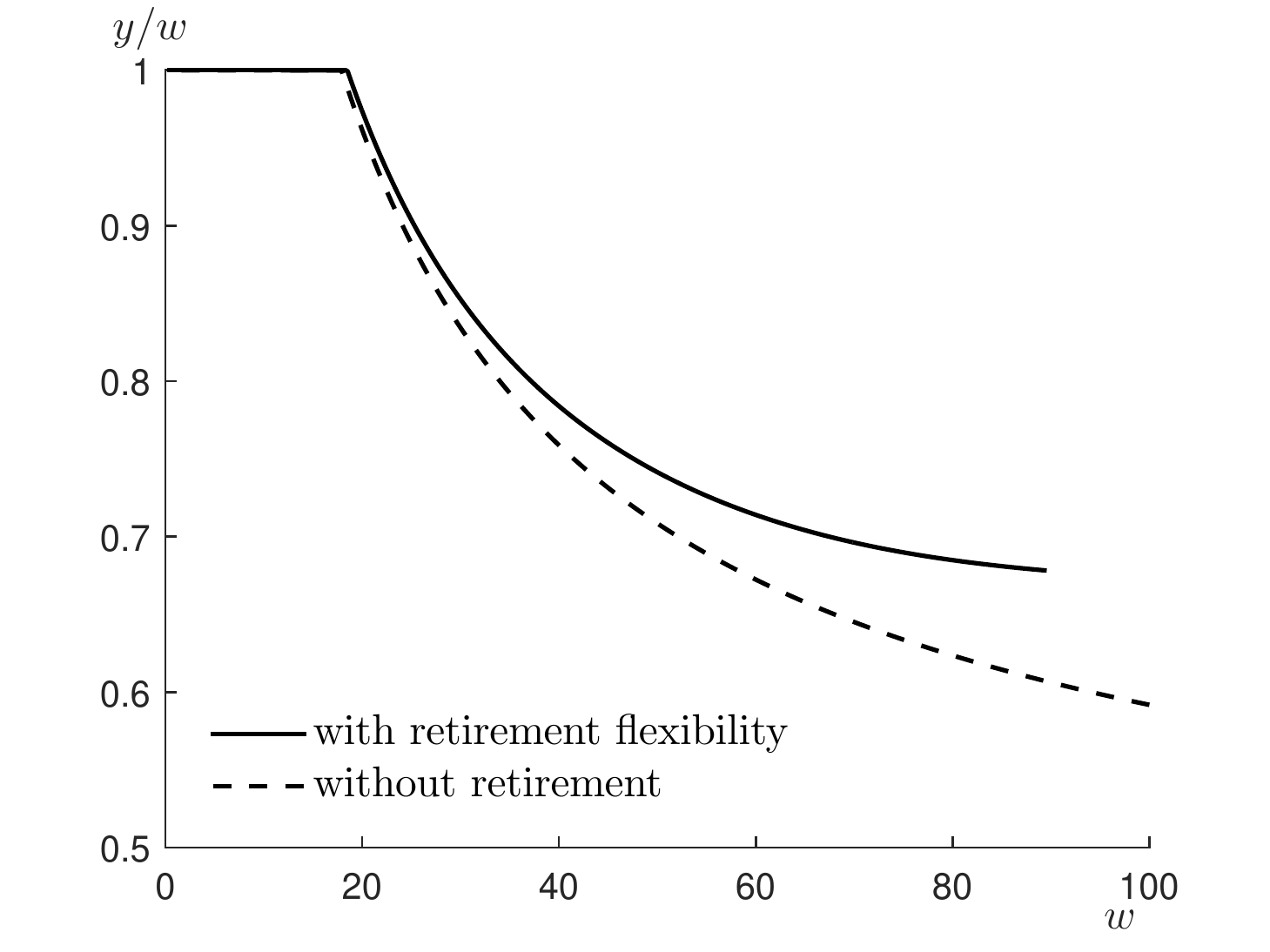} & \hspace{-.5cm}
\includegraphics[width=0.46\textwidth]{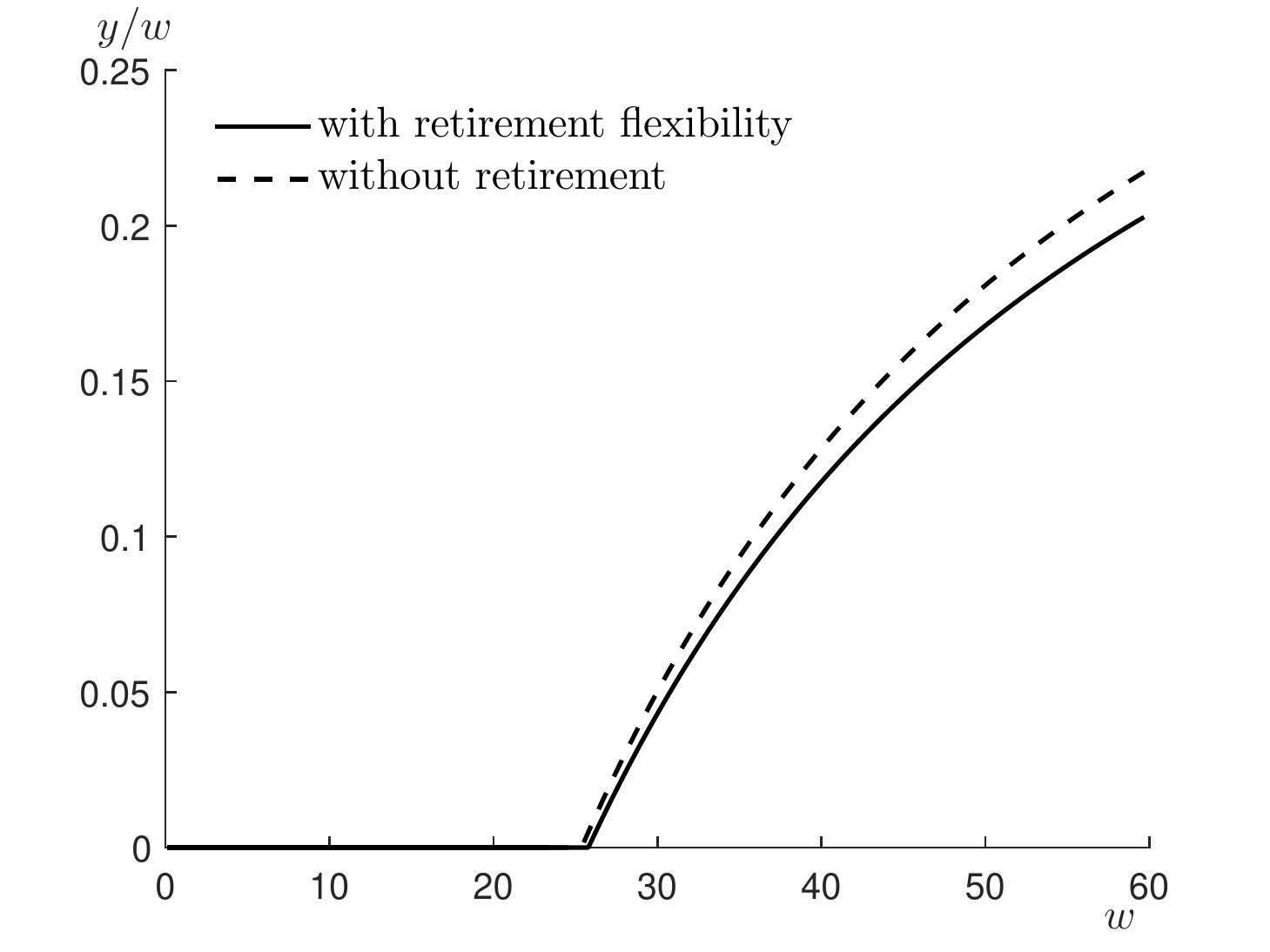} \\
{\small (i) DL (without cointegration)} & {\small (ii) Our model (with conintegration)}
\end{tabular}
\caption[Effects of retirement flexibility on portfolio share]{\textbf{Effects of retirement flexibility on portfolio share.} The solid line and the dotted line 
represent the cases with and without retirement flexibility, respectively. The left panel denotes the case 
without cointegration, whereas the right panel stands for the case with cointegration. 
To reflect empirical observations in the data, the contemporaneous correlation $\rho$ in DL (\cite{DL10}) between the stock and labor markets is zero. Initial value $z$ of 
additional state variable representing cointegration between the stock and labor markets is assumed to be zero. Without cointegration retirement flexibility increases the portfolio share, whereas with 
cointegration it decreases the portfolio share. Basic parameters are chosen from Table \ref{coint_table_summary}. }
\label{coint_portfolio share_retirement_flexibility}
\end{figure}

\subsection{Early Retirement}
Most of the analysis in this subsection is focused on providing economic justification for voluntary
retirement. Early retirement is quite reasonable and even numerically plausible, especially when
wages are expected to fall in the long run. The numerical result is related to empirical evidence
in Issues in Labor Statistics (2000), \cite{GS02}, and \cite{GST10} that early retirement was quite possible between 1995 and 2000 during
which the U.S. economy experienced a stock market boom and a quick rise in the stock market
returns.
In order to dig up economic plausibility of early retirement, we introduce an economic notion
of implicit value of human capital (\cite{K98}), which can be defined by the marginal rates of
substitution between labor income and financial wealth as follows:
$$
\frac{\partial{V(w,I,z)}}{\partial{I}}\Big/\frac{\partial{V(w,I,z)}}{\partial{w}}.
$$
That is, human capital's implicit value is an investor's subjective marginal value of his
future labor income and can become a proxy for the investor's early retirement demand. Once
the investor accumulates wealth, a higher (lower) implicit value implies that he  tends to work
more (less).

Consistent with empirical observation (\cite{CS97} and \cite{CGM05}), 
our model also generates 
the empirically plausible hump-shaped implicit value of human capital (Figure \ref{coint_implicit
value}). More importantly, cointegration results in an earlier peak point in the implicit value
compared to DL. Therefore, wealth at retirement is lower with cointegration
than without it.

\begin{figure}[H]
\centering
\begin{tabular}{cc}
\includegraphics[width=0.46\textwidth]{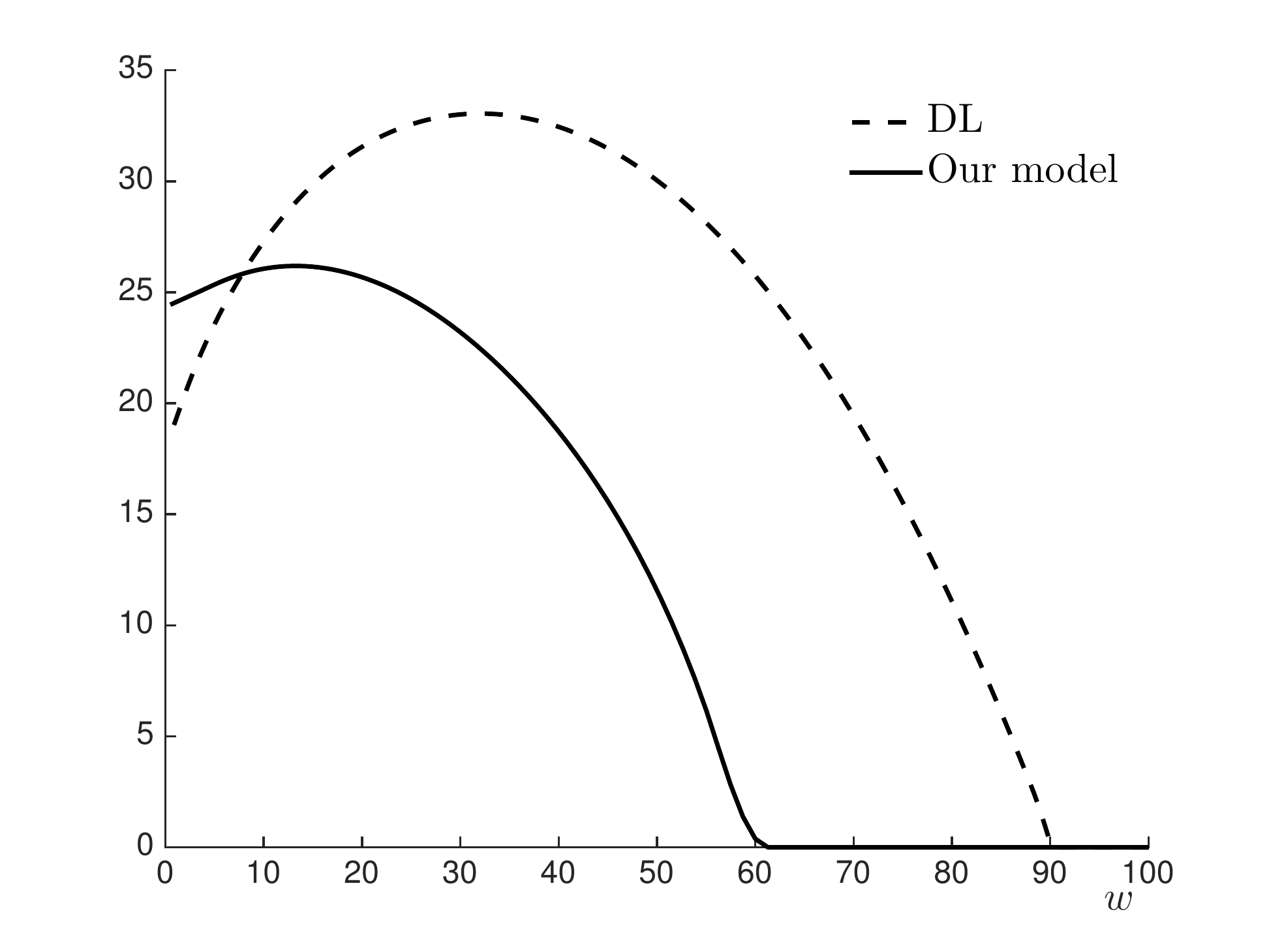} & \hspace{-.5cm}
\includegraphics[width=0.46\textwidth]{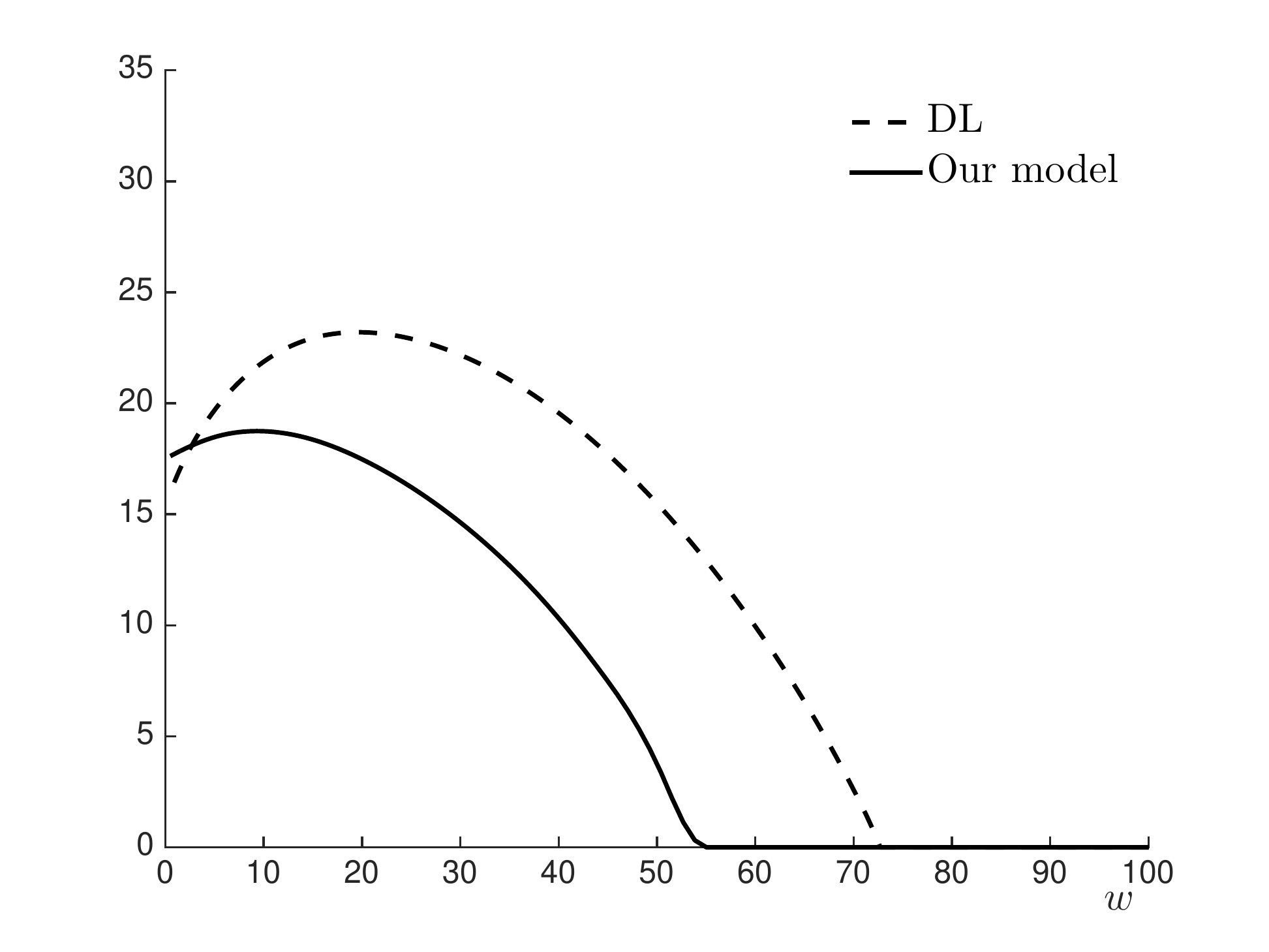} \\
{\small (i) $\delta_{D}=0, z=0$} & {\small (ii) $\delta_{D}=5\%, z=0$}
\end{tabular}
\caption[Implicit value of human capital as a function of wealth-to-income ratio]{\textbf{Implicit value of human capital as a function of wealth-to-income ratio.} The dotted line and the solid line represent DL (\cite{DL10}) result and the our result, respectively. The 
left panel denotes the case without downward jumps in labor income, whereas the right panel stands for the case with downward jumps in labor income.
Basic parameters are chosen from Table \ref{coint_table_summary}. }
\label{coint_implicit value}
\end{figure}

Consistent with \cite{FP07} and \cite{DL10}, our life-cycle model shows that there exists a certain wealth
threshold for voluntary retirement over which an investor is optimal to retire permanently.
However, the existing retirement literature including  \cite{FP07} and \cite{DL10} just concentrates on generating
the threshold instead of finding out the economic justification for optimal retirement. We
find that in the presence of cointegration, the wealth threshold for voluntary retirement
becomes smaller when labor income will be decreased than when it will be increased in the
long run (Figure \ref{coint_early retirement}). This is evidence that early retirement becomes
numerically plausible, especially when wages are expected to decline in the long term.

\begin{figure}[H]
\centering
\begin{tabular}{cc}
\includegraphics[width=0.46\textwidth]{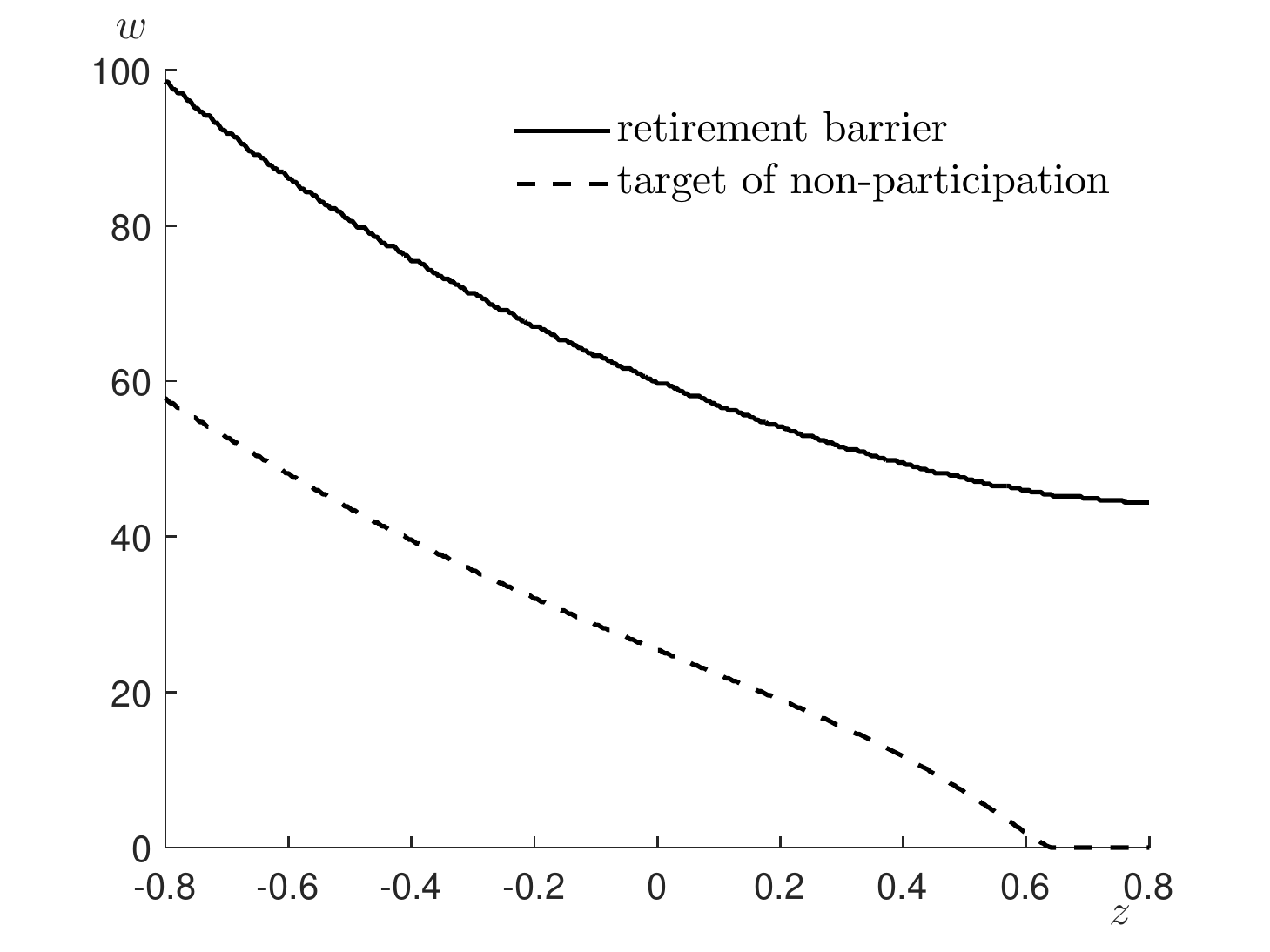} & \hspace{-.5cm}
\includegraphics[width=0.46\textwidth]{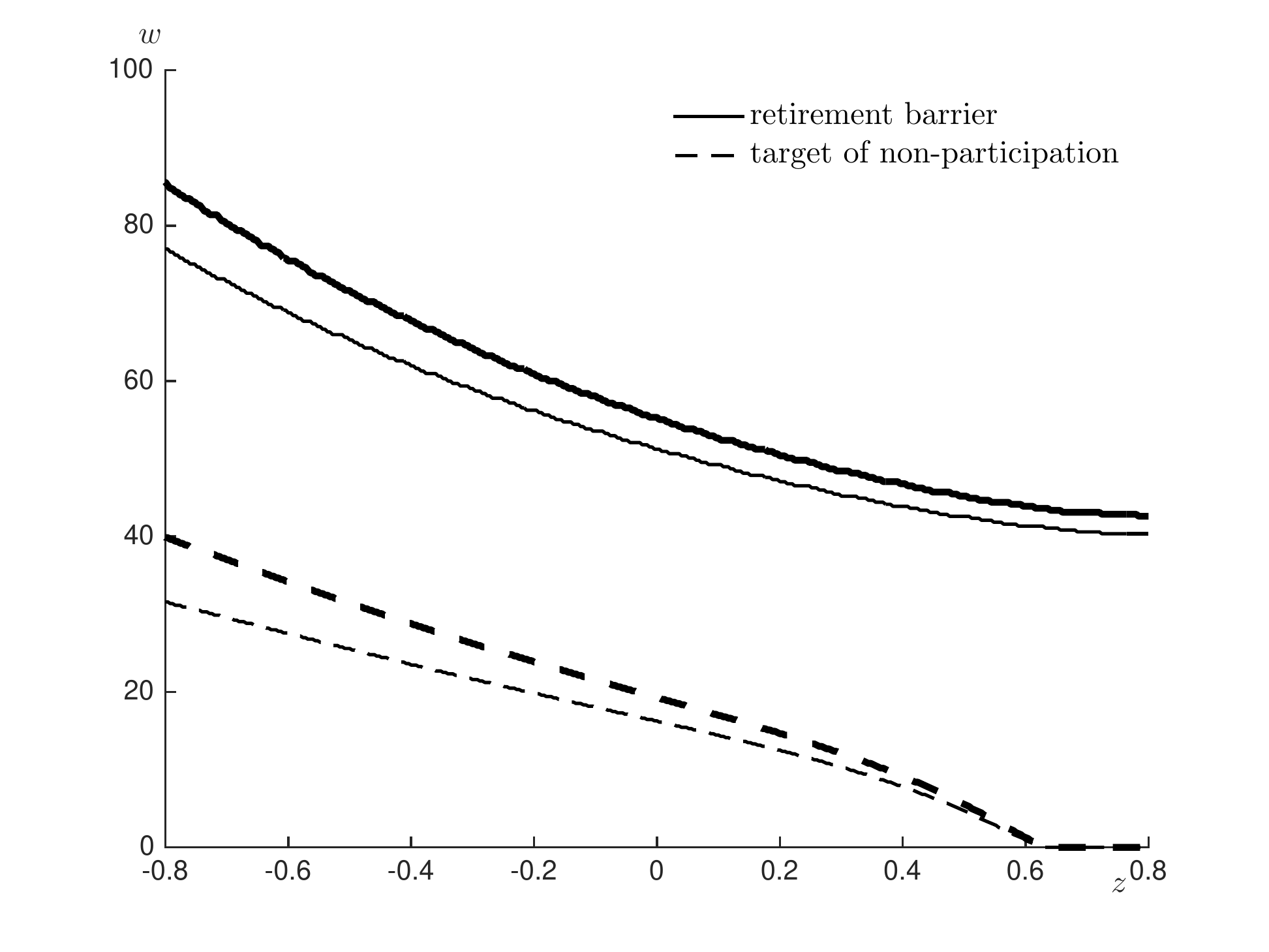} \\
{\small (i) $\delta_{D}=0$} & {\small (ii) $z=0$, $\delta_{D}=5\%$ for thick line; $\delta_{D}=3\%$ for thin line}
\end{tabular}
\caption[Sensitivity analysis of wealth threshold for voluntary retirement and target wealth-to-income ratio with respect to $z$]{\textbf{Sensitivity analysis of wealth threshold for voluntary retirement and target wealth-to-income ratio with respect to $z$.} The solid line and the dotted line represent the wealth threshold for retirement and the target
wealth-to-income ratio, respectively. The left panel denotes the case without downward jumps in labor income,
whereas the right panel stands for the case with those jumps in labor income. 
Basic parameters are chosen from Table \ref{coint_table_summary}. }
\label{coint_early retirement}
\end{figure}

A majority of analysis for early retirement runs much deeper than existing life-cycle literature. We analyze the retirement through comparative statics by changes in a range of fundamental parameters in the financial market.

\noindent{\textbf{Changes in investment opportunity and risk aversion}}

Better investment opportunity lowers the wealth threshold for voluntary retirement (Figure 
\ref{coint_early retirement_investment_opportunity}), which can be related to empirical observation 
that early retirement is quite possible during up markets.

\begin{figure}[H]
\centering
\begin{tabular}{cc}
\includegraphics[width=0.46\textwidth]{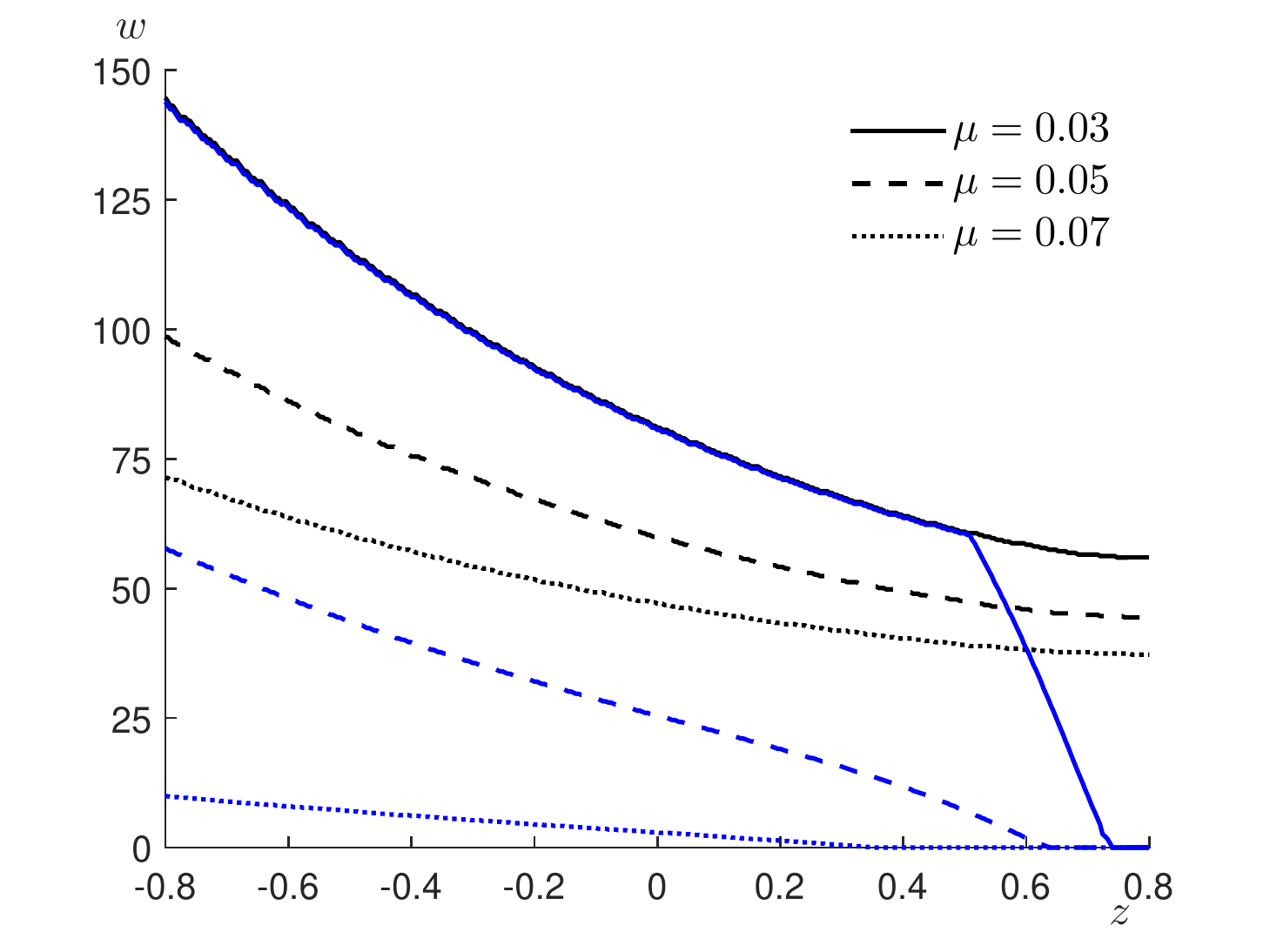} & \hspace{-.5cm}
\includegraphics[width=0.46\textwidth]{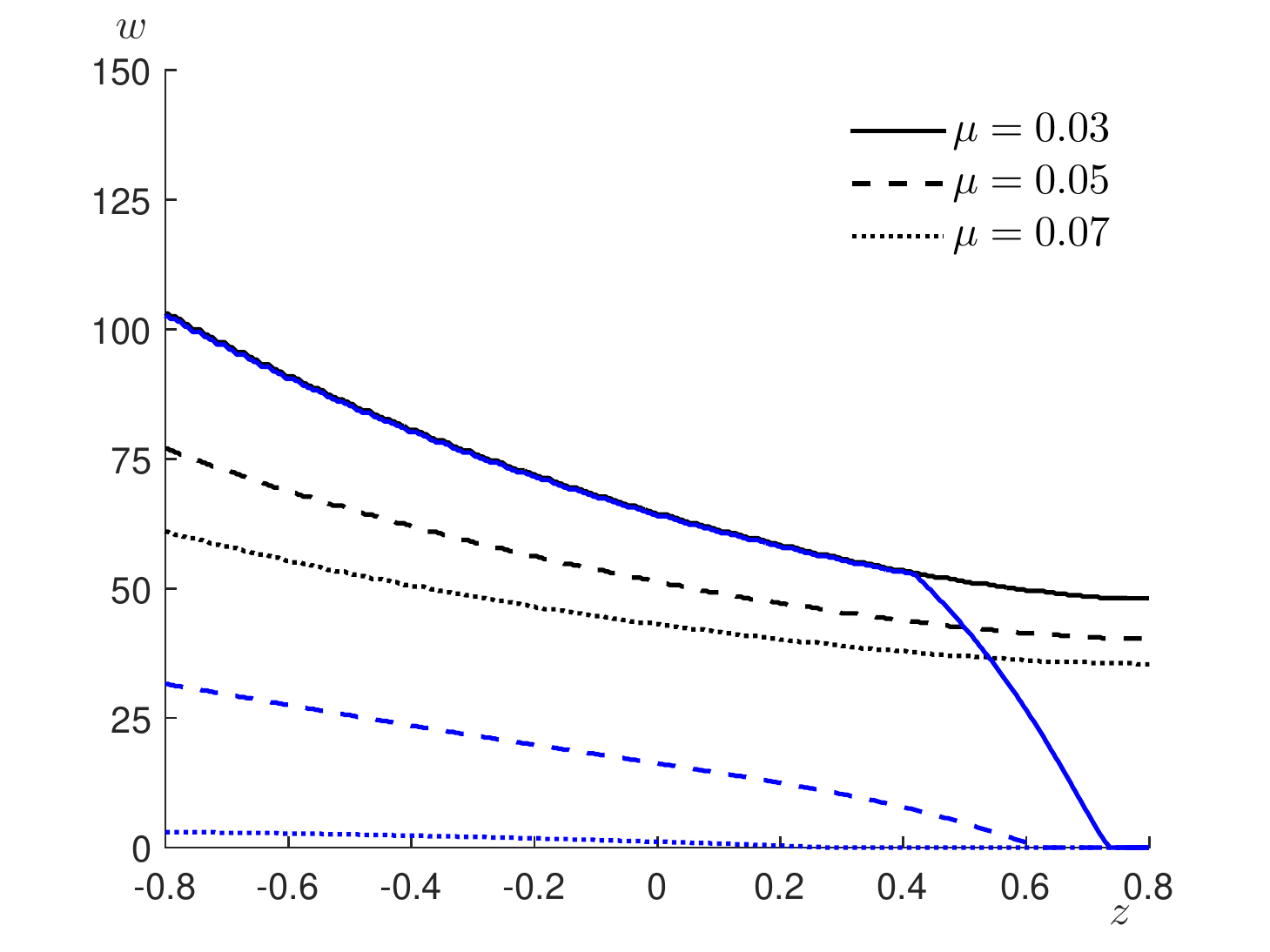} \\
{\small (i) $\delta_{D}=0$} & {\small (ii) $z=0$, $\delta_{D}=5\%$} \\
\includegraphics[width=0.46\textwidth]{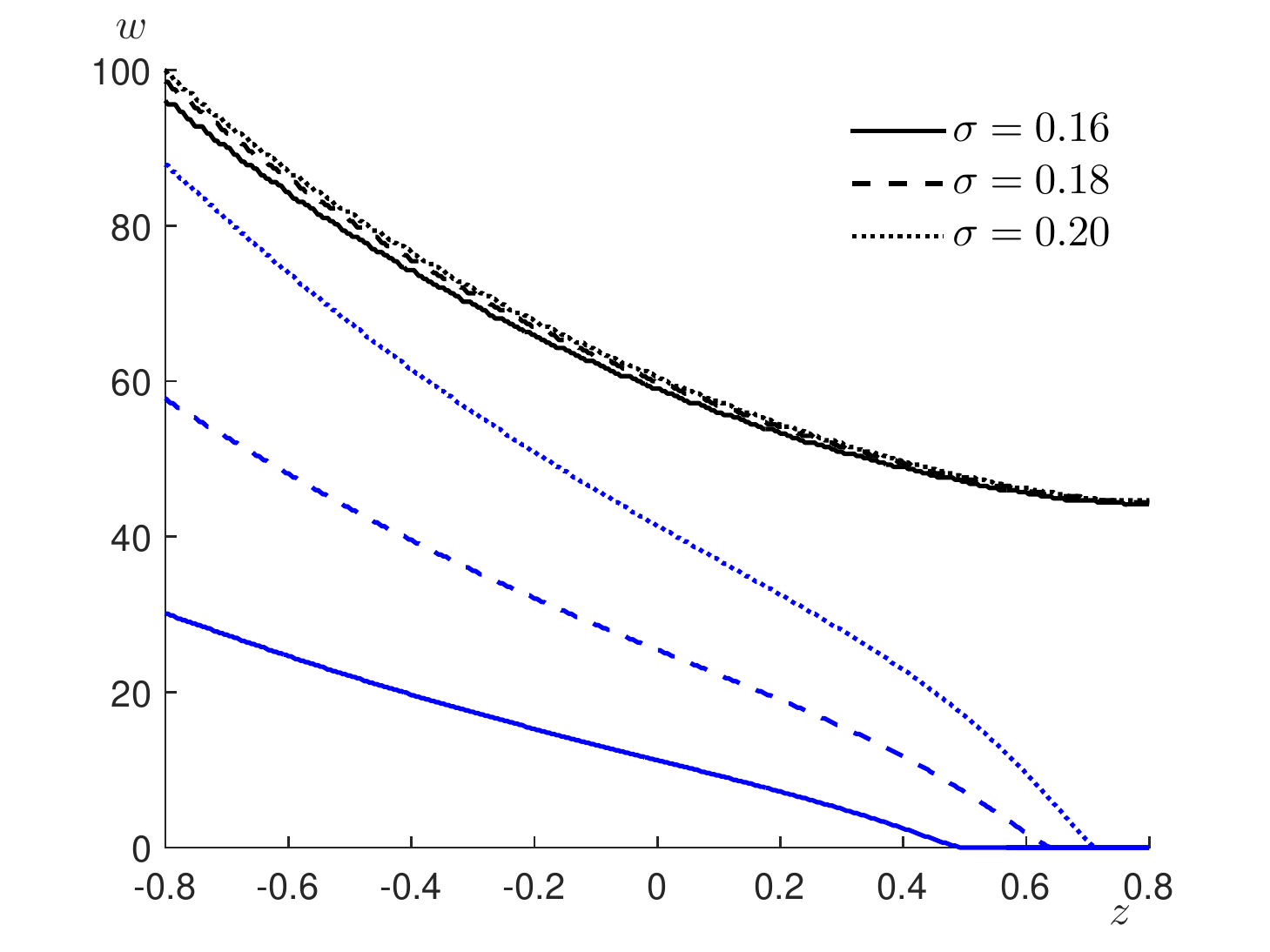} & \hspace{-.5cm}
\includegraphics[width=0.46\textwidth]{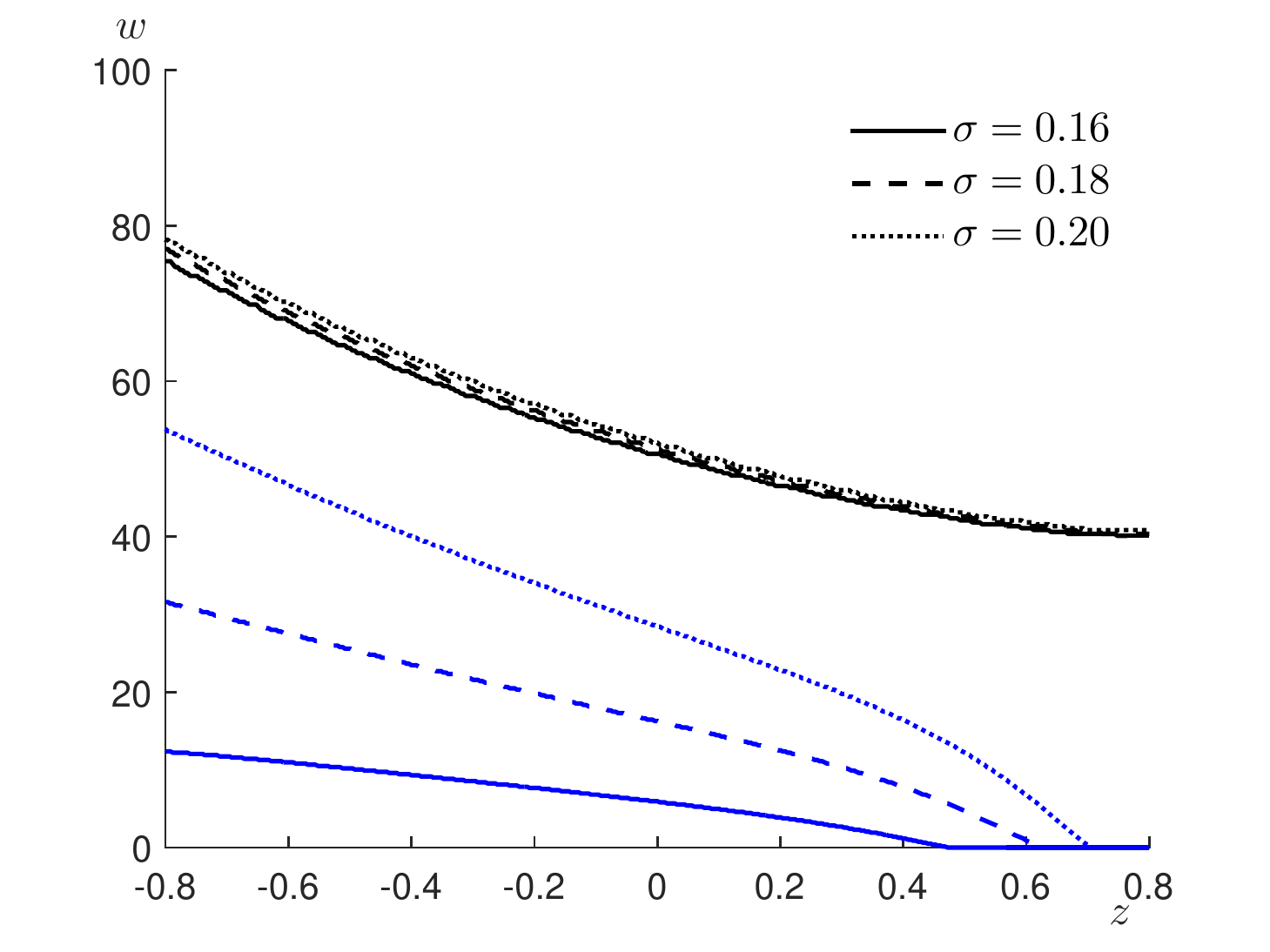} \\
{\small (iii) $\delta_{D}=0$} & {\small (iv) $\delta_{D}=5\%$}
\end{tabular}
\caption[Sensitivity analysis of wealth threshold for voluntary retirement with respect to opportunity set]{\textbf{Sensitivity analysis of wealth threshold for voluntary retirement with respect to opportunity set.} The black lines (upper three lines) and the blue lines (lower three lines) represent the wealth threshold for retirement and the target
wealth-to-income ratio, respectively. The left panels denote the case without downward jumps in labor 
income, whereas the right panels stand for the case with downward jumps in labor income. 
Basic parameters are chosen from Table \ref{coint_table_summary}. }
\label{coint_early retirement_investment_opportunity}
\end{figure}

Interestingly, higher risk aversion decreases the wealth threshold for voluntary retirement (Figure \ref{coint_early retirement_gamma}),
which seems to be conflicting as opposed to  \cite{FP07},  \cite{DL10}. As soon as the standard real option analysis is
applied,\footnote{\cite{HM07} argue that investment time of entrepreneurial
activities can be delayed as an agent becomes more prudent with a larger coefficient of risk aversion.} 
risk aversion is found to further delay optimal retirement. Existing studies such as \cite{FP07} and \cite{DL10} predict 
that risk aversion appears to raise up the wealth threshold for voluntary retirement when labor income risks are fully diversified.  Intuitively, the option to work is more valuable to investors who are more risk averse and they are apparently  trying to avoid the risk of losing such option value. 
However, our model predicts that this result can be reversed when income risks are not spanned by the
market. A natural intuition is that the risk averse investors require an additional
premium for holding the unspanned income risk. Working shorter becomes more
attractive, commanding a lower premium for the undiversifiable income risks. Further, higher risk aversion decreases the amount of future human capital in the presence of uninsurable income risks. Thus, the option for retiring earlier becomes progressively more attractive for more risk averse investors.

\begin{figure}[H]
\centering
\begin{tabular}{cc}
\includegraphics[width=0.46\textwidth]{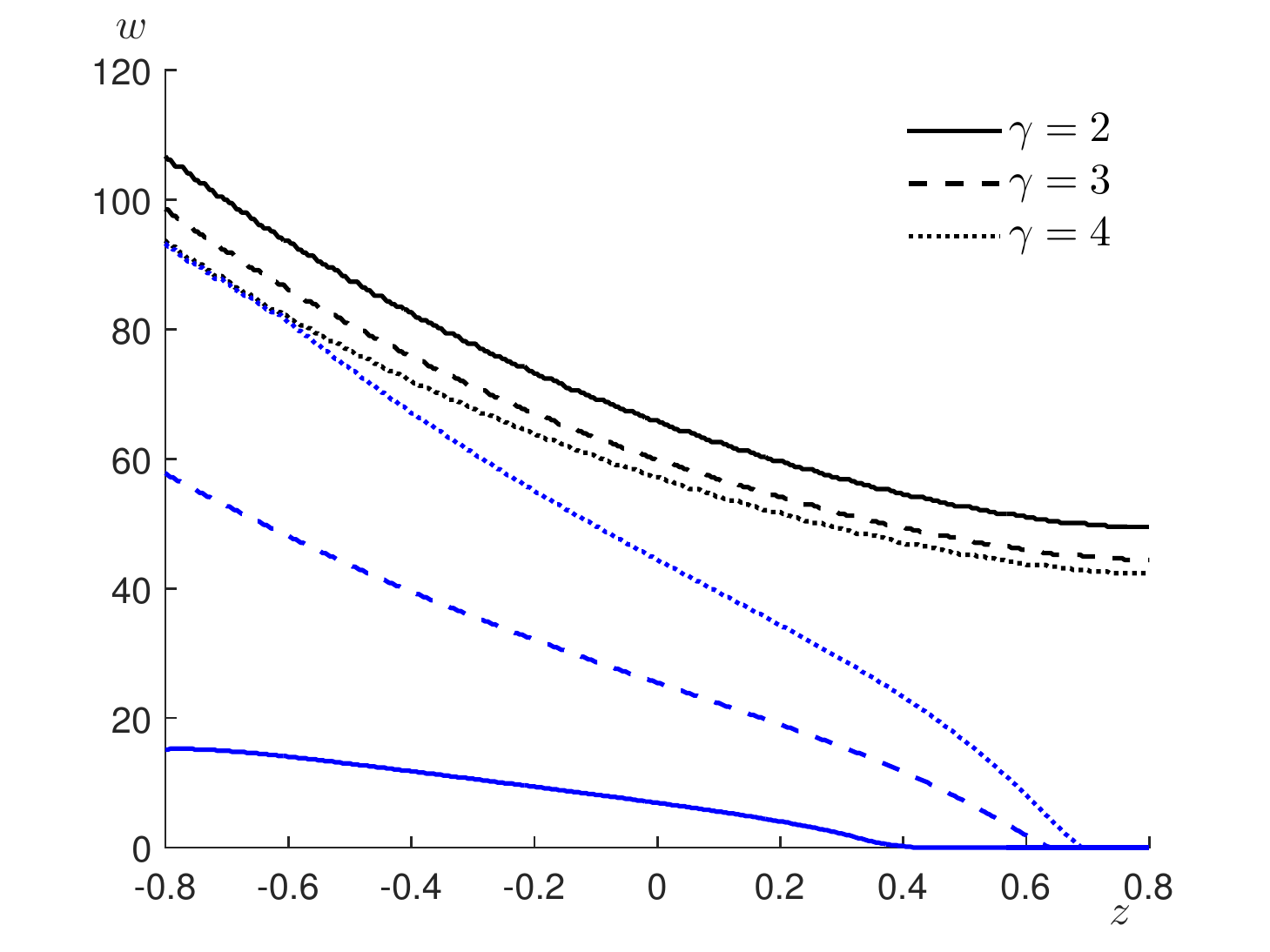} & \hspace{-.5cm}
\includegraphics[width=0.46\textwidth]{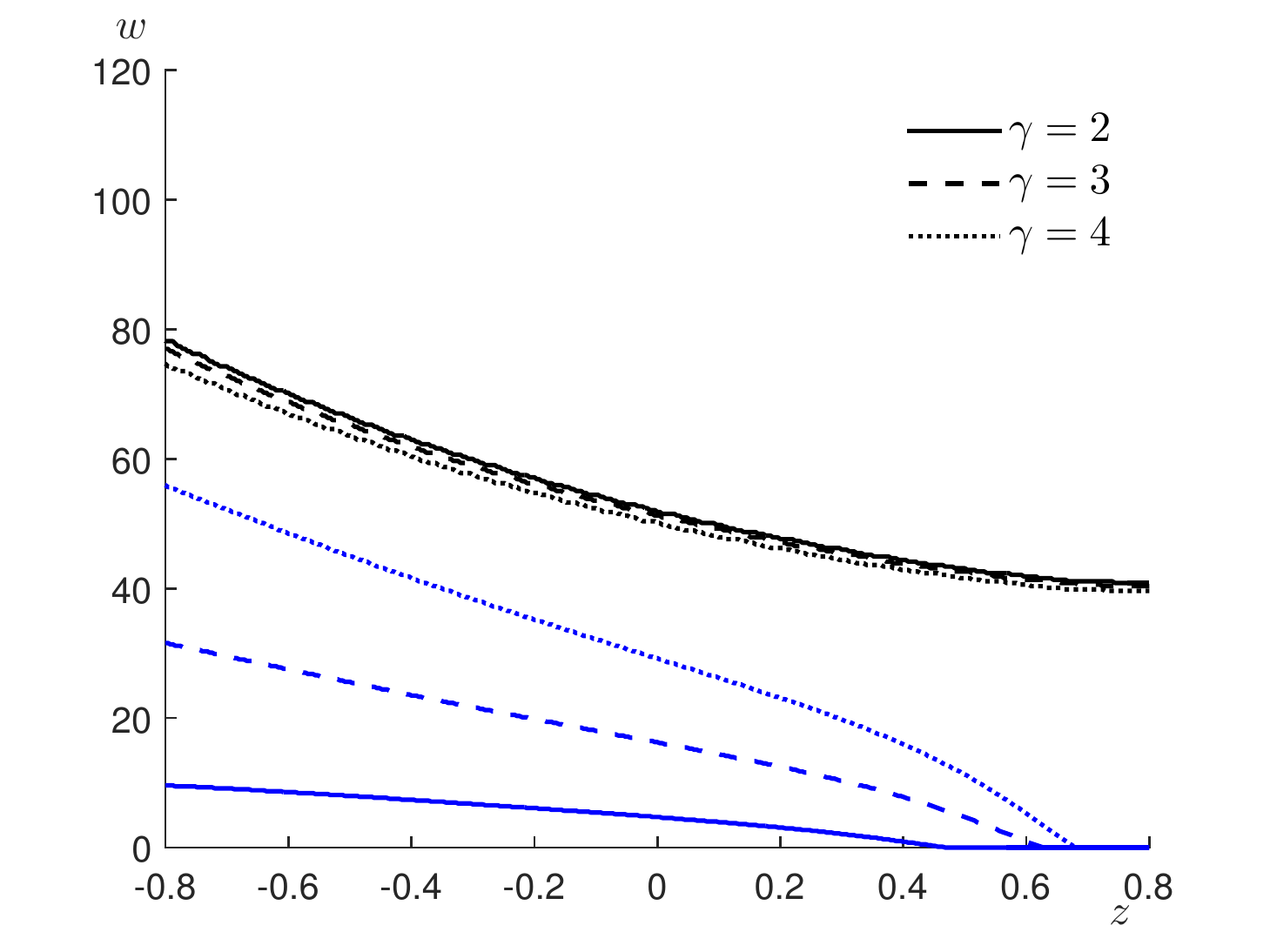} \\
{\small (i) $\delta_{D}=0$} & {\small (ii) $\delta_{D}=5\%$}
\end{tabular}
\caption[Sensitivity analysis of wealth threshold for voluntary retirement with respect to risk aversion $\gamma$]{\textbf{Sensitivity analysis of wealth threshold for voluntary retirement and target wealth-to-income ratio with respect to $z$.} The black lines (upper three lines) and the blue lines (lower three lines) represent the wealth threshold for voluntary retirement
and the target wealth-to-income ratio, respectively. The left panel denotes the case without downward 
jumps in labor income, whereas the right panel stands for the case with downward jumps in labor income. 
Basic parameters are chosen from Table \ref{coint_table_summary}. }
\label{coint_early retirement_gamma}
\end{figure}

\noindent{\textbf{Changes in volatility on income growth}}

In the interest of effects of labor income riskiness on early retirement, a riskier stream of future 
labor income (or higher $\sigma_{I}$) is associated with lower wealth threshold for voluntary 
retirement (Figure \ref{coint_early retirement_income_volatility}). 
As we mentioned in the analysis of portfolio share, an investor shows strong additional hedging demand against uninsurable risks related to labor income. Intuitively, investing more in the
stock market is followed by working shorter when the labor market is down, reducing the labor income risk. 


\begin{figure}[H]
\centering
\begin{tabular}{cc}
\includegraphics[width=0.46\textwidth]{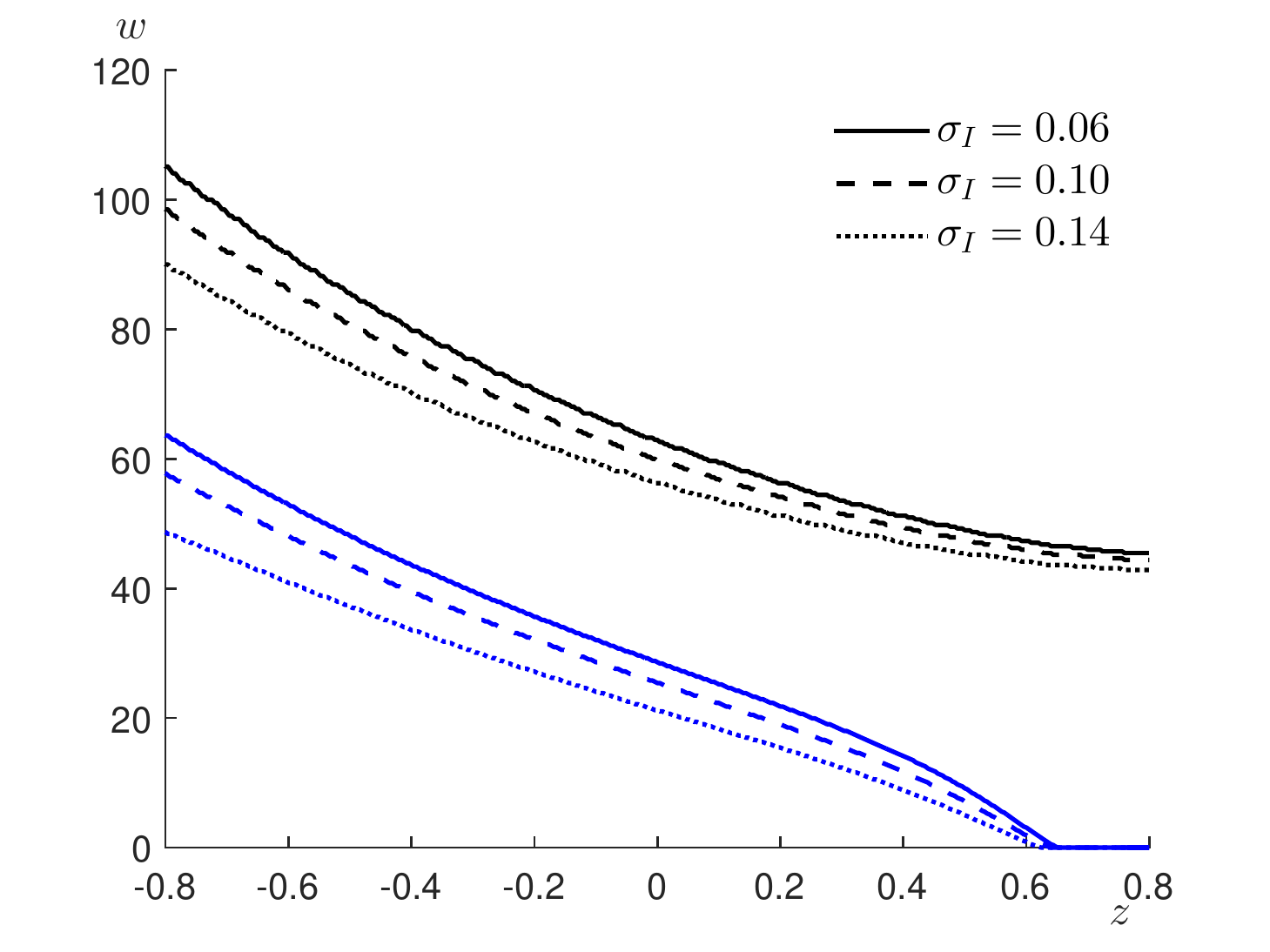} & \hspace{-.5cm}
\includegraphics[width=0.46\textwidth]{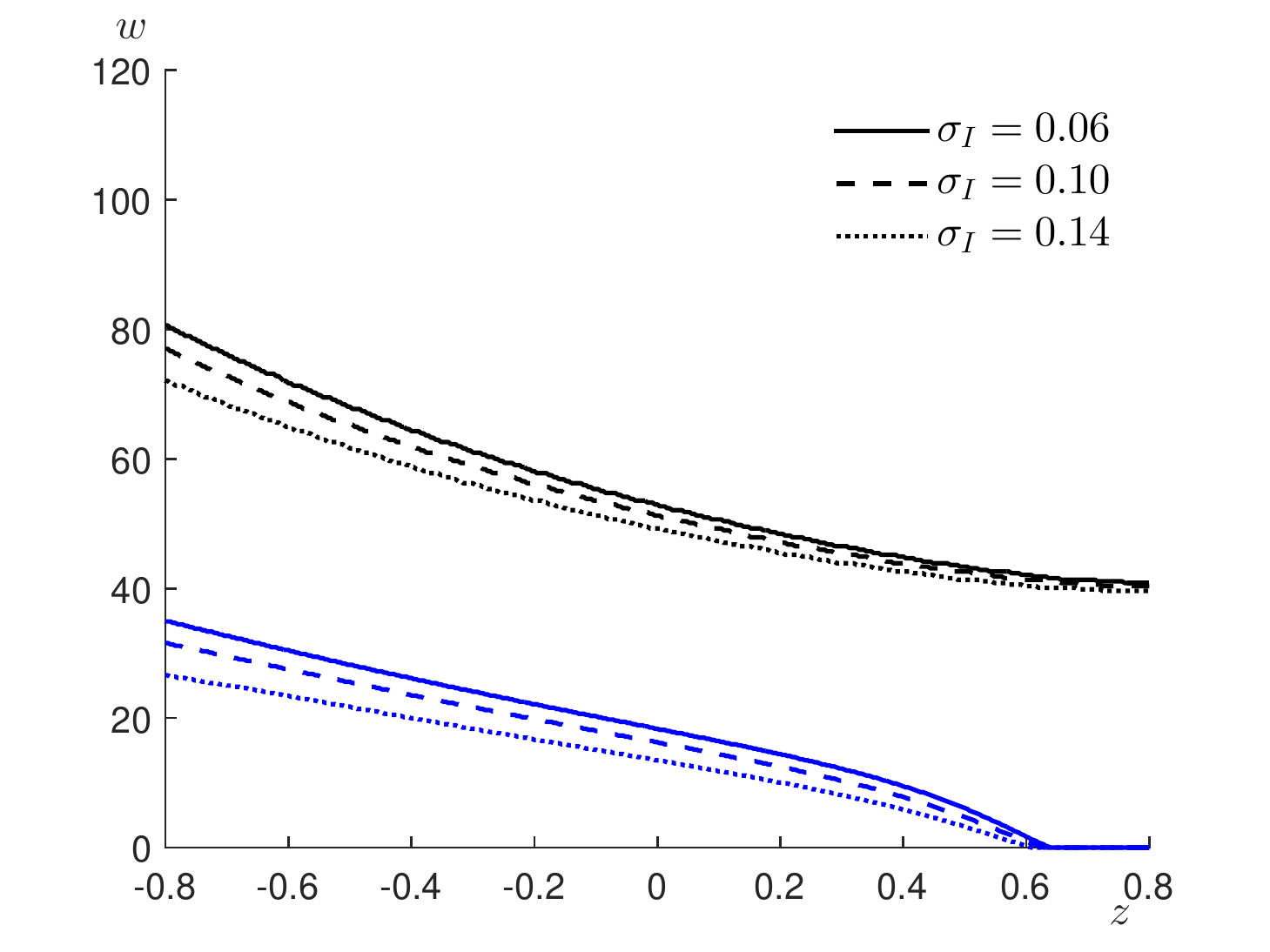} \\
{\small (i) $\delta_{D}=0$} & {\small (ii) $\delta_{D}=5\%$}
\end{tabular}
\caption[Sensitivity analysis of wealth threshold for voluntary retirement with respect to income growth volatility $\sigma_{I}$]{\textbf{Sensitivity analysis of wealth threshold for voluntary retirement with respect to income growth volatility $\sigma_{I}$.} 
The black lines (upper three lines) and the blue lines (lower three lines) represent  the wealth threshold for voluntary retirement and the target wealth-to-income ratio, respectively. The left panel denotes the  case without downward jumps in labor income, whereas the right panel stands for the case with downward jumps in labor income.
Basic parameters are chosen from Table \ref{coint_table_summary}.}
\label{coint_early retirement_income_volatility}
\end{figure}

\noindent\textbf{Changes in degree of mean reversion}

Let's uncover more details of the effects of cointegration on voluntary retirement by changing the
degree of mean reverted process $z_t$ (Figure \ref{coint_early 
retirement_alpha}). An important prediction of the model is that cointegration has two opposing
effects on stock market participation and early retirement, as opposed to the effects of changing 
in fundamental parameter values such as investment opportunity set and income growth volatility. 
It is a done deal that once labor income is more transitory (or higher $\alpha$), optimal portfolio 
takes the form of more riskless assets to reserve enough wealth to finance future consumption, 
delaying an investor's stock market participation. 

Our earlier discussions based on changes in 
investment opportunity or income growth volatility imply that wealth threshold for voluntary 
retirement will be expected to rise when an investor delays stock market participation (Figure
\ref{coint_early retirement_investment_opportunity}, \ref{coint_early retirement_income_volatility}). In line 
with this, as far as more transitory labor income is concerned, higher levels of wealth threshold 
seem to be expected. However, contrary to this projection, more transitory labor income is likely
to decline wealth threshold for voluntary retirement, providing an explanation of early retirement. 
Therefore, the cointegration between the stock and labor markets, i.e., the long-run dependence between two 
markets has the first-order effect on early retirement policy.

\begin{figure}[H]
\centering
\begin{tabular}{cc}
\includegraphics[width=0.46\textwidth]{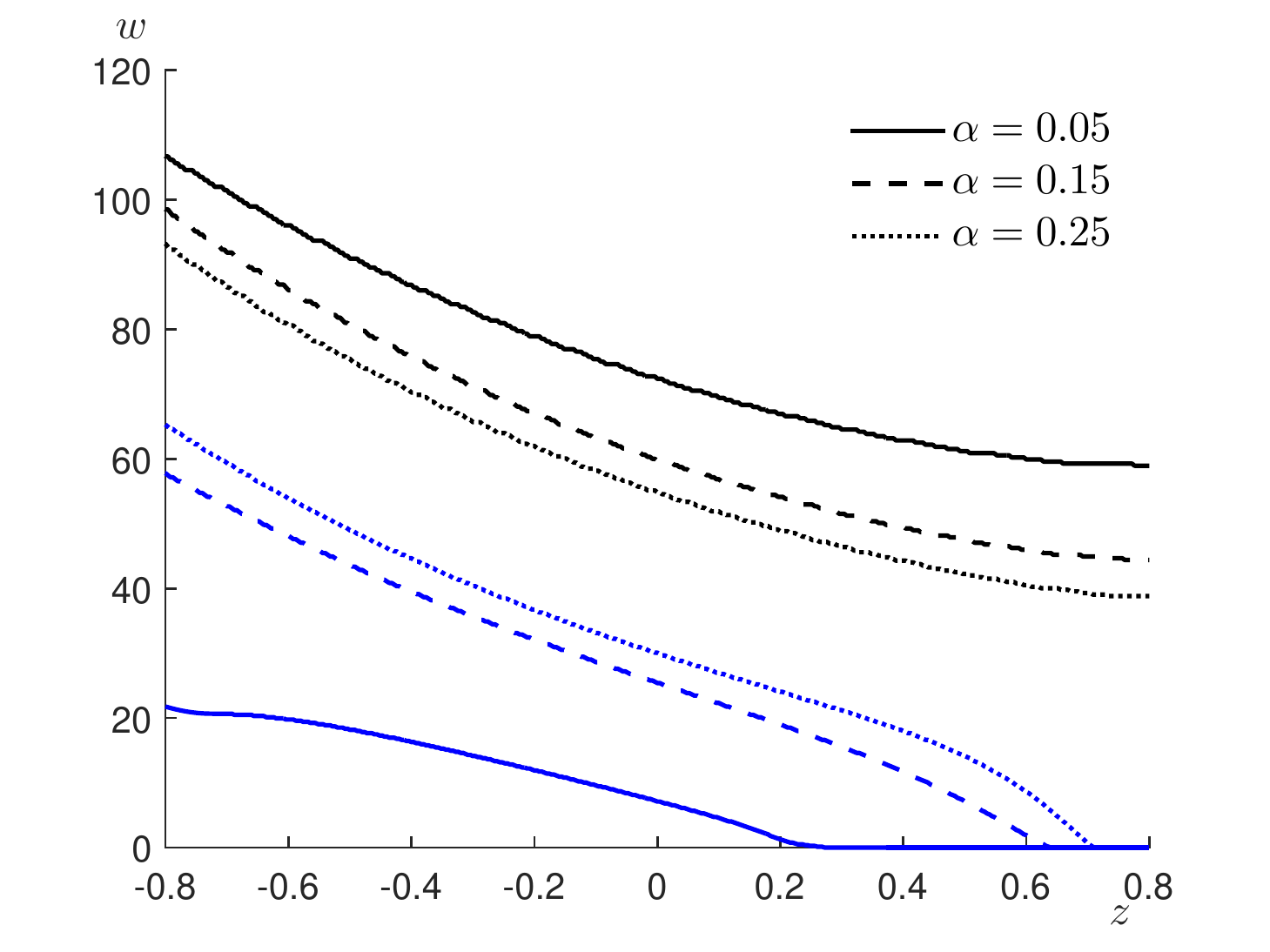} & \hspace{-.5cm}
\includegraphics[width=0.46\textwidth]{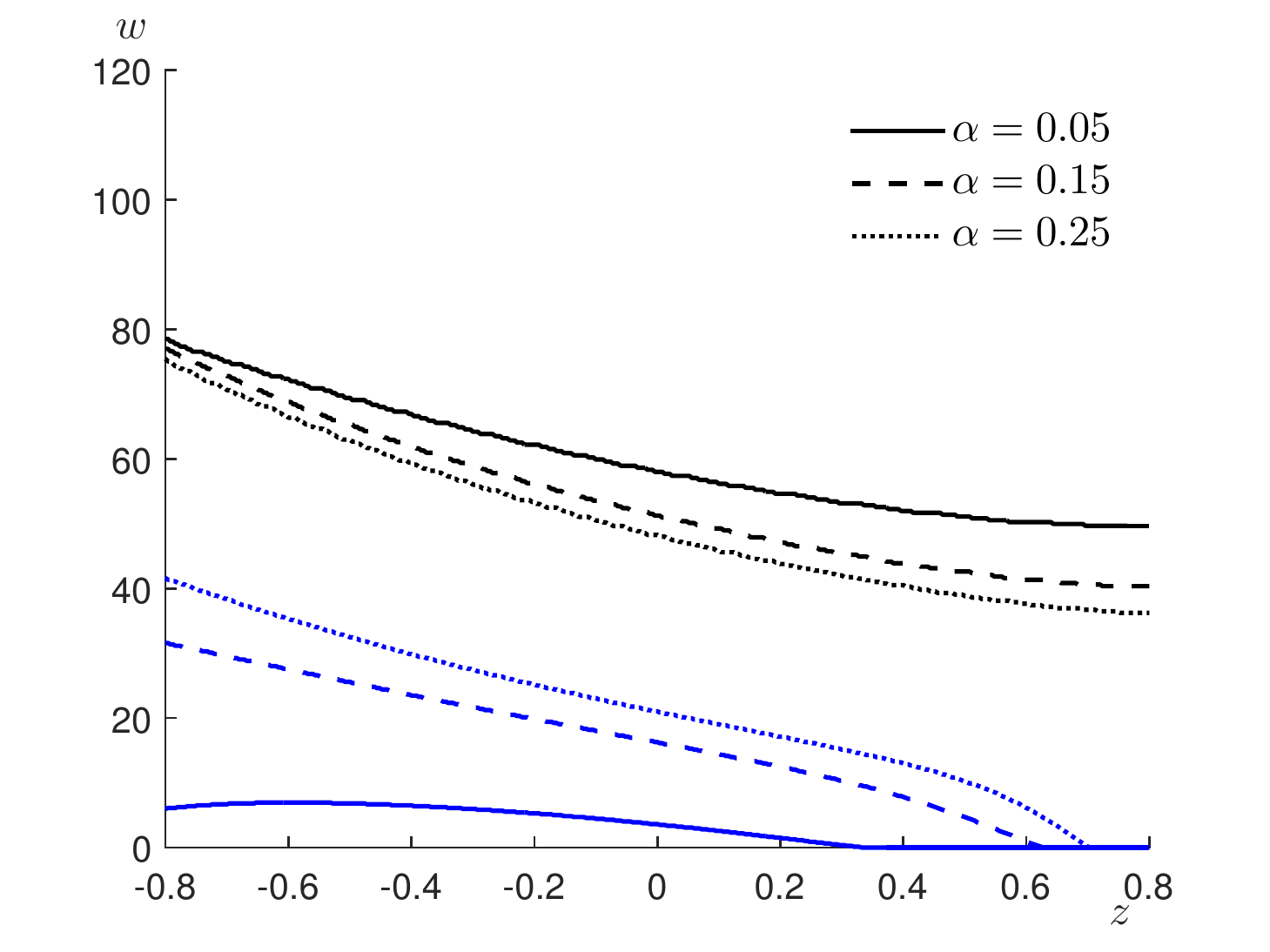} \\
{\small (i) $\delta_{D}=0$} & {\small (ii) $\delta_{D}=5\%$}
\end{tabular}
\caption[Sensitivity analysis of wealth threshold for voluntary retirement with respect to mean reversion $\alpha$]{\textbf{Sensitivity analysis of wealth threshold for voluntary retirement with respect to mean reversion $\alpha$.} 
The black lines (upper three lines) and the blue lines (lower three lines) represent  the wealth threshold for voluntary retirement and the target wealth-to-income ratio, respectively. The left panel denotes the case without downward jumps in labor income, whereas the right panel stands for the case with downward  jumps in labor income.
Basic parameters are chosen from Table \ref{coint_table_summary}.
}
\label{coint_early retirement_alpha}
\end{figure}

To further compare the difference between the expected time to retirement among the optimal policy in our model and DL's policy when the stock and labor markets are cointegrated, we perform a Monte-Carlo simulation. That is, the income process follows (\ref{lip2}) with $z_t$ given in (\ref{labor_zprocess}). 
We are interested in the expected time to retirement and the expected portfolio share invested in stocks before retirement. Table \ref{table coint simulation} shows the results obtained from 10,000 simulated paths under different initial wealth-to-income ratio and cointegration factor $\alpha$. In the presence of cointegration, DL's policy is suboptimal compared to our model as DL's policy totally ignore the cointegration effect between stock and income markets and thus performs naively. Therefore, investors who follow DL's policy have longer expected time to retire. Consistent with Fig \ref{coint_portfolio share}, investors who follow the optimal policy in our model will have much lower expected portfolio share before retirement. Moreover, the expected time to retirement under DL's policy is monotone increasing in the strength of cointegration (factor $\alpha$), while the expected time to retirement under our model's optimal policy does not show monotonicity in $\alpha$. Indeed, for those investors with low initial wealth-to-income ratio, stronger cointegration increase the expected time to retirement because of low future human capital to accumulate wealth. For those investors with high initial wealth-to-income ratio, stronger cointegration shrinks the distance from initial wealth-to-income ratio to the optimal threshold of retirement and thus induces shorter expected time to retirement. Therefore, the cointegration between stock and income markets significantly reduces the allocation in stocks and accelerates earlier retirement for richer (or older) investors. 

\begin{table}[H]
\caption[Model simulation]{\textbf{Model simulation.} The expected time to retirement (in year) and expected portfolio share in stocks before retirement under different policies are simulated when the stock and labor markets are cointegrated. Basic parameters are chosen from Table \ref{coint_table_summary}.}
\label{table coint simulation}
\begin{center}
\resizebox{\linewidth}{!}{%
\begin{tabular}{l c c  c c c}
\hline  
\hline
&& \multicolumn{2}{ c }{Under optimal policy in our model } & \multicolumn{2}{ c }{ Under DL's policy  }\\
&& Expected time to retire  & Expected portfolio share & Expected time to retire  & Expected portfolio share  \\
\hline
\multirow{3}{*}{$\alpha=0.05$} &$w/I=10$ &$106$ &$0.14$ &$128$ &$0.89$ \\  
& $w/I=30$ &$55$  &$0.20$ & $74$  &$0.78$\\
& $w/I=50$ & $23$ &$0.27$ & $41$ &$0.71$ \\
\hline
\multirow{3}{*}{$\alpha=0.15$} &$w/I=10$ &$129$ &$0.04$ &$180$ &$0.85$ \\  
& $w/I=30$ &$58$  &$0.07$ & $105$  &$0.78$\\
& $w/I=50$ & $15$ &$0.19$ & $57$ &$0.71$ \\
\hline
\multirow{3}{*}{$\alpha=0.25$} &$w/I=10$ &$141$ &$0.03$ &$230$ &$0.75$ \\  
& $w/I=30$ &$57$  &$0.07$ & $133$  &$0.76$\\
& $w/I=50$ & $7$ &$0.18$ & $74$ &$0.71$ \\
\hline\hline
\end{tabular}}
\end{center}
\end{table}

\subsection{Mandatory Retirement Age}
The expected time to retirement in Table \ref{table coint simulation} seems too long for individuals with low initial wealth-to-income ratio if we consider human's finite working and living age. The  mandatory retirement age is an age-dependent factor and  makes the model time-dependent if we include mandatory retirement age. We claim that our main results still hold in the presence of mandatory retirement age. For example, to mainly focus on the effect the mandatory retirement age, we rewrite our value function as
\begin{equation}\label{valfmandretire}
V(w,I,z,t) \equiv\sup_{(c,y,\tau)\in\mathcal{A}(w,I,z)}\mathbb{E}\Big[\int^{\tau\wedge T}_{t}e^{- \beta (u-t)}
\df{c^{1-\gamma}_{u}}{1-\gamma} du+e^{- \beta(\tau\wedge T-t)}G(W_{\tau\wedge T})\Big],
\end{equation} 
where $T$ is the mandatory retirement deadline and $\tau\wedge T$ is the shorthand notation for $\min(\tau, T)$. 
Now all the optimal investment and retirement strategies depend on time $t$.  At time $t$, an investor has $T-t$ years to mandatory retirement.
We assume that throughout an investor begins to work at age 20 ($t=0$) and the mandatory retirement age is 70 ($T=50$). The quantitive results are shown in Fig \ref{coint_FiniteTime}.
Panel (i) shows that the wealth threshold for voluntary retirement declines as an investor nears mandatory retirement. Intuitively, as an investor ages, the incentive to keep the option of working ``alive'' is reduced and hence the wealth threshold for voluntary retirement declines. Moreover, at the same age, the optimal wealth threshold for voluntary retirement in the presence of cointegration (the solid line) is lower than the threshold in the absence of cointegration (the dashed line), which implies the early retirement option.  
Panel (ii) shows the portfolio share invested in stocks at $z=0$. It can be seen that the portfolio share increases not only in financial wealth $w$ but also in the physical age, so that the younger investors are not as aggressive as expected and they should hold less or no stock holdings than the older investors. This indicates a plausible hump-shaped profile for stock holdings over life-cycle, which is also shown in other papers such as  \cite{DL10}, \cite{P07}, \cite{WY10}, and \cite{CS14}. Panel (iii) and (iv) depict the target of non-participation and wealth threshold for voluntary retirement as functions of age and cointegration factor $z$. We find that non-participation is more obvious among younger investors and early retirement is more plausible when income is expected to decline in the long term ($z>0$).  By panel (ii) and panel (iii), investors near mandatory retirement always have positive portfolio share in stocks. That is because near the mandatory retirement age,  cointegration does not have sufficient time to act and the option of working is less attractive so that effect of income features on optimal consumption and investment declines. 

\begin{figure}[H]
\centering
\begin{tabular}{cc}
\includegraphics[width=0.46\textwidth]{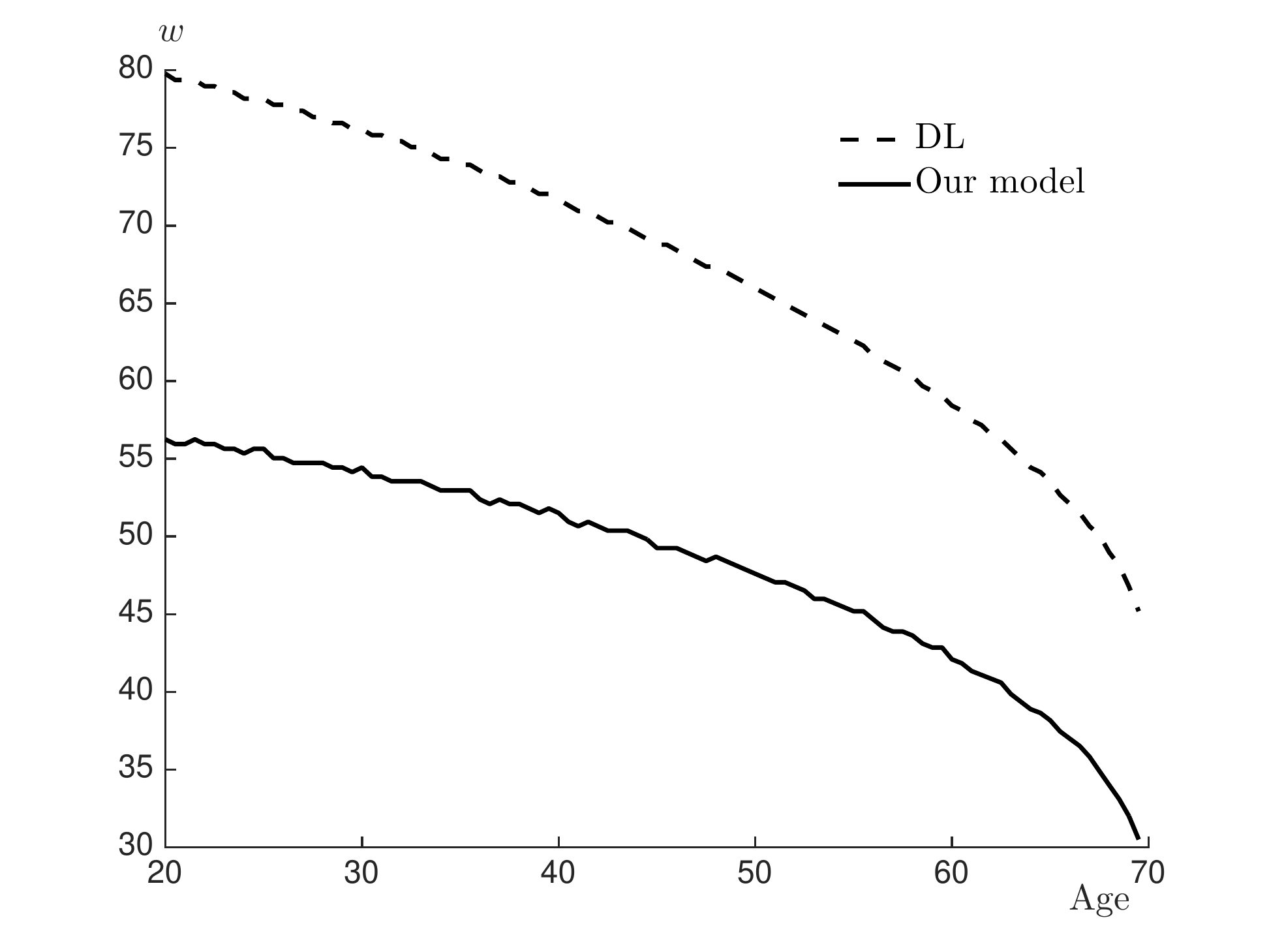} & \hspace{-.5cm}
\includegraphics[width=0.46\textwidth]{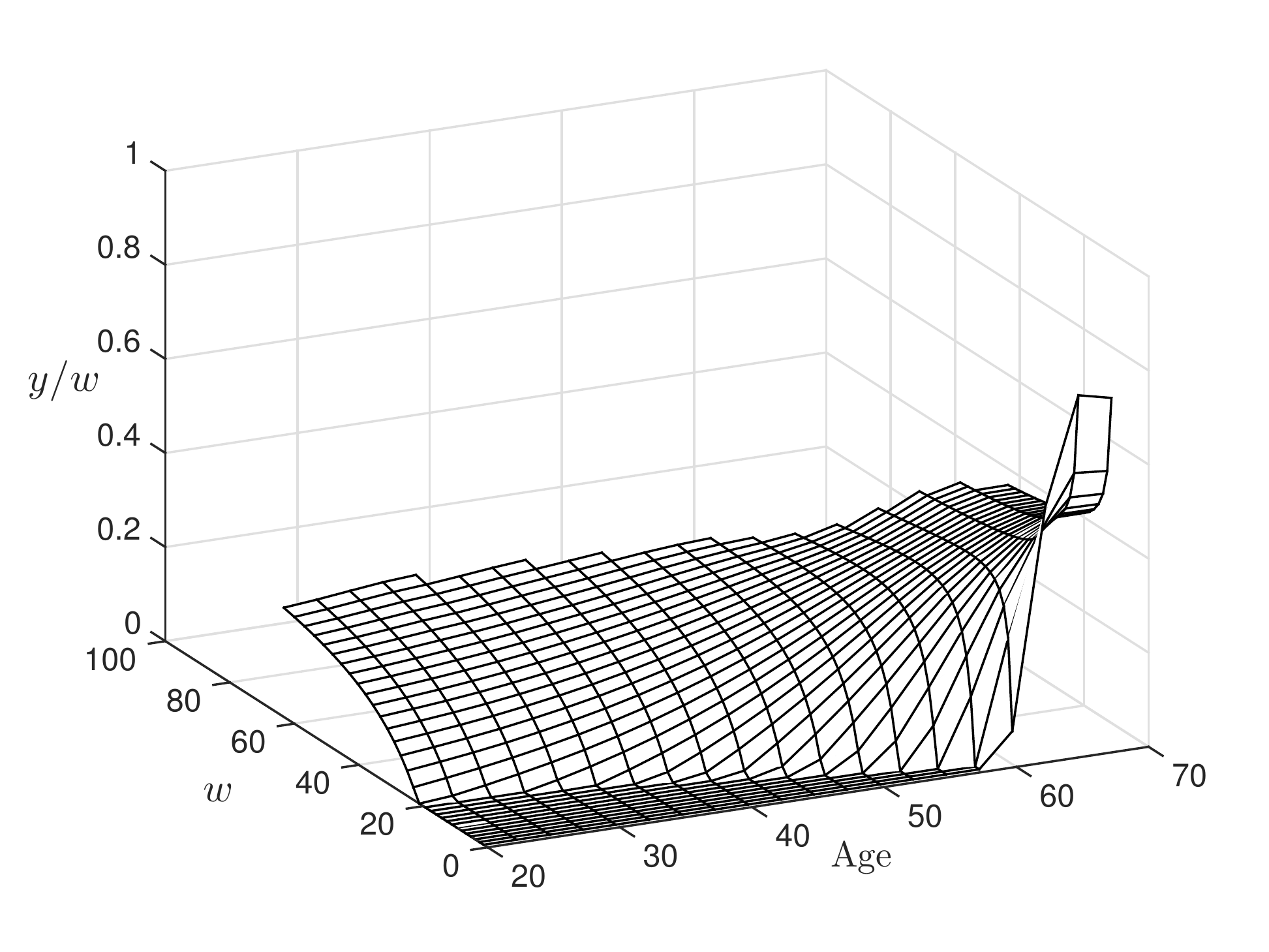} \\
{\small (i) Wealth threshold for voluntary retirement at $z=0$} & {\small (ii) Portfolio share at $z=0$} \\
\includegraphics[width=0.46\textwidth]{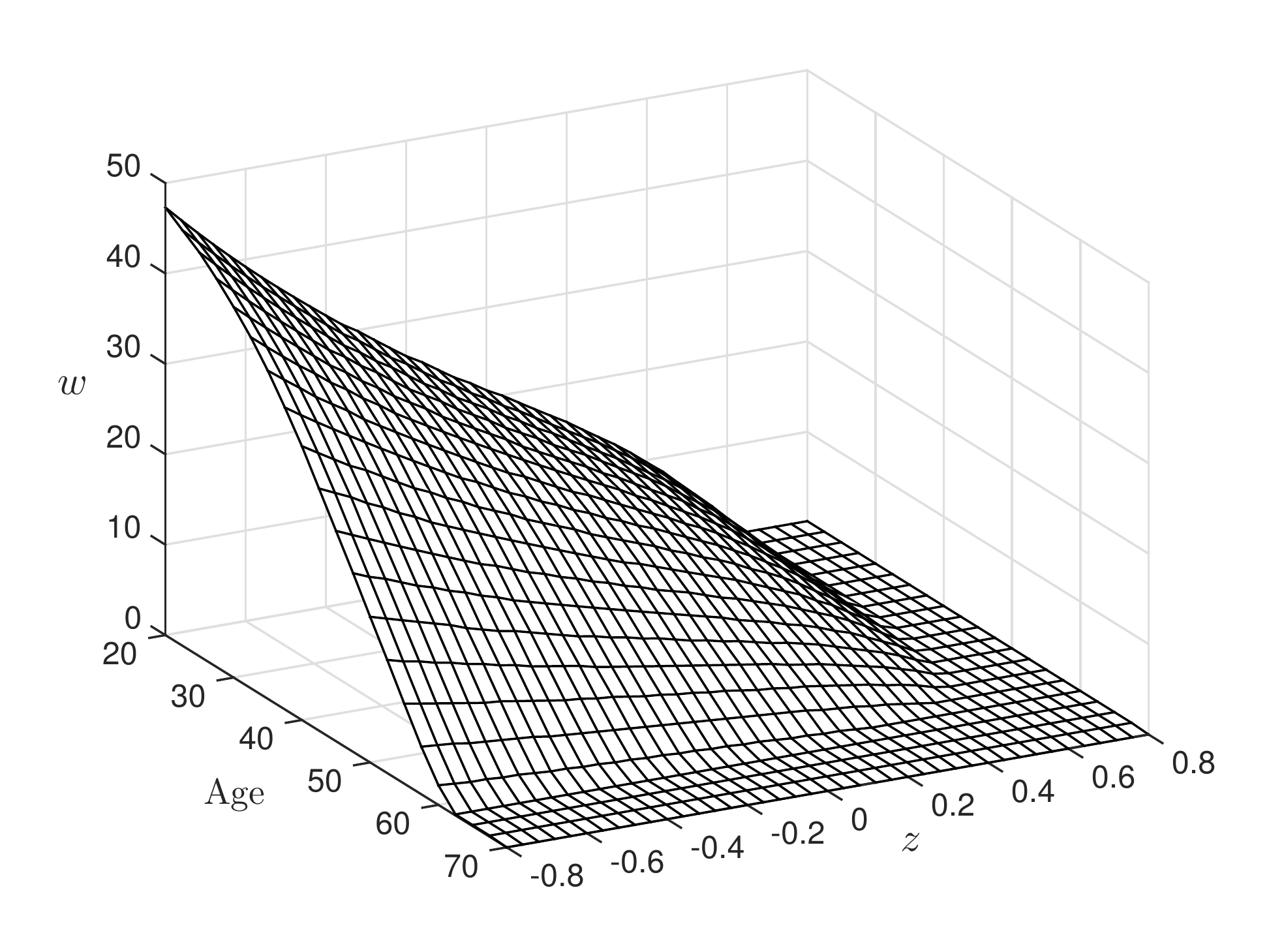} & \hspace{-.5cm}
\includegraphics[width=0.46\textwidth]{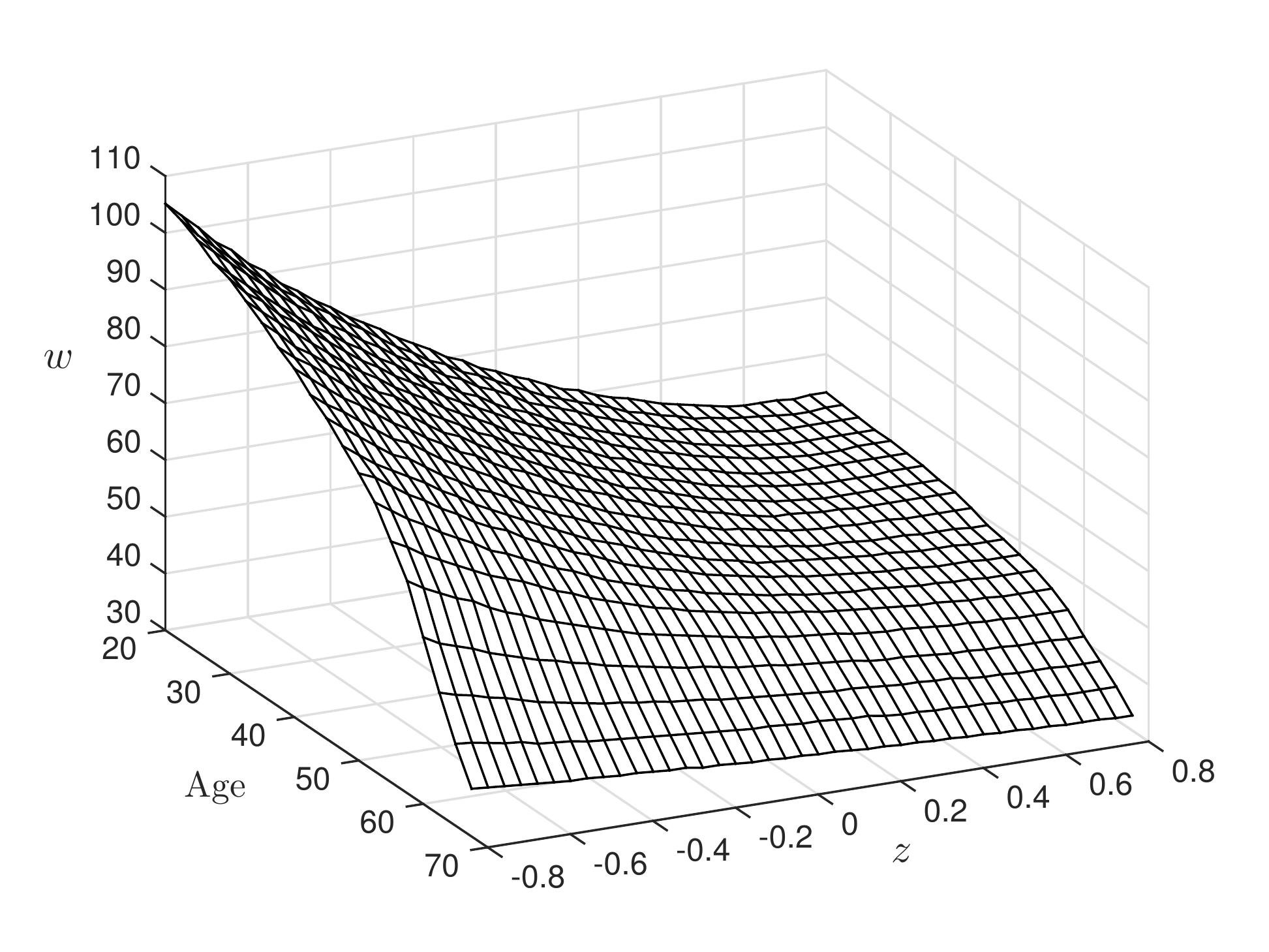} \\
{\small (iii) Target of non-participation} & {\small (iv) Wealth threshold for voluntary retirement}
\end{tabular}
\caption[Optimal investment and retirement with mandatory retirement age]{\textbf{Optimal investment and retirement with mandatory retirement age.}  Since we normalize  annual rate of labor income as one, i.e., $I=1$, wealth-to-income ratio $w/I$ reduces to financial wealth $w$. Basic parameters are chosen from Table \ref{coint_table_summary}. }
\label{coint_FiniteTime}
\end{figure}

\subsection{Robustness}
We have shown by a numerical analysis that the risk aversion speeds up retirement when the
risks associated with labor income are uninsurable. This result is based on the assumption
of CRRA utility preferences. We now relax this assumption. We derive the wealth threshold
for voluntary retirement using more elaborate non-expected recursive utility preferences 
(\cite{EZ89}, \cite{W90}, \cite{han2021disentangled, han2018deep}) than standard CRRA utility function. We use the 
continuous-time formulation of this non-expected utility developed by \cite{DE92}. It has been widely known that risk aversion should differ from elasticity of 
intertemporal substitution (EIS). Specifically, an investor's consumption, investment, and voluntary 
retirement problem with cointegration between the stock and labor markets, and short sale and 
borrowing constraints is to maximize his non-expected recursive utility preference by controlling per-period consumption $c$, risky investment $y$, 
and voluntary retirement time $\tau$:
\begin{equation*}
\begin{aligned}
V(w,I,z)\equiv\sup_{(c,y,\tau)\in\mathcal{A}(w,I,z)}\mathbb{E}\Big[&\int^{\tau}_{0}e^{-\delta_{D}t}\Big\{f\big(c_{t},V(W_{t},I_{t},Z_{t})\big)\\
&+\delta_{D}V(W_{t},\kappa I_{t},Z_{t})\Big\}dt+\int^{\infty}_{\tau}f\big(Bc_{t},V(W_{t},0)\big)dt\Big],
\end{aligned}
\end{equation*}
where $\mathbb{E}$ is the expectation taken at time $0$, $f(c,V)$ is the continuous-time formulation of the 
non-expected recursive utility and it is given by
$$
f(c,V)=\df{\beta}{1-\psi^{-1}}\Big\{c^{1-\psi^{-1}}[(1-\gamma)V]^{\frac{\psi^{-1}-\gamma}{1-\gamma}}
-(1-\gamma)V\Big\}.
$$
Here, $\psi>0$ is the coefficient of EIS. When we set $\psi=1/\gamma$, the recursive utility preference 
reduces to the widely used standard CRRA utility preference.

Under the CRRA specification, since the inverse relationship between the EIS and the relative risk 
aversion, higher EIS implies lower relative risk aversion, delaying retirement. By using nonexpected
recursive utility, we try to separately study the effects of the EIS and relative risk aversion.
We find that the EIS produces lower wealth threshold for voluntary retirement than that under CRRA utility
preferences. That is, higher EIS reduces the wealth threshold for voluntary retirement (Figure 
\ref{coint_recursive}). Intuitively, investors with a higher EIS are inclined to make more investments in risky
assets rather than cutting back their consumption, resulting in accumulating more wealth through which 
early retirement becomes more likely to be achieved. 
 
\begin{figure}[H]
\centering
\begin{tabular}{cc}
\includegraphics[width=0.46\textwidth]{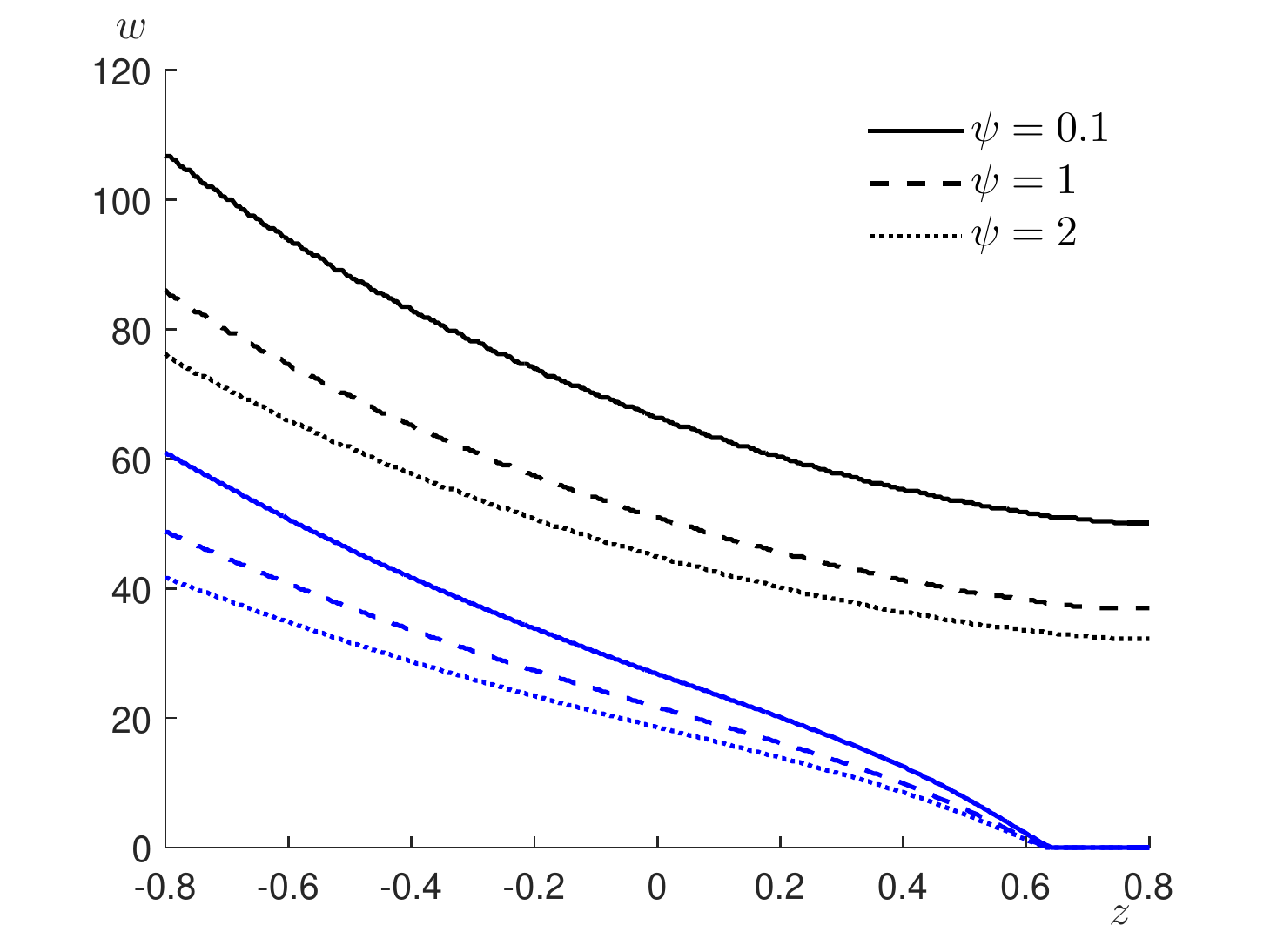} & \hspace{-.5cm}
\includegraphics[width=0.46\textwidth]{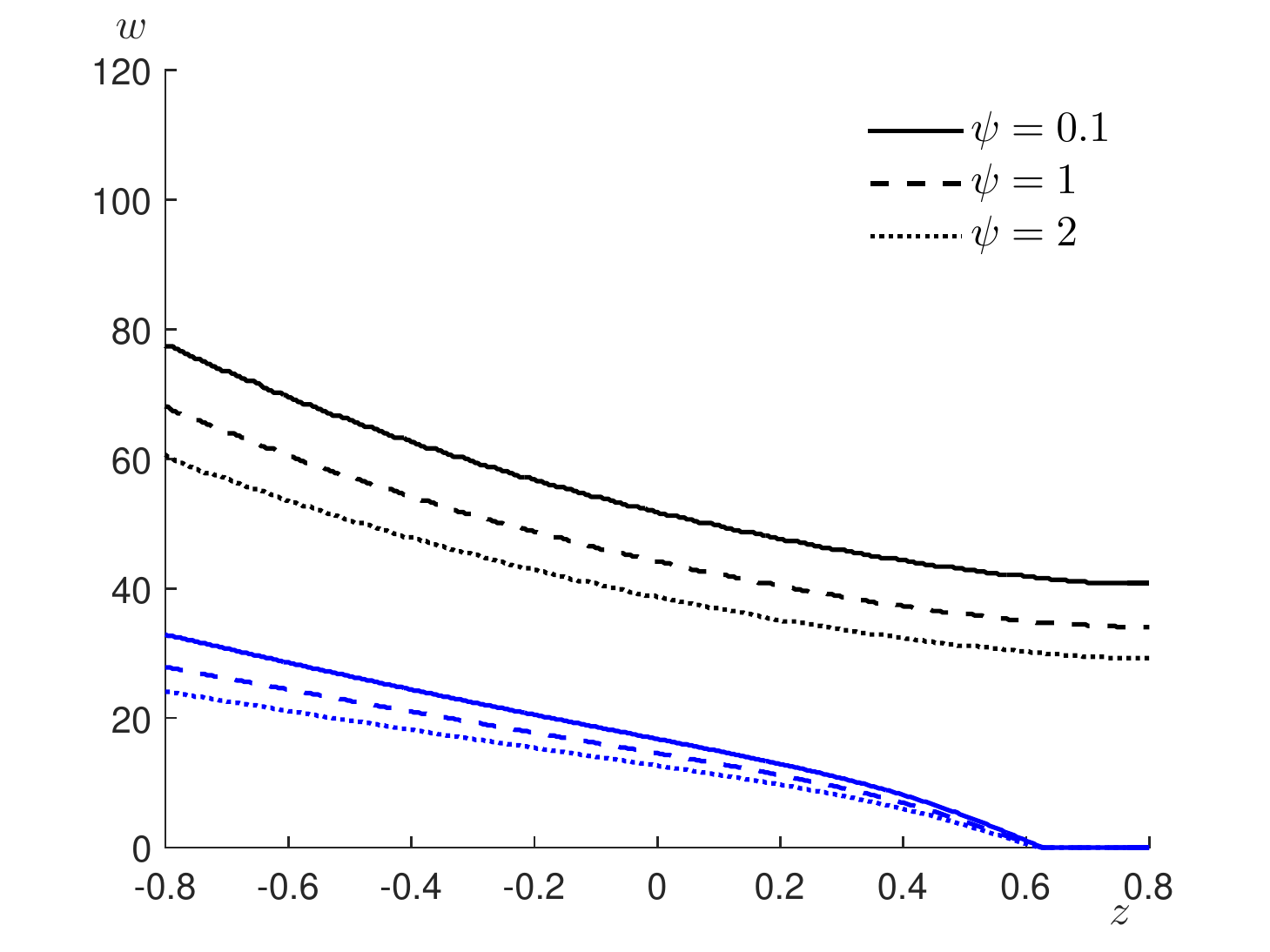} \\
{\small (i) $\delta_{D}=0$} & {\small (ii) $\delta_{D}=5\%$}
\end{tabular}
\caption[Sensitivity analysis of wealth threshold for voluntary retirement with respect to EIS $\psi$]{\textbf{Sensitivity analysis of wealth threshold for voluntary retirement with respect to EIS $\psi$.} 
The black lines (upper three lines) and the blue lines (lower three lines) represent the wealth threshold for voluntary retirement and the target wealth-to-income ratio, respectively. The left panel denotes the 
case without downward jumps in labor income, whereas the right panel stands for the case with downward jumps in labor income. Basic parameters are chosen from Table \ref{coint_table_summary}.
}
\label{coint_recursive}
\end{figure}

We find that risk aversion still speeds up retirement for a fixed EIS value (not reported).
Our other main results, also, are quantitatively similar under non-expected recursive utility 
specification for a multitude of parameter values.

\chapter{Conclusion}\label{section conclusion}

In this thesis, we consider two stochastic control problems in capital structure and individual's life-cycle portfolio choice. 
 
In Chapter \ref{part 1}, the reported bank accounting values do not necessarily correspond to the current business situation because banks' assets are difficult to assess and banks have an incentive to smooth their earnings. In this part,  we derived a new partially observed model for this situation and calibrated that to a sample of U.S. banks.
Given the noisy accounting reports, in our model the shareholders and the regulators obtain the conditional probability distribution of the true accounting values. Then the bank shareholders solve for the optimal dividend and recapitalization policy of the bank, and the bank regulators decide to close the bank, if the expected equity conditional on the accounting reports falls too low.
The threshold for the expected book equity when this happens is low because the regulators weight more the risk of liquidating a solvent bank than the risk of not liquidating an insolvent bank.
We focused on the bank's dividends and recapitalization option since banks mostly use those to adjust their equity capital (see e.g.  \cite{BH14} and  \cite{BLS16}).

On average, the noise in the reported accounting asset values hides about one-third of the true asset return volatility and raises the banks' market equity value by 7.8\%.
Particularly, those banks with a high level of loan loss provisions, nonperforming assets, and real estate loans, and with a low volatility of reported total assets have noisy accounting asset values.
Because of the substantial shock on the true asset values, the banks' assets were more opaque during the recent financial crisis than outside that.
The partially observed model explains the banks' actions significantly better than the corresponding fully observed model, indicating that the shareholders and bank regulators view accounting reports as noisy and act accordingly. 
Our model has been published on Journal of Economic Dynamics and Control (See \cite{DHK19}).

In Chapter \ref{part 2},   we present an optimal portfolio choice model for voluntary retirement in an economic situation, where the borrowing and short sale constrained investor is encountering uninsurable income risks and cointegration between the stock and labor markets.
The model can generate
empirically plausible investment strategy that proportions of financial wealth invested 
in stocks rise with financial wealth. The model can also provide rational explanations 
of early retirement that has been empirically observed in a strand of  literature
that explores retirement decisions. Our main results still hold when considering the mandatory retirement age. 

One of the important possible extensions  is to include general equilibrium
asset pricing consideration based on our partial equilibrium results. We strongly
believe that the model presented here could become a big advantage for better understandings 
toward sound financial advices on pension, insurance, and retirement, and the basis of 
policy design in order to resolve the issues associated with global retirement funding 
problems in the short and long runs.
\renewcommand{\thesection}{A.\arabic{section}}
\renewcommand{\theequation}{A.\arabic{section}-\arabic{equation}}
\renewcommand{\thesubsection}{A.\arabic{section}.\arabic{subsection}}

\renewcommand{\thelemma}{A.\arabic{lemma}}
\renewcommand{\theproposition}{A.\arabic{proposition}}

\setcounter{equation}{0}
\setcounter{section}{0}

\chapter*{Appendix}\label{bank_appendix}
\addcontentsline{toc}{chapter}{Appendix}
\section{Existance of $u_1$ and $u_2$ Under Conditions (\ref{appendix_fully_assumption})}\label{fully_proof_uniqueness_u1_u2}
Let us prove that under the conditions (\ref{appendix_fully_assumption}), there exists a pair solution $(u_1,u_2)$ to (\ref{fully_u1_u2_formula})
such that $\kappa<u_1<u_2<u_0$ and $H(X;u_2)\leq f_1(X;u_2)$ for all $X\in[\kappa, u_2]$.

We first introduce a lemma.
\begin{lemma} \label{appendix_compare_H_f2}
When $\omega \kappa \leq e^{-(\delta-\alpha)\Delta}\frac{\alpha-\mu}{\delta-\mu}(1+\kappa)$,
we have
$H(X;\theta)< f_2(X;\theta)$ for any $\theta\in(\kappa,u_0]$ and $X>\kappa$, where $f_2(X; \theta)$ is given in (\ref{appendixf2}),  $H(X; \theta)$ is given in (\ref{fully M}), and $u_0$ is given in Lemma \ref{lemma_u0}.
\end{lemma}
\noindent\textit{Proof}. $\omega \kappa \leq e^{-(\delta-\alpha)\Delta}\frac{\alpha-\mu}{\delta-\mu}(1+\kappa)$ implies $\delta>\alpha>\mu$.
By (\ref{p(x,t)}), we have
\begin{eqnarray*}
& & \mathbb{E}\left[(X_{\Delta}+1)\mathbf{1}_{\{\tau>\Delta\}}\right]=\mathbb{E}[X_\Delta+1]-\mathbb{E}\left[(X_{\Delta}+1)\mathbf{1}_{\{\tau\leq \Delta\}}\right] \nonumber \\
&=&e^{(\alpha-\mu)\Delta}(X+1)   -\int_0^\Delta \mathbb{E}[X_\Delta+1| X_t=\kappa]\frac{\partial }{\partial t} p(X,t)dt \nonumber\\
&=& e^{(\alpha-\mu)\Delta}(X+1)-(1+\kappa)p(X,\Delta)-(\alpha-\mu)(\kappa+1)\int_0^\Delta e^{(\alpha-\mu)(\Delta-t)}p(X,t)dt.\nonumber
\end{eqnarray*}
Combining with \begin{eqnarray*}
\mathbb{E}^{X}[\omega \kappa e^{-(\delta-\mu)\tau}\mathbf{1}_{\{\tau\leq \Delta\}}] &=& \omega \kappa\int_0^\Delta e^{-(\delta-\mu)t} \frac{\partial }{\partial t}p(X,t)dt 
\end{eqnarray*}
 \begin{eqnarray*}
&=& \omega \kappa e^{-(\delta-\mu)\Delta}p(X,\Delta) +\omega \kappa(\delta-\mu)\int_0^\Delta e^{-(\delta-\mu)t} p(X,t)dt,
\end{eqnarray*}
we obtain
\begin{eqnarray}
 \begin{aligned}\label{appendix_Mx}
 H(X;\theta)=&
e^{-(\delta-\mu)\Delta}\Big\{(f_1(\theta;\theta)-\theta-K-1)(1-p(X,\Delta))+e^{(\alpha-\mu)\Delta}(X+1)\\
&-(1+\kappa)p(X,\Delta)-(\alpha-\mu)(\kappa+1)\int_0^\Delta e^{(\alpha-\mu)(\Delta-t)}p(X,t)dt  \\
&+\omega \kappa p(X,\Delta) +\omega \kappa(\delta-\mu)\int_0^\Delta e^{(\delta-\mu)(\Delta-t)} p(X,t)dt\Big\}.
 \end{aligned}\end{eqnarray}
Using (\ref{appendix_Mx}),  we have
 \begin{eqnarray*}
H(X; \theta) &<& e^{-(\delta-\mu)\Delta}\Big\{ \left( f_1(\theta; \theta)-\theta -K-1\right)+e^{(\alpha-\mu)\Delta} (X+1)\Big\}\\
&\leq &  e^{-(\delta-\mu)\Delta}\Big\{ \left(f_1(\theta; \theta)+X-\theta\right)+(X+1)(e^{(\alpha-\mu)\Delta} - 1) \Big\}\\
&\leq & X-\theta+e^{-(\delta-\mu)\Delta}\Big\{ \frac{\alpha-\mu}{\delta-\mu}(1+\theta)+(X+1)(e^{(\alpha-\mu)\Delta} - 1)\Big\}\\
&\leq& \frac{\alpha-\mu}{\delta-\mu}(1+\theta)+X-\theta= f_2( X; \theta), \quad \text{ for } X\in(\kappa,\theta],
\end{eqnarray*}
where the first inequality follows from $(1-p(X,\Delta))\leq 1$, $\omega \kappa \leq e^{-(\delta-\alpha)\Delta}\frac{\alpha-\mu}{\delta-\mu}(1+\kappa) <(1+\kappa)$, and $\omega\kappa(\delta-\mu)e^{(\delta-\mu)(\Delta-t)}\leq (\alpha-\mu)(1+\kappa)e^{(\alpha-\mu)(\Delta-t)}$ for all $t\in[0,\Delta]$, the second inequality comes from $K\geq 0$,
the third inequality comes from $e^{-(\delta-\mu)\Delta}(X-\theta)\leq X-\theta$ and $f_1(\theta; \theta)=\frac{\alpha-\mu}{\delta-\mu}(1+\theta)$, and the last inequality comes from $(1+X)\frac{e^{(\alpha-\mu)\Delta}-1}{\alpha-\mu}\leq (1+\theta)\frac{e^{(\delta-\mu)\Delta}-1}{\delta-\mu}$ for all $X\leq \theta$. 

When $X>\theta$, we have
\begin{eqnarray*}
H(X; \theta)
&< &  e^{-(\delta-\mu)\Delta}\Big\{ \left(f_1(\theta; \theta)+X-\theta\right)+(X+1)(e^{(\alpha-\mu)\Delta} - 1) \Big\}\\
&\leq & e^{-(\delta-\mu)\Delta}\Big\{ \frac{\alpha-\mu}{\delta-\mu}(1+\theta)+X-\theta+(X+1)(e^{(\delta-\mu)\Delta} - 1)\frac{\alpha-\mu}{\delta-\mu}\Big\}\\
&=& \frac{\alpha-\mu}{\delta-\mu}(1+\theta)+\left(\frac{\alpha-\mu}{\delta-\mu}+\frac{\delta-\alpha}{\delta-\mu} e^{-(\delta-\mu)\Delta}  \right)(X-\theta)\\
&\leq & \frac{\alpha-\mu}{\delta-\mu}(1+\theta)+ (X-\theta) = f_2( X; \theta), \quad \text{ for } X>\theta,
\end{eqnarray*}
where the  second inequality comes from $f_1(\theta; \theta)=\frac{\alpha-\mu}{\delta-\mu}(1+\theta)$ and $\frac{e^{(\alpha-\mu)\Delta}-1}{\alpha-\mu}\leq \frac{e^{(\delta-\mu)\Delta}-1}{\delta-\mu}$,
and the last inequality comes from $\frac{\alpha-\mu}{\delta-\mu}+\frac{\delta-\alpha}{\delta-\mu} e^{-(\delta-\mu)\Delta}\leq 1$.
This completes the proof of Lemma \ref{appendix_compare_H_f2}.\\

Now we give a lemma to show the concavity of $H(X;\theta)$.
\begin{lemma}\label{appendix_H_concave}
 Assume $\alpha-\mu-\frac{1}{2}\sigma^2\geq 0$,  $\omega \kappa \leq e^{-(\delta-\alpha)\Delta}\frac{\alpha-\mu}{\delta-\mu}(1+\kappa)$, and $u_0<\frac{\alpha-\mu}{\delta-\alpha}+\frac{\delta-\mu}{\delta-\alpha}(\kappa-\omega\kappa-K)$, where $u_0$ is given in Lemma \ref{lemma_u0}. Then $H(X;\theta)$ in (\ref{fully M}) is concave w.r.t $X$ for any $\theta \in(\kappa, u_0]$.
\end{lemma}
\noindent\textit{Proof}.  When $\theta\in(\kappa,u_0]$, we have $\theta\leq u_0<\frac{\alpha-\mu}{\delta-\alpha}+\frac{\delta-\mu}{\delta-\alpha}(\kappa-\omega\kappa-K)$, which is equivalent to $f_1(\theta;\theta)-\theta+\kappa-K>\omega \kappa$ as $f_1(\theta; \theta)=\frac{\alpha-\mu}{\delta-\mu}(1+\theta)$.
Because $\zeta_t:=\ln(X_t+1)$ is a Brownian motion with drift $\alpha-\mu-\frac{1}{2}\sigma^2\geq 0$ and volatility $\sigma$,  the corresponding cumulative distribution  for stopping time $\tau$, denoted by $\tilde p(\zeta,t)$, is a concave function in $X$.  $p(X,t)$ in (\ref{p(x,t)}) can be rewritten as $p(X,t)=\tilde p(\ln(X+1), t)$. Hence the cumulative distribution $p(X,t)$ satisfies
\begin{eqnarray*}
\frac{\partial }{\partial X}p(X,t)<0,\quad \frac{\partial ^2}{\partial X^2}p(X,t)>0, \quad \frac{\partial }{\partial t}p(X,t)>0,
\end{eqnarray*}
for all $(X,t)\in \mathbb{R}^+ \times \mathbb{R}^+.$ Therefore, from (\ref{appendix_Mx}) we get
 \begin{eqnarray*}\begin{aligned}
&\frac{\partial }{\partial X}H(X;\theta)= - e^{-(\delta-\mu)\Delta}(f_1(\theta;\theta)-\theta+\kappa-\omega\kappa-K)\frac{\partial }{\partial X}p(X,\Delta)+e^{-(\delta-\alpha)\Delta}\\
&-e^{-(\delta-\mu)\Delta}\int_0^\Delta \left[(\alpha-\mu)(1+\kappa)e^{(\alpha-\mu)(\Delta-t)}- \omega\kappa( \delta-\mu) e^{(\delta-\mu)(\Delta-t)} \right] \frac{\partial}{\partial X} p(X,t)dt> 0,\\
&\frac{\partial^2}{\partial X^2}H(X;\theta)=- e^{-(\delta-\mu)\Delta}(f_1(\theta;\theta)-\theta+\kappa-\omega\kappa-K)\frac{\partial^2 }{\partial X^2} p(X,\Delta)\\
&-e^{-(\delta-\mu)\Delta}\int_0^\Delta \left[(\alpha-\mu)(1+\kappa)e^{(\alpha-\mu)(\Delta-t)}- \omega\kappa( \delta-\mu) e^{(\delta-\mu)(\Delta-t)} \right]  \frac{\partial^2}{\partial X^2}p(X,t)dt < 0,
\end{aligned}\end{eqnarray*}
as $f_1(\theta;\theta)-\theta+\kappa-\omega \kappa-K\geq 0$ and $(\alpha-\mu)(1+\kappa)e^{(\alpha-\mu)(\Delta-t)}-\omega\kappa(\delta-\mu)e^{(\delta-\mu)(\Delta-t)} \geq 0$ for all $t\in[0,\Delta]$.
This completes the proof of Lemma \ref{appendix_H_concave}.\\

We then prove the existence of $(u_1,u_2)$ in two steps.\\
i). If $\frac{\partial H(X;u_0)}{\partial X}\big|_{X=\kappa}> \frac{\partial f_1(X;u_0)}{\partial X}\big|_{X=\kappa}$ holds,  the equation $ H(\kappa; u_0)= f_1(\kappa; u_0)=\omega \kappa$ implies that there exists an $X\in(\kappa,u_0)$ such that  $H(X; u_0)>f_1(X;u_0)$. On the other hand, we know $H(u_0;u_0)<f_1(u_0;u_0)=\frac{\alpha-\mu}{\delta-\mu}(1+u_0)$ by Lemma \ref{appendix_compare_H_f2} under the condition $\omega\kappa<e^{-(\delta-\alpha)\Delta}\frac{\alpha-\mu}{\delta-\mu}(1+\kappa)$. This implies that $H(X;u_0)$ must cross $f_1(X;u_0)$ from above within the interval $\kappa<X<u_0.$\\
 ii). 
We know $\frac{\partial }{\partial \theta}f_1(x;\theta)<0$ in the proof of Lemma \ref{lemma_u0}. Then we have $H(\kappa; \theta) = \omega\kappa=f_1(\kappa;u_0) < f_1(\kappa; \theta)$ for $\theta\in(\kappa,u_0)$.
By Lemma \ref{appendix_compare_H_f2}, we infer $H(\theta;\theta)<f_1(\theta;\theta)=\frac{\alpha-\mu}{\delta-\mu}(1+\theta)$ for $\kappa<\theta\leq u_0$
Further,
 \begin{eqnarray*}\lim_{\theta\rightarrow \kappa}H(\theta;\theta)=\omega\kappa<\frac{\alpha-\mu}{\delta-\mu}(1+\kappa)=\lim_{\theta\rightarrow \kappa} f_1(\kappa;\theta).\end{eqnarray*}
 Notice $H(X;\theta)$ is concave and  increasing by Lemma \ref{appendix_H_concave}, while $f_1(X;\theta)$ is also increasing in $X$.
 These implies that there exists a $\theta'\in(\kappa,u_0)$ such that $H(X;\theta')<f_1(X;\theta')$ for all $X\in[\kappa, \theta'].$

From i) and ii), by the continuity of $H(X;\theta)$ and $f_1(X;\theta)$ w.r.t $X$ and $\theta$, there exist a $u_2$ in the interval $(\kappa,u_0)$ such that $H(u_1;u_2)=f_1(u_1;u_2)$ for some $u_1\in(\kappa, u_2)$ while $H(X;u_2)\leq f_1(X;u_2)$ for all $X \in[\kappa,u_2]$. The continuous differentiability of $H(X;u_2)$ and $f_1(X;u_2)$ w.r.t $X $ implies that
$\frac{\partial H(X;u_2)}{\partial X}\Big|_{X=u_1}= \frac{\partial f_1(X;u_2)}{\partial X}\Big|_{X=u_1}$. This pair  of $(u_1, u_2)$ is what we are looking for.

\section{Verification of $V(X)$ in (\ref{value of recap}) in Theorem \ref{semi-explicit solution with recap}}\label{verification of conjectured v}
In this section, we prove that $V(X)$ in (\ref{value of recap}) in Theorem \ref{semi-explicit solution with recap} satisfies the HJB equation (\ref{fully_hjb3}), i.e., it is a solution to (\ref{fully_hjb3})  and, thus, is indeed the value function of the fully observed model by the uniqueness of viscosity solution.

We  first give a lemma to show that $H(X;\theta)$ is decreasing in delay time $\Delta$.
\begin{lemma}\label{appendix_H_monotone_delay}
When $\alpha-\mu-\frac{1}{2}\sigma^2\geq 0$, $\omega \kappa \leq e^{-(\delta-\alpha)\Delta}\frac{\alpha-\mu}{\delta-\mu}(1+\kappa)$,
$u_0<\frac{\alpha-\mu}{\delta-\alpha}+\frac{\delta-\mu}{\delta-\alpha}(\kappa-\omega\kappa-K)$, and $ \frac{\delta-\alpha}{\delta-\mu} e^{(\alpha-\mu)\Delta} \geq \big[\frac{\delta-\alpha}{\delta-\mu} e^{(\alpha-\mu)\Delta}(1+\kappa) -\frac{\delta-\alpha}{\delta-\mu}(1+u_0)-K\big] \frac{\partial }{\partial X} p(X,\Delta)\big |_{X=\kappa}$, where $u_0$ is given in Lemma \ref{lemma_u0} and $p(X,t)$ is defined in (\ref{p(x,t)}), the function $H(X;\theta)$ in (\ref{fully M}) satisfies $\frac{\partial}{\partial \Delta} H(X;\theta)\leq 0$ for all $X\geq \kappa$ and $\theta\in(\kappa,u_0]$.
\end{lemma}
\begin{proof}
Due to $\alpha-\mu-\frac{1}{2}\sigma^2\geq 0$, we have $\alpha>\mu$.
Let us rewrite $H(X;\theta)$ in (\ref{appendix_Mx}) as
\begin{align*}
&H(X;\theta) = e^{-(\delta-\mu)\Delta}(f_1(\theta;\theta)-\theta-1-K)(1-p(X,\Delta)) + e^{-(\delta-\alpha)\Delta}(1+X)\\
&-e^{-(\delta-\alpha)\Delta}(1+\kappa)\int_0^\Delta e^{-(\alpha-\mu)t}\frac{\partial }{\partial t} p(X,t)dt + \omega\kappa \int_0^\Delta e^{-(\delta-\mu)t} \frac{\partial }{\partial t} p(X,t)dt.
\end{align*}
Taking derivative with respect to $\Delta$ in the above equation, we obtain
\begin{eqnarray*}
\begin{aligned}
&\frac{\partial}{\partial \Delta}H(X;\theta) = e^{-(\delta-\mu)\Delta}(f_1(\theta;\theta)-\theta-1-K)[-(\delta-\mu)(1-p(X,\Delta)) -\frac{\partial }{\partial t} p(X,t)\big|_{t=\Delta}] \\
&- (\delta-\alpha) e^{-(\delta-\alpha)\Delta}(1+X)+e^{-(\delta-\alpha)\Delta}(\delta-\alpha)(1+\kappa)\int_0^\Delta e^{-(\alpha-\mu)t}\frac{\partial }{\partial t} p(X,t)dt \\
& -(1+\kappa)e^{-(\delta-\mu)\Delta} \frac{\partial }{\partial t} p(X,t)\big|_{t=\Delta} + \omega\kappa e^{-(\delta-\mu)\Delta} \frac{\partial }{\partial t} p(X,t)\big|_{t=\Delta}\\
\leq &
 (\delta-\mu) e^{-(\delta-\mu)\Delta} \left(\frac{\delta-\alpha}{\delta-\mu}(1+\theta)+K\right)  (1-p(X,\Delta))\\
 & - (\delta-\alpha) e^{-(\delta-\alpha)\Delta}\big[1+X -(1+\kappa)p(X,\Delta) \big]\\
 &-e^{-(\delta-\mu)\Delta}\left(\frac{\alpha-\mu}{\delta-\mu}(1+\theta)-\theta-K+\kappa -\omega\kappa \right)\frac{\partial }{\partial t} p(X,t)\big|_{t=\Delta},
\end{aligned}\end{eqnarray*}
where the above inequality comes from $\int_0^\Delta e^{-(\alpha-\mu)t}\frac{\partial }{\partial t}p(X,t)dt
\leq \int_0^\Delta \frac{\partial }{\partial t} p(X,t)dt = p(X,\Delta)$.  Because $u_0<\frac{\alpha-\mu}{\delta-\alpha}+\frac{\delta-\mu}{\delta-\alpha}(\kappa-\omega\kappa-K)$, we know $\frac{\alpha-\mu}{\delta-\mu}(1+\theta)-\theta-K+\kappa-\omega\kappa \geq 0$ for all $\theta\in(\kappa,u_0]$. In order to prove $\frac{\partial}{\partial \Delta}H(X;\theta) \leq 0$, we only need to verify that $G(X; \theta, \Delta) \geq 0$ for all $X\geq \kappa$ and $\theta\in(\kappa,u_0]$, where
{\fontsize{10.5pt}{10.5pt}
\begin{eqnarray*}\begin{aligned}
G(X; \theta,\Delta): = &  \frac{\delta-\alpha}{\delta-\mu} e^{(\alpha-\mu) \Delta}\big[1+X -(1+\kappa)p(X,\Delta) \big] -  \left(\frac{\delta-\alpha}{\delta-\mu}(1+\theta)+K\right)  (1-p(X,\Delta)).
\end{aligned}\end{eqnarray*}}
It is easy to check $G(\kappa; \theta, \Delta) = 0$. Moreover, by calculation,
{\fontsize{10.5pt}{10.5pt}
\begin{eqnarray*}\begin{aligned}
\frac{\partial}{\partial X}G(X;\theta,\Delta) = &  \frac{\delta-\alpha}{\delta-\mu} e^{(\alpha-\mu) \Delta} + \left(\frac{\delta-\alpha}{\delta-\mu}(1+\theta)+K - \frac{\delta-\alpha}{\delta-\mu} e^{(\alpha-\mu) \Delta}(1+\kappa) \right) \frac{\partial }{\partial X} p(X,\Delta).
\end{aligned}\end{eqnarray*}}
By lemma \ref{appendix_H_concave}, $p(X,\Delta)$ is a concave function with respect to $X$.
If $\frac{\delta-\alpha}{\delta-\mu}(1+\theta)+K - \frac{\delta-\alpha}{\delta-\mu} e^{(\alpha-\mu) \Delta}(1+\kappa)\leq 0$, we have $\frac{\partial}{\partial X}G(X;\theta,\Delta) \geq 0$ because $\frac{\partial }{\partial X} p(X,\Delta)\leq 0$.  If $\frac{\delta-\alpha}{\delta-\mu}(1+\theta)+K - \frac{\delta-\alpha}{\delta-\mu} e^{(\alpha-\mu) \Delta}(1+\kappa)>0$, since $\frac{\partial }{\partial X^2} p(X,\Delta)\geq 0$, we have $\frac{\partial }{\partial X} p(X,\Delta)|_{X=\kappa} \leq \frac{\partial }{\partial X} p(X,\Delta)<0$ for all $X\geq \kappa$ and 
\begin{eqnarray*}\begin{aligned}
&\frac{\partial}{\partial X}G(X;\theta,\Delta) \\
\geq &  \frac{\delta-\alpha}{\delta-\mu} e^{(\alpha-\mu) \Delta}  + \left(\frac{\delta-\alpha}{\delta-\mu}(1+\theta)+K - \frac{\delta-\alpha}{\delta-\mu} e^{(\alpha-\mu) \Delta}(1+\kappa) \right)  \frac{\partial }{\partial X} p(X,\Delta)\big |_{X=\kappa}  \\
\geq &  \frac{\delta-\alpha}{\delta-\mu} e^{(\alpha-\mu) \Delta}  + \left(\frac{\delta-\alpha}{\delta-\mu}(1+u_0)+K - \frac{\delta-\alpha}{\delta-\mu} e^{(\alpha-\mu) \Delta}(1+\kappa) \right)\frac{\partial }{\partial X} p(X,\Delta)\big |_{X=\kappa}    \geq 0.
\end{aligned}\end{eqnarray*}
Hence $G(X;\theta,\Delta)\geq  G(\kappa;\theta,\Delta) = 0$ for all $X\geq\kappa$ and $\theta\in(\kappa,u_0]$. This completes the proof of Lemma \ref{appendix_H_monotone_delay}.
\end{proof}

We now verify $V(X)$ in (\ref{value of recap}) satisfies the HJB equation (\ref{fully_hjb3}). When $X=\kappa$, we have $V(\kappa) = H(\kappa; u_2) = \omega \kappa$, so that the boundary condition is satisfied. \\
$(i).$ We prove that $V\geq \mathcal{P}^0V$ for all $X\in[\kappa, +\infty)$. When $X\in[\kappa, u_1]$, we know $V(X) = H(X; u_2)=\mathcal{P}^0V(X)$ by construction. When $X\in[u_1, u_2]$, we know $V(X)=f_1(X;u_2) \geq H(X; u_2) = \mathcal{P}^0V(X)$ by construction and the proof in online Appendix \ref{fully_proof_uniqueness_u1_u2}.  When $X\in[u_2, +\infty]$, we have $V(X) = f_2(X; u_2) > H(X; u_2)$ by construction and Lemma \ref{appendix_compare_H_f2}.  \\
$(ii).$ We prove that $\mathcal{L}V\leq 0$ for all $X\in[\kappa, +\infty)$.  When $X\in (\kappa, u_1)$, we know $V(X)=H(X; u_2)$ by construction.
From Ito's formula,
$$
e^{-(\delta-\mu)\tau_\varepsilon} H(X_{\tau_\varepsilon}; u_2) = M(X; u_2) + \int_0^{\tau_\varepsilon}\mathcal{L}H(X_t; u_2)dt + \int_0^{\tau_\varepsilon} H_{X}(X_t;u_2) \sigma (1+X_t) d W_t,
$$
where $\tau_\varepsilon:= \varepsilon\wedge \inf\{ t\geq 0: X_t\notin (X-\varepsilon,X+\varepsilon)\} $ for small $\varepsilon$ such that $\kappa<X-\varepsilon$ and $X+\varepsilon<u_1$. As $H_{X}(X_t;u_2)$ is bounded during $[0,\tau_\varepsilon]$, the expected stochastic integrand above is zero and we have
\begin{equation}\label{partial_H_disH}
\mathbb{E}[e^{-(\delta-\mu)\tau_\varepsilon} H(X_{\tau_\varepsilon}; u_2)] = H(X; u_2) + \mathbb{E}\Big[\int_0^{\tau_\varepsilon}\mathcal{L}H(X_t; u_2)dt \Big].
\end{equation}
Notice that $\mathbb{E}[e^{-(\delta-\mu)\tau_\varepsilon} H(X_{\tau_\varepsilon}; u_2)]$ means the value of waiting  until  $\tau_\varepsilon$ prior to ordering a new equity. As there is no liquidation during time period $[0,\tau_\varepsilon]$, the expectation $\mathbb{E}[e^{-(\delta-\mu)\tau_\varepsilon} H(X_{\tau_\varepsilon}; u_2)]$ is equivalent to the case where new equity is ordered at time zero and will be issued at time $(\Delta+\varepsilon)$. That is $\mathbb{E}[e^{-(\delta-\mu)\tau_\varepsilon} H(X_{\tau_\varepsilon}; u_2)] = H(X; u_2, \Delta+\varepsilon)$ where $H(X; u_2, \Delta+\varepsilon)$ refers to $H(X;u_2)$ in (\ref{fully M}) with $\Delta$ replaced by $\Delta+\varepsilon$. By Lemma \ref{appendix_H_monotone_delay}, we have $H(X;u_2) \geq H(X; u_2, \Delta+\varepsilon)$ for small $\varepsilon>0$ and therefore, $H(X;u_2)\geq \mathbb{E}[e^{-(\delta-\mu)\tau_\varepsilon} H(X_{\tau_\varepsilon}; u_2)]$.
Taking $\varepsilon\rightarrow 0$ in (\ref{partial_H_disH}) gives
$\mathcal{L}H(X;u_2) = \mathcal{L}V(X) \leq 0$ for $X\in(\kappa,u_1)$. When $X\in[u_1, u_2]$, obviously $\mathcal{L}V=0$ by construction.  When $X>u_2$, we have $\mathcal{L}V= \mathcal{L}f_2(X;u_2)=-(\delta-\alpha)(X-u_2)\leq 0$. \\
$(iii).$ We prove that $V_{X} \geq 1$ for all $X>\kappa$.  It is easy to check that the corresponding $V(X)$ in (\ref{value of recap}) is globally concave and $C^{2}$ except at $X=u_1$ by construction and Lemma \ref{appendix_H_concave}.
When $X\in[u_2, +\infty)$, $V(X) = f_2(X; u_2)=f_1(u_2; u_2) + X-u_2$  so that $V_X = 1$.  When $X\in(\kappa, u_2)$, we obtain $V_X\geq 1$  by the concavity of $V(X)$.

Therefore, $V(X)$ in (\ref{value of recap}) is a solution to HJB equation (\ref{fully_hjb3}) and satisfies the linear growth condition and the boundary condition, which leads to the desired result by the uniqueness of viscosity solution.

\section{Proof of Proposition \ref{theorem of s(t)}.}  \label{appendix_theorem of s(t)}
To prove Proposition  \ref{theorem of s(t)}, we define $d\tilde Z_t:=\frac{M_t}{m}dt+d \mathcal{B}_t,$ where $M_t = \log Y_t$ and $m$ is the noise level in $(\ref{signalprocess})$.
Next, we create a new probability measure \begin{eqnarray*}\frac{d\tilde{\mathbb{P}}}{d \mathbb{P}}= \nu_t,\quad  d\nu_t= -\nu_t\frac{M_t}{m}d\mathcal{B}_t.\end{eqnarray*}
Thus, $\nu_t=\exp\left(-\frac{1}{2}\int_0^t\frac{M_s^2}{m^2}ds-\int_0^t\frac{M_s}{m}d\mathcal{B}_s\right).$ We first give two lemmas as follows.
\begin{lemma}
\label{wiener process}Process $\widehat W_t, \widetilde W_t, \tilde Z_t$ are all standard Wiener processes under $\tilde{\mathbb{P}}$, where \begin{eqnarray*}
\widehat W_t:=W_t+\int_0^t\rho\frac{M_s}{m}ds,\quad \widetilde W_t:= \frac{\widehat W_t-\rho \tilde Z_t}{\sqrt{1-\rho^2}}. \end{eqnarray*} Moreover, process $\widehat W_t$ and $\tilde Z_t$ have correlation $\rho$ while process $\widetilde W_t$ and $\tilde Z_t$ are independent.
\end{lemma}
\begin{proof}
By the change of measure, it is easy to obtain  that $\tilde Z_t$ and $\widehat W_t$ are standard Wiener Processes under $\tilde{\mathbb{P}}$. Moreover, $d\widehat W_t d \tilde Z_t=\rho dt$, thus, $\widehat W_t$ and $\tilde Z_t$ have correlation $\rho$ and $\widetilde W_t$ is independent of $\tilde Z_t$ by the definition of $\widetilde W_t$.
\end{proof}

\begin{lemma}\label{dynamic of q}
Define $q(M, t)$ as the conditional density of $M_t$ under $\mathcal{G}_t$ that solves $\mathbb{E}[\psi(M_t,t)|\mathcal{G}_t]=\int q(M,t)\psi(M,t)dM/ \int \psi(M,t)dM$ for any test function $\psi\in C^{2,1}_0,$ where  $M_t=\log Y_t$ is the log total assets. Then $q(M,t)$ satisfies 
\begin{equation*}dq=[-q_M(\alpha-\frac{1}{2}\sigma^2)+\frac{1}{2}\sigma^2q_{MM}]dt+(q\frac{M}{m}-\sigma\rho q_M)d\tilde Z_t.\end{equation*}
\end{lemma}

\begin{proof} In order to prove lemma \ref{dynamic of q}, we divide the whole process into several steps. More details can be referred to in Chapter 4 in  \cite{B04}.\\
$\textit{Step 1}$: Un-normalized conditional probability.
Let us introduce a new information filtration $\tilde{\mathcal{G}}_t:=\sigma\{\tilde Z(s),s\leq t\}$.  Obviously $\tilde{\mathcal{G}}_t=\mathcal{G}_t$ as $\tilde Z_t = Z_t/m$, so that
\begin{eqnarray*}\Pi(t)(\varphi_t)=\mathbb{E}[\varphi(M_t,t)|\mathcal{G}_t]=\mathbb{E}[\varphi(M_t,t)|\tilde{\mathcal{G}}_t],\ \varphi\in C_0^{2,1}.\end{eqnarray*}
It is convenient to use $\tilde{\mathbb{P}}$ defined in Lemma \ref{wiener process} instead of $\mathbb{P}$, because the noise signal $\tilde Z_t$ is standard Wiener process under $\tilde{\mathbb{P}}$. Therefore, we need also the Radon-Nikodym derivative
\begin{eqnarray*}\frac{d \mathbb{P}}{d\tilde{\mathbb{P}}}=\frac{1}{\nu_t}=\eta_t,\quad 
\eta_t=\exp\left(\frac{1}{2}\int_0^t\frac{M_s^2}{m^2}ds+\int_0^t\frac{M_s}{m}d\mathcal{B}_s\right).\end{eqnarray*}
The un-normalized conditional probability is defined as $p(t)(\varphi_t)=\tilde{ \mathbb{E}}[\varphi(M_t,t)\eta_t| \tilde{\mathcal{G}}_t].$ Then we have \begin{eqnarray*}\Pi(t)(\varphi_t)=\mathbb{E}[\varphi(M_t,t)| \tilde{\mathcal{G}}_t]=\frac{\tilde{\mathbb{E}}[\varphi(M_t,t)\eta_t| \tilde{\mathcal{G}}_t]}{\tilde{\mathbb{E}}[\eta_t| \tilde{\mathcal{G}}_t]}=\frac{p(t)(\varphi_t)}{p(t)(1)}.\end{eqnarray*}
$\textit{Step 2}$: Zakai equation. To proceed, we note that \begin{eqnarray*}d\eta_t&=&\eta_t\left[\frac{M_t^2}{m^2}dt+\frac{M_t}{m}d\mathcal{B}_t \right]=\eta_t\frac{M_t}{m}d\tilde Z_t,\\  d M_t&=&\left(\alpha-\frac{1}{2}\sigma^2-\sigma\rho\frac{M_t}{m}\right)dt+\sigma d\widehat W_t\\
&=&\left(\alpha-\frac{1}{2}\sigma^2-\sigma\rho\frac{M_t}{m}\right)dt+\sigma \rho d\tilde Z_t+\sigma\sqrt{1-\rho^2}d\widetilde W_t.\end{eqnarray*}
Therefore, by Ito's lemma,
\begin{eqnarray*}d[\eta_t\varphi(M_t,t)]&=&\eta_t\left[\frac{\partial\varphi}{\partial t}+\left(\alpha-\frac{1}{2}\sigma^2-\sigma\rho\frac{M_t}{m}\right)\frac{\partial \varphi}{\partial M}+\frac{1}{2}\sigma^2\frac{\partial^2\varphi}{\partial M^2}\right]dt\\
&+&\eta_t\frac{\partial \varphi}{\partial M}\left[\sigma \rho d\tilde Z_t+\sigma\sqrt{1-\rho^2}d\widetilde W_t\right]+\eta_t\varphi\frac{M_t}{m}d\tilde Z_t+\eta_t\sigma\rho\frac{M_t}{m}\frac{\partial \varphi}{\partial M}dt\\
&=& \eta_t\left\{\left[\frac{\partial\varphi}{\partial t}-\mathcal{A}_1\varphi\right]dt+\sigma\sqrt{1-\rho^2}\frac{\partial \varphi}{\partial M}d\widetilde W_t+\left(\sigma\rho\frac{\partial \varphi}{\partial M}+\varphi\frac{M_t}{m}\right)d\tilde Z_t\right\},
\end{eqnarray*}
where $\mathcal{A}_1=-(\alpha-\frac{1}{2}\sigma^2)\frac{\partial }{\partial M}-\frac{1}{2}\sigma^2\frac{\partial^2}{\partial M^2}.$
To compute the conditional expectation, we use test functions which are $\tilde{\mathcal{G}}_t$-measurable. Because the generating processes are Wiener processes, it is sufficient to test with stochastic processes of the form\begin{eqnarray*}d\gamma(t)=i\gamma(t)\left(\beta_1(t)d\widetilde W_t+\beta_2(t)d\tilde Z_t\right), \gamma(0)=1
\end{eqnarray*}
where $i=\sqrt{-1},$ and $\beta_1, \beta_2 \in \mathbb{R}$ are arbitrary deterministic bounded functions. By the process of $\eta_t\varphi(M_t, t)$ the definition of $p(t)(\varphi_t)$, and $\tilde{\mathbb{E}}[\widetilde W_t|\tilde{\mathcal{G}}_t]=0,$ \begin{eqnarray*}\tilde{\mathbb{E}}[\gamma(t)p(t)(\varphi_t)]&=&\tilde{\mathbb{E}}[\gamma(t)\Pi(0)(\varphi_0)]+\tilde {\mathbb{E}}\left[\gamma(t)\int_0^t p(s)\left(\frac{\partial \varphi}{\partial s}-\mathcal{A}\varphi\right)ds\right]\\
& & +\tilde{\mathbb{E}}\left[\gamma(t)\int_0^t p(s)\left(\sigma\rho\frac{\partial \varphi}{\partial M}+\varphi\frac{M_s}{m}\right)d\tilde Z_s\right].\end{eqnarray*}
Because this relation holds for all $\gamma(t)$ (defined above), we get the Zakai equation\begin{eqnarray*}p(t)(\varphi_t)&=& \Pi(0)(\varphi_0)+ \int_0^t p(s)\left(\frac{\partial \varphi}{\partial s}-\mathcal{A}\varphi\right)ds +\int_0^t p(s)\left(\sigma\rho\frac{\partial \varphi}{\partial M}+\varphi\frac{M_s}{m}\right)d\tilde Z_s.\end{eqnarray*}
$\textit{Step 3}$: Un-normalized density. We look for a density that solves the equation above, i.e $q(M,t)$ such that\begin{eqnarray*}p(t)(\varphi_t)=\int q(M,t)\varphi(M,t)dM.\end{eqnarray*}
Then we get\begin{eqnarray*}\int q(M,t)\varphi(M,t)dM&=&\int q(M,0)\varphi(M,0)dM+\int_0^t\int q(M,s)\left(\frac{\partial \varphi}{\partial s}-\mathcal{A}\varphi\right)dMds\\
& & +\int_0^t\int q(M,s)\left(\sigma\rho\frac{\partial \varphi}{\partial M}+\varphi\frac{M_s}{m}\right)dMd\tilde Z_s.\end{eqnarray*}
Using integration by parts in $t$ and $M$, we have\begin{equation*}\int\left[dq+\mathcal{A}_1^*qdt+(\sigma\rho q_M-q\frac{M}{m})d\tilde Z_t\right]\varphi dM=0.\end{equation*}
where $\mathcal{A}_1^*$, the adjoint of $A_1$, is given as $\mathcal{A}_1^*=(\alpha-\frac{1}{2}\sigma^2)\frac{\partial }{\partial M}-\frac{1}{2}\sigma^2\frac{\partial^2}{\partial M^2}.$
This completes the proof of Lemma \ref{dynamic of q}.
\end{proof}
Now after giving two lemmas above, we divide the proof of Proposition  \ref{theorem of s(t)} into several steps as well.\\
$\textit{Step 1}$: Un-normalized density processes. Let the initial density be normal with mean $M_0$ and variance $S_0$, i.e $p_0(M)=\frac{1}{\sqrt{2\pi S_0}}e^{-\frac{1}{2}(M-M_0)^2/S_0}$.  We postulate \begin{eqnarray*}q(M,t)=\exp\left(-\frac{1}{2}[\Gamma_t M^2-2v_t M+b_t ]\right),\end{eqnarray*}
where $\Gamma$ is deterministic and $v,b$ are Ito processes. Thus we write \begin{eqnarray*}dv=v_0dt+v_1d\widetilde W_t+v_2d\tilde Z_t, \quad db=v_0dt+b_1d \widetilde W_t+b_2d\tilde Z_t. \end{eqnarray*}
By Ito's lemma,  \begin{eqnarray*}dq=q\left[-\frac{1}{2}\Gamma'M^2dt+Mdv-\frac{1}{2}db+\frac{1}{2}(Mv_1-\frac{1}{2}b_1)^2dt+\frac{1}{2}(Mv_2-\frac{1}{2}b_2)^2dt\right].\end{eqnarray*}
$\textit{Step 2}$: Note that $q_M=q(-M\Gamma+v)$ and $q_{MM}=q(-M\Gamma+v)^2-q\Gamma.$ Comparing the coefficient of diffusion and drift term, we obtain \begin{eqnarray*}& & v_2=\frac{1}{m}+\sigma\rho\Gamma, \  -\frac{1}{2}b_2=-\sigma\rho v, \  v_1=b_1=0,\\
& & \Gamma'=(\frac{1}{m}+\sigma\rho\Gamma)^2-\sigma^2\Gamma^2, \\
& & v_0=(\alpha-\frac{1}{2}\sigma^2)\Gamma+\frac{1}{m}\sigma\rho v+\sigma^2\Gamma v(\rho^2-1),\\
& & b_0=\sigma^2v^2(\rho^2-1)+v(2\alpha-\sigma^2)+\sigma^2\Gamma.\end{eqnarray*}
The equation $q(M,0)=p_0(M)=\frac{1}{\sqrt{2\pi S_0}}e^{-(M-M_0)^2/(2S_0)}$ gives $\Gamma_0=1/S_0$. Thus, by setting $S_t=1/\Gamma_t$, we obtain the Riccati equation (\ref{dSt}),  which has solution as shown in equation (\ref{S(t)}).\\
$\textit{Step 3}$: Kalman filter. By sep 2, we obtain \begin{eqnarray*}dv&=&\left[(\alpha-\frac{1}{2}\sigma^2)\Gamma+\frac{1}{m}\sigma\rho v+\sigma^2\Gamma v(\rho^2-1)\right]dt+\left(\frac{1}{m}+\sigma\rho\Gamma\right) d\tilde Z_t.
\end{eqnarray*}
Let $\hat M_t=v_tS_t$. Because $q(M,0)=p_0(M)$ implies $v_0=L_0/S_0,$ we obtain the initial condition $\hat M_0=M_0$. By Ito's lemma, \begin{equation*}d\hat M_t=\left(\alpha-\frac{1}{2}\sigma^2\right)dt+\frac{S_t}{m}\left(d\tilde Z_t-\frac{\hat M_t}{m}dt\right)+\sigma\rho \left (\tilde Z_t-\frac{\hat M_t}{m}dt\right).\end{equation*}
If we define an innovation process $d\tilde{\mathcal{B}}_t=d\tilde Z_t-\frac{\hat M_t}{m}dt, \tilde{\mathcal{B}}_0=0$, then \begin{eqnarray*}d\hat M_t=\left(\alpha-\frac{1}{2}\sigma^2\right)dt+\left(\frac{S_t}{m}+\sigma\rho\right) d\tilde{\mathcal{B}}_t.\end{eqnarray*}
$\textit{Step 4}$: Conditional Probability density. By step 2, \begin{eqnarray*}db_t=[\sigma^2v^2(\rho^2-1)+v(2\alpha-\sigma^2)+\sigma^2\Gamma]dt+2\sigma\rho v d\tilde Z_t.\end{eqnarray*}
Initial condition $q(M,0)=p_0(M)$ gives $e^{-b(0)/2}=\frac{1}{\sqrt{2\pi S_0}}e^{-L_0^2/(2S_0)}.$
Using $\Gamma_t=1/S_t$ and $v_t=\hat M_t/ S_t$, we rewrite $q(M,t)$ as  \begin{eqnarray}\label{partial_q(M,t)}
q(M,t)=\frac{K_t}{\sqrt{2\pi S_t}}\exp\left(-\frac{1}{2S_t}(M-\hat M_t)^2\right),\end{eqnarray}
where $K_t =\sqrt{2\pi S_t }e^{\frac{1}{2}(-b_t+\hat M^2_t/ S_t)}$. Notice that the initial condition  $e^{-b_0/2}=\frac{1}{\sqrt{2\pi S_0}} e^{-M_0^2/(2S_0)}$ and $\hat M_0=M_0$ yield $K_0=1$. By Ito's lemma, we get $d [\frac{1}{2}(-b_t+\hat M^2_t/ S_t)]=-\frac{\hat M_t^2}{m^2}dt+\frac{\hat M_t}{m} d\tilde Z_t.$
Thus $K_t$ satisfies \begin{eqnarray*}K_t=\exp\left(-\frac{1}{2}\int_0^t\frac{\hat M_s^2}{m^2}ds+\int_0^t\frac{\hat M_s}{m}d\tilde Z_s\right).\end{eqnarray*}
Then we complete the proof of Proposition  \ref{theorem of s(t)}. 

\section{Derivation  of  $I(S)$ and $\psi(x,y)$}\label{Appendix A1}
From Appendix \ref{appendix_theorem of s(t)} and (\ref{partial_q(M,t)}), 
$M_t$ follows a normal distribution with mean $\hat M_t$ and variance $S_t$. We know that $Y_t = e^{M_t}$ follows a lognormal distribution with mean $(\log \hat Y_t - S_t/2)$ and variance $S_t$. As the fraction of $E_t+D_t$ and $Y_t$ is constant when there are no controls, $E_t+D_t$ follows a lognormal distribution with mean $(\log (\hat E_t+D_t) - S_t/2)$ and variance $S_t$. Thus, 
{\fontsize{10pt}{10pt}
\begin{equation*}\begin{aligned}
\mathbb{P}(\hat E/D \leq \kappa ) = \mathbb{P} \left( \frac{\log(E_t+D_t) -(\log(\hat E_t+D_t) -S_t/2)}{\sqrt{S_t}} \leq \frac{\log((1+\kappa)D_t) - (\log(\hat E_t+D_t) -S_t/2)}{\sqrt{S_t}}\right).
\end{aligned}
\end{equation*}}
We can obtain 
$\frac{\log((1+\kappa)D_t) - (\log(\hat E_t+D_t) -S_t/2)}{\sqrt{S_t}}= \Phi^{-1}(a)$, 
which yields formula (\ref{hat tao}) with $I(S_t)=-1+(1+\kappa)e^{\frac{1}{2}S_t-\Phi^{-1}(a)\sqrt{S_t}}$.
Similarly, \begin{equation*}\begin{aligned}
\mathbb{E}[X_{\hat\tau^\pi}^+ | \mathcal{G}_{\hat\tau^\pi}]  = \mathbb{E} \left[ \left(\frac{\exp\{\log( \hat E_{\hat\tau^\pi} +D_{\hat\tau^\pi}) -\frac{1}{2}S_{\hat\tau^\pi}+ \sqrt{S_{\hat\tau^\pi}}Z\}}{D_{\hat\tau^\pi}} -1\right)^+ \Big | \mathcal{G}_{\hat\tau^\pi}\right],
\end{aligned}
\end{equation*}
where $Z$ is a standard normal random variable. This gives the formula for $\psi(x,y)$.

\section{Proof of Proposition \ref{partial_growth condition}}\label{appendix_linear growth condition}

Using the transformation $(\ref{v(x,S)})$, $\hat V(\hat X,S)$ as defined in (\ref{v(x,S)}) satisfies
\begin{eqnarray*}\label{partial_valuefunctionxs}
 \hat V (\hat X,S)=\sup_{\pi\in\Pi}\mathbb{E}\Big[ \int_0^{\hat \tau^{\pi}} e^{-\delta_1 u}dL^{\pi}_u
 -\sum_i e^{-\delta_1(t^{\pi}_i+\Delta)}
\left(s^{\pi}_i+K\right)  \mathbf{1}_{\{t^{\pi}_i+\Delta< \hat \tau^{\pi}\}} +e^{-\delta_1 \hat \tau^{\pi} }\omega X_{\hat \tau^\pi}^+ \Big],
\end{eqnarray*}
where $\delta_1=\delta-\mu$.
It suffices to prove that $\hat V(\hat X,S)$ grows linearly. Let us first introduce two lemmas.

\begin{lemma}\label{parital_v_delta}
Define
\begin{equation} \label{partial_V^0}
\hat V^0(\hat X,S): =  \sup_{\pi\in\Pi^0}\mathbb{E}\Big[ \int_0^{\hat \tau^{\pi}} e^{-\delta_1 u}dL^{\pi}_u  -\sum_i e^{-\delta_1 t^{\pi}_i}
\left(s^{\pi}_i+K\right)  \mathbf{1}_{\{t^{\pi}_i< \hat \tau^{\pi}\}} +e^{-\delta_1 \hat \tau^{\pi} }\omega X_{\hat \tau^\pi}^+ \Big]
\end{equation}
in $(\hat X,S)\in \Omega$, where the set of admissible strategy $\Pi^0$ is defined as
\begin{equation*}
\Pi^0: =\left\{ \pi=(L_t^\pi, (s_i^\pi,t_i^\pi)_i)\Bigg| 
\begin{aligned}& s_i^\pi \text{ is issued at } t_i^\pi \text{ with no time delay},\\
& \hat X_t^{\pi} \geq I(S_t)\text{ for } t\geq0, \text{ and } \sum_{i: t^{\pi}_i< \hat \tau^{\pi}}e^{-\delta_1 t_i^\pi} (s_i^\pi+K)<+\infty \end{aligned} \right\}.
 \end{equation*}
Then
$\hat V(\hat X,S)\leq \hat V^0(\hat X,S),\quad  \forall (\hat X,S)\in \Omega.$
\end{lemma}
\begin{proof}
For any admissible strategy $\pi=\{L_t^\pi, (s_i^\pi, t_i^\pi)_i\} \in \Pi$, we construct a new strategy $\tilde \pi$ as follows:
\begin{eqnarray*}\begin{aligned}
 L_t^{\tilde \pi} = L_t^\pi, \quad  t_i^{\tilde \pi} = t_i^{\pi} +\Delta, \quad s_i^{\tilde \pi} = s_i^{\pi}.
\end{aligned}\end{eqnarray*}
Apparently $\hat X_t^{(\hat X,S), \tilde \pi} = \hat X_t^{(\hat X,S),\pi}$,   $\hat \tau^{\tilde \pi} = \hat \tau^\pi$, and $\sum_{i}e^{-\delta_1t_i^{\tilde\pi}} (s_i^{\tilde\pi}+K)\mathbf{1}_{\{  t^{\tilde \pi}_i< \hat \tau^{\tilde \pi}\}}\leq\sum_{i}e^{-\delta_1t_i^{\tilde\pi}} (\bar{s}+K)<  +\infty$. Hence, $\tilde  \pi \in \Pi^0$, which leads to the desired result.
\end{proof}

\begin{lemma}\label{partial_v_v0}
Let $\varphi$ be a nonnegative $C^{2,1}$ supersolution to the HJB equation
\begin{eqnarray*}\min\left\{\left[\left(\left(S/m+\sigma\rho\right)^2-\sigma^2\right)\frac{\partial}{\partial S}-\mathcal{L}\right] \hat U^0, \quad  \frac{\partial }{\partial \hat X}\hat U^0-1, \quad  \hat U^0 - \mathcal{H}\hat U^0\right\} =0 \end{eqnarray*}
in $(\hat X,S)\in \Omega$, with boundary condition $\hat U^0(I(S),S) = \omega \psi(I(S),S)$, where $\mathcal{H}\hat U^0(\hat X,S) = \underset{s\in(0,\bar{s})}{\sup} \big\{ \hat U^0(\hat X+s,S) -s -K\big\}$.
Let $\hat V^0$ be as given in (\ref{partial_V^0}).
Then,
\begin{eqnarray*}
\hat V^0(\hat X,S) \leq \varphi(\hat X,S), \  \  \forall (\hat X,S) \in \Omega.
\end{eqnarray*}
\end{lemma}
\begin{proof}
For any admissible strategy $\pi \in \Pi^0$, set
$\hat \tau_n^\pi = \inf\{t\geq 0: \hat X_t^{\pi} < I(S_t) +1/n, \text{ or }  \hat X_t^{\pi} > n\}\wedge n$, $n\in \mathbb{N}$, and apply Ito's formula for the supersolution $\varphi(\hat X_t,S_t)$ between $0$ and $\hat \tau_n^\pi$. Then, taking expectation and noting that the integrand in the stochastic integral is bounded on $[0,\hat\tau_n^\pi]$, we get
\begin{eqnarray*}\begin{aligned}
&\mathbb{E} [e^{-\delta_1 \hat \tau_n^\pi}\varphi(\hat X_{\hat \tau_n^\pi},S_{\hat \tau_n^\pi}) ] \\
=&\varphi(\hat X,S)  + \mathbb{E} \bigg[ \int_0^{\hat \tau_n^\pi} e^{-\delta_1 t}    \left[\mathcal{L} - \left(\left(S/m+\sigma\rho\right)^2-\sigma^2\right)\frac{\partial}{\partial S}\right]  \varphi (\hat X_t, S_t)dt    \\
&-\int_0^{\hat \tau_n^\pi} e^{-\delta_1 t}\varphi_{\hat X}(\hat X_t,S_t) d(L_t^\pi)^c +\sum_{0\leq u\leq \hat \tau^\pi_n} e^{-\delta_1 u}[\varphi(\hat X_u,S_u) -\varphi(\hat X_{u-},S_{u-})]  \bigg],
\end{aligned}\end{eqnarray*}
where $(L_t^\pi)^c = L_t^\pi - \underset{0\leq u\leq  t }{\sum}(L_{u}^\pi - L_{u-}^\pi)$. Notice $\hat X_{u}-\hat X_{u-} = -(L_{u}^\pi - L_{u-}^\pi)$ if $u\neq t_i^\pi$ and $\hat X_{u}-\hat X_{u-} = s_i^\pi$ if $u=t_i^\pi$.
Since $\varphi_{\hat X} \geq 1$ and $\varphi(\hat X_{t_i^\pi-},S_{t_i^\pi-}) \geq \varphi(\hat X_{t_i^\pi},S_{t_i^\pi}) -s_i^\pi-K$ from the property of supersolution, by the mean value theorem, we have
\begin{eqnarray*}\begin{aligned}
&\varphi(\hat X_u,S_u) -\varphi(\hat X_{u-},S_{u-}) \leq -(L_{u}^\pi - L_{u-}^\pi), \text{ if } u\neq t_i^\pi; \\
&\varphi(\hat X_{u},S_{u}) - \varphi(\hat X_{u-},S_{u-}) \leq s_i^\pi +K, \text{ if } u =t_i^\pi.
\end{aligned}\end{eqnarray*}
Recall that $\left[\mathcal{L} - \left(\left(S/m+\sigma\rho\right)^2-\sigma^2\right)\frac{\partial}{\partial S}\right]  \varphi  \leq 0$. We then obtain
\begin{eqnarray*}\begin{aligned}
&\mathbb{E} [e^{-\delta_1 \hat \tau_n^\pi}\varphi(\hat X_{\hat \tau_n^\pi},S_{\hat \tau_n^\pi}) ] +\mathbb{E} \bigg[ \int_0^{\hat \tau_n^\pi} e^{-\delta_1 t}  dL_t^\pi \bigg]
 \leq \varphi(\hat X,S)  + \mathbb{E} \bigg[  \sum_{t_i^\pi \leq \hat \tau_n^\pi } e^{-\delta_1 t_i^\pi} (s_i^\pi+K) \bigg]\\
\leq & \varphi(\hat X,S)  + \mathbb{E} \bigg[  \sum_{t_i^\pi < \hat \tau^\pi } e^{-\delta_1 t_i^\pi} (s_i^\pi+K) \bigg] <+\infty.
\end{aligned}\end{eqnarray*}

Notice that $\hat\tau_n^\pi \rightarrow \hat\tau^\pi$. we have $\underset{n\rightarrow \infty}{\underline{\lim}}  e^{-\delta_1 \hat \tau_n^\pi}\varphi(\hat X_{\hat\tau_n^\pi},S_{\hat\tau_n^\pi}) \geq \omega e^{-\delta_1 \hat \tau^\pi}\psi(I(S_{\hat\tau^\pi}),S_{\hat\tau^\pi})$ if  $\hat \tau^\pi<+\infty$, and $\underset{n\rightarrow \infty}{\underline{\lim}}  e^{-\delta_1 \hat \tau_n^\pi}\varphi(\hat X_{\hat\tau_n^\pi},S_{\hat\tau_n^\pi}) \geq 0 =\omega e^{-\delta_1 \hat \tau^\pi}\psi(I(S_{\hat\tau^\pi}),S_{\hat\tau^\pi})$ If $\hat\tau^\pi=+\infty$  as $\psi(I(S_{\hat\tau^\pi}),S_{\hat\tau^\pi})<+\infty$.

Applying Fatou's Lemma and sending $n\rightarrow \infty$, we get
 \begin{eqnarray*}\begin{aligned}
 \varphi(\hat X,S)  \geq \mathbb{E} \bigg[ \int_0^{\hat \tau^\pi} e^{-\delta_1 t}  dL_t^\pi - \sum_{t_i^\pi < \hat \tau^\pi } e^{-\delta_1 t_i^\pi} (s_i^\pi+K) +  e^{-\delta_1 \hat \tau^\pi}\omega \psi(I(S_{\hat\tau^\pi}),S_{\hat\tau^\pi})\bigg ],
 \end{aligned}\end{eqnarray*}
which yields the desired result due to the arbitrariness of admissible strategy $\pi\in\Pi^0$.\\
\end{proof}

We are ready to prove the linear growth property of $\hat V$ using Lemma \ref{parital_v_delta} and \ref{partial_v_v0}. One the one hand,  shareholders can choose to pay $\hat X-I(S)$ amount of dividend at initial time and then the regulators liquidate the bank. Thus, we have $V(\hat X,S)\geq \hat X-I(S)+\omega\psi(I(S),S)$. Due to $\omega \psi(I(S),S)\geq 0$ and $I(S)\leq I(\bar S)\vee \kappa$, we can derive  $V(\hat X,S)\geq \hat X- C_0$ with any positive constant $C_0\geq I(\bar S)\vee \kappa$.
On the other hand, we consider a smooth function $\varphi(\hat X,S) = \hat X +C_1 +C_2 S$, where $C_1$ and $C_2$ are nonnegative constants to be determined later.
 In order to make $\varphi$ a supersolution to the HJB equation in Lemma \ref{partial_v_v0}, we need
\begin{eqnarray*}\begin{aligned}
&\left[\left(\left(S/m+\sigma\rho\right)^2-\sigma^2\right)\frac{\partial}{\partial S}-\mathcal{L}\right]  \varphi \\
&= \left(\left(S/m+\sigma\rho\right)^2-\sigma^2\right)C_2+(\delta-\mu)(\hat X+C_1+C_2 S) -(\alpha-\mu)(1+\hat X) \geq 0,\\
&  \frac{\partial }{\partial \hat X} \varphi-1 =0 \geq 0,\\
 & \mathcal{H} \varphi(\hat X,S) = \underset{s>0}{\sup} \big\{ \hat X+s+C_1+C_2 S -s -K\big\} \leq \hat X+C_1+C_2 S = \varphi(\hat X,S),\\
 & \varphi(I(S),S) = I(S) + C_1+C_2S \geq \omega \psi(I(S),S).
\end{aligned}\end{eqnarray*}

Noticing $\delta>\max\{\alpha,\mu\}$, we infer that $\varphi $ is indeed  a supersolution to the HJB equation provided that $C_1 \geq \frac{C_2\sigma^2 +\alpha-\mu -(\delta-\alpha)I(S)}{\delta-\mu}$ and $I(S) + C_1+C_2S \geq \omega \psi(I(S),S)$. Because $\psi(I(S),S)\leq I(S)+1$, $\omega \in [0,1]$,  and $S$ is bounded in $[0,\bar{S}]$,  we can find many such pairs of $(C_1,C_2)$.
By Lemma \ref{parital_v_delta} and Lemma  \ref{partial_v_v0},  the desired result follows.

\section{Proof of Weak Dynamic Programming}\label{appendix_proof_wdpp}
Based on Theorem 3.3 in Chapter 3 and Theorem 4.3 in Chapter 4 in  \cite{T12}, 
we give the proof of the weak dynamic programming principle in Proposition \ref{prop_wdpp}.

The weak DPP is trivial for the stopping time $\theta$ with value in $[\hat\tau^\pi, +\infty)$ by the definition of value function $\hat V$. Hence we only consider the stopping time $\theta$ with value in $\theta\in [0, \hat\tau^\pi)$.  
Define the objective function under strategy $\pi$ up to time $\theta$ as
\begin{eqnarray*} \begin{aligned}
&J_{\pi}(\hat X_\theta,S_\theta) \\
=& \mathbb{E}\Big[\int_\theta^{\hat \tau^{\pi}} e^{-\delta_1 (u-\theta)}dL^{\pi}_u
 -\sum_i e^{-\delta_1(t^{\pi}_i+\Delta-\theta)}
\left(s^{\pi}_i+K  \right)  \mathbf{1}_{\{\theta\leq t^{\pi}_i+\Delta< \hat \tau^{\pi}\}} +e^{-\delta_1 (\hat \tau^{\pi}-\theta) }\omega X_{\hat \tau^{\pi}}^+  \Big| \mathcal{G}_\theta\Big].
\end{aligned}\end{eqnarray*}
It is known that $J_{\pi}(\hat X_\theta,S_\theta)  \leq \hat V^{*}(\hat X_\theta,S_\theta)$.  By the conditional expectation,
{\fontsize{10.5pt}{10.5pt}
\begin{eqnarray*} \begin{aligned}
&\mathbb{E}\Big[ \int_0^{\hat \tau^{\pi}} e^{-\delta_1 u}dL^{\pi}_u
 -\sum_i e^{-\delta_1(t^{\pi}_i+\Delta)}
\left(s^{\pi}_i+K  \right)  \mathbf{1}_{\{ t^{ \pi}_i+\Delta< \hat \tau^{ \pi}\}} +e^{-\delta_1 \hat \tau^{\pi} }\omega X_{\hat \tau^{\pi}}^+ \Big] \\
=& \mathbb{E}\Bigg[\mathbb{E}\Big[ \int_0^{\hat \tau^{\pi}} e^{-\delta_1 u}dL^{\pi}_u
 -\sum_i e^{-\delta_1(t^{\pi}_i+\Delta)}
\left(s^{\pi}_i+K  \right)  \mathbf{1}_{\{ t^{ \pi}_i+\Delta< \hat \tau^{ \pi}\}} +e^{-\delta_1 \hat \tau^{\pi} }\omega X_{\hat \tau^{\pi}}^+ \Big| \mathcal{G}_\theta\Big] \Bigg]\\
=& \mathbb{E}\Bigg[ \int_0^{\theta} e^{-\delta_1 u}dL^{\pi}_u
 -\sum_i e^{-\delta_1(t^{\pi}_i+\Delta)}
\left(s^{\pi}_i+K  \right)  \mathbf{1}_{\{  t^{\pi}_i+\Delta< \theta\}}  \\
&+e^{-\delta_1\theta}\mathbb{E}\Big[\int_\theta^{\hat \tau^{\pi}} e^{-\delta_1 (u-\theta)}dL^{\pi}_u
 -\sum_i e^{-\delta_1(t^{\pi}_i+\Delta-\theta)}
\left(s^{\pi}_i+K  \right)  \mathbf{1}_{\{\theta\leq t^{\pi}_i+\Delta< \hat \tau^{\pi}\}} +e^{-\delta_1 (\hat \tau^{\pi}-\theta) }\omega X_{\hat \tau^{\pi}}^+   \Big| \mathcal{G}_\theta\Big]\Bigg] \\
=&\mathbb{E}\Bigg[ \int_0^{\theta} e^{-\delta_1 u}dL^{\pi}_u
 -\sum_i e^{-\delta_1(t^{\pi}_i+\Delta)}
\left(s^{\pi}_i+K  \right)  \mathbf{1}_{\{t^{\pi}_i+\Delta< \theta\}}  \\
&+e^{-\delta_1 (t^{\pi}_{k_{\theta}^{\pi}}+\Delta) }\mathbb{E}\Big[ J_\pi(\hat X_{t^{\pi}_{k_{\theta}^{\pi}}+\Delta} +s^{\pi}_{k_{\theta}^{\pi}},S_{t^{\pi}_{k_{\theta}^{\pi}}+\Delta}) -s^{\pi}_{k_{\theta}^{\pi}}-K\Big| \mathcal{G}_{\theta} \Big] \mathbf{1}_{\{\theta\geq t_{k_\theta^{\pi}}^{\pi},  t_{k_\theta^{\pi}}^{\pi} +\Delta < \hat\tau^\pi\}}\\
& + e^{-\delta_1\hat\tau^\pi}\mathbb{E}\Big[ \omega X_{\hat\tau^\pi}^+  \Big| \mathcal{G}_{\theta} \Big] \mathbf{1}_{\{\theta\geq t_{k_\theta^{\pi}}^{\pi},  t_{k_\theta^{\pi}}^{\pi} +\Delta \geq \hat\tau^\pi\}} + e^{-\delta_1\theta}J_{\pi}(\hat X_\theta,S_\theta) \mathbf{1}_{\{\theta<t_{k_\theta^{\pi}}^{\pi} \}}  \Bigg]\\
\leq & \mathbb{E}\Bigg[ \int_0^{\theta} e^{-\delta_1 u}dL^{\pi}_u
 -\sum_i e^{-\delta_1(t^{\pi}_i+\Delta)}
\left(s^{\pi}_i+K  \right)  \mathbf{1}_{\{ t^{\pi}_i+\Delta< \theta\}}  \\
&+e^{-\delta_1 (t^{\pi}_{k_{\theta}^{\pi}}+\Delta) }\mathbb{E}\Big[ \hat V^{*}(\hat X_{t^{\pi}_{k_{\theta}^{\pi}}+\Delta} +s^{\pi}_{k_{\theta}^{\pi}},S_{t^{\pi}_{k_{\theta}^{\pi}}+\Delta}) -s^{\pi}_{k_{\theta}^{\pi}}-K\Big| \mathcal{G}_{\theta} \Big] \mathbf{1}_{\{\theta\geq t_{k_\theta^{\pi}}^{\pi},  t_{k_\theta^{\pi}}^{\pi} +\Delta < \hat\tau^\pi \}} \\
&  + e^{-\delta_1\hat\tau^\pi}\mathbb{E}\Big[ \omega X_{\hat\tau^\pi}^+  \Big| \mathcal{G}_{\theta} \Big] \mathbf{1}_{\{\theta\geq t_{k_\theta^{\pi}}^{\pi},  t_{k_\theta^{\pi}}^{\pi} +\Delta \geq \hat\tau^\pi\}} + e^{-\delta_1\theta}\hat V^{*}(\hat X_\theta,S_\theta) \mathbf{1}_{\{\theta<t_{k_\theta^{\pi}}^{\pi} \}}  \Bigg].
\end{aligned}\end{eqnarray*}}
We obtain the first inequality in Proposition \ref{prop_wdpp} by the arbitrariness of admissible strategy $\pi\in \Pi$.

Next we prove the second inequality in Proposition \ref{prop_wdpp}. We can substitute $\hat V_{*}$ by an arbitrary function
$$\varphi: \Omega \rightarrow \mathbb{R}, \text{ such that  } \varphi \text{  is upper-semicontinuous and } \hat V_{*}\geq \varphi.$$
Similar to the method used in Theorem 3.3 in Chapter 3 in  \cite{T12}, we can find a countable sequence $(\hat X_i, S_i, r_i)_{i\geq1}$ such that $\Omega \subseteq \underset{i\geq 1}{\cup}B(\hat X_i, S_i; r_i)$, where $B(\hat X_i, S_i; r_i): =\Big\{(\hat X, S) \in \Omega:  |\hat X-\hat X_i|<\varepsilon, |S-S_i|<r_i \Big\}$.
Define
$$A_{i+1} = B(\hat X_{i+1}, S_{i+1}; r_{i+1})\backslash C_i, \quad C_i = C_{i-1}\cup A_i, \quad i=0,1,2,\dots,$$
where $A_0: =\emptyset$ and $C_{-1}: =\emptyset$.
Then $A_i\cap A_j = \phi$ for $i\neq j$ and $\Omega= \underset{i\geq 0}{\cup}A_i$.  By the lower semicontinuity of $(J_{\pi^{i,\varepsilon}})_{*}$ and the upper semicontinuity of $\varphi$, we can also find a sequence of strategy $\pi^{i,\varepsilon} \in \Pi((\hat X_i,S_i))$ for $i\geq 1$
 such that $(J_{\pi^{i,\varepsilon}})_{*} \geq \varphi - 3\varepsilon$ on $A_i$.
 Now set $A^n = \underset{i\leq n}{\cup} A_i$ for $n\geq 1$. For any admissible strategy $\pi\in\Pi$ and stopping time  $\theta$ with value in $[0, \hat\tau^\pi)$, we define a sequence of admissible strategy for $s\in[0,\hat\tau^\pi]$:
\begin{eqnarray*} \begin{aligned}
\pi_s^{\varepsilon,n}: =& \pi_s \mathbf{1}_{\{ s\leq (t^\pi_{k_\theta^{\pi}}+\Delta)\wedge \hat\tau^\pi,  t^\pi_{k_\theta^{\pi}} \leq \theta\}} +  \pi_s \mathbf{1}_{\{ s\leq \theta,  t^\pi_{k_\theta^{\pi}} >\theta\}}  \\
&+ \mathbf{1}_{\{(t^\pi_{k_\theta^{\pi}}+\Delta)\wedge \hat\tau^\pi<s\leq\hat\tau^\pi,   t^\pi_{k_\theta^{\pi}}  \leq \theta\}} \Big(\pi_s \mathbf{1}_{\{(\hat X_\theta^{\pi}, S_\theta)\in (A^n)^c\}} + \sum_{i=1}^n \pi_s^{i,\varepsilon} \mathbf{1}_{\{(\hat X_\theta^{\pi}, S_\theta)\in A_i\}}  \Big) \\
&+ \mathbf{1}_{\{ \theta<s\leq\hat\tau^\pi,   t^\pi_{k_\theta^{\pi}} >\theta \}} \Big(\pi_s \mathbf{1}_{\{(\hat X_\theta^{\pi}, S_\theta)\in (A^n)^c\}} + \sum_{i=1}^n \pi_s^{i,\varepsilon} \mathbf{1}_{\{(\hat X_\theta^{\pi}, S_\theta)\in A_i\}}  \Big), \  \ n=1,2,\dots.
\end{aligned}\end{eqnarray*}
Then we have $t^{\pi^{\varepsilon,n}}_{k_\theta^{\pi^{\varepsilon,n}}}= t^\pi_{k_\theta^{\pi}}$ and $\pi^{\varepsilon,n}_s=\pi_s$ for all $s\leq  (t^\pi_{k_\theta^{\pi}}+\Delta)\wedge \hat\tau^\pi$ when $ t^\pi_{k_\theta^{\pi}}  \leq \theta $. By Fatou's Lemma, it follows
{\fontsize{10.5pt}{10.5pt}
\begin{eqnarray*} \begin{aligned}
&\hat V (\hat X,S) \geq J_{\pi^{\varepsilon,n}} (\hat X, S) \geq  (J_{\pi^{\varepsilon,n}} (\hat X, S))_{*} \\
=& \underset{(\hat X',S')\rightarrow (\hat X,S)}{\underline{\lim}} \mathbb{E}^{\hat X',S'}\Bigg[ \int_0^{\theta} e^{-\delta_1 u}dL^{\pi^{\varepsilon,n}}_u
 -\sum_i e^{-\delta_1(t^{\pi^{\varepsilon,n}}_i+\Delta)}
\left(s^{\pi^{\varepsilon,n}}_i+K  \right)  \mathbf{1}_{\{  t^{\pi^{\varepsilon,n}}_i+\Delta< \theta\}}  \\
& +e^{-\delta_1\theta}\mathbb{E}\Big[ \int_\theta^{\hat \tau^{\pi^{\varepsilon,n}}} e^{-\delta_1(u-\theta)}dL^{\pi^{\varepsilon,n}}_u
 -\sum_i e^{-\delta_1(t^{\pi^{\varepsilon,n}}_i+\Delta-\theta)}
\left(s^{\pi^{\varepsilon,n}}_i+K \right)  \mathbf{1}_{\{\theta \leq t^{ \pi^{\varepsilon,n}}_i+\Delta< \hat \tau^{ \pi^{\varepsilon,n}}\}} \\
&+e^{-\delta_1 (\hat \tau^{\pi^{\varepsilon,n}}-\theta) }\omega X_{\hat \tau^{\pi^{\varepsilon,n}}}^+  \Big| \mathcal{G}_\theta\Big]   \Bigg] \\
\geq &  \mathbb{E}\Bigg[ \int_0^{\theta} e^{-\delta_1 u }dL^{\pi}_u
 -\sum_i e^{-\delta_1(t^{\pi}_i+\Delta)}
\left(s^{\pi}_i+K\right)  \mathbf{1}_{\{t^{\pi}_i+\Delta< \theta\}} \\
&+e^{-\delta_1 (t^{\pi^{\varepsilon,n}}_{k_{\theta}^{\pi^{\varepsilon,n}}}+\Delta) }\mathbb{E}\Big[ (J_{\pi^{\varepsilon,n}})_{*}(\hat X_{t^{\pi^{\varepsilon,n}}_{k_{\theta}^{\pi^{\varepsilon,n}}}+\Delta} +s^{\pi^{\varepsilon,n}}_{k_{\theta}^{\pi^{\varepsilon,n}}},S_{t^{\pi^{\varepsilon,n}}_{k_{\theta}^{\pi^{\varepsilon,n}}}+\Delta}) -s^{\pi^{\varepsilon,n}}_{k_{\theta}^{\pi^{\varepsilon,n}}}-K\Big| \mathcal{G}_{\theta} \Big] \mathbf{1}_{\{\theta\geq t_{k_\theta^{\pi^{\varepsilon,n}}}^{\pi^{\varepsilon,n}}, t_{k_\theta^{\pi^{\varepsilon,n}}}^{\pi^{\varepsilon,n}}+\Delta < \hat \tau^{\pi^{\varepsilon,n}} \}}  \\
& +e^{-\delta_1 \hat \tau^{\pi^{\varepsilon,n}} }\mathbb{E}\Big[ \omega X_{\hat \tau^{\pi^{\varepsilon,n}}}^+ \Big| \mathcal{G}_{\theta} \Big]   \mathbf{1}_{\{\theta\geq t_{k_\theta^{\pi^{\varepsilon,n}}}^{\pi^{\varepsilon,n}}, t_{k_\theta^{\pi^{\varepsilon,n}}}^{\pi^{\varepsilon,n}}+\Delta \geq \hat \tau^{\pi^{\varepsilon,n}} \}}  + e^{-\delta_1 \theta }(J_{\pi^{\varepsilon,n}})_{*}(\hat X_\theta,S_\theta) \mathbf{1}_{\{\theta<t_{k_\theta^{\pi}}^{\pi} \}}  \Bigg]\\
\geq & \mathbb{E}\Bigg[ \int_0^{\theta} e^{-\delta_1 u}dL^{\pi}_u
 -\sum_i e^{-\delta_1(t^{\pi}_i+\Delta)}
\left(s^{\pi}_i+K \right)  \mathbf{1}_{\{t^{\pi}_i+\Delta< \theta\}} \\
& +e^{-\delta_1 (t^{\pi}_{k_{\theta}^{\pi}}+\Delta ) }\mathbb{E}\Big[ (J_{\pi^{\varepsilon,n}})_{*} (\hat X_{t^{\pi}_{k_{\theta}^{\pi}}+\Delta} +s^{\pi}_{k_{\theta}^{\pi}},S_{t^{\pi}_{k_{\theta}^{\pi}}+\Delta}) -s^{\pi}_{k_{\theta}^{\pi}}-K\Big| \mathcal{G}_{\theta} \Big] \mathbf{1}_{\{\theta\geq t_{k_\theta^{\pi}}^{\pi}, t_{k_\theta^{\pi}}^{\pi}+\Delta < \hat\tau^\pi , (\hat X_\theta^{\pi}, S_\theta)\in A^n \}} \\
&+e^{-\delta_1\hat \tau^{\pi} }\mathbb{E}\Big[ \omega X_{\hat \tau^{\pi}}^+ \Big| \mathcal{G}_{\theta} \Big] \mathbf{1}_{\{\theta\geq t_{k_\theta^{\pi}}^{\pi}, t_{k_\theta^{\pi}}^{\pi}+\Delta \geq \hat\tau^\pi\}} + e^{-\delta_1\theta }(J_{\pi^{\varepsilon,n}})_{*}(\hat X_\theta,S_\theta) \mathbf{1}_{\{\theta<t_{k_\theta^{\pi}}^{\pi}, (\hat X_\theta^{\pi}, S_\theta)\in A^n\}} \Bigg].
\end{aligned}\end{eqnarray*}}
Remember that $(J_{\pi^{i,\varepsilon}})_{*} \geq \varphi - 3\varepsilon$ on each $A_i$ for $i\geq 1$.  Then we have
\begin{eqnarray*} \begin{aligned}
&\hat V (\hat X,S)
\geq  \mathbb{E}\Bigg[ \int_0^{\theta} e^{-\delta_1 u}dL^{\pi}_u
 -\sum_i e^{-\delta_1(t^{\pi}_i+\Delta)}
\left(s^{\pi}_i+K \right)  \mathbf{1}_{\{ t^{\pi}_i+\Delta< \theta\}} \\
& +e^{-\delta_1 (t^{\pi}_{k_{\theta}^{\pi}}+\Delta) }\mathbb{E}\Big[ \varphi (\hat X_{t^{\pi}_{k_{\theta}^{\pi}}+\Delta} +s^{\pi}_{k_{\theta}^{\pi}},S_{t^{\pi}_{k_{\theta}^{\pi}}+\Delta}) -s^{\pi}_{k_{\theta}^{\pi}}-K\Big| \mathcal{G}_{\theta} \Big]  \mathbf{1}_{\{\theta\geq t_{k_\theta^{\pi}}^{\pi}, t_{k_\theta^{\pi}}^{\pi}+\Delta < \hat\tau^\pi, (\hat X_\theta^{\pi},   S_\theta)\in A^n \}}\\
& +e^{-\delta_1 \hat \tau^{\pi} }\mathbb{E}\Big[ \omega X_{\hat \tau^{\pi}}^+ \Big| \mathcal{G}_{\theta} \Big] \mathbf{1}_{\{\theta\geq t_{k_\theta^{\pi}}^{\pi}, t_{k_\theta^{\pi}}^{\pi}+\Delta \geq \hat\tau^\pi \}}+ e^{-\delta_1\theta }\varphi(\hat X_\theta, S_\theta) \mathbf{1}_{\{\theta<t_{k_\theta^{\pi}}^{\pi}, (\hat X_\theta^{\pi}, S_\theta)\in A^n\}}   \Bigg] - 3\varepsilon.
\end{aligned}\end{eqnarray*}
Sending $n\rightarrow \infty$ and by the arbitrariness of $\varepsilon>0$, we obtain
\begin{eqnarray*} \begin{aligned}
&\hat V (\hat X,S)
\geq    \mathbb{E}\Bigg[  \int_0^{\theta} e^{-\delta_1 u }dL^{\pi}_u
 -\sum_i e^{-\delta_1(t^{\pi}_i+\Delta)}
\left(s^{\pi}_i+K \right)  \mathbf{1}_{\{  t^{\pi}_i+\Delta< \theta\}} \\
& +e^{-\delta_1 (t^{\pi}_{k_{\theta}^{\pi}}+\Delta)} \mathbb{E}\Big[ \varphi(\hat X_{t^{\pi}_{k_{\theta}^{\pi}}+\Delta} +s^{\pi}_{k_{\theta}^{\pi}},S_{t^{\pi}_{k_{\theta}^{\pi}}+\Delta}) -s^{\pi}_{k_{\theta}^{\pi}}-K\Big| \mathcal{G}_{\theta} \Big] \mathbf{1}_{\{\theta\geq t_{k_\theta^{\pi}}^{\pi},t_{k_\theta^{\pi}}^{\pi}+\Delta < \hat\tau^\pi \}} \\
& +e^{-\delta_1\hat \tau^{\pi} }\mathbb{E}\Big[ \omega X_{\hat \tau^{\pi}}^+ \Big| \mathcal{G}_{\theta} \Big] \mathbf{1}_{\{\theta\geq t_{k_\theta^{\pi}}^{\pi}, t_{k_\theta^{\pi}}^{\pi}+\Delta \geq \hat\tau^\pi \}}+ e^{-\delta_1 \theta } \varphi (\hat X_\theta,S_\theta)\mathbf{1}_{\{\theta<t_{k_\theta^{\pi}}^{\pi} \}}
 \Bigg].
\end{aligned}\end{eqnarray*}

We can find a sequence $\{\varphi_n\}_n$ such that $\varphi_n \leq \hat V_{*} \leq \hat V$ and $\varphi_n\rightarrow \hat V_{*}$ pointwise. Define $\phi_N: = \underset{n\geq N}{\min}\varphi_n$. Then $\phi_{N}$ is non-decreasing and converges to $\hat V_{*}$ pointwise on $\Omega$. By the monotone convergence theorem,
\begin{eqnarray*} \begin{aligned}
&\hat V (\hat X,S,t) \geq \lim_{N\rightarrow \infty}  \mathbb{E}\Bigg[ \int_0^{\theta} e^{-\delta_1 u }dL^{\pi}_u
 -\sum_i e^{-\delta_1(t^{\pi}_i+\Delta)}
\left(s^{\pi}_i+K \right)  \mathbf{1}_{\{  t^{\pi}_i+\Delta< \theta\}} \\
& +e^{-\delta_1 (t^{\pi}_{k_{\theta}^{\pi}}+\Delta)} \mathbb{E}\Big[\phi_N(\hat X_{t^{\pi}_{k_{\theta}^{\pi}}+\Delta} +s^{\pi}_{k_{\theta}^{\pi}},S_{t^{\pi}_{k_{\theta}^{\pi}}+\Delta}) -s^{\pi}_{k_{\theta}^{\pi}}-K\Big| \mathcal{G}_{\theta} \Big] \mathbf{1}_{\{\theta\geq t_{k_\theta^{\pi}}^{\pi},t_{k_\theta^{\pi}}^{\pi}+\Delta < \hat\tau^\pi \}} \\
& +e^{-\delta_1\hat \tau^{\pi} }\mathbb{E}\Big[ \omega X_{\hat \tau^{\pi}}^+ \Big| \mathcal{G}_{\theta} \Big] \mathbf{1}_{\{\theta\geq t_{k_\theta^{\pi}}^{\pi}, t_{k_\theta^{\pi}}^{\pi}+\Delta \geq \hat\tau^\pi \}}+ e^{-\delta_1 \theta } \phi_N (\hat X_\theta,S_\theta)\mathbf{1}_{\{\theta<t_{k_\theta^{\pi}}^{\pi} \}}
 \Bigg]\\
 =& \mathbb{E}\Bigg[ \int_0^{\theta} e^{-\delta_1 u }dL^{\pi}_u
 -\sum_i e^{-\delta_1(t^{\pi}_i+\Delta)}
\left(s^{\pi}_i+K \right)  \mathbf{1}_{\{  t^{\pi}_i+\Delta< \theta\}} \\
& +e^{-\delta_1 (t^{\pi}_{k_{\theta}^{\pi}}+\Delta)} \mathbb{E}\Big[ \hat V_*(\hat X_{t^{\pi}_{k_{\theta}^{\pi}}+\Delta} +s^{\pi}_{k_{\theta}^{\pi}},S_{t^{\pi}_{k_{\theta}^{\pi}}+\Delta}) -s^{\pi}_{k_{\theta}^{\pi}}-K\Big| \mathcal{G}_{\theta} \Big] \mathbf{1}_{\{\theta\geq t_{k_\theta^{\pi}}^{\pi},t_{k_\theta^{\pi}}^{\pi}+\Delta < \hat\tau^\pi \}} \\
& +e^{-\delta_1\hat \tau^{\pi} }\mathbb{E}\Big[ \omega X_{\hat \tau^{\pi}}^+ \Big| \mathcal{G}_{\theta} \Big] \mathbf{1}_{\{\theta\geq t_{k_\theta^{\pi}}^{\pi}, t_{k_\theta^{\pi}}^{\pi}+\Delta \geq \hat\tau^\pi \}}+ e^{-\delta_1 \theta } \hat V_* (\hat X_\theta,S_\theta)\mathbf{1}_{\{\theta<t_{k_\theta^{\pi}}^{\pi} \}}
 \Bigg].
 \end{aligned}\end{eqnarray*}
By the arbitrariness of admissible strategy $\pi\in \Pi$, we obtain  the second inequality in Proposition \ref{prop_wdpp}.
Therefore, the weak DPP holds for value function  $\hat V(\hat X, S)$.

\backmatter
\bibliographystyle{note}
\bibliography{note}

\end{document}